\documentclass[12pt]{amsart}
\usepackage{amsmath,amsthm,amsfonts,amssymb,eucal, hyperref}
\usepackage{xcolor}
\usepackage{graphicx}

\newcommand{\newred}{}
\renewcommand {\a}{ \alpha }
\renewcommand{\b}{\beta}

\newcommand{\ga}{\gamma}

\renewcommand{\k}{\kappa}
\newcommand{\ka}{\kappa}

\newcommand{\p}{\partial}
\newcommand{\om}{\omega}
\newcommand{\Om}{\Omega}

\newcommand{\oq}{\ {\raise 7pt\hbox{${\scriptstyle\circ}$}}
\kern -7pt{
\hbox{$Q$}}}

\newcommand{\R}{ \mathbb R}
\newcommand{\Q}{ \mathbb Q}
\newcommand{\C}{ \mathbb C}
\renewcommand{\S}{ \mathbb S^{d-1}}
\newcommand{\N}{ \mathbb N}
\newcommand{\E}{\mathbb E}

\newcommand{\Zl}{ \mathbb Z^l}

\newcommand {\GH}{\mathfrak H}

\newcommand {\GV}{\mathfrak V}

\newcommand {\GU}{\mathfrak U}

\newcommand {\ba}{\mathbf a}

\newcommand {\BB}{\mathbf B}

\newcommand{\BN}{\mathbf N}

\newcommand {\BL}{\mathbf L}
\newcommand {\BS}{\mathbf S}

\newcommand {\bx}{\mathbf x}

\newcommand {\be}{\mathbf e}

\newcommand {\bk}{\mathbf k}
\newcommand {\bp}{\mathbf p}
\newcommand {\bq}{\mathbf q}
\newcommand {\bm}{\mathbf m}

\newcommand {\bz}{\mathbf z}
\newcommand {\by}{\mathbf y}
\newcommand {\bt}{\mathbf t}

\newcommand {\bs}{\mathbf s}
\newcommand {\bu}{\mathbf u}
\newcommand {\bv}{\mathbf v}
\newcommand {\bn}{\mathbf n}

\newcommand {\bnu}{\boldsymbol\nu}

\newcommand {\bmu}{\boldsymbol\mu}
\newcommand {\bka}{\boldsymbol\kappa}
\newcommand {\bth}{\boldsymbol\theta}

\newcommand {\boldeta}{\boldsymbol\eta}
\newcommand {\boldom}{\boldsymbol\om}

\newcommand {\BPhi}{\boldsymbol\Phi}
\newcommand {\bxi}{\boldsymbol\xi}

\newcommand {\BPsi}{\boldsymbol\Psi}

\newcommand{\lu}{\langle}
\newcommand{\ru}{\rangle}

\newcommand{\CN}{\mathcal N}

\newcommand{\CG}{\mathcal G}
\newcommand{\CR}{\mathcal R}

\newcommand{\CB}{\mathcal B}
\newcommand{\CL}{\mathcal L}

\newcommand{\CP}{\mathcal P}

\newcommand{\CA}{\mathcal A}

\newcommand{\CM}{\mathcal M}

\newcommand{\CC}{\mathcal C}

\newcommand{\CE}{\mathcal E}
\newcommand{\CD}{\mathcal D}

\newcommand{\1}
{{\,\vrule depth3pt height9pt}{\vrule depth3pt height9pt}
{\vrule depth3pt height9pt}{\vrule depth3pt height9pt}\,}

\DeclareMathOperator {\diam} {{diam}}

\theoremstyle{plain}
\newtheorem{thm}{Theorem}[section]
\newtheorem{theorem}{Theorem}[section]
\newtheorem{cor}[thm]{Corollary}
\newtheorem{corollary}[thm]{Corollary}
\newtheorem{cla}[thm]{Claim}
\newtheorem{lem}[thm]{Lemma}

\newtheorem{defn}[thm]{Definition}
\newtheorem{defin}[thm]{Definition}

\newtheorem{rem}[thm]{Remark}

\numberwithin{equation}{section}

\newcommand{\bee}{\begin{equation}}
\newcommand{\ene}{\end{equation}}
\newcommand{\bees}{\begin{equation*}}
\newcommand{\enes}{\end{equation*}}
\newcommand{\bes}{\begin{split}}
\newcommand{\ens}{\end{split}}

\newcommand{\bet}{\begin{thm}}
\newcommand{\ent}{\end{thm}}
\newcommand{\bel}{\begin{lem}}
\newcommand{\enl}{\end{lem}}
\newcommand{\bec}{\begin{cor}}
\newcommand{\enc}{\end{cor}}
\newcommand{\becl}{\begin{cla}}
\newcommand{\encl}{\end{cla}}
\newcommand{\bep}{\begin{proof}}
\newcommand{\enp}{\end{proof}}
\newcommand{\ber}{\begin{rem}}
\newcommand{\enr}{\end{rem}}
\newcommand{\ep}{\varepsilon}
\newcommand{\la}{\lambda}
\newcommand{\La}{\Lambda}
\newcommand{\de}{\delta}

\newcommand{\al}{\alpha}
\newcommand{\Z}{\mathbb Z}
\newcommand{\Ga}{\Gamma}

\newcommand {\BUps}{\boldsymbol\Upsilon}

\makeatletter
\def\square{\RIfM@\bgroup\else$\bgroup\aftergroup$\fi
  \vcenter{\hrule\hbox{\vrule\@height.6em\kern.6em\vrule}\hrule}\egroup}
\makeatother

\numberwithin{equation}{section}

\begin{document}

\hoffset -4pc

\title[Bethe-Sommerfeld Absolute Continuity]
{Bethe-Sommerfeld Conjecture and Absolutely Continuous Spectrum of Multi-dimensional Quasi-periodic Schr\"odinger Operators}
\author[YK,LP,RS]
{Yu. Karpeshina, L. Parnovski, R. Shterenberg}

\numberwithin{equation}{section}


\date{\today}


\begin{abstract}
We consider Schr\"odinger operators $H=-\Delta+V(\bx)$ in
$\R^d$, $d\geq2$, with quasi-periodic potentials $V(\bx)$. We prove that
the absolutely continuous spectrum of a generic $H$ contains a semi-axis $[\lambda_*,+\infty)$. 
We also construct a family of eigenfunctions of the absolutely continuous spectrum; these eigenfunctions are small perturbations of the exponentials. 
 The proof is based on a version  of the multi-scale analysis in the momentum
space with several new ideas introduced along the way. 
\end{abstract}



\maketitle
\vskip 0.5cm

\newcommand{\meas}{{\mathrm {meas}}}
\newcommand{\conf}{{\mathrm {\bk}}}
\newcommand{\corner}{\phi}
\newcommand{\angl}{{\BPhi}}
\newcommand{\sph}{\BPhi}
\newcommand{\nk}{{k}}
\newcommand{\nbka}{{\bka}}
\newcommand{\nka}{{\kappa}}
\newcommand{\freq}{{\mathrm {freq}}}
\newcommand{\rank}{{\mathrm {Rank}}}
\newcommand{\spa}{{\mathrm {span}}}
\newcommand{\ext}{{\mathrm {ext}}}
\newcommand{\patch}{{\mathrm {res}}}
\newcommand{\m}{{\mathbf {\bn}}}
\newcommand{\n}{{\mathbf {n}}}
\newcommand{\q}{{\mathbf {q}}}
\renewcommand{\th}{\theta}
\newcommand{\OO}{{\mathcal O}}
\newcommand{\MM}{{\mathcal M}}
\newcommand{\no}{{|}}
\renewcommand{\p}{{\vec p}}
\newcommand{\cs}{\gamma_0}
\newcommand{\oldj}{p}
\newcommand{\cont}{{\bf {K}}}
\newcommand{\En}{E}
\newcommand{\cube}{[-1/2,1/2]}
\newcommand{\eig}{\lambda^{(\infty)}}
\newcommand{\bkappa}{{\bka}}
\newcommand{\CMM}{\CR^{\bm}}
\newcommand{\local}{\mathrm {local}}
\newcommand{\Czero}{{C_{6.14}}}
\newcommand{\Cone}{{C_{6.11}}}
\newcommand{\Ctwo}{{C_{6.21}}}
\newcommand{\Cthree}{{c_{6.18}}}
\newcommand{\Cfour}{{C_{6.18}}}
\newcommand{\Cfive}{{C_{6.47}}}
\newcommand{\Csix}{{C'_{6.47}}}
\newcommand{\tim}{\tilde m}
\newcommand{\ham}{\hat m}
\newcommand {\BOm}{\boldsymbol\Omega}
\newcommand{\Bourgain}{{multiscale }}
\newcommand{\CNN}{\mathcal B}
\newcommand{\tildeB}{L}
\newcommand{\dis}{e}
\newcommand{\sm}{{\mathrm {small}}}
\newcommand{\bi}{{\mathrm {big}}}
\renewcommand{\Re}{{\mathrm {Re \ }}}
\renewcommand{\Im}{{\mathrm {Im \ }}}
\newcommand{\Piy}{{\boldsymbol\Xi}}
\newcommand{\sh}{{\mathrm{sh}}}
\newcommand{\Cun}{\check Z}
\newcommand{\cun}{\check z}
\newcommand{\fun}{U}

\section{Introduction}\label{section0}
We study multidimensional Schr\"odinger operators acting on $L^2(\R^d)$, $d\ge 2$, defined in the following way. 
Let $\boldom_1,\dots,\boldom_l\in\R^d$, $l>d$, be a collection of vectors that we will call {\it the basic frequencies}. It will be convenient to form a `vector' out of the basic frequencies: 
$\vec\boldom:=(\boldom_1,\dots,\boldom_l)$. 
We consider the operator 
\bee\label{H}
H:=H_0+V,
\ene
where
\bee
H_0:=-\Delta
\ene
and $V$ is a real-valued potential of the form
\bee\label{V_q=0}
V:=\sum_{|\bn|\le Q} V_{\bn}\be_{\bn\vec\boldom},\ \ \ V_{\bn}=\bar V_{-\bn}.
\ene
The last sum is finite and taken over all vectors $\bn=(n_1,\dots,n_l)\in\Z^l$ with 
\bee\label{l_infinity}
|\bn|:=\max_{j=1,...,l}|n_j|<Q, \ \ \ Q\in\N.
\ene
We have also denoted 
\bee\label{be}
\be_{\bth}(\bx):=e^{i\lu\bth,\bx\ru}, \ \ \bth,\bx\in\R^d 
\ene
and
\bee
\bn\vec\boldom:=\sum_{j=1}^l n_j\boldom_j\in\R^d; 
\ene
these vectors  $\bn\vec\boldom$ are called {\it the frequencies}. 
For convenience and without loss of generality, we assume that the basic frequencies 
$\boldom_j\in \cube^d$ and thus $\vec\boldom\in\cube^{dl}$ (so that the Lebesgue measure of this set is one; obviously, we can always achieve this by rescaling if necessary) and they are linearly independent over rationals.
Our main result is the following theorem.

\begin{theorem}\label{mainth} For any finite set $\{V_\bn\}$, $V_{\bn}=\bar V_{-\bn}$ ($\bn\in\Z^l$, $|\bn|<Q$) of Fourier coefficients 
there exists a subset $\BOm_*=\BOm_*(\{V_\bn\})\subset[-1/2,1/2]^{dl}$ of basic frequencies with $meas(\BOm_*)=1$ such that for any $\vec\boldom\in\BOm_*$ the absolute continuous spectrum of the operator $H$ contains a semi-axis $[\lambda_*,\infty)$, where   $\lambda_*=\lambda_*(\vec\boldom,\{V_\bn\})$ is sufficiently large.
\end{theorem}

The one-dimensional situation 
is thoroughly
   investigated in both the discrete and continuous settings,
   see e.g. \cite{2a}-\cite{6a}, \cite{9a}-\cite{4a}, \cite{5a}, \cite{7a}, \cite{3b}, \cite{10a}. In particular, in the one-dimensional continuous case, the spectrum is pure absolutely continuous and generically
Cantor at high energies, \cite{3}. In other situations, the spectrum of one-dimensional quasi-periodic operators can be of any nature (absolutely continuous, singular continuous, or pure point), and a transition between different types  of spectrum can happen even with a small change of coefficients (see e.g. \cite{Avila1},  \cite{6a}, \cite{9a}). The multidimensional case is much less studied, some important results being  \cite{BLS}--\cite{ChD}, \cite{FP}, \cite{GSV2}, \cite{JiLiSh}, \cite{PSh}, \cite{Sh1}-\cite{Shu}; see also recent papers \cite{WL} and  \cite{Shi}. 


 If $l=d$, which (generically) means that $V$ is periodic, the spectrum is known to be purely absolutely continuous everywhere    \cite{Th}. 
 Moreover, the Bethe-Sommerfeld conjecture  states that for $d\geq 2$, the spectrum of \eqref{H} contains a semi-axis. When this is the case (e.i. when the spectrum of an operator $H$ contains a semi-axis $[\la_0,+\infty)$), we will say that $H$ satisfies the Bethe-Sommerfeld property. 
  A variety of proofs of this property in the periodic case   have been developed over decades. For the most general and recent results see \cite{Par}, \cite{PS}. 
    For a limit-periodic potential that is periodic in one direction,  this property is established  for  $2\leq d \leq 4$ in \cite{SS}. For the general case
    of limit-periodic potentials  this property is established for $d=2$ in \cite{KL3}; see also \cite{DamG} for results on multidimensional discrete limit-periodic operators.
 Concerning the quasi-periodic case, a recent paper \cite{KS} (see also \cite{KS1} and \cite{KaSh}) has established the result in the case $d=2$, $l=4$. The result in that paper was formulated for prescribed basic frequencies: $\boldom_1=(1,0)$, $\boldom_2=(0,1)$, $\boldom_3=(\alpha,0)$, $\boldom_4=(0,\alpha)$ with Diophantine $\alpha$. 
Nevertheless, the methods of \cite{KS} are robust enough so that they are likely to work in other cases when $d=2$ and $l=4$, without the need to make many genericity assumptions.  Unfortunately, the approach in \cite{KS} could not be extended to higher dimensions or a larger number of frequencies.   
 

We will prove our theorem by constructing (generalized) eigenfunctions of the absolutely continuous spectrum as  small perturbations of exponential functions  $\be_{\bk}$, $\bk\in\R^d$, with large $\|\bk\|$. These eigenfunctions, denoted by $U^{(\infty)}(\bk,\bx)$, are a natural generalisation of the Bloch-Floquet solutions. Each such function  $U^{(\infty)}(\bk,\bx)$ is a solution of the equation $HU^{(\infty)}(\bk,\bx)=\lambda U^{(\infty)}(\bk,\bx)$ and has a form $U^{(\infty)}(\bk,\bx)=\be_{\bk}(1+u^{(\infty)}(\bk,\bx))$, where $u^{(\infty)}(\bk,\bx)$ is a small almost-periodic function: $||u^{(\infty)}||_{L^{\infty}}<C||\bk||^{-\de}$ for some positive $\de$. We obtain many finer  properties of these Floquet solutions, but we postpone formulating them exactly until Section  \ref{8.1'} (see Theorem \ref{T:Dec10}) to avoid introducing complicated notation. Theorem \ref{T:Dec10} can be considered as another main result of our paper. In fact, Theorem  \ref{mainth} is a relatively straightforward corollary of it. Another corollary of Theorem \ref{T:Dec10} is the long time behaviour of evolutionary equations (ballistic transport), which is discussed in \cite{KPS}.

To construct these solutions $U^{(\infty)}$, it is natural to consider the action of $H$ not in $L^2(\R^d)$, but in a linear space spanned by $\{\be_{\bk+\bn\vec\boldom}(\bx)\}_{\bn \in \Z^l,\bk\in \R^d}$; notice that this space is invariant under the action of $H$. Sometimes this linear space (or rather the closure of it, see Section \ref{section1}), denoted by $\GH(\bk)$, is called {\it {the fibre}} generated by $\bk$. 
Of course, $\GH(\bk)$ is not a subspace of $L^2(\R^d)$, but it is a subspace 
of the Besicovitch space $B^2(\R^d)$ -- the space that contains all the exponential functions \eqref{be} (see Section \ref{section1} for more details). The action of $H$ in $B^2(\R^d)$ is sometimes called the {\emph {Aubry dual}} of $H$. 
It is known that, as a set, the spectrum of $H$ on $L^2(\R^d)$ is the same as on  $B^2(\R^d)$ (see \cite{Shu0} and \cite{Shu}), but the nature of the spectrum is entirely different. As a result, the generalized eigenfunctions of the continuous spectrum of $H$ acting in $L^2$ will be proper eigenfunctions of $H$ acting on (the fibres in) $B^2$. The fact that these eigenfunctions will produce absolutely continuous spectrum in $L^2$ will follow from more or less standard estimates, assuming we have good control on the dependence of these eigenfunctions on the `initial momentum', $\bk$.  Thus,  we are going to  prove a partial localisation (i.e. the existence of the point spectrum) of $H$ acting in $B^2$, together with control on the behaviour of the eigenvalues and eigenfunctions (which includes making sure that every energy high enough is an eigenvalue). 

This localisation makes our results morally close to some theorems where complete localization (i.e. the spectrum being pure point) are established for the discrete quasi-periodic Schr\"odinger (or Schr\"odinger type) operators. Since such operators are usually bounded, the high energy regime does not exist. Therefore, the regime of a large coupling constant (or equivalently a small constant in front of the non-diagonal terms) is often  considered instead. One should note that this regime, although looking quite similar to the high energy regime, is not exactly the same, and  there are many differences between the two. The important results in the discrete setting are contained in papers \cite{74}, \cite{Bou1}, and \cite{JiLiSh}. In those papers, a complete localization is established for all dimensions and all numbers of frequencies; however, no control on the spectrum as a set is established and the spectrum could, in principle, be either Cantor, or the interval, or anything in-between. The proper `translation' of the Bethe-Sommerfeld Conjecture from our original setting (all high energies are in the spectrum) to the discrete setting would probably be that the spectrum contains an interval. Such a result was established in papers \cite{GSV1} and \cite{GSV2} in the case of one frequency (so that operators act on $\Z^1$) and large coupling constant. These papers established even more, namely that the spectrum {\it is} an interval (no gaps) and is completely localized. 

The methods we use to prove Theorem \ref{mainth} are based on a multi-scale analysis in the momenta space. As we consider the action of $H$ in $B^2$, it is natural to look at the invariant subspaces of $H$ generated by one exponential $\be_{\bk}$; we have called these subspaces fibres and denoted them by $\GH(\bk)$. There is a natural identification of each fibre with $l^2(\Z^l)$, so we are effectively considering a family of actions of $H$  on $l^2(\Z^l)$, parametrized by a point $\bk\in\R^d$ (see Section \ref{section1} for more details). Let $\lambda=\rho^2$ be a large value of energy. We want to prove that for most $\bk$ with $\|\bk\|\sim\rho$, we can find an eigenvalue of such an action with an eigenfunction that is a small perturbation of $\be_{\bk}$  (that corresponds to the delta-function at the origin of $\Z^l$). 

To begin with, we introduce a range of `scales' $r_n$ ($n=0,1,2,...$) defined by $r_0=10^{-6}$, $r_{n+1}=\rho^{r_n}$ (these, as well as some other notions introduced in this and the next few paragraphs, are not the exact definitions we will use in our article, but they give a fair idea of the ideas). We consider a collection of `central cubes' (or `boxes') $\hat K^{(n)}$; each such central cube is a ball of radius $\rho^{r_n}$ in the $l^{\infty}$-metric on $\Z^l$ centered at the origin. We will often use the name `cubes' for balls in $l^{\infty}$-norms in $\Z^l$ (since they look like cubes). Let $H^{(n)}=H^{(n)}(\bk)$ be the restriction of $H$ onto the linear subspace of $\Z^l$ spanned by $\hat K^{(n)}$ (of course, the projection onto this linear subspace does not commute with $H$, so by the restriction we mean $H$ multiplied by the projections onto $\hat K^{(n)}$ on both sides). We aim to achieve the following: 

{\bf Our Goal.} Each $H^{(n)}$ has one simple eigenvalue located sufficiently far (at least $\sim \rho^{-r_n}$-away) from the rest of the spectrum of $H^{(n)}$. This eigenvalue (denoted by $\la^{(n)}(\bk)$) behaves like $\|\bk\|^2$ plus smaller terms that are controlled, together with their derivatives, via perturbation theory. As a result, $\la^{(n)}$ is a continuous increasing function of $\|\bk\|$ and so, takes value $\rho^2$ for a large collection of $\bk$'s; we call the set of $\bk$ with this property   ($\la^{(n)}(\bk)=\rho^2$) the {\it isoenergetic surface} of $H^{(n)}$. 

Of course, we will be unable to achieve our goal for all $\bk$. For example, if $\|\bk\|=\|\bk+\boldom_1\|$ (or $\|\bk\|-\|\bk+\boldom_1\|$ is small, which on the physical level means that $\bk$ lies close to the diffraction (hyper)plane of $H$), we may have problems already on the very first step (to be precise, we start from $n=0$, so the very first step will have an official name {\it step zero}); in this case we would say that $\bk$ and $\bk+\boldom_1$ are in resonance. At each further step of our procedure, the definition of what the expression `in resonance' means exactly  will change and will include the spectrum of the restriction of $H$ to bigger and bigger cubes. However,  our goal is still achievable if we throw away some collection of `bad' initial points $\bk$; these bad points will take away only a small proportion of every isoenergetic surface with large enough energy (and this proportion will decay exponentially with $n$).     
    
The way we  achieve this is by induction. Suppose, the restriction, $H^{(n)}$, of $H$ to $\hat K^{(n)}$ satisfies our goal. Consider the restriction $H^{(n+1)}$ of $H$ to the next cube $\hat K^{(n+1)}$. We would like to treat $H^{(n+1)}$ as a perturbation of $H^{(n)}$ (extended somehow to $\hat K^{(n+1)}$). What can prevent us from doing this are the points inside  $\hat K^{(n+1)}\setminus \hat K^{(n)}$ that are in resonance with $\bk$. We would like to cover all such resonant points by cubes $\{K^{(n)}_m\}_{m=1}^M$ of size at most $r_n$ and then perform two tasks: 

Task 1.  Prove that for most $\bk$,   
$H^{(n)}$ is not in resonance with the restriction of $H$ onto all such cubes $K^{(n)}_m$, i.e. the spectra of the restrictions of $H$ onto $K^{(n)}_m$ lie sufficiently far away from the eigenvalue of $H^{(n)}$ we are interested in (the closest eigenvalue to $\lambda$).

 and 

Task 2. Prove that $H^{(n+1)}$ can be treated as a perturbation of the direct sum of  the restrictions of $H$ to all $K^{(n)}_m$ plus $H^{(n)}$. 

For this approach to work, we have to make sure that the sizes of $K^{(n)}_m$ are not too large and, moreover, that they are located sufficiently far away from each other in the lattice $\Z^l$. The tool to achieve this is a version of Bourgain's Lemma \cite{Bou1}, modified for our situation. This lemma is quite robust and can be treated as a `black box'; it tells us that we can achieve that $K^{(n)}_m$ are well-separated by throwing away a small proportion of `bad' frequency vectors $\vec\boldom$  from $[-1/2,1/2]^{dl}$. This lemma leads to the construction of a structure that we call a {\it \Bourgain structure}, see definition \ref{8.1}. This structure 
tells us that for each cube $K^{(j)}_{m_j}$ at any level $j\le n$ there are two possibilities: either this cube is good (or non-resonant; this means that the norm of the resolvent of the restriction of $H$ to this cube is not too large), or, if it is bad (resonant), then it is covered by some cube $K^{(j+1)}_{m_{j+1}}$ at the next level. We also know that all these cubes $\{K^{(j)}_{m_j}\}$ of all levels are located far away from each other (unless one of them is inside the other). Unfortunately, the existence of a  \Bourgain structure at level $n$ is not something that can be used as an induction assumption to imply the existence of a \Bourgain structure at the next level $n+1$ (even if we add some standard measure estimates like \eqref{Smn} to the assumption). In order to construct a structure that properly persists at the next level, we introduce a bigger structure that we call {\it an enlarged \Bourgain structure}, see definition \ref{8.2}; in this structure there are two different types of cubes at each level $j$: usual cubes $K^{(j)}$ and enlarged (much bigger) cubes $\tilde K^{(j)}$. The usual cubes are not too big, so that estimates there are good enough for perturbation theoretical arguments, whereas the enlarged cubes play an auxiliary role. The enlarged cubes of level $n$ are used only at the $n$-th step of induction and are discarded afterwards; in particular, we never cover enlarged cubes with any next level cubes. {\newred See also remark \ref{newrem5} for more comments on this structure.} While the standard \Bourgain structure (or similar structures) has appeared in many articles on this subject, we believe that the enlarged \Bourgain structure is novel.  This enlarged \Bourgain structure (and its use to establish the inductive step) is the first major new idea introduced in our paper.

After constructing the enlarged \Bourgain structure  at each level, we can perform tasks 1 and 2 indicated above.  The first task is done by means of Cartan's lemma (formulated in the Appendix), where again the enlarged \Bourgain structure comes in handy. Task 2 is performed using a tedious perturbation theoretical lemma \ref{abstractlemma}, which essentially consists of a resolvent identity written down many times.  

A serious problem with applying our approach is that both Bourgain's lemma and Cartan's lemma give us good enough estimates to be used in the induction procedure only once we have moved sufficiently far in the scale of approximations, namely when the power of $\rho$ becomes much larger than one.  
 Until we reach this level (i.e. until step two), we need to perform the first two steps -- step zero and step one -- by different means. Step zero is very simple: we do not have any cubes, but we declare any point $\bn$ inside the initial central cube  $\hat K^{(0)}$ good or bad, depending on the size of $|\|\bk+\bn\vec\boldom\|-\|\bk\||$, and then we keep only starting points $\bk$ for which all non-zero points inside  $\hat K^{(0)}$ are good. However, step one is something that, apparently, cannot be done using either a straightforward perturbation theory, or ideas related to the \Bourgain structure. Therefore,  we use a different approach. First, we assume that the basic frequencies $\{\boldom_j\}$ satisfy not just the standard Diophantine condition (meaning that their linear combinations cannot lie very close to the origin), but also something that we call the Strong Diophantine Condition (SDC). This condition means that also the angles between two different integer linear combinations of basic frequencies cannot be too small, as well as the angles between linear subspaces generated by the integer linear combinations of $\{\boldom_j\}$, see section \ref{section2} for precise definitions. We believe this condition is new; we prove that it is generic, i.e. it is satisfied on a set of basic frequencies of full measure. Once we impose SDC, the structure of the resonant boxes (or `clusters') $K^{(0)}_m$ at step one becomes manageable. Each such cluster is generated by a periodic lattice in a proper affine subspace $\GV$ of $\R^d$. The fact that these clusters are well-separated in $\Z^l$ is a consequence of the SDC (and the strong convexity of the Euclidean ball). The next observation is that the size of these clusters is much smaller than $\rho$, which implies that when we move $\bk$ in a certain direction, the restriction of $H$ to each such cluster is monotone. This makes the estimates of the measure of the set of bad $\bk$'s  (those are defined as $\bk$'s that may come in resonance with the restriction of $H$ to one of the clusters  $K^{(0)}_m$) quite straightforward. The vague idea of approximating a quasi-periodic operator by a direct sum of operators that are, effectively, periodic operators in proper subspaces of $\R^d$ has been used before (e.g. in \cite{PS} or \cite{PSh}), but the constructions there are completely different from ours. We consider the construction used in Step one (together with the SDC) the second main idea introduced in this paper.

Now we describe the structure of the paper. In Section \ref{section1}, we discuss notation and some major conventions we will be using in our paper. Section \ref{section2} gives the definition of the Strong Diophantine Condition and proves that it is generic. Section \ref{section3} performs the zeroth step of our procedure, while in Section \ref{section4} we construct resonant clusters $K^{(0)}_m$ necessary for the first step of our construction. In Section \ref{section5} we actually perform the first step. In Section \ref{goodset-3} we prove the version of the Bourgain's Lemma that we need for our construction. Section  \ref{section7} is one of the most important in our paper (together with section \ref{section4}): in Section \ref{section7} we set up the induction process (to be kicked off at  step 2). The complete inductive statement that is pushed onto the next level (modulo sets of small measure) is rather involved and includes several definitions and estimates. The main inductive statements, Theorem \ref{MGT} and Theorem \ref{Thm7.4}, are proved in Sections \ref{section8} and \ref{section9}. A short Section \ref{8.1'} then describes how to finish the proof of our main theorem, using mostly quite standard tools.  Finally, Section  \ref{section10} contains various Appendices; in order to help the reader, we have also included the Index of Notation at the end of this section. 

Let us now briefly discuss possible generalisations of our results. First, it seems that our results could be extended to infinite range potentials $V$ -- when Fourier coefficients $V_{\bn}$ decay super-exponentially in $\bn$ (very likely), or even exponentially (quite possible). We decided that, since our paper is quite long and  technical as it is, we will not consider these cases in it. Also, our approach seems to work for other differential operators with constant coefficients perturbed by a potential; pseudo-differential operators are a completely different ballgame, and our methods are likely be less effective  when dealing with them. 
We also think our approach can be used to study the regime of large coupling constant in discrete Schr\"odinger type operators. {\newred Finally, see remark \ref{newrm1} concerning the possibility to prove the complete absolute continuity of the spectrum of $H$ for large energies. 
}

\subsection*{Acknowledgments} The results were partially obtained during the programme Periodic and
Ergodic Spectral Problems in January  -- June 2015, supported by EPSRC Grant EP/K032208/1.
We are grateful to the Isaac Newton Institute for Mathematical Sciences, Cambridge,  for their
support and hospitality. YK is grateful to  Mittag-Leffler Institute their support and hospitality (April 2019).   The authors are grateful to Michael Goldstein, Wilhelm Schlag, and Mircea Voda  for useful discussions and to Jeffrey Galkowski for reading the early version of this article and making useful suggestions. The research of LP was partially supported by EPSRC grants EP/J016829/1 and EP/P024793/1. YK and RS were partially supported by NSF grant DMS-1814664. {\newred We are very grateful to the referees for reading this (very technical indeed) text and making many useful remarks and suggestions.}

\section{Notation and general conventions}\label{section1}
 We assume, as we can without loss of generality, that $V_{\bf 0}=0$. Let frequencies $\boldom_1,...,\boldom_l$ be  rationally independent and span the entire space $\R^d$. 
We denote by $\Z^l\vec\boldom$ the collection of all vectors $\bth\in\R^d$ that can be expressed as a linear combination of $\{\boldom_j\}$ with integer coefficients, i.e. the collection of all vectors of the form $\bn\vec\boldom$ with $\bn\in\Z^l$. This set is countable and non-discrete. We will use similar notation in other situations: if $A\subset\Z^l$, then $A\vec\boldom$ is the collection of all vectors of the form $\bn\vec\boldom$ with $\bn\in A$.
Of course, there is a natural isomorphism between $\Z^l$ and $\Z^l\vec\boldom$, 
\bees
(n_1,\dots,n_l)\mapsto \sum_{j=1}^l n_j\boldom_j. 
\enes
This isomorphism allows us to extend the $l^\infty$-norm \eqref{l_infinity} of elements of $\Z^l$ to the elements of $\Z^l\vec\boldom$. 
This notation is chosen to avoid confusion with the Euclidean length of $\bn\vec\boldom$ as an element of $\R^d$: the length of a vector $\bxi\in\R^d$ is denoted by $\|\bxi\|$.

We denote 
$$\Z_{\bk}^l=\Z^l(\bk):=\bk+\Z^l\vec\boldom\subset\R^d.
$$
We also denote by  
\bees
\Om(a;r),\ \ \BB(a;r)
\enes
balls with centre $a$ and radius $r$ in these norms ($\Om$ is the ball in $|\cdot|$-norm in $\Z^l$, and $\BB$ is the $||\cdot||$-ball in $\R^d$); we put 
\bee\label{newballs}
\Om(r):=\Om(0;r), \ \ \BB(r):=\BB(0;r), \ \ \Om'(r):=\Om(0;r)\setminus\{0\}.
\ene
Sometimes, we will be using superscripts $\R$ and $\Z$ to indicate that corresponding objects are in $\R^d$ or $\Z^l$ resp., so that, for example, $\Om(r)=\Om^{\Z}(r)$.
To indicate that we are dealing with $l_{\infty}$-norms in $\Z^l$, we often use expressions like $\Z$-distance, $\Z$-ball, etc. 
Since the balls $\Om$ in $\Z^l$ are taken in $\infty$-norm, we will sometimes refer to these balls as cubes. 
\begin{defn}\label{extended}
We define an extended cube (or extended ball) of size (or radius) $r$ in $\Z^l$ as any set that contains $\Om(a;r)$ and is contained in $\Om(a;2r)$ for some $a$; in this case, we will refer to $a$ as a centre of our extended ball and $r$ as its size (obviously, an extended ball could have several centres and several sizes). 
\end{defn}

Given a vector $\bk\in\R^d$ and a linear subspace $\GV\subset\R^d$, we denote by $\bk_{\GV}$ the orthogonal projection of $\bk$ onto $\GV$ and put $\bk_{\GV}^{\perp}:=\bk-\bk_{\GV}$. Given several vectors $\bx_1,\dots,\bx_n$ (from any vector space), we denote by 
\bee
R(\bx_1,\dots,\bx_n)=\spa \{\bx_1,\dots,\bx_n\}
\ene
the collection of all linear combinations of $\bx_1,\dots,\bx_n$ with real coefficients, and by
\bee
Z(\bx_1,\dots,\bx_n)
\ene
their linear combinations with integer coefficients. 

In this text there will be many different constants, and the letters by which we denote them indicate differences in their statuses. By letters $C$, $c$ we denote positive constants, the exact value of which is not important and can change each time these constants occur in the text (sometimes even each time they occur in one formula). The constants $C$ are assumed to be large, and the constants $c$ small. If, on the other hand, we use the expressions like $c_1$ or $C_{17}$, this means that the corresponding (positive) constant is fixed throughout the text. In section \ref{goodset-3} we will use constants $\Cun_j$ and $\cun_j$ with similar status. All our constants may depend on $d$ (the dimension of the Euclidean space), $l$ (the number of the basic frequencies $\boldom_j$), the norm of the potential $||V||_{\infty}$, $Q$ (the diameter of the support of $\hat V$, see \eqref{l_infinity}), and $\mu$ (the constant from the Strong Diophantine condition, discussed in Section \ref{section2}); moreover, we always assume that $d$, $l$, and $\mu$ are fixed and will frequently omit writing explicitly how our estimates depend on these quantities. If our fixed constants $C_j$ depend on other parameters ($x,y,z$, etc), we will express this by writing $C_j(x,y,z,...)$. 
Finally, sometimes we will use special constants, like $Z_0$ and $\gamma_0$. Those are some positive quantities that will be fixed at some stage of the proof (usually towards the end of the paper), but in the meantime it would be convenient for us to treat these quantities as parameters. When we write, for some positive quantities $A$ and $B$, that  $A\ll B$, this means $A <CB$ for some $C$ (that $C$, as we declared, can depend on $d$, $l$, $\mu$, and $V$).  

When $A\subset\R^a$ is a measurable set, by $\meas(A)$ we denote its Lebesgue $a$-dimensional measure.  

Our aim is to construct a mapping, putting into correspondence to vectors $\bk\in\R^d$ of sufficiently large length ($||\bk||>\rho_*$) a solution $u_{\bk}(\bx)$ of the equation 
\bee\label{solutions}
Hu_{\bk}=\eig(\bk)u_{\bk}.
\ene
The notation $\la=\eig(\bk)$ is used because it will be a limit of a sequence of approximations $\la^{(n)}(\bk)$; each $\la^{(n)}(\bk)$ is a small perturbation of $\|\bk\|^2$. 
This solution $u_{\bk}$ will be bounded (but obviously will not belong to $L_2(\R^d)$); moreover, it will be a small perturbation of the exponential $\be_{\bk}$.
More precisely, we will show that 
\bee
||u_{\bk}-\be_{\bk}||_{\infty}\ll ||\bk||^{-\delta},\ \ \delta>0. 
\ene
Then, establishing control of the behaviour of 
$\eig(\bk)$ as a function of $\bk$, allows us to deduce that the absolutely continuous spectrum of $H$ is non-empty and $\{u_{\bk}\}$ are the generalized eigenfunctions of this spectrum. Further, we will show that the equation 
\bee
\eig(\bk)=\la=:\rho^2
\ene
 has a solution $\bk$ for all sufficiently large $\rho$. This means that all large energies belong to the (absolutely continuous) spectrum. 
 Unfortunately, 
we will not be able to construct $u_{\bk}$ and $\eig(\bk)$ for all $\bk$ with large length -- even in the simpler periodic settings such mappings would not exist for $\bk$ located near the diffraction (Voronoi) planes. 
Nevertheless, we will construct these mappings for the majority of $\bk$. Namely, we will construct two sets $\CG^{\conf},\CNN^{\conf}\subset\R^d $; notation stands for {\it good} (sometimes also called non-resonant) and {\it bad} (or resonant) sets of vectors from the configuration space. These sets will satisfy the following properties:

(1) $\{\bk\in\R^d, ||\bk||>E_*\}=\CG^{\conf}\sqcup\CNN^{\conf}$ (a disjoint union; here, $E_*$ is a large fixed parameter); 

(2) Solutions $u_{\bk}$ and $\eig(\bk)$ of \eqref{solutions} are defined for $\bk\in\CG^{\conf}$; 

(3) For each $\rho>E_*+1$ the pre-image $(\eig)^{-1}(\rho)$ is non-empty (and contains a large proportion of the isoenergetic surface, or the `wobbly sphere', see below for more details). 

\ber\label{2.1}
Often, as here, we will use the notation with the letter(s) in the superscript being supposed to help the reader by telling them which variables are involved in the objects considered. For example, the set $\CG^{\conf}$ is a good set of $\conf$'s. The notation $\conf\in\CG^{\conf}$ should therefore be read as `$\conf$ is in a good set of variables $\conf$'. We also remark that notation $\CG$ will be used for the `good' sets in various variables, and $\CB$ is used for the `bad' sets.
\enr

The proof of our Theorem will consist in performing infinitely many steps; at each step $n$ we will throw away some bad set $\CNN^{\bk(n)}$. At the end we will put 
\bee
\CNN^{\bk}=\CNN^{\bk(\infty)}:=\cup_n \CNN^{\bk(n)}
\ene
and show that the set that we have thrown away is not too large in some sense. 

\ber
Most of the objects we construct (like good or bad sets) will depend on the Fourier coefficients  $V_{\bn}$ of the potential, which we consider to be fixed. The frequencies $\vec\boldom$, on the other hand, will be varying at some point. Indeed, we will construct the set $\BOm_*$ as a `good' set of frequencies:
\bee
\BOm_*=\CG^{\vec\boldom(\infty)}:=\cap_n \CG^{\vec\boldom(n)},
\ene
where $\CG^{\vec\boldom(n)}$ is a `good' set of frequencies we are keeping at step $n$. Of course, we will have
\bee
\CG^{\vec\boldom(n)}:=[-1/2,1/2]^{ld}\setminus \CNN^{\vec\boldom(n)},
\ene
where the bad sets $\CNN^{\vec\boldom(n)}$ have measures that quickly decay in $n$; we will also show that, moreover, the measure of the total bad set
\bee
\CNN^{\vec\boldom(\infty)}:=\cup_n \CNN^{\vec\boldom(n)}
\ene
can eventually be made arbitrarily small. We also emphasise that these bad sets (except the zeroth one) will depend on the Fourier coefficients  $\{V_{\bn}\}$. 
\enr

\ber
One of the difficulties for a reader in trying to follow the line of thought in this paper is the fact that we often switch between the objects that belong to $\R^d$ and $\Z^l$. To make things simpler, we will use `integer' letters (like $\bn$, $\bm$,  etc) for vectors from $\Z^l$ and Greek letters (like $\bxi$, $\boldeta$, etc) for vectors from $\R^d$. The only exception from this rule is the `initial' vector $\bk\in\R^d$; this choice was made  to make as few changes from the notation of \cite{KS} as possible. We will have another variable, $\bka\in\R^d$, that will have the same status as $\bk$; on some occasions it would be convenient to use both these letters simultaneously.   

\enr

Another convention we will sometimes follow is this: for a set $A\subset\Z^l\vec\boldom\subset\R^d$ we denote by $A^{\Z}\subset\Z^l$ a set for which we have $A=A^{\Z}\vec\boldom$.

The solution $u_{\bk}$ that we want to construct will belong to the fibre $\GH(\bk)$ `generated' by $\be_{\bk}$. 
 This subspace is defined like this:
\bee\label{7fibre}
\GH(\bk):=\{\sum_{\bn\in \Z^l}a_{\bn
}\be_{\bk+\bn\vec\boldom}\},
\ene
where the collection of complex coefficients $\{a_{\bn
}\}_{\bn\in\Z^l}$ belongs to $l^2(\Z^l)$. Of course, $\GH(\bk)$ is not a subspace of $L^2(\R^d)$, but it is a subspace of the Besikovitch space 
$B^2(\R^d)$. The Besicovitch space is defined as the collection of all (formal, countable) linear combinations of exponential functions with square-summable coefficients: 
\bee\label{B2}
B^2(\R^d):=\{\sum_{j\in\N}a_{j
}\be_{\bxi_j}: \ \bxi_j\in\R^d, \ (a_1,a_2,...)\in l^2\}.
\ene
The inner product is defined by $\lu\be_{\bxi},\be_{\bxi'}\ru_{B^2}=\de_{\bxi,\bxi'}$
This inner product makes $B^2(\R^d)$ a non-separable Hilbert space with the canonical (uncountable) orthonormal basis 
$\{\be_{\bxi}\}_{\bxi\in\R^d}$. 
The action of $H$ on individual exponentials (which is clearly defined) can naturally be extended to (a dense subspace of) $B^2(\R^d)$. Then we obviously have that for each $\bk\in\R^d$ the space $\GH(\bk)$ is a separable subspace of $B^2(\R^d)$ and is invariant under the action of $H$.  The `overall' action of $H$ on  $B^2(\R^d)$ can be expressed as a (non-countable) direct sum of $H$ acting on individual fibres.  
There is a natural isometry between $\GH(\bk)$ and $l^2(\Z^l)$; we will often identify these two spaces. We will also need a space $B^1(\R^d)$; this space is defined similarly to \eqref{B2}, but requiring that the vector of coefficients $(a_1,a_2,...)\in l^1$. The $l^1$ norm of $(a_1,a_2,...)$ makes $B^1(\R^d)$ a (non-separable) Banach space. 

\begin{defn}\label{proj}
Given a set $\La\subset\Z^l$, we denote by $l^2(\La)$ (resp. $\GH_{\La;\bk}=\GH(\La;\bk)$)
the closed subspace of $l^2(\Z^l)$ (resp. $\GH(\bk)$) spanned by the elements of $\La$ (resp. $\bk+\La\vec\boldom$). We also denote by  
$P(\La)$ (resp. $\CP(\La)=\CP(\La;\bk)$) the orthogonal 
projection onto this subspace. Given any operator $H$ acting on  $B^2(\R^d)$, we denote by $H(\La;\bk)$ the operator $\CP(\La;\bk)H\CP(\La;\bk)$ acting on $\GH_{\La;\bk}$. Similarly, if $H$ acts on $l^2(\Z^l)$, we denote $H(\La):=P(\La)HP(\La)$.  We also put $H(\bk):=H(\Z^l;\bk)$. 
\end{defn}

Obviously, $H_0(\bk)$ is diagonal in the natural basis consisting of $\{\be_{\bk+\bth}\}_{\bth\in \Z^l\vec\boldom}$. Given any operator $H(\bk)$ acting on $\GH(\bk)$ and $\bn_1,\bn_2\in \Z^l$, by $H(\bk)_{\bn_1\bn_2}$ we denote the matrix element:
\bee
H(\bk)_{\bn_1\bn_2}:=\lu H(\bk)\be_{\bk+\bn_1\vec\boldom},\be_{\bk+\bn_2\vec\boldom}\ru_{B^2}. 
\ene

Let us introduce the polar coordinates for vectors $\bk$: we put 
\bee\label{coordinates1}
\bk=\nk\BPhi, 
\ene
where 
\bee\label{neweq1}
\nk=||\bk||
\ene 
and 
\bee\label{neweq2}
\BPhi=\BPhi(\bk):=\bk/||\bk||\in\S.
\ene
We will need more detailed coordinates than $\BPhi$ on a sphere, but unless $d=2$ there are no convenient global coordinates on $\S$. Nevertheless, most of the time we will work not on the entire sphere, but on a small part of it. If this is the case, we can introduce coordinates on an individual patch of a sphere in the following way. Suppose, we work on a neighbourhood of a given point $\BPhi^*$ on a sphere. Then we can introduce Cartesian coordinates in $\R^d$ in such a way that the last basis vector $e_d$ coincides with $\BPhi^*$. Then we denote the first $d-1$ coordinates of a point $\BPhi\in\S$ by $(\phi_1,\dots,\phi_{d-1})$ so that 
\bee\label{coordinates2}
\BPhi=(\Phi_1,...,\Phi_d)=\BPsi(\phi_1,\dots,\phi_{d-1}):=(\phi_1,\dots,\phi_{d-1},\sqrt{1-\phi_1^2-\dots-\phi_{d-1}^2}). 
\ene
These coordinates make sense for $\BPhi$ lying near $\BPhi^*$, and $\BPhi^*$ has coordinates $(0,\dots,0,1)$. We call such coordinates {\it natural coordinates on $\S$ around $\BPhi^*$}. We start by covering the entire sphere with patches of the type described above and then restricting our attention to $\BPhi(\bk)$ lying in one such patch. Thus, without loss of generality we will often assume that $\bk$ has a form \eqref{coordinates1} and \eqref{coordinates2} with 
\bee\label{coordinates3}
\vec\phi:=(\phi_1,\dots,\phi_{d-1})\in (-\tilde\phi,\tilde\phi)^{d-1}=:\Pi 
\ene
and $\tilde\phi$ is a sufficiently small positive number. We will cover the sphere by multiple patches $\{(\Pi_m,\BPsi_m)\}$, meaning that 
$\Pi_m=(-\tilde\phi_m,\tilde\phi_m)^{d-1}$, $\BPsi_m:\Pi_m\to\S$ is a mapping of the form \eqref{coordinates2} (of course with the centre $\BPhi^*=\BPhi^*_m$ different for each $m$), and each point of $\S$ belongs to $\BPsi_m(\Pi_m)$ for at least one $m$. 
Throughout the proof we will obtain various estimates valid only for points from $\BPsi_m(\Pi_m)$ for a fixed $m$. 
Of course, our constructions will thus depend on $m$, but we will often omit this dependence in the notation. Nevertheless, each time we talk about points $\BPhi$ or $\phi$, we assume that we have fixed a certain patch $\{(\Pi_m,\BPsi_m)\}$. When we want to emphasise that a particular patch $\{(\Pi,\BPsi)\}$ is centred at $\BPhi^*$, 
we use notation $\Pi(\BPhi^*)$; 
 we also use notation 
$\BPsi_{\BPhi^*}$ 
for the corresponding mapping $\BPsi$, so that $\BPsi_{\BPhi^*}:\Pi(\BPhi^*)\to\S$. 
We denote $\CA^{\BPhi}_m:=\BPsi_m(\Pi_m)$ so that $\S=\cup_m \CA^{\BPhi}_m$. We also remark that sometimes we will make variables $\phi$ complex. This will lead to complexification of correspondent $\Pi$; the (analytic) mappings $\BPsi$ will be still defined by \eqref{coordinates2}. 

In the process of proving of our main Theorem, we will be obtaining finer and finer approximations of the eigenvalues and generalized eigenfunctions of $H$. Thus, at each step $n$ of the procedure we will consider a cover of $\S$ by smaller and smaller patches. We will indicate this by writing a superscript $(n)$ to indicate the objects we consider at the $n$-th step. Thus, we will have a patch at the step $n$ being defined as 
\bee \CA^{\BPhi(n)}_m:=\BPsi_m(\Pi_m^{(n)}), \label{patches} \ene
where $\Pi_m^{(n)}=(-\tilde\phi_m^{(n)},\tilde\phi_m^{(n)})^{d-1}$.
The `size' of the patch of order $n$, $\tilde\phi^{(n)}=\tilde\phi_m^{(n)}$, is chosen to be independent of $m$ and decay exponentially with $n$; the explicit formulas for these sizes will be given later. We will define a complexification of $\Pi_m^{(n)}$, 
\bee
\Pi_{m,\C}^{(n)}:=( D(\tilde\phi^{(n)}))^{d-1},
\ene
where
\bee
D(r):=\{z\in\C:\ |z|<r\};
\ene
we also define
\bee\label{D}
D:=D(1)=\{z\in\C:\ |z|<1\}. 
\ene
\ber
Of course, all the patches $\Pi_m^{(n)}$ of $\vec\phi$-variables for fixed $n$ and different $m$ are identical. We still want to treat them as different objects, since the way we will split a patch $\Pi_m^{(n)}$ into good and bad parts will depend on $m$. 
\enr

Thus, for each $n$ we have a cover of $\S$ by patches $\{\CA^{\BPhi(n)}_m\}$; we will often refer to $n$ as the order or the level of the patches. We will be assuming that any patch of any level $n\ge 1$ is covered by at least one patch of a previous level. Therefore, given any point $\BPhi$, we can choose several patches of order $n$ that cover this point, then several patches of order $n-1$ that cover the patch of order $n$, etc. Sometimes it will be necessary to keep track of the `allegiance' of a given point (i.e. the patches it has been assigned to) for all levels. 
This leads to the following notion:

\begin{defin}\label{def7.1}
A matryoshka $\CM^{\BPhi(n)}$ of patches of level $n$ is a collection of patches $\{\CA^{\BPhi(j)}_{m_j}\}_{j=0}^n$, where each patch of level $j+1$ lies inside a patch of level $j$. We say that matryoshka $\CM^{\BPhi(n)}$ is good, if each of its patches is good.
\end{defin} 

\ber
The last part of the previous definition involves the notion of a good patch. The definition of a good patch is quite involved, and we will have to wait until section \ref{section7} to give it completely. At the moment, we notice that a patch at any level $n>0$ can possibly be good only if it is covered by a good patch at the previous level. Thus our procedure will look like this: at each step $n$ we will declare some patches bad, throw them away, then consider only good patches, cover them by next level patches, declare some of them bad, throw them away, etc.  
\enr

\ber
At a certain stage in the proof, we will have to introduce the complexified good sets. This would mean that the set of parameters $\vec\phi$ will be made complex. {\newred The reason for this is that we will be using several results from the theory of complex variables, like Rouch\'e's Lemma \ref{complexlemma}  or Cartan's Lemma \ref{Cartan}. 
This makes most of the tools we apply in this paper ``intrinsically complex'', perhaps the only exception being Bourgain's type Lemma \ref{MGL}.}
We will distinguish between real and complex good sets by writing $\CG^{\vec\phi}_{\R}$ or $\CG^{\vec\phi}_{\C}$, resp. Sometimes we will also complexify other parameters, like $\rho$.
\enr

Usually, in this paper we will be working in a  fixed window of energies: we assume that $\la=\rho^2$ with $\rho\in [E-1,E+1]$ and $E$ for convenience is assumed to be integer. Then all our estimates will be made in terms of functions (powers, exponentials, etc) of $E$, where $E$ is assumed to be fixed, but large: $E\ge E_*=E_*(\{V_{\bn}\},\vec\boldom)$. The number $E_*$ (which we will also conveniently assume to be an integer) is our `initial' straightforward lower bound, above which all our constructions will work. We will define $E_*$ through several lower bounds it should satisfy. To begin with, we  assume that $E_*$ satisfies inequalities  \eqref{E_*}; later, we will add one more, \eqref{E_*1}.  The `final' bound, $\lambda_*$ (with the property that the interval $[\lambda_*,+\infty)$ is covered by absolutely continuous spectrum) will be chosen at the very end of the proof, in section \ref{8.1'}, and it will depend on  $\vec\boldom$ in an uncontrolled way. 

We will use these patches in $\BPhi$ coordinates in Section \ref{section3}; in further Sections, we will also need patches in other variables.
The variables that will be covered in patches will be chosen from the following 
list: the energy $\lambda=\rho^2\in\R$, 
 the spherical angle $\BPhi\in\S$, the frequency vector $\vec\boldom\in\R^{ld}$, and $\bxi\in\R^d$. The latter is an auxiliary variable that runs through the spherical layer $\|\bxi\|\in [E-1,E+1]$; at some stage we will put $\bxi=\bk+\bn\vec\boldom$, but it is convenient to consider $\bxi$ as a separate independent variable for a while. We assume that  the following `region of interest' is covered by the patches: $\{\bxi\in\R^d, \|\bxi\|\in[E-1,E+1]\}$, $\rho\in [E-1,E+1]$, $\vec\boldom\in[-1/2,1/2]^{ld}$.
 
{\newred
\ber\label{newrm2}
As stated above, $\bxi$ will play a role of a point on a quasi-periodic lattice $\CL:=\{\bk+\bn\vec\boldom, \bn\in\Z^l\}$ and introducing it as a separate parameter will enable us to respond to the question ``Suppose that $\CL$ has passed through a point $\bxi\in\R^d$, i.e. there is $\bn\in\Z^l$ such that  $\bxi=\bk+\bn\vec\boldom$. What can we say about a  neighbourhood of $\bn$ in $\Z^l$ then?'' We want to emphasise that, despite a similarity between variables $\bk$ and $\bxi$, their roles are completely different. Indeed, while we are allowed to throw away some bad regions for $\bk$, we are not allowed to do this for $\bxi$, since $\CL$ is generically dense in $\R^d$. Also, sometimes we will shift $\bk$ by a small vector, but we have no freedom to do the same with each $\bxi$ -- we need to deal with them no matter where they are located.    
\enr
} 
 
\ber 
We will not consider patches in $k:=||\bk||$ as such. The reason is the following. Given a good spherical angle $\BPhi$ and the energy $\rho$, we will be considering consecutive approximations of the `true' eigenvalue $\la^{(\infty)}(\bk)$, given by $\la^{(n)}(\bk)$ at step $n$. Then, we will consider only points $\bka^{(n)}(\BPhi)=\kappa^{(n)}(\BPhi)\BPhi$ 
 satisfying the property $\la^{(n)}(\bka^{(n)}(\BPhi))=\la$ and small neighbourhoods of such points. In other words, once we fix a patch for $\BPhi$ and a patch for $\rho$, this would determine a `domain of interest' for $\bk$ and allow us to stick to considering only $\bk$ from these domains. These domains would not form, strictly speaking, patches in $\bk$, but they will play an important role in our constructions and so we will call them quasi-patches in $\bk$ and denote in the same way as `proper' patches, i.e. we use notation $\CA^{\bk(n)}$ for these quasi-patches. We will also say that these quasi-patches in $\bk$ are associated with corresponding patches in $\BPhi$. The rigorous definition of these (quasi-)patches will be given later (see \eqref{CAbk} or \eqref{CAbkn}).     
\enr

\ber
The reader may also wonder why we have introduced an extra parameter $E$ that is, essentially, of the same size as $\rho$ and then have defined our objects (like good or bad lattice points, boxes, etc.) in terms of $E$ rather than $\rho$. The reason is that we do not want these objects to change if we start varying $\rho$. Thus, $E$ is a parameter that roughly defines `the size of $\rho$', but does not change if $\rho$ varies within a small interval.
\enr
 
 The patches will be denoted by the symbol $\CA$ with superscripts describing which variables are participating in this patch (and, possibly, another superscript $(n)$ indicating that those are the patches at step $n$). Initially, we assume that each patch is a product of balls in each of the participating variables; the radii of these balls (the size of the patch in corresponding variable) will depend on $n$ and will be defined each time we introduce a particular patch. For example, $\CA^{\rho,\bxi,\vec\boldom}$ of size $\En^{-1}$ in $\rho$ and $\bxi$ and $\En^{-2}$ in $\vec\boldom$ is a product of three balls: $B(\rho^*,\En^{-1})$,  $B(\bxi^*,\En^{-1})$, and $B(\vec\boldom^*,\En^{-2})$ (it is usually clear in which spaces these balls are located); the corresponding letter with a star ($\rho^*$,$\bxi^*$, etc.) denotes a centre of such ball. Sometimes we will have to assume that the centre (say, $\vec\boldom^*$) of our patch satisfies a certain property. 
This will be done by choosing a point 
$\vec\boldom^{*,{\textrm {new}}}\in B(\vec\boldom^*,\En^{-2})$ 
that satisfies our property (assuming such a point exists) and then considering a new patch $B(\vec\boldom^{*,{\textrm {new}}},2\En^{-2})$ instead of $B(\vec\boldom^*,\En^{-2})$. This process gives us the possibility to assume, without loss of generality, that if a property is satisfied at some point of a patch, it is satisfied at its centre. Sometimes it will be convenient to assume that the patches do not intersect pairwise, in which case we will assign points that belong to the intersection of two (or more) patches to just one of these patches. After this procedure the patches are no longer products of balls, but subsets of such products  (however, they do not intersect). 
Sometimes we will write, slightly abusing the notation, $\BPhi\in\CA^{\BPhi,\rho,\bxi}$, meaning of course that $\BPhi$ belongs to the projection of $\CA$ onto the $\BPhi$-variables. The $n$-th level patches should, ideally, be labelled in the following way:  
$\{\CA^{\bxi (n)}_{m^{\bxi}_n}\}$,  $\{\CA^{\rho (n)}_{m^{\rho}_n}\}$, etc. Here, $m^{\bxi}_n$, say, is a natural number running between $1$ and the overall number of patches in $\bxi$ of order $n$. Since this notation is too cumbersome, we will often denote by $\tim$ a universal label of any patch in any variable; the meaning of this label can be different each time we use it; for example, $\CA^{\bxi (n)}_{\tim}$ is a patch 
$\CA^{\bxi (n)}_{m^{\bxi}_n}$.

Here is the list of the conventions we use.  The patches (of each level) in $\vec\boldom$ do not intersect. At step zero, these patches cover $[-1/2,1/2]^{ld}$. Starting from the next step, we will declare some patches in $\vec\boldom$ bad; those patches are not covered by further patches, but the good patches are.   The patches in $\rho$ do overlap, so that each point $\rho\in [E-1,E+1]$ is covered by (at least) two patches at any step. The patches in $\bxi$ have a special arrangement. First, for fixed $n$ we cover all $\bxi$ by `small' patches -- balls of certain radius $r=r^{(n)}$. As we stated earlier, these $r^{(n)}$ depend on $n$ and will be properly defined later. 
Then we increase the size of each `small' patch by a factor of $10$, so that each point $\bxi,\ \|\bxi\|\in[E-1,E+1]$, is now covered by multiple patches $\CA^{\bxi(n)}$, each of size $10r^{(n)}$. These scaled up patches are the patches we will be working with. 

\ber\label{bxibk}
Suppose that we are considering points $\bk$ inside a quasi-patch $\CA^{\bk(n)}$
of size  $r^{(n)}$.
Suppose now that we shift this quasi-patch by a certain vector $\bn\vec\boldom$. Then various points of the shifted quasi-patch $\CA^{\bk(n)}+\bn\vec\boldom$ are covered by different patches $\CA^{\bxi(n)}_m$, but our construction ensures that there is at least one patch,   say $\CA^{\bxi(n)}_{m_0}$, that covers the shifted quasi-patch completely. In such situations it will be convenient for us to assign all the points $\bxi\in (\BB(\bk_0,r^{(n)})+\bn\vec\boldom)$ to this one particular big patch $\CA^{\bxi(n)}_{m_0}$ of size $10r^{(n)}$.
\enr

 We will often say that a certain estimate is stable on a patch. This would mean that if an estimate holds at one point of a patch, it will hold at all other points, possibly with a different constant (usually new constant is an old constant times $2$ or $1/2$). 
 
Sometimes, when we talk about points $\bk$, $\rho$, or $\bxi$ we will need to keep track of all the patches it belonged to for all levels. This leads to the following definition, similar to definition \ref{def7.1}  

\begin{defin}\label{def7.1xi}
A matryoshka $\CM^{\bxi(n)}$ of patches of level $n$ is a collection of patches $\{\CA^{\bxi(j)}_{\tim_j}\}_{j=0}^n$, where each patch of level $j+1$ lies inside a patch of level $j$. The definition of matryoshkas in $\rho$ or $\bk$ is similar (only in the case of $\bk$ we consider quasi-patches here). 
\end{defin} 

\ber\label{rem1}
The reason why we need these tedious constructions is the following. It is not enough for our purposes to just construct the solutions to \eqref{solutions}, but to prove various regularity properties of $\lambda^{(\infty)}(\bk)$. In order to achieve this, we have to establish regularity of approximations $\lambda^{(n)}(\bk)$. And to do this, we need to make sure that certain parts of our construction are stable when $\bk$ runs along patches of different levels. The notion of matryoshka formalises these properties, see also remark \ref{rem2}. 
\enr

The size of the patches of order $n$ will decay with $n$ so fast that the number of all possible matryoshkas of order $n$ can be estimated by the square of the number of patches of order $n$. We also will use the `universal index' $\tim$ to label either matryoshkas or patches within one matryoshka.

Sometimes one (or more) of coordinates of our $d$-dimensional vectors will become complex. In this case we put $||\a||^2_\R:=\langle\a,\a\rangle _\R$, where
$\langle \a,\b\rangle_\R:=\sum _{j=1}^da_jb_j$ when $\a, \b \in \C^d$. Recall that 
$||\a||^2=\langle\a,\a\rangle $, where
$\langle \a,\b\rangle:=\sum _{j=1}^da_j\bar b_j$

The bad  set $\CNN^{\BPhi}$ described above will be constructed as a union of infinitely many bad sets of smaller and smaller size. Our procedure consists of infinitely many approximating steps, and at each step $n$ we will throw away some bad set $\CNN^{\BPhi(n)}$; the measure of these sets will decay exponentially in $n$. 
We will often call these bad sets the {\it resonant sets}, since the term resonant set is often used in physics in analogous situations. The procedure will be slightly different at step zero and step one compared to further inductive steps. We start by fixing sufficiently large $\la$ and will find a solution (in fact, lots of them) of the equation 
\bees
Hu=\la u.
\enes
During our construction we will need to assume that $\rho$ is sufficiently large, $\rho\geq E_*$. Several necessary conditions (i.e. some lower bounds on $E_*$) are summarized in \eqref{E_*}. Later, we will add one more condition \eqref{E_*1}. Then, in section 11, we explain how to choose the final  $\lambda_*=\rho_*^2$ (where $\rho_*\ge E_*$) for Theorem~\ref{mainth}. 

We will  assume 
that the basic frequency vector $\vec\boldom$ satisfies the Strong Diophantine Condition (SDC) introduced in the next Section. Our constructions will depend on the choice of $E_*$ and the parameters in the SDC, but we will omit this dependence in the notation.  
Since we are trying to find $\bk$ satisfying $\eig(\bk)=\la$ with $\eig(\bk)$ a perturbation of $\|\bk\|^2$, we will consider vectors $\bk$ satisfying $|\|\bk\|-\rho|\ll \la^{-1/2}$. Suppose, $\bxi=\bk+\bn\vec\boldom$ with $\bn\in \Z^l$ (here, $\bn$, $\vec\boldom$, and $\nk$ are real, but $\BPhi$ may be complex; {\newred recall definitions \eqref{neweq1} and \eqref{neweq2}}). We obviously have
\bees
||\bxi||^2=\nk^2+||\bn\vec\boldom||^2+2\nk\Re\lu\BPhi,\bn\vec\boldom\ru 
\enes
and
\bees
||\bxi||_\R^2=\nk^2+||\bn\vec\boldom||^2+2\nk\lu\BPhi,\bn\vec\boldom\ru_\R ; 
\enes
the latter expression is obviously holomorphic in $(\phi_1,\dots,\phi_{d-1})$. 

Sometimes we will want to consider operators $H(\bkappa)$ (or $H(\La;\bkappa)$) for complex $\bkappa\in\C^d$, assuming that $H=H_0+V$.  In this case, we define the operator $H(\bkappa)$ as an operator acting in $\Z^l$ with  
\bee
H(\bkappa)_{\bn_1\bn_2}=V_{\bn_2-\bn_1}, \ \ \bn_2\ne\bn_1 
\ene 
and
\bee
H(\bkappa)_{\bn_1\bn_1}=||\bkappa+\bn_1\vec\boldom||_{\R}^2. 
\ene 
The operator $H(\La;\bkappa)$ is defined similarly: we just additionally assume that $$
H(\bkappa)_{\bn_1\bn_2}=0
$$ if $\bn_1\not\in\Lambda$ or $\bn_2\not\in\Lambda$. We will usually assume that $\bkappa$ is a small (but possibly complex) perturbation of $\bk$ in this context. 

If $A$ is an operator acting in a Hilbert space, by $||A||$ we denote the operator norm of $A$, whereas by $||A||_1$ and $||A||_2$ we denote correspondingly the trace or Hilbert-Schmidt norm of $A$. 

We will often use a convention that if $P$ is an orthogonal projection and $H$ is an operator, then $(PHP)^{-1}$ is an inverse of $PHP$ restricted to the range of $P$.  

\ber
There is one more thing we should warn the reader about. Since we consider our paper to be rather technical and difficult to read, we try to simplify the exposition. In particular, in situations when the precise values of constants/powers are not important, sometimes we use values that look simpler but are perhaps not optimal. Therefore, if a reader thinks that some estimates can be easily improved a little bit, this may well be the case.  
\enr

\section{Diophantine conditions}\label{section2}
{\newred In this section, we will  give the definition
of the Strong Diophantine Condition, discuss several of its immediate consequences, and prove that it is generic (Lemma \ref{new:generic}).}

Let $\vec\boldom=\boldom_1,\dots,\boldom_l$ ($l\ge d+1$) be a vector of basic frequencies as above. We assume that $\{\boldom_j\}$  are linearly independent over $\Q$. The first condition we have to assume is that they are Diophantine. This means the following. Suppose, $\bn=(n_1,\dots,n_l)$ is an integer vector. Recall that  $\bn\vec\boldom=\sum_{j=1}^l n_j\boldom_j\in \R^d$ is the `inner product' between $\bn$ and $\vec\boldom$. Then we require that 
\bee\label{weak}
||\bn\vec\boldom||>A|\bn|^{-\tilde\mu}
\ene  
for some positive $A$ and $\tilde\mu$.  
It is well-known (and will be proven later anyway) that for sufficiently large $\tilde\mu$ this condition is generic in a sense that it is satisfied for $\vec\boldom\in\tilde\Omega$, where $\tilde\Omega$ is a subset of $\cube^{ld}$ of full measure. This condition 
is not sufficient for our purposes as we need to control not just the lengths of various linear combinations of $\boldom_j$, but also the angles between these linear combinations. In order to do this, we will impose a much stronger condition, which will still be generic. We call this new assumption the {\it strong Diophantine condition} and will sometimes call the condition \eqref{weak} the {\it weak Diophantine condition}. Let us introduce more notation first. We denote by 
\bee\label{BN}
\BN=(n_{jk})_{j,k=1}^{j=d,k=l}
\ene
 a matrix with integer entries $n_{jk}$. We denote its norm by 
\bee 
 ||\BN||_*:=\sum_{j=1}^d\sum_{k=1}^l|n_{jk}|
\ene 
  and form the following linear combinations: 
\bee \label{thetas} 
  \bth_j=\bth_j(\vec\boldom,\BN):=\sum_{k=1}^l n_{jk}\boldom_k\in\R^d, \ \ \ j=1,\dots,d.
\ene  
    Obviously, if $Ran(\BN)<d$ (in which case we call $\BN$ {\it degenerate}), then $\bth_j$ are linearly dependent. Let us denote by ${\mathcal {ND}}(l,d)$ the  collection of all non-degenerate matrices \eqref{BN}.  
    Now we formulate our condition:

{\bf Strong Diophantine Condition} (SDC)

Assume that $\BN\in{\mathcal {ND}}(l,d)$, i.e. $Ran(\BN)=d$. Then the determinant of $\{\bth_j\}$ is bounded below:
\bee\label{strong}
|\det(\bth_1,\dots,\bth_d)|>B_0||\BN||_*^{-\hat\mu}
\ene
for some positive $B_0$ and $\hat\mu$. 

\ber
It is clear that this condition implies the weak diophantine condition. The opposite is not true: for example, the  following three vectors: $(1,0)$, $(0,1)$, and $(a,b)$ with Diophantine $a$ and Liouville $b$ in $\R^2$ satisfy weak, but not strong Diophantine condition. We also notice that Condition A from \cite{PSh} follows from SDC (Condition A states that if $\bth_1,\dots,\bth_d$ are as above, then either $\{\bth_j\}$ are linearly independent (over $\R$), or they are linearly dependent with integer coefficients). In fact, SDC may be considered as a `quantified version' of Condition A. 
\enr

\bel\label{angles}
SDC implies the following statement: suppose, $\{\bth_j\}$ are as in \eqref{thetas} with non-degenerate $\BN$. Suppose, $2\le n\le d$. Then the angle between $\bth_n$ and the subspace
spanned by $\bth_1,...,\bth_{n-1}$ is bounded below by $B_0||\BN||_*^{-\mu}$, where
$\mu$ is any number satisfying $\mu>\hat\mu+d$.  
\enl
\bep
This follows from the fact that $|\det(\bth_1,\dots,\bth_d)|$ is bounded above by the product of the sine of the angle we discuss and $\prod_{j=1}^d|\bth_j|$, and $|\bth_j|\le ||\BN||_*$. 
\enp

\ber
We note that SDC implies even more: if we form any two strongly different linear subspaces out of the vectors $\{\bth_j\}$, (two subspaces are strongly different if none of them is contained in the other), then the angle between them is also bounded below by 
$C||\BN||_*^{-\mu}$. Since we do not need this statement in our paper, we omit the proof.  
\enr

\bel\label{determinant}
SDC implies the following statement: for any linear independent system of $s\leq d$ integer vectors $\bm_j\in\Z^l$, $j=1,\dots,s$, the corresponding vectors $\bth_j:=\bm_j\vec\boldom$ generate an  $s$-dimensional parallelepiped with the $s$-dimensional volume bounded below by   $B_0(\sum_j|\bm_j|)^{-\hat\mu }$.
\enl
\bep
We just notice that, since vectors $\vec\boldom_1,...,\vec\boldom_l$ span the entire space $\R^d$, we can  add $d-s$ vectors $\vec\boldom_p$ to our collection $\{\bth_j\}$
so that the resulting set of $d$ vectors is linearly independent. Now the statement follows from the SDC. 
\enp

Now, let us prove that SDC is generic. Namely, we will prove the following statement:
\bel\label{new:generic}
Suppose, $\hat\mu>ld^2$ is fixed. Then there exists $\BOm_0$ a subset of $\cube^{ld}$ of full measure such that for any $\vec\boldom\in\BOm_0$ the inequality \eqref{strong} is satisfied for some positive $B_0=B_0(\vec\boldom)$. Moreover, for each $B_0>0$ there is a set $\BOm_0(B_0)$ of measure at least $1-C(\hat\mu)B_0^{1/d}$ such that \eqref{strong} holds. 
\enl

\bep

Suppose, $\BN\in{\mathcal {ND}}(l,d)$ is  fixed.

\bel
Let $a>0$. Then 
\bee
\meas\ \{\vec\boldom\in \cube^{ld}:\ |\det(\bth_1,\dots,\bth_d)|<a\}< 2da^{1/d}.
\ene
\enl
\bep
Since $\BN=\{n_{jk}\}$ is non-degenerate, there is a $d\times d$ non-degenerate minor of $\BN$. Without loss of generality, we can assume that $\det(n_{jk})_{j,k=1}^d$ is non-zero, which means that 
\bee
|\det(n_{jk})_{j,k=1}^d|\ge 1. 
\ene
Denote by  $(\om_{j1},\dots,\om_{jd})$  the coordinates of $\boldom_j$. We want to study $\tilde D:=\det(\bth_1,\dots,\bth_d)$ as a function of $\{\om_{jk}\}$. 
As a function of  $\om_{11}$, $\tilde D$ is linear, say $\tilde D=s_1\om_{11}+t_1$, where $s_1$ and $t_1$ are functions of the rest of coordinates $\{\om_{jk}\}$. Then we can write 
\bee
\bes
&\{\vec\boldom\in \cube^{ld}:\ |\det(\bth_1,\dots,\bth_d)|<a\}\\
&= \{\vec\boldom\in \cube^{ld}:\ |\det(\bth_1,\dots,\bth_d)|<a\ \& \ |s_1|> a^{(d-1)/d}\}\\
&\cup
\{\vec\boldom\in \cube^{ld}:\ |\det(\bth_1,\dots,\bth_d)|<a\ \&\  |s_1|\le a^{(d-1)/d}\}. 
\end{split}
\ene
Obviously, if $|s_1|> a^{(d-1)/d}$, then 
\bee
\meas\ \{\om_{11}\in \cube:\ |s_1\om_{11}+t_1|<a\}<2a^{1/d}
\ene
and thus 
\bee
\meas\ \{\vec\boldom\in \cube^{ld}:\ |\det(\bth_1,\dots,\bth_d)|<a\ \&\  |s_1|> a^{(d-1)/d}\}< 2a^{1/d}. 
\ene
Now let us write $s_1$ as a function of $\om_{22}$. It is again linear, say 
$s_1=s_2\om_{22}+t_2$, where $s_2$ and $t_2$ are functions of all the coordinates $\om_{jk}$ except $\om_{11}$ and $\om_{22}$. As before, we can write 
\bee
\bes
&\{\vec\boldom\in \cube^{ld}:\ |\det(\bth_1,\dots,\bth_d)|<a\ \& \ |s_1|\le a^{(d-1)/d}\}\\&=
\{\vec\boldom\in \cube^{ld}:\ |\det(\bth_1,\dots,\bth_d)|<a\ \& \ |s_1|\le a^{(d-1)/d}\ \& \ |s_2|> a^{(d-2)/d}\}\\
&
\cup
\{\vec\boldom\in \cube^{ld}:\ |\det(\bth_1,\dots,\bth_d)|<a\ \& \ |s_1|\le a^{(d-1)/d}\ \& \ |s_2|\le a^{(d-2)/d}\}
\end{split}
\ene
and estimate
\bee
\meas\ \{\vec\boldom\in \cube^{ld}:\ |\det(\bth_1,\dots,\bth_d)|<a\ \&\  |s_1|<a^{(d-1)/d}\ \&\  |s_2|> a^{(d-2)/d}\}< 2a^{1/d}. 
\ene
We carry on this process until we have to express $s_{d-1}$ as a function of $\om_{dd}$, at which stage we notice that we have $s_{d-1}=s_d\om_{dd}+t_d$ with 
\bee
s_d=\det(n_{jk})_{j,k=1}^d. 
\ene
This means that $|s_d|\ge 1$ and 
\bee
\bes
&\meas\ \{\vec\boldom\in \cube^{ld}:\ |\det(\bth_1,\dots,\bth_d)|<a\\
&\ \& \ |s_1|\le a^{(d-1)/d}\ \& \ |s_2|\le a^{(d-2)/d}\dots\ \& \ |s_{d-1}|\le a^{1/d} \}\le 2a^{1/d}. 
\end{split}
\ene
Summing all these inequalities, we arrive at
\bee
\meas\ \{\vec\boldom\in \cube^{ld}:\ |\det(\bth_1,\dots,\bth_d)|<a\}< 2da^{1/d}.
\ene
\enp

Now we notice that for each positive $N$ the number of non-degenerate matrices $\BN\in$ with $||\BN||_*= N$ is bounded above by $ld(2N)^{ld-1}$. Therefore, 
\bee
\bes
&\meas\ \{\vec\boldom\in \cube^{ld}:\exists \BN\in{\mathcal {ND}}(l,d), ||\BN||_*= N,\ |\det(\bth_1(\vec\boldom,\BN),\dots,\bth_d(\vec\boldom,\BN))|<a\}\\ 
&< 2^{ld}ld^2a^{1/d}N^{ld-1}.
\end{split}
\ene
Putting $a=B_0N^{-\hat\mu}$, we see that 
\bee
\bes
&
\meas\ \{\vec\boldom\in \cube^{ld}:\exists \BN\in{\mathcal {ND}}(l,d), ||\BN||_*= N,\\ &|\det(\bth_1(\vec\boldom,\BN),\dots,\bth_d(\vec\boldom,\BN))|<B_0||\BN||_*^{-\hat\mu}\}
< B_0^{1/d}2^{ld}ld^2N^{ld-\hat\mu/d-1}.
\end{split}
\ene
Choosing $\hat\mu$ so that $ld-\hat\mu/d-1<-1$ (or equivalently $\hat\mu>ld^2$), we obtain 
\bee
\bes
&
\meas\ \{\vec\boldom\in \cube^{ld}:\exists \BN\in{\mathcal {ND}}(l,d), \ |\det(\bth_1(\vec\boldom,\BN),\dots,\bth_d(\vec\boldom,\BN))|<B_0||\BN||_*^{-\hat\mu}\}\\
&< B_0^{1/d}2^{ld}ld^2\sum_{N=1}^{\infty} N^{ld-\hat\mu/d-1}=:B_0^{1/d}C(\hat\mu).
\end{split}
\ene
Therefore, 
\bee
\bes
&
\meas\ \{\vec\boldom\in \cube^{ld}:\\& \forall B_0>0\ \exists \BN\in{\mathcal {ND}}(l,d), \ 
|\det(\bth_1(\vec\boldom,\BN),\dots,\bth_d(\vec\boldom,\BN))|<B_0||\BN||_*^{-\hat\mu}\}=0. 
\end{split}
\ene
\enp

\section{Step zero}\label{section3}
{\newred
In this Section, we will perform the zeroth step of our procedure. This step is, essentially, a somewhat sophisticated perturbation theory applied to our operator. 
The main results here are Theorem \ref{Thm1} and its Corollary \ref{Cor4.8}, where we construct the zeroth approximation to the isoenergetic surface and list its properties. The error obtained here is $\la^{-1+\epsilon}$.
}
 
We start by fixing the value of $\hat\mu>ld^2$: for definiteness, we put for the rest of the paper 
\bees
\hat\mu:=ld^2+1
\enes 
and  
\bees
\mu:=ld^2+d+2. 
\enes
We also temporarily fix some positive value for $B_0$ (it will be fixed until Section \ref{8.1'}) such that $C(\hat\mu)B_0^{1/d}<1/100$. Then we define the good (and bad, respectively) sets of frequencies at step zero to be
\bee\label{CG0}
\CG^{\vec\boldom(0)}=\CG^{\vec\boldom(0)}_{B_0}:=\BOm_0(B_0)
\ene
and
\bee\label{CN0}
\CNN^{\vec\boldom(0)}=\CNN^{\vec\boldom(0)}_{B_0}:=[-1/2,1/2]^{ld}\setminus\BOm_0(B_0). 
\ene
The measure of the bad set is at most $C(\hat\mu)B_0^{1/d}$. At step one we will not change these sets, but the bad set will start growing (albeit slowly) starting from step two. We usually will consider $B_0$ being fixed and will skip writing $B_0$ as the subscript; at the very end of the proof we will return to checking how the measure of all important sets depends on $B_0$.  
We assume that $\vec\boldom\in \CG^{\vec\boldom(0)}_{B_0}$. 

\subsection{Resonant sets}

Consider a point $\bk=k\BPsi_j(\vec\phi)\in\R^d$. Given any $\bxi=\bk+\bn\vec\boldom\in \Z_{\bk}^l$ with $|\bn|$ being not very large, we want to exclude the possibility of $||\bk||$ and $||\bxi||$ being close to each other. The `closeness' will be measured by the parameter $B\in\R$ which later will be chosen of order $E^{1-\ep}$ (with positive but small $\ep$). 
Recall that we assume $\rho\in [E-1,E+1]$ and $E$ is fixed at the moment,
whereas $\rho$ is allowed to vary.  

{\newred
\begin{defn}\label{new:def3}
For each $\bn\in \Z^l$ we define the resonant  set $\CNN^{\angl}(\bn;B)\in\S$:
\bee\label{resonance}
\CNN^{\angl}(\bn;B)=\{\BPhi\in\S,\ |\lu\BPhi,\bn\vec\boldom\ru|\le B E^{-1}\}.
\ene
\end{defn}}

The next two estimates are based on simple geometry.
\bel
Suppose, $|\nk-\rho|\leq E^{-1}$, $\rho\in [E-1,E+1]$, $|\bn|\leq E^{1/5}$, $B\ge E^{4/5}$ 
and $\BPhi\not\in \CNN^{\angl}(\bn;B)$. Then for $E>d^5$ we have 
\bee\label{new:good1}
\bigm| ||\nk\BPhi+\bn\vec\boldom||^2-
\rho^2\bigm|> B
\ene
\enl
\bep
This follows from:
\bee
\bes
&||\nk\BPhi+\bn\vec\boldom||^2-
\rho^2\bigm|=|2\nk\lu\BPhi,\bn\vec\boldom\ru+||\bn\vec\boldom||^2+(k^2-\rho^2)|\\
\ge 
& 2 k |\lu\BPhi,\bn\vec\boldom\ru|-(||\bn\vec\boldom||^2+|k^2-\rho^2|)\ge 2 B - d |\bn|^2-2>B. 
\end{split}
\ene

\enp

\bel We have
\bee
\meas(\CNN^{\angl}(\bn;B))\leq  C ||\bn\vec\boldom||^{-1}E^{-1}B.
\ene
\enl
{\newred
\bep
This is an elementary geometry. 
\enp
}
Now the Diophantine property implies 
\bec
\bee
\meas(\CNN^{\angl}(\bn;B))\leq C B_0^{-1} |\bn|^{\hat\mu}E^{-1}B.
\ene
\enc

We also need to introduce another parameter $\tildeB\in\R$ which will measure how many vectors $\bn\vec\boldom$ we force to be `good'. This parameter later will be chosen to be of order $E^{\ep}$. Now, 
we define a resonant set {\newred (recall definition \eqref{newballs})}:
\bee
\CNN^{\angl}(B,\tildeB):=\cup_{\bn\in\Om'(\tildeB)}\CNN^{\angl}(\bn;B)\subset \S. 
\ene

We obviously have:
\bee
\meas(\CNN^{\angl}(B,\tildeB))\leq  C B_0^{-1}\tildeB^{l+\hat\mu}E^{-1}B. 
\ene

We will use the following convention throughout this paper. 
By ${\iota}$ we denote a finite ordered collection of integers, like $(3)$, or $(1,2,-5)$. We assume that adding $0$ at the end of any such collection does not change it and then 
we introduce a lexicographic order on the resulting equivalence classes, so that say $(3,-2)<(3,0)=(3)<(3,1)<(3,1,2)<(4)$. To each such collection $\iota$ we put into correspondence a positive number $\sigma_{\iota}$. We also assume that if $\iota_1<\iota_2$, then 
$C_0\sigma_{\iota_1}< \sigma_{\iota_2}$, where
\bee\label{7sigma}
C_0:=1000 l^3\mu^2d^2. 
\ene
We also assume that $\sigma_{1,d}<C_0^{-1}$, where $(1,d)$ is the biggest allowed value of $\iota$ for $\sigma_\iota$ in our paper (we write $\sigma_{1,d}$ for $\sigma_{(1,d)}$ for brevity). 
Since the total number of all possible values of $\iota$ we use in this paper will be finite, we always can find numbers $\sigma_{\iota}$ satisfying these properties. {\newred
The reason we introduce this convention is that we will have a tremendous amount of different exponents at each scale. This convention will indicate which of these exponents is larger and, moreover, will guarantee that a ratio between any two different exponents is sufficiently large to guarantee the estimates needed.} 



In what follows we will assume that $E\geq E_*$ where $E_*$ is an integer satisfying 
\bee\label{E_*}
E_*^{\sigma_0}>\max\{B_0^{-1},100\|V\|_\infty,\frac{100Q}{\sigma_0}\}.
\ene

Now, finally, we can define the zero step  resonant (or bad) set
\bee
{\CNN^{\angl (0)}}:=\cup_{s=0}^{d-1}\bigl[\CNN^{\angl}(E^{1-(l+\mu+1)\sigma_{0,s,1}},E^{\sigma_{0,s,1}})
\cup\CNN^{\angl}(E^{1-(l+\mu+1)\sigma_{1,s,1}},E^{\sigma_{1,s,1}})\bigr].  
\ene
The necessity for taking such a complicated set will be clear later. The index $(0)$ on the l.h.s. signifies that this is a bad set obtained at the  step zero of our procedure.  We also call
the complement of this set the first non-resonant (or good) set: 
\begin{equation} \label{W1} 
\CG^{\angl (0)}=\CG^{\angl (0)}(\rho):=
\S\setminus \CNN^{\angl (0)}. 
\end{equation} 
Clearly, it  is open. 
From the above we have the following estimate:
\bee\label{W1'}
\meas(\CNN^{\angl (0)})\leq  E^{-\sigma_{0,0,1}}.
\ene

We call the angles $\BPhi$ good, or non-resonant if they belong to $\CG^{\angl(0)}$ and $\vec\phi\in\Pi_m$ is non resonant if $\BPsi_m(\vec\phi)$ is good, a patch $\Pi _m$ being defined after (\ref{coordinates2}).

We have to introduce now the covering of $\S$ by patches of the zeroth level $\{\CA^{\BPhi(0)}_m\}$, $\CA^{\BPhi(0)}_m=\BPsi_m(\Pi^{(0)}_m)$, and the size of $\Pi^{(0)}_m$ is $\frac{E^{-2}}{10}$, {\newred which is the proper size to guarantee the stability of our estimates.} We will discuss these patches in more detail in the next section; now we just state that $\vec\phi\in\Pi_m^{(0)}$ is called good if $\BPsi_m(\vec\phi)$ is good. We also put
\bee
\CG^{\vec\phi (0)}_m:=\{\vec\phi\in\Pi_m:\ \BPsi_m(\vec\phi)\in\CG^{\angl (0)}\},
\ene
and then  we define 
$\CG^{\vec\phi (0)}_{m,\C}$ as the complex $E^{-(l+\mu+3)\sigma_{1,d-1,1}}$-neighbourhood of
$\CG^{\vec\phi (0)}_m$.


Let us list the properties of the good set that either directly follow from the definition, or are corollaries of standard perturbation results.

\bel\label{L:G1} 
\begin{enumerate}
\item For every $\BPhi \in \CG^{\angl (0)}$, $\rho\in[E-1,E+1]$, $s=0,...,d-1$, 
 and $\bn\in \Z^l$ with $0<|\bn|\le E^{\sigma_{0,s,1}}$ the following inequality
holds:
\begin{equation} \label{G1-1}
|\|\bk+\bn\vec\boldom\|^2-
\rho^2|\ge E^{1-(l+\mu+1)\sigma_{0,s,1}},
\end{equation} 
where $\bk=\rho\BPhi$.

\item The estimate above is stable in $E^{-(l+\mu+2)\sigma_{1,d-1,1}}$-neighbourhood of $\rho$. Recall that this means that if $\nka\in \C:
|\nka-\rho|<E^{-(l+\mu+2)\sigma_{1,d-1,1}}$ and $\vec\phi \in \CG^{\vec\phi(0)}_{j,\C}$, a slightly weaker
inequality holds for $\nbka:=\nka\BPsi_j(\vec\phi)$: 
\begin{equation}\label{G1-2}
\left|\|\nbka+\bn\vec\boldom\|^{2}-{\rho}^{2}\right|>\frac{E^{1-(l+\mu+1)\sigma_{0,s,1}}}{2}. 
\end{equation}

\end{enumerate}
Analogous statements hold when we replace $\sigma_{0,s,1}$ with $\sigma_{1,s,1}$.
\enl

\bec \label{C:L:G1} If $\vec\phi \in \CG^{\vec\phi(0)}_{j,\C}$, $|\nka-\rho|<E^{-(l+\mu+2)\sigma_{1,d-1,1}}$, $\nbka=\nka\BPsi_j(\vec\phi)$, and $z$ is on the circle
\begin{equation} \label{G-4}
\cont_0:=\{z:|z-\rho^{2}|=\frac14 E^{1-(l+\mu+1)\sigma_{1,d-1,1}}\},
\end{equation}
then  the following inequality holds for all $\bn\in \Z^l$ with $|\bn|\le E^{\sigma_{1,d-1,1}}$:
\begin{equation} \label{G1-5}
|\|\nbka+\bn\vec\boldom\|^2-
z|\ge \frac14 E^{1-(l+\mu+1)\sigma_{1,d-1,1}}.
\end{equation} 

\enc

\subsection{Perturbation results}

Let $r=1,2...$ and recall that $\cont_0$ is a contour given by \eqref{G-4} and $\nbka$ is a point lying close to $\bk$. We put $\hat K^{(0)}:=\Omega(\frac12 E^{\sigma_{1,d-1,1}})$ (this is the zeroth of {\it central cubes} that will be properly defined in section \ref{section7}) and define (recall definition \ref{proj})
\bee
\CP^{(0)}=\CP^{(0)}(\nbka):=\CP(\hat K^{(0)};\nbka)
\ene
and
\bee\label{newH_0}
H^{(0)}=H^{(0)}(\nbka):=H(\hat K^{(0)};\nbka)=\CP^{(0)}(\nbka)H(\nbka)\CP^{(0)}(\nbka)=:H_0^{(0)}+V_0^{(0)},
\ene
where 
\bee
H_0^{(0)}=H_0^{(0)}(\nbka):=\CP^{(0)}(\nbka)H_0(\nbka)\CP^{(0)}(\nbka)
\ene
and
\bee
V^{(0)}_0:=H^{(0)}(\nbka)-H_0^{(0)}(\nbka).
\ene

{\newred We will use the following objects that appear in successive iterations of resolvent identities:}

\begin{equation}\label{g}
 g^{(0)}_r({\nbka}):=\frac{(-1)^r}{2\pi
ir}\hbox{Tr}\oint_{\cont_0}\left((H^{(0)}_0(\nbka)
-z)^{-1}V^{(0)}_0\CP^{(0)}(\nbka)\right)^rdz
\end{equation} 
and
\begin{equation}\label{G}
G^{(0)}_r({\nbka})
:=\frac{(-1)^{r+1}}{2\pi
i}\oint_{\cont_0}\left((H^{(0)}_0(\nbka)
-z)^{-1}V^{(0)}_0\CP^{(0)}(\nbka)\right)^r\\
(H^{(0)}_0(\nbka)
-z)^{-1}dz.
\end{equation}
Note that $g^{(0)}_1(\nbka)=0$. {\newred Indeed, note that $V^{(0)}_0\CP^{(0)}(\nbka)$ has zeroes on the diagonal since the zeroth Fourier coefficient of $V$ vanishes. On the other hand, $(H^{(0)}_0(\nbka)
-z)^{-1}$ is a diagonal operator. Thus, the product of these operators has zeroes on the diagonal and a zero trace.}

 The coefficient
$g^{(0)}_2(\nbka)$ admits the following representation:
\begin{equation}
g^{(0)}_2(\nbka)
    =\sum _{\bn\in \Z^l\setminus \{0\}}| V_{\bn}| ^2(\|\nbka\|^{2}-
        \|{\nbka}+\bn\vec\boldom\|^{2})^{-1}.
 \end{equation}

\begin{thm} \label{Thm1} 
Suppose, $\vec\phi \in \CG^{\vec\phi(0)}_{j,\C}\cap\R^{d-1}$   and   $\nka\in\R$,
$|\nka-\rho|\leq E^{-(l+\mu+2)\sigma_{1,d-1,1}}$,
$\nbka=\nka\BPsi_j(\vec\phi)$. Then there exists a single {\newred (i.e. unique and simple)} eigenvalue of
$H^{(0)}({\bka})$ in the interval $I_0:=(
\rho^{2}-E^{1-(l+\mu+2)\sigma_{1,d-1,1}}, \rho^{2}+E^{1-(l+\mu+2)\sigma_{1,d-1,1}})$. It is given by the absolutely convergent
series:
\begin{equation}\label{eigenvalue}
\lambda^{(0)}({\nbka})=\lambda^{(0)}({\nbka};\rho)=\kappa^{2}+\sum\limits_{r=2}^\infty
g^{(0)}_r({\nbka}).
\end{equation} 
The coefficients $g^{(0)}_r({\nbka})$
satisfy the following estimates:
\begin{equation}\label{estg} 
|g^{(0)}_r({\nbka})|\ll 
E^{-(r-1)(1-(l+\mu+1)\sigma_{1,d-1,1})+l\sigma_{1,d-1,1}}. 
\end{equation}
 Moreover,
 \begin{equation}\label{estg_2}
|g^{(0)}_2({\nbka})|\ll E^{-2+2(l+\mu+1)\sigma_{0,0,1}+\sigma_{0}}. 
\end{equation} 
The
corresponding spectral projection is given by the series:
\begin{equation}\label{sprojector}
\E^{(0)}({\nbka})=\E_{\mathrm {unp}}({\nbka})+\sum\limits_{r=1}^\infty G^{(0)}_r({\nbka}),
\end{equation} 
where $\E_{\mathrm {unp}}({\bka})$ is the unperturbed spectral
projection (onto the linear span of $\be_{\bka}$). The operators $G^{(0)}_r({\nbka})$ satisfy the estimates:
\begin{equation}
\label{jan27}
\left\|G^{(0)}_r({\nbka})\right\|_1\ll E^{-r(1-(l+\mu+1)\sigma_{1,d-1,1})+l\sigma_{1,d-1,1}}.
\end{equation}
Matrix elements of $G^{(0)}_r({\nbka})$ satisfy the following
relations:
\begin{equation}
G^{(0)}_r({\nbka})_{\bn\bn'}=0,\ \ \mbox{if}\ \ \
rQ<|\bn|+|\bn'|. \label{zeros} 
\end{equation}
In particular, we have:
 \begin{equation}\label{perturbation}
\lambda^{(0)}({\nbka})=\nka^{2}+O\left(E^{-2+(3l+2\mu+2)\sigma_{1,d-1,1}}\right),
\end{equation}
\begin{equation}\label{perturbation*}
\left\|\E^{(0)}({\nbka})-\E_{\mathrm {unp}}({\nbka})\right\|_1\ll E^{-1+(2l+\mu+1)\sigma_{1,d-1,1}}.
\end{equation}
Matrix elements of the spectral projection $\E^{(0)}(\nbka)$ also satisfy
the estimate:
\begin{equation}\label{matrix elements}
\left|\E^{(0)}({\nbka})_{\bn\bn'}\right|<E^{-\dis^{(0)}(\bn,\bn')}\ \
\mbox{when}\ |\bn|>2Q \mbox{\ or }
|\bn'|>2Q,
\end{equation}
where we have defined 
$$\dis^{(0)}(\bn,\bn'):=(|\bn|+|\bn'|)(2Q)^{-1}.$$
The  coefficients $g^{(0)}_r({\nbka})$ and operators
$G^{(0)}_r({\nbka})$ can be analytically extended  as functions of $\vec \phi $ to $\CG^{\vec\phi(0)}_{j,\C}$  and   as functions of $\nka $ to the disk,
$|\nka-\rho|\leq E^{-(l+\mu+2)\sigma_{1,d-1,1}}$, $\nka\in\C$, the estimates \eqref{estg}, \eqref{estg_2}, \eqref{jan27} being preserved.
\end{thm}

\begin{proof} While \eqref{sprojector} and \eqref{G} are standard formulae based on expansion of the resolvent in
perturbation series, the representation \eqref{eigenvalue} and \eqref{g} may look a little bit less obvious though it still follows from similar arguments. The details can be found, for example, in \cite{Ka}, {\newred Theorem 2.1 or \cite{KS} Theorem 3.3. We also remark that later in our paper we prove a similar result (for the next step) in a more subtle situation, Theorem  \ref{Thm2}, so a reader who does not want to look at external sources, could find a proof there. In particular, the proof uses the fact that at this step all $\bn\in\Z^l$ that are not too large are good and \eqref{new:good1} is satisfied.}
\end{proof}

\begin{corollary}\label{Cor4.8}
Expansions \eqref{eigenvalue}, \eqref{sprojector} can be analytically extended as functions of $\vec \phi $ to $\CG^{\vec\phi(0)}_{j,\C}$  and   as functions of $\nka $ to the disk,
$|\nka-\rho|\leq E^{-(l+\mu+2)\sigma_{1,d-1,1}}$, $\nka\in\C$. The following estimate holds:
\begin{equation} \label{2021}
\partial_{\nka}\lambda^{(0)}({\nbka})=2\nka+O\left(E^{-2+(3l+2\mu+2)\sigma_{1,d-1,1}+(l+\mu+2)
\sigma_{1,d-1,1}}\right).
\end{equation}
\end{corollary}
The analyticity follows from that of $g^{(0)}_r({\nbka})$, $G^{(0)}_r({\nbka})$   and estimates \eqref{estg}, \eqref{estg_2}, \eqref{jan27}. Estimate \eqref{2021} follows from the Cauchy formula with respect to $\nka $ in the $ E^{-(l+\mu+2)\sigma_{1,d-1,1}}$-neighbourhood of $\rho $. Similar estimates can be written for all derivatives of $\lambda^{(0)}$ and $\E^{(0)}$ with respect to $\nka$ and $\vec \phi $.

Theorem \ref{Thm1} implies, in particular, that 
the function $\lambda^{(0)}(\nka)$ is increasing for real $\nka$ {\newred and $\vec\phi$}. Therefore, for $\vec\phi \in \CG^{\vec\phi(0)}_{j,\C}\cap\R^{d-1}$ the equation 
\bee
\lambda^{(0)}({\nka\BPsi_j(\vec\phi)})=\la=\rho^2
\ene
has a unique solution $\nka=\nka^{(0)}=\nka^{(0)}(\vec\phi;\rho)$. Let us denote 
\bees
\nbka^{(0)}=\nbka^{(0)}(\vec\phi;\rho):=\nka^{(0)}(\vec\phi)\BPsi_j(\vec\phi).
\enes
Then it follows from the definition that we have
\bee\label{tau}
\lambda^{(0)}(\nbka^{(0)})=\la. 
\ene
We can extend this definition in an obvious way to define the function 
\bee
\nka^{(0)}(\vec\phi):=\nka^{(0)}(\BPsi_j(\vec\phi))
\ene
whenever $\vec\phi\in\Pi_j$. 
Moreover, we can continue the function $\nka^{(0)}(\vec\phi)$ analytically (again, as locally convergent power series) 
to $\CG^{\vec\phi(0)}_{j,\C}$ so that \eqref{tau} is still satisfied there. Also, we have
\bee\label{new17}
|\nka^{(0)}(\vec\phi)-\rho|+|\nabla_{\vec\phi}(\nka^{(0)})|+|\nabla^2_{\vec\phi}(\nka^{(0)})|= o(E^{-2}); 
\ene
the proof is analogous to the proof of Lemma 2.11 in \cite{KL1}.

Let us consider the set of points in $\R^d$ given by the formula:
$\{\nbka=\nbka^{(0)} (\BPsi_j(\vec\phi);\rho), \ \ \vec\phi \in \CG^{\vec\phi(0)}_{j,\C}\cap\R^{d-1}\} $. This set is a slightly distorted 
sphere 
with holes. All the points of this set satisfy
the equation $\lambda^{(0)} (\nbka ^{(0)}(\BPhi;\rho))=\rho^2$.
We call this set the isoenergetic surface of the operator $H^{(0)}$ and denote it 
by ${\CD}^{(0)}(\rho)$. The ``radius"
$\nka^{(0)}(\vec \phi;\rho )$ of ${\CD}^{(0)}(\rho)$
monotonously increases with $\rho $.

Now we define the quasi-patches in $\bk$ of level $0$ (corresponding to a choice of $\rho$). First, it will be convenient to define $\nbka^{(-1)}(\BPhi,\rho):=\rho$. Then we put
\bee\label{CAbk0}
\CA^{\bk(0)}=\CA^{\bk(0)}(\BPhi^*):=\{\bk\in\R^d:\ \bk\|\bk\|^{-1}\in\CA^{\BPhi(0)}(\BPhi^*)\ \&\ 
|\|\bk\|-\nbka^{(-1)}(\BPhi,\rho)|<E^{-1}\}
\ene
and call this set $\CA^{\bk(0)}$ a (quasi-)patch in $\bk$ associated to $\CA^{\BPhi(0)}$.
The (quasi-)patches in $\bk$ at higher levels $n\ge 1$ will be defined analogously, once we have constructed isoenergetic functions  $\nbka^{(n)}(\BPhi,\rho)$ and patches in $\BPhi$ at higher levels: we will put
\bee\label{CAbk}
\CA^{\bk(n)}=\CA^{\bk(n)}(\BPhi^*):=\{\bk\in\R^d:\ \bk\|\bk\|^{-1}\in\CA^{\BPhi(n)}(\BPhi^*)\ \&\ 
|\|\bk\|-\nbka^{(n-1)}(\BPhi,\rho)|< r^{(n)} \},
\ene
where  $r^{(n)}$ is the size of $\CA^{\BPhi(n)}$. The set $\CA^{\bk(n)}$ will be called a  (quasi-)patch in $\bk$ associated to $\CA^{\BPhi(n)}$ (corresponding to a choice of $\rho$). 
Strictly speaking, we would have to put
 $Er^{(n)}$ to the RHS of \eqref{CAbk} to make it consistent with \eqref{CAbk0}, but we can afford not to do this (if $n\ge 1$) to simplify formulas a little bit. 

\section{Step zero: preparation for step one}\label{section4}

In this step, we will improve the error of our estimates from $\la^{-1+\epsilon}$ in \eqref{perturbation} to an exponentially small error  in \eqref{perturbation-2}. {\newred
We will split all the points of the lattice $\{\bk+\bn\vec\boldom, \bn\in\Z^l\}$ 
into good and bad groups. We will prove that the bad lattice points are grouped into clusters (definition \ref{new:clusters}), that these clusters form a {\em periodic} lattice of rank $<d$ and of size that is not too big (Lemmas \ref{1} and \ref{2'}), and that these clusters are well-separated (Lemma \ref{Z0}). These estimates will imply that our operator restricted to each of the clusters is monotone in certain parameter (Lemma \ref{new:5.26} and its Corollary \ref{new:5.27}). This, in turn, will allow us to control the number of poles of the resolvent of this restriction (Lemmas \ref{rescluster} and \ref{resclustershrink}). These lemmas are crucial in proving  
the main estimate of Step one, Theorem \ref{Thm2}, where our error will become exponentially small.
}

We assume that the spherical angle $\BPhi=\bk\|\bk\|^{-1}$ corresponding to our starting point $\bk$ is good at the previous step (i.e., $\BPhi\in\CG^{\BPhi(0)}$) and proceed  to reduce our operator $H(\bk)$ further.

\subsection{Pre-clusters}

In this section, the resonant zones will be characterised by two parameters, that we now denote by $L$ and $\tilde L$. The role of $L$ will be slightly different than in the previous section, but it will still be of order $E^{\epsilon}$; more precisely, we will assume that $E^{\sigma_0}\leq L,\,\tilde L\leq E^{\sigma_{1,d}}$. 
We also put 
\bee \label{new5.1}
r_{1,3}:=E^{\sigma_0},\ \ \ r_{1,2}:= \cs^{2} r_{1,3},\ \ \ r_{1,1}:= \frac{\cs}{10} r_{1,2},
 \ene
 where $\cs$ is a small positive constant to be defined later,  
 and try to construct an approximation of the eigenvalue $\lambda ^{(\infty)}(\bka)$ with an error $O(E^{-r_{1,1}})$.

Recall the convention stated in remark \ref{2.1} and {\newred definition \ref{new:def3}.} Our next definition is somewhat similar, only instead of talking about bad angles $\BPhi$, we will define bad vectors $\bn\in\Z^l$.  
Suppose, $\bxi\in\R^d$ and $\rho,L\in\R_+$. We put ($\CR$ stands for `resonant'):
\bee\label{badtheta}
\CMM=\CMM(L,\bxi)=\CMM(L,\bxi,\vec\boldom,\rho)
:=\{\bm\in\Z^l
:\ \ |\|\bxi+\bm\vec\boldom\|^2-
\rho^2
|\le L\}.
\ene
  An equivalent definition is this:
\bee
\CMM=\{\bm\in\Z^l
:\ \ \bxi+\bm\vec\boldom\in\BS\BL(\rho,L)\}
\ene
where by
\bee\label{spherical_layer}
\BS\BL(\rho,L)
:=\BB(0,\sqrt{\rho^2+L})\setminus\BB(0,\sqrt{\rho^2-L})
\ene
we have denoted a spherical layer.
Later, we will put 
\bee\label{bxibn}
\bxi=\bxi(\bn):=\bk+\bn\vec\boldom\in\Z^l_{\bk}, 
\ene
but at the moment we want to treat $\bxi$ on its own right. This means that estimates for $\bxi$ will be given for all $\bxi$ belonging to a patch of certain size, and then we will be using these estimates for $\bxi$ given by \eqref{bxibn} for various $\bn\in\Z^l$. 

\begin{defn}\label{newdef1}
We will say that $\bxi\in\R^d$  is $L$-bad if $\bxi\in\BS\BL(\sqrt{\rho^2-L},\sqrt{\rho^2+L})$ and $\bn\in\Z^l$ is $L$-bad if $\bxi=\bk+\bn\vec\boldom\in\BS\BL(\sqrt{\rho^2-L},\sqrt{\rho^2+L})$. These points are $L$-good otherwise. The definition of $L$-bad $\bn$ depends on the choice of the `initial point' $\bk$, and sometimes, in order to emphasize this, we will say that 
$\bn$ is $L$-bad with respect to $\bk$.
\end{defn}

 Of course, $\bm$ is $L$-bad with respect to $\bxi=\bk+\bn\vec\boldom$ if and only if  $\bn+\bm$ is $L$-bad with respect to $\bk$.

Before giving main definitions, let us formulate one important property of $\BS\BL$; this property can be proved using school-level geometry (see Figure~\ref{fig1} for illustration).

\bel
Suppose, $\GU$ is a hyperplane in $\R^d$ and $\bu\in\R^d$ is a vector. Suppose, 
\bee
\bx,\by\in\BS\BL (\rho,L)
\cap\bigl[\GU+\bu\bigr].
\ene
Denote by $\bx_{\GU}$ and $\by_{\GU}$ the projections of $\bx$ and $\by$ onto $\GU$. Then 
\bee\label{bxby}
\bigm|\|\bx_{\GU}\|-\|\by_{\GU}\|\bigm|\ll L^{1/2}.
\ene
\enl
\bep
We obviously have:
\bee
\bigm|\|\bx_{\GU}\|-\|\by_{\GU}\|\bigm|(\|\bx_{\GU}\|+\|\by_{\GU}\|)\le 2L. 
\ene
Together with the inequality $\bigm|\|\bx_{\GU}\|-\|\by_{\GU}\|\bigm|\le(\|\bx_{\GU}\|+\|\by_{\GU}\|)$ this concludes the proof. 
\enp

\begin{figure}[ht]
\centering
\includegraphics[scale=0.35]{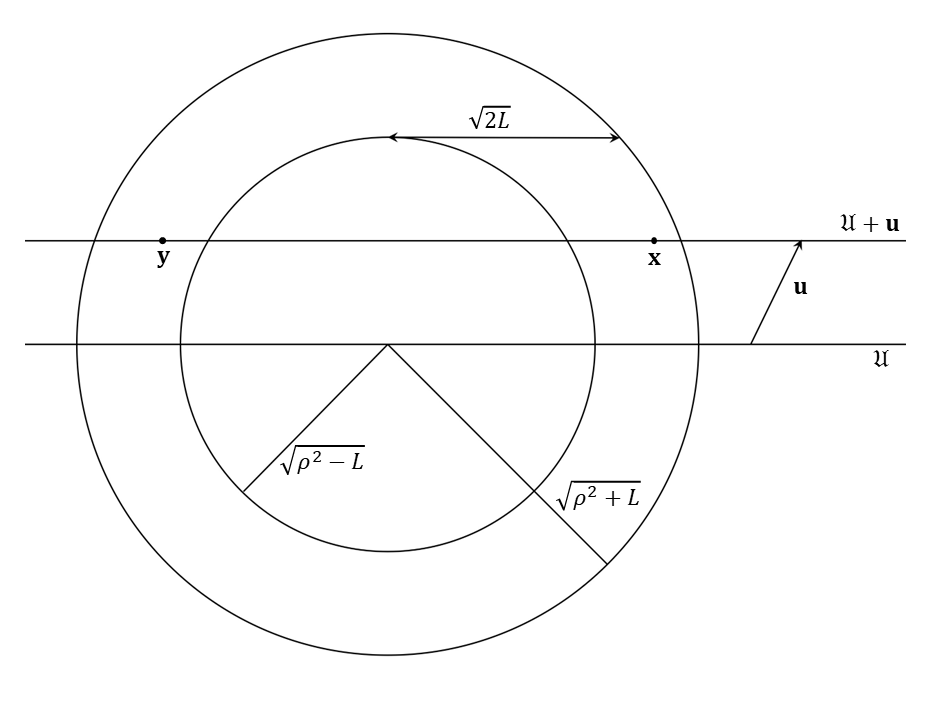}
\caption{}
\label{fig1}
\end{figure}

Now we proceed with our main definitions. 
\begin{defn}
We say that two frequencies $\bm_1\vec\boldom$ and $\bm_2\vec\boldom$ are {\it $(L,\tilde L)$-coupled} (sometimes adding `with respect to $\bxi$') and write $\bm_1\sim_{L,\tilde L}\bm_2$, if:

1. $\bm_1,\bm_2\in  \CMM(L,\bxi)$;

2. $|\bm_1-\bm_2|<\tilde L$. 

In this case we will also say that corresponding points $\bxi_j=\bxi+\bm_j\vec\boldom$ ($j=1,2$) are $(L,\tilde L)$-coupled.
\end{defn}
\ber
In this (and further) definitions we do not distinguish between a pair of frequencies $(\bm_1\vec\boldom,\bm_2\vec\boldom)$ and a pair of integer vectors  $(\bm_1,\bm_2)$ being coupled. 
\enr
\begin{defn}
We say that two frequencies $\bm\vec\boldom$ and $\bm'\vec\boldom$ are {\it $(L,\tilde L)$-conjugate} (with respect to $\bxi$), if 
there is a finite collection of points $\bm_0:=\bm,\bm_1,\dots,\bm_m=\bm'$
such that for each $j=0,\dots,m-1$ we have $\bm_j\sim_{L,\tilde L}\bm_{j+1}$. 
\end{defn}
\begin{defn}\label{new:clusters}
Given any point $\bxi\in\R^d$, we denote by $\BUps_{L,\tilde L}^{\Z}(\bxi)$ the collection of vectors from $\Z^l$ that are $(L,\tilde L)$-conjugate to $0\in\Z^l$ with respect to $\bxi$. If $\bxi$ is an $L$-good point (i.e. if $\bxi\not\in\BS\BL(\rho,L)
$), we define this set to be empty. 
We also denote
\bee
\BUps_{L,\tilde L}(\bxi):=\bxi+\BUps_{L,\tilde L}^\Z(\bxi)\vec\boldom:=\{\bxi+\bm\vec\boldom,\ \bm\in \BUps_{L,\tilde L}^\Z(\bxi)\}.
\ene
\end{defn}

Obviously, this definition gives us a disjoint union of $\BS\BL(\rho,L)
$ into equivalence classes: if $\bxi'\in \BUps_{L,\tilde L}(\bxi)$, then $\BUps_{L,\tilde L}(\bxi')=\BUps_{L,\tilde L}(\bxi)$. 

\begin{defn}
By the {\it rank} $\rank(\BUps^{\Z}_{L,\tilde L}(\bxi))$ we call the dimension of the linear subspace of $\R^l$ spanned by $\{\bm,\ \bm\in \BUps_{L,\tilde L}^\Z(\bxi)\}$.  
\end{defn}

{\newred
\ber
Here and below, of course, by the linear subspace in $\R^l$ spanned by integer vectors $\{\bm\}$ we mean the subspace spanned by the images of $\{\bm\}$ under the natural embedding of $\Z^l$ into $\R^l$.
\enr}

\bel\label{1}
Suppose, $E^{\sigma_0}\le L,\tilde L\le E^{\sigma_{1,d}}$. Then 
$$s:=\rank(\BUps_{L,\tilde L}^{\Z}(\bxi))<d,$$ 
$$\hbox{\rm diam}(\BUps_{L,\tilde L}(\bxi))<L\tilde L^{s\mu},$$ 
and the number of elements $\#\BUps_{L,\tilde L}^\Z(\bxi)$ satisfies 
 $$\#\BUps_{L,\tilde L}^\Z(\bxi)<L^s\tilde L^{(s^2+1)\mu}.$$
\enl
\bep
Let $\GV$ be a subspace of $\R^l$ generated by $s'\le l$ linearly independent vectors from $\Z^l$. We denote by $\BUps_{L,\tilde L}^\Z(\bxi;\GV)$ the collection of points $\bm'$ such that there exists a sequence $\bm_0:=0,\bm_1,\dots,\bm_m=\bm'$ of vectors from $\GV$ 
satisfying the usual property: for each $j=0,\dots,m-1$ we have $\bm_j\sim_{L,\tilde L}\bm_{j+1}$.  
We call the {\it rank} $\rank(\BUps^{\Z}_{L,\tilde L}(\bxi;\GV))$ the dimension of the linear span of all vectors in $\BUps_{L,\tilde L}^\Z(\bxi;\GV)$. 
Obviously, 
\bee
\rank(\BUps^{\Z}_{L,\tilde L}(\bxi;\GV))\le\dim\GV. 
\ene
We can also look at the linear subspace of $\R^d$:  
\bee
\GV\vec\boldom:=\{\bv\vec\boldom:\ \bv\in\GV\}\subset\R^d. 
\ene
\ber
Note that $\GV$ is a subspace of $\R^l$, whereas $\GV\vec\boldom$ is a subspace of $\R^d$. 
\enr
Let us prove the following statement:
\bel\label{2}
 Suppose, 
$\rank\BUps_{L,\tilde L}^\Z(\bxi;\GV)=\dim\GV$.  
Then we have:
\bee\label{projection}
||\bxi_{\GV\vec\boldom}||\ll L\tilde L^{s'\mu}. 
\ene
Moreover, 
\bee\label{projection1}
||\boldeta_{\GV\vec\boldom}||\ll L\tilde L^{s'\mu}  
\ene
for any $\boldeta\in(\bxi+\BUps_{L,\tilde L}^\Z(\bxi;\GV)\vec\boldom)=:\BUps_{L,\tilde L}(\bxi;\GV)$, 
\bee\label{diameter}
\diam(\BUps_{L,\tilde L}(\bxi;\GV))\ll L\tilde L^{s'\mu},
\ene
 and
\bee \label{number}
\#(\BUps_{L,\tilde L}^\Z(\bxi;\GV))\ll L^{s'}\tilde L^{({s'}^2+1)\mu}.
\ene
\enl
\bep 
The proof goes by induction in $s'=\dim\GV$. The case $s'=0$ is trivial, so we assume $s'\geq1$. 
Suppose, $s'=1$, so $\bxi\in\BUps_{L,\tilde L}(\bxi;\GV)$ (or, equivalently,  $\bxi\in\BS\BL(\rho,L)$).  Then there is a vector $\bm\in(\GV\cap\Om(\tilde L))$, $\bm\ne 0$, such that $\bm\in \BUps_{L,\tilde L}^\Z(\bxi)$ (otherwise the rank of $\BUps_{L,\tilde L}^\Z(\bxi;\GV)$ is zero). This means that $\boldeta:=\bxi+\bm\vec\boldom\in \BS\BL(\rho,L)$. 
Since $|\bm|\le \tilde L$, the diophantine condition implies that either $||\bxi_{\GV\vec\boldom}||\le ||\bm\vec\boldom||<\tilde L$ (in which case \eqref{projection} is trivial), or $||\bxi_{\GV\vec\boldom}||>||\bm\vec\boldom||$, in which case we have {\newred (since $\dim \GV\vec\boldom=1$):}
\bee
\bigm| ||\bxi_{\GV\vec\boldom}||-||\boldeta_{\GV\vec\boldom}|| \bigm|=
\bigm| ||\bxi_{\GV\vec\boldom}||-||\bxi_{\GV\vec\boldom}+(\bm\vec\boldom)|| \bigm|
=||\bm\vec\boldom||\ge \tilde L^{-\mu}. 
\ene
Therefore, \eqref{badtheta}  implies 
\bee
||\bxi_{\GV\vec\boldom}||\tilde L^{-\mu}\ll \bigm| ||\bxi_{\GV\vec\boldom}||^2-||\boldeta_{\GV\vec\boldom}||^2\bigm|{\newred
=\bigm| ||\bxi||^2-||\boldeta||^2\bigm|
}
\ll 
 L,
\ene
which implies 
\eqref{projection}. Since the same argument can be repeated for all $\boldeta\in\BUps_{L,\tilde L}(\bxi;\GV)$, this also implies \eqref{projection1}.  This in turn implies that $\diam( \BUps_{L,\tilde L}(\bxi))\ll L\tilde L^{\mu}$ and 
$$\#(\BUps_{L,\tilde L}^\Z(\bn;\GV))\ll L\tilde L^{2\mu}$$ 
(we use the Diophantine condition again).

Suppose now we have proved our statement for $\dim\GV=s'$; let us prove it for $\dim\GV=s'+1$. The first assumption of Lemma implies that there exists a sequence of points $\{\bxi_j=\bxi+\bm_j\}$, $j=0,\dots,m$ such that $\bm_0=0$, $\bm_j\in\CMM(L)$, $\bm_j-\bm_{j-1}\in\Omega(\tilde L)$, and 
$\spa\{\bm_j\}=\GV$. 
Let $n$ be the smallest index for which $\dim\spa(\bm_1,\dots,\bm_n)=s'+1$. This definition implies that $\GV':=\spa(\bm_1,\dots,\bm_{n-1})$ has dimension $s'$. Moreover, $\rank\BUps_{L,\tilde L}^\Z(\bxi;\GV')=s'=\dim\GV'$. Therefore, the induction assumption implies that 
$||(\bxi_{n-1})_{\GV'\vec\boldom}||\ll L\tilde L^{s'\mu}$. The result for $s'=1$ implies that 
$||(\bxi_{n-1})_{\spa((\bm_{n}-\bm_{n-1})\vec\boldom)}||\ll L\tilde L^{s'\mu}$. The strong diophantine condition implies that the angle between $(\bm_{n}-\bm_{n-1})\vec\boldom$ and $\GV'\vec\boldom$ is at least $\tilde L^{-\mu}$. Therefore, $||(\bxi_{n-1})_{\GV\vec\boldom}||\ll L\tilde L^{(s'+1)\mu}$. The induction assumption implies that $||\bxi-\bxi_{n-1}||=||(\bxi-\bxi_{n-1})_{\GV'\vec\boldom}||\ll L\tilde L^{s'\mu}$ and therefore $||(\bxi)_{\GV\vec\boldom}||\ll L\tilde L^{(s'+1)\mu}$. Since we can repeat the same arguments starting from any point $\boldeta\in\BUps_{L,\tilde L}(\bxi;\GV)$, this implies \eqref{projection1} and \eqref{diameter}. We postpone the proof of \eqref{number} 
until Lemma \ref{2'}, where we prove a stronger estimate \eqref{number'}. 
\enp
Now, in order to finish the proof of Lemma \ref{1}, we put $\GV$ to be a span of $\BUps _{L,\tilde L}^\Z(\bxi)$. Then $s'=\dim\GV=\rank\BUps_{L,\tilde L}^\Z(\bxi)=s$, and it remains to show that $s<d$. Suppose that $s\ge d$. Then $\GV\vec\boldom=\R^d$ and $||\bxi||=||\bxi_{\GV\vec\boldom}||\ll L\tilde L^{d\mu}$, which contradicts our assumption $|\ ||\bxi||^2-\rho^2|\le L$. 
\enp

\ber
This lemma and Strong Diophantine Condition imply that $\rank(\BUps^{\Z}_{L,\tilde L}(\bxi))$  is equal to $\rank(\BUps_{L,\tilde L}(\bxi))$ -- the dimension of the subspace of $\R^d$ spanned by $\{\bm\vec\boldom,\ \bm\in \BUps_{L,\tilde L}^\Z(\bxi)\}$. 
\enr

We will also need a statement slightly stronger than Lemma \ref{2}. Let us fix $\bxi$, $\rho$, 
$L$ and $\tilde L$, and denote 
\bee\label{GaZ-1}
\GV=\GV(\bxi)=\mbox{span}\big(\BUps _{L,\tilde L}^\Z(\bxi)\big)\subset\R^l, 
\ \ \dim\GV=\hbox{Rank}\,\BUps_{L,\tilde L}^\Z(\bxi)=:s.
\ene
We also define
\bee\label{GaZ}
\Ga^{\Z}(\bxi):=\GV(\bxi)\cap\Z^l.
\ene

Before formulating our estimate, we will prove a simple technical result.
\bel\label{sublattice}
Suppose, $\Ga\subset\R^s$ is a lattice of full rank, and $\ga_1,\dots,\ga_s\in\Ga$ are linearly independent with $||\ga_j||<\tilde L$. Then 
there is a basis $\bmu_1,\dots,\bmu_s$ of $\Ga$ with $||\bmu_j||<C(s)\tilde L$. The same result holds if we replace $||\cdot||$ by any other norm in $\R^s$. 
\enl
\bep
Since all norms in $\R^s$ are equivalent, it is enough to prove this statement for the Euclidean norm. 
We choose $\bmu_1$ to be any element of $\Ga\setminus\{0\}$ with the smallest length. Suppose, $\bmu_1,\dots,\bmu_{j}$ are chosen. Denote 
$\GV=\GV_j:=R(\bmu_1,\dots,\bmu_{j})$ and $\GU=\GU_j:=\GV_j^\perp$. For each $\ga\in\Ga$ we denote by $\ga_{\GU}$ its orthogonal projection onto $\GU$. We have $\ga_{\GU}=(\ga+\GV)\cap\GU$.  We denote the set of such projections (i.e. projections of all elements $\ga\in\Ga$ onto $\GU$) by $\Ga_{\GU}$. It is easy to see that $\Ga_{\GU}$
 is a lattice of full rank in $\GU$. Indeed, $\Ga_{\GU}$ is, obviously, an abelian  group, isomorphic to the quotient $\Gamma/Z(\bmu_1,...,\bmu_j)$, and this quotient group has rank $s-j$. Also, it is clear that the span of $\{\gamma_{\GU}\}$ equals $\GU$. Therefore, the collection $\{\gamma_{\GU}\}$ is a (periodic) lattice. 

Next, let us choose $\bnu_{j+1}$ to be any non-zero element of $\Ga_{\GU}$ with the smallest length (note that $\bnu_{j+1}$ does not always belong to $\Ga$). We define  $\bmu_{j+1}$ to be an element of $\Ga\cap (\bnu_{j+1}+\GV)$ such that $\bmu_{j+1}-\bnu_{j+1}$ is in the first Brillouin zone of the lattice $Z(\bmu_1,\dots,\bmu_{j})$.  This means:
\bee
||\bmu_{j+1}-\bnu_{j+1}||\le ||\bmu_{1}||+\dots+||\bmu_{j}||. \label{Sept15}
\ene 
We repeat this procedure to obtain a collection of vectors $\bmu_1\dots,\bmu_s$. 

Our construction implies that for all $j=1,...,s$ we have 
\bee\label{induction}
\Ga\cap R(\bmu_1,\dots,\bmu_{j})=Z(\bmu_1,\dots,\bmu_{j}).
\ene
Indeed, it is clear that the RHS is a subset of the LHS. The opposite inclusion is proved by induction. The base is obvious. Suppose, 
$\Ga\cap\GV_j=Z(\bmu_1,\dots,\bmu_{j})$ and put $\GV_{j+1}:=R(\bmu_1,\dots,\bmu_{j+1})$. Keeping in mind that $\bnu_{j+1}$ is a non-zero element of $\Ga_{\GU}$ with the smallest length, it is easy to see that $\ga_{\GU_j}=n\bnu_{j+1}$, $n\in \Z$, for any $\gamma\in \Ga\cap\GV_{j+1}$. Hence,
\bee
\Ga\cap\GV_{j+1}\subset \Ga\cap\cup_{n\in\Z}(n\bnu_{j+1}+\GV_j)=\Ga\cap \cup_{n\in\Z}(n\bmu_{j+1}+\GV_j)= \cup_{n\in\Z}(n\bmu_{j+1}+\GV_j\cap \Ga).
\ene
Using the induction assumption that $\Ga\cap\GV_j=Z(\bmu_1,\dots,\bmu_{j})$, we obtain $\Ga\cap\GV_{j+1}=Z(\bmu_1,\dots,\bmu_{j+1})$.
This proves \eqref{induction} for all $j=1,...,s$. Putting $j=s$ there shows that $\{\bmu_1\dots,\bmu_s\}$ is a basis of $\Ga$. 

It remains to prove  that every $||\bmu_j||\ll\tilde L$. 
Our construction implies $\|\bmu_1\|\leq \tilde L$. Assume we have proved that $||\bmu_{l}||\leq C\tilde L$, $l=1,...,j$. Note that there is a $\ga _i$ with a non-zero projection $\ga _{i,j}$ on $\GU _j$. By the definition of $\bnu_{j+1}$, 
 $||\bnu_{j+1}||\le||\ga_{i,j}||\le||\ga_{i}||<\tilde L$.
Using this inequality and (\ref{Sept15}), by induction, we obtain $||\bmu_{j+1}||<C\tilde L$. 
\enp
Now we formulate another estimate we will need later on. 
\bel\label{2'}
We have:
\bee\label{diameter'}
\diam((\Ga^\Z(\bxi)\cap\CMM(L,\bxi))\vec\boldom)\ll L\tilde L^{s\mu},
\ene 
\bee \label{number'}
\#(\Ga^\Z(\bxi)\cap\CMM(L,\bxi))\ll L^s\tilde L^{(s^2+1)\mu},
\ene
\bee \label{diameter''}
\diam(\Ga^\Z(\bxi)\cap\CMM(L,\bxi))\ll L\tilde L^{(s+2)\mu}.
\ene
\enl
\ber
Note that $\Ga^\Z(\bxi)\cap\CMM(L,\bxi)\subset\Z^l$, while  $(\Ga^\Z(\bxi)\cap\CMM(L,\bxi))\vec\boldom\subset\R^d$.
\enr
\bep
We know that $\Ga^\Z(\bxi)$ is an $s$-dimensional lattice containing $s$ linearly independent vectors $\ga_1,\dots,\ga_s$ of length smaller than $\tilde L$. By Lemma \ref{sublattice}, we can find a basis $\bmu_1,\dots,\bmu_s$ of $\Ga^{\Z}(\bxi)$ with 
$|\bmu_j|\ll\tilde L$. 
Lemma \ref{determinant} now implies that 
\bee \label{3.17}
\hbox{vol}\,(\bmu _1\vec\boldom,\dots,\bmu_s\vec\boldom)\gg\tilde L^{-\mu}. 
\ene
Here, $\hbox{vol}\,(\bmu _1\vec\boldom,\dots,\bmu_s\vec\boldom)$ means the volume of the corresponding $s$-dimensional parallelepiped. In particular, this implies that the lattice generated by $(\bmu _1\vec\boldom,\dots,\bmu_s\vec\boldom)$ is periodic. 
Next, we notice that \eqref{projection} together with \eqref{bxby}
implies \eqref{diameter'}. This, together with \eqref{3.17} and standard covering arguments yields \eqref{number'} and \eqref{diameter''}; recall that $\mu>d$. 

\enp

Now we divide all the bad points $\bn\in\CMM(E^{\sigma_{0}},\bxi)$ into sets of different ranks in the following way. For $s=0,1,\dots,d-1$ and $\oldj=0,1$ we denote $\BUps_{\oldj,s}^\Z(\bxi):=\BUps_{E^{\sigma_{\oldj,s}},E^{\sigma_{\oldj,s}}}^\Z(\bxi)$. Note that $\rank(\BUps_{\oldj,s}^\Z(\bxi))$ is a non-decreasing function of $s$ taking values $0,1,...,d-1$. It is a simple exercise that this implies that for each $\oldj$ there is at least one number $s=0,1,\dots,d-1$ with the property that 
 $\rank(\BUps_{p,s}^\Z(\bxi))=s$. We want to call one of these numbers $s$ the $\oldj$-rank of $\bxi$, and, as will be clear soon, the proper choice is the biggest of such numbers. This leads to the following definition: 

\begin{defn} \label{def3.11}
 Let $\oldj=0,1$. The biggest number $s$ satisfying 
 \bee\label{new:s}
 \rank(\BUps_{p,s}^\Z(\bxi))=s
\ene 
  is called the $\oldj$-rank of $\bxi$ and denoted by $\rank_\oldj(\bxi)$.  By $\CR_{\oldj,s}$ we denote the collection of all points $\bxi$ with $\oldj$-rank being equal to $s$. By 
 $\Ga^{\Z}(\bxi)=\Ga^{\Z}_p(\bxi)$ we denote the set given by \eqref{GaZ-1}-\eqref{GaZ}, assuming that $L=\tilde L=E^{\sigma_{\oldj,s}}$ with $s=\rank_\oldj(\bxi)$.  
\end{defn}
The immediate (and very important) consequence of this definition is that if $\rank_\oldj(\bxi)=s$, then $\rank(\BUps_{\oldj,s+1}^\Z(\bxi))=s$. {\newred 
Indeed, clearly $s=\rank(\BUps_{\oldj,s}^\Z(\bxi))\le\rank(\BUps_{\oldj,s+1}^\Z(\bxi))$ and we cannot have $\rank(\BUps_{\oldj,s+1}^\Z(\bxi))\ge s+1$ since then $s$ would not be the biggest number satisfying \eqref{new:s}.
}
This leads to the following corollaries (that hold for both $\oldj=0,1$; we will often omit mentioning the dependence on $\oldj$ in what follows): 
\bel\label{Z0}
Suppose, $\bxi\in\CR_{\oldj,s}$. Then $\BUps_{\oldj,s}(\bxi)\subset\CR_{\oldj,s}$ and 
$\rank\BUps_{\oldj,s}^\Z(\bxi)=s$. Any point $\bm$ which is within  
$E^{\sigma_{\oldj,s}}$ $\Z$-distance from $\BUps_{\oldj,s}^\Z(\bxi)$ 
is either inside $\BUps_{\oldj,s}^\Z(\bxi)$, or is $E^{\sigma_{\oldj,s}}$-good. 
\enl
\bel\label{Z}
Suppose, $\bxi\in\CR_{\oldj,s}$. Then any $E^{\sigma_{\oldj,s+1}}$-bad point $\bm$ that is $(E^{\sigma_{p,s+1}},E^{\sigma_{p,s+1}})$-conjugate to $0$ is inside $\Ga^\Z(\bxi)$. 
\enl
\begin{defn}\label{7BUps}
Given a point $\bxi\in\CR_{\oldj,s}$, we will call $\BUps_{p,s}^\Z(\bxi)$ the {\it primitive pre-cluster} corresponding to $\bxi$. Often we will omit writing the subscript $s$. 
\end{defn}
In order to proceed further with our procedure, we need to enlarge this pre-cluster a little bit. 
\begin{defn}\label{new:ext}
Suppose, $\bxi\in\CR_{\oldj,s}$. By the {\it extended pre-cluster} of $\bxi$ we call the following set:
\bee\label{ext}
\check\BUps^{\Z}_p(\bxi):=\CMM(E^{\sigma_{p,s}},\bxi)\cap\Gamma^\Z(\bxi)+
\Om(E^{\sigma_{\oldj,0}})\cap\Ga^\Z(\bxi). 
\ene
\end{defn}

Estimate \eqref{number'} implies
\bee\label{importantest}\#(\check\BUps^{\Z}_{\oldj}(\bxi))\ll E^{(s+(s^2+1)\mu)\sigma_{\oldj,s}}E^{l\sigma_{\oldj,0}}
\leq E^{d^2(l+\mu)\sigma_{\oldj,s}}.\ene
Definition {\newred \ref{new:ext}}, estimate \eqref{diameter''} and Lemmas \ref{Z0}, \ref{Z} imply
\bel\label{Z0new}
Suppose, $\bxi\in\CR_{\oldj,s}$. Any point $\bm$ which is within  
$E^{\sigma_{\oldj,s+1}}/2$ $\Z$-distance to $\check\BUps_{\oldj,s}^\Z(\bxi)$ 
is either inside $\check\BUps_{\oldj,s}^\Z(\bxi)$, or is $E^{\sigma_{\oldj,s}}$-good. In addition, if $\bm\not\in\Gamma^\Z(\bxi)$ then it is $E^{\sigma_{\oldj,s+1}}$-good.
\enl

\begin{defn}\label{def5.19}
We also define (see Figure~\ref{fig2} below for illustration) the super-extended pre-cluster 
\bee\label{superext}
\tilde \BUps^{\Z}_{\oldj}(\bxi):=\check\BUps^{\Z}_{\oldj}(\bxi)+\Om(E^{\sigma_{\oldj,s,1}}) 
\ene
and the intermediate pre-cluster:
\bee\label{haterext}
\hat \BUps^{\Z}_\oldj(\bxi):=\tilde\BUps_p^{\Z}(\bxi)\cap\Ga^\Z(\bxi)=\check\BUps^{\Z}_{\oldj}(\bxi)+
(\Om(E^{\sigma_{\oldj,s,1}})\cap\Ga^\Z(\bxi)). 
\ene
Obviously, we have $\check\BUps^{\Z}_\oldj(\bxi)\subset\hat\BUps^{\Z}_\oldj(\bxi)
\subset\tilde\BUps^{\Z}_\oldj(\bxi)$. 
 
\end{defn}

\begin{figure}[ht]
\centering
\includegraphics[scale=0.30]{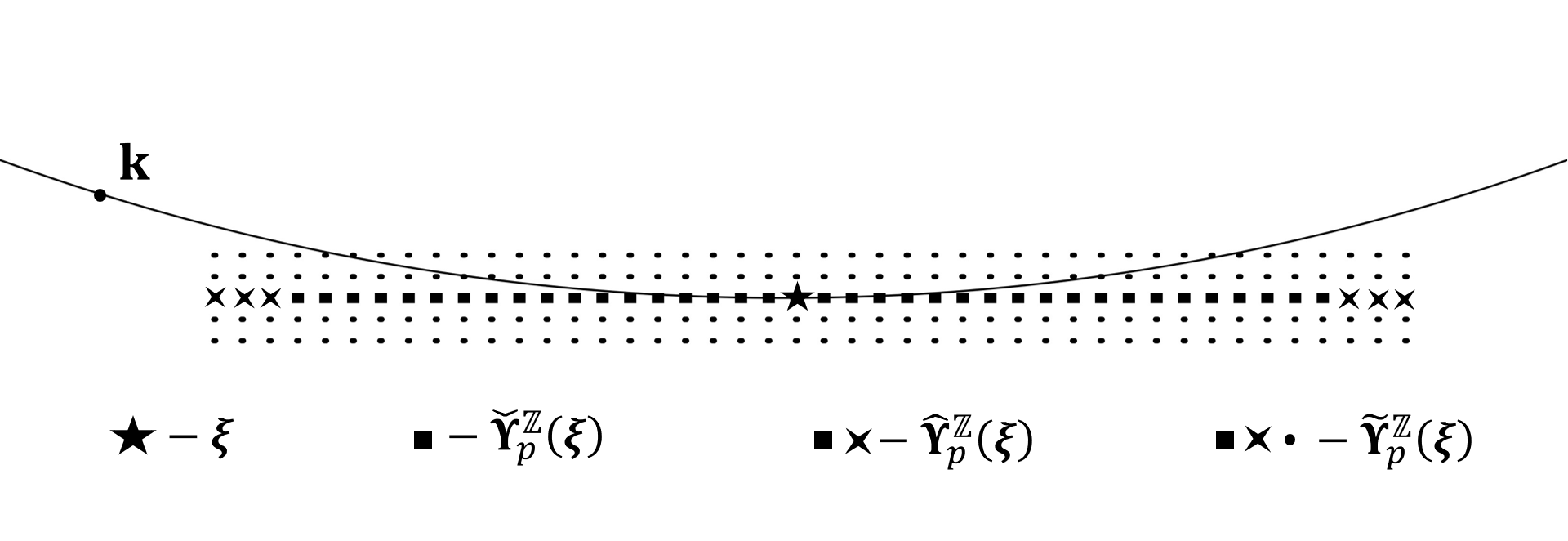}
\caption{}
\label{fig2}
\end{figure}

\subsection{Operator restricted to clusters\label{ORC}}


Fix $p=0,1$. 
Consider patches $\{\CA^{\rho,\bxi(0)}_{\tim}=\CA^{\rho(0)}_{\tim}\times \CA^{\bxi(0)}_{\tim}\}$ of size $\En^{-1}$ in both variables. For each such patch $\CA^{\rho,\bxi(0)}_{\tim}$ we denote by $\rho^*=\rho^*_{\tim}$ and $\bxi^*=\bxi^*_{\tim}$ the corresponding centres. From now on, we will consider only pre-clusters of the form $\BUps(\rho^*_{\tim},\bxi^*_{\tim})$; if we have any other point $(\bxi,\rho)\in\CA^{\rho,\bxi}_{\tim}$, then instead of a pre-cluster  $\BUps(\rho,\bxi)$ we will be considering $\BUps(\rho^*_{\tim},\bxi^*_{\tim})$.  

We also introduce the patches $\{\CA^{\BPhi(0)}_j\}$ of size $\frac{\En^{-2}}{10}$; note that a change of $E^{-2}$ of $\BPhi$ results in a change of $E^{-1}$ of $\bk$, so the size of the patches $\CA^{\BPhi(0)}$ is consistent with the size of $\CA^{\rho,\bxi(0)}$.  Consider a point $\BPhi$ from any fixed patch $\CA^{\BPhi(0)}_j$ with centre $\BPhi^*_j$. We will only be interested in the `good' patches, i.e. we assume that $\BPhi_j^*=\BPsi(\vec\phi)$ with $\vec\phi\in\CG^{\vec\phi(0)}_{\C}\cap\R^{d-1}$. We fix $\BPhi^*=\BPhi^*_j$ and $\rho^*=\rho_{\tim}^*$ for the rest of this subsection. Denote  $\bk^*=\bk^*_j:=\bka^{(0)}(\BPhi^*_j,\rho^*)=\ka^{(0)}(\BPhi^*)\BPhi^*$. 
Suppose now that for some $\bn_0\in\Z^l$ we have 
$\|\bk^*_j+\bn_0\vec\boldom\|\in  [\En-1,\En+1]$. 
By the discussion in section \ref{section1}, this implies that the entire shifted ball
$
\BB(\bk^*_j,\frac{E^{-1}}{10})+\bn_0\vec\boldom
$ 
is inside at least one patch
$\CA^{\bxi}_{m}$. Let $\bxi^*=\bxi^*_m=\bxi^*_m(\bn_0)$ be  its centre (if there are several patches in $\bxi$ with this property, we choose any of them). 
This gives us the possibility to define the ranks of the integer vectors $\bn_0$ as the ranks of the corresponding points $\bxi_m^*(\bn_0)$. Suppose, $\rank(\bn_0)=s$, so that $\bxi_m^*\in\CR_{p,s}$. Given any pre-cluster defined in the previous section (say, $\BUps^{\Z}_{\oldj}(\bxi_m^*)$), we define the corresponding cluster in the following way:
\bee\label{7CC}
\CC_\oldj(\bn_0,\bk)=\CC_\oldj(\vec\boldom,\rho^*,\bk,\bn_0):=\BUps^{\Z}_\oldj(\bxi_m^*)+\bn_0\subset\Z^l; 
\ene
the extended, intermediate and super-extended clusters are defined similarly, e.g.
\bee\label{7CC1}
\tilde\CC_\oldj(\bn_0,\bk)=\tilde\CC_\oldj(\vec\boldom,\rho^*,\bk,\bn_0):=\tilde\BUps^{\Z}_\oldj(\bxi_m^*)+\bn_0\subset\Z^l.
\ene
\ber\label{neweverything}
Note that our clusters cover all the bad lattice points: if $\bm\not\in\cup_{\bn}\CC_p(\bn)$, then $\bm$ is $E^{\sigma_{p,0}}$-good. 
\enr
\ber The reason why we need both types of clusters: $\tilde\CC_0$ and $\tilde\CC_1$ is rather technical. 
The extended and super-extended clusters $\tilde\CC_1$ are the zeroth level objects of what we will later define as the \Bourgain structure (see definition \ref{8.1}). 
We will use $\tilde\CC_0$, essentially, only in sections 5.3 and \ref{section5} when we construct the approximation $\la^{(1)}(\bk)$ to the eigenvalue $\la^{(\infty)}(\bk)$ based on the restriction of $H$ to the central cube $\hat K^{(1)}$.  Afterwards, when we run the induction, we will use only the clusters $\tilde\CC_1$ (except when considering `the prodigal sons' -- cubes that look very  similar to $\hat K^{(1)}$). We also comment that the proofs of various properties for $p=0$ and $p=1$ usually are completely the same. See Remark \ref{p=0} and a comment just before \eqref{nonsimple} for examples of specific statements where we needed to have construction with both $p=0$ and $p=1$ in place.  
\enr

The results from this section  imply that these extended clusters do not intersect pairwise; moreover, if $\bn_2\not\in\check\CC_p(\bn_1)$, then 
\bee\label{intersect}
d(\check\CC_p(\bn_1),\check\CC_p(\bn_2))\ge \frac12 \En^{\sigma_{\oldj,s+1}},\ \ s={\max\rank(\bn_1,\bn_2)}.
\ene 
Similarly, if $\bn_2\not\in\tilde\CC_p(\bn_1)$, then 
\bee\label{intersecttilde}
d(\tilde\CC_p(\bn_1),\tilde\CC_p(\bn_2))\gg  \En^{\sigma_{\oldj,s+1}},\ \ s={\max\rank(\bn_1,\bn_2)}.
\ene

Next, we 
fix a point $\bn_0$ (again, until the end of this subsection) and put for $\bk\in\BB(\bk^*_j,\frac{E^{-1}}{10})$ (recall definition \ref{proj}):
\bee
\bes
\CP=\check\CP&=\CP_{\oldj}(\bk)=\CP_{\oldj}(\vec\boldom,\rho,\BPhi^*,\bn_0;\bk):=\CP(\check\CC_{\oldj}(\bn_0,\bk);\bk),\\ 
\hat\CP&=\hat\CP_{\oldj}(\bk)=\hat\CP_{\oldj}(\vec\boldom,\rho,\BPhi_j^*,\bn_0;\bk):=\CP(\hat\CC_p(\bn_0,\bk);\bk),\\  
\tilde\CP&=\tilde\CP_{\oldj}(\bk)=\tilde\CP_{\oldj}(\vec\boldom,\rho,\BPhi_j^*,\bn_0;\bk):=\CP(\tilde\CC_p(\bn_0,\bk);\bk); \label{3.27}
\end{split}
\ene
this notation will be kept throughout this subsection. Obviously, $\check \CP<\hat \CP<\tilde \CP$. 
We will study how the resolvents of the operators $H(\check\CC_{\oldj}(\bn_0,\bk),\bk)=\CP(\bk)H(\bk)\CP(\bk)$ and $H(\tilde\CC_{\oldj}(\bn_0,\bk),\bk)=\tilde\CP(\bk)H(\bk)\tilde\CP(\bk)$ depend on $\bk$. Recall that the definition of all the clusters $\CC(\bn_0)$ was given with respect to a fixed point $\bxi_m^*$ and thus the clusters (as subsets of $\Z^l$) do not change when we vary $\bk\in \BB(\bk^*_j,\frac{E^{-1}}{10})$; these clusters stay fixed even when we make $\bk$ 
complex. 


Now we need to introduce coordinates $(\phi_1,\dots,\phi_{d-1})$ in $\Pi^{(0)}$ a bit more carefully than before. Namely, we request that $\phi_1$ measures the angle in the $2$-dimensional plane containing three points: $0$, $\bk^*$, and $\bxi_0:=\bk^*+\bn_0\vec\boldom$. More precisely, we choose the coordinates so that $\BPhi^*=\bk^*||\bk^*||^{-1}$ has coordinates 
$(0,\dots,0,1)$ (as we did before), and, moreover, $\bxi_0$ has coordinates $\|\bxi_0\|(\sin\al,0,\dots,0,\cos\al)$, where $\al$ is the angle between $\bk^*$ and $\bxi_0$. Once we have these coordinates, we put 
$\hat\phi:=(\phi_2,\dots,\phi_{d-1})$; thus, $(\phi_1,\dots,\phi_{d-1}):=(\phi_1 ,\hat\phi )$.
\bel\label{eta}
Suppose, the $\oldj$-rank ($\oldj=0,1$) of $\bn_0$ is $s\ge 1$ and $\BPhi^*\in\CG^{\BPhi(0)}$. Then 
\bee\label{theta}
\|\bn_0\vec\boldom\|\gg E^{1-(l+\mu+2)\sigma_{\oldj,s,1}},
\ene  
\bee\label{bphi1}
\al\gg E^{-(l+\mu+2)\sigma_{\oldj,s,1}},
\ene and 
\bee\label{bphi2}
\pi-\al\gg E^{-(l+\mu+2)\sigma_{\oldj,s,1}}.
\ene
\enl
\bep
Since $\bxi_0=\bk^*+\bn_0\vec\boldom$ is $E^{\sigma_{\oldj,s}}$-bad, we have
 \bee
\bigm|||\bk^*+\bn_0\vec\boldom||^2-||\bk^*||^2\bigm|<E^{\sigma_{\oldj,s}}.
\ene
Our assumption that $s\ge 1$ implies that there exists a vector $\bm\in \Z^l$, $0<|\bm|\le  E^{\sigma_{\oldj,s}}$ such that $\bxi_0+\bm\vec\boldom$ is also $E^{\sigma_{\oldj,s}}$-bad, i.e.
\bee
\bigm|||\bk^*+\bn_0\vec\boldom+\bm\vec\boldom||^2-||\bk^*||^2\bigm|<E^{\sigma_{\oldj,s}}. 
\ene
These two estimates imply that 
\bee\label{5.40}
|\lu \bk^*+\bn_0\vec\boldom,\bm\vec\boldom\ru|\ll E^{2\sigma_{\oldj,s}}.
\ene
Since $\BPhi^*\in\CG^{\BPhi(0)}$ and $|\bm|\le E^{\sigma_{\oldj,s,1}}$, we have
\bee\label{n5.41}
|\lu \bk^*,\bm\vec\boldom\ru|> E^{1-(l+\mu+1)\sigma_{\oldj,s,1}}. 
\ene
These two inequalities imply 
\bee
|\lu \bn_0\vec\boldom,\bm\vec\boldom\ru|> E^{1-(l+\mu+1)\sigma_{\oldj,s,1}}/2.
\ene
Since $||\bm\vec\boldom||\le|\bm|\le E^{\sigma_{\oldj,s,1}}$, we have
\bee
||\bn_0\vec\boldom||> E^{1-(l+\mu+2)\sigma_{\oldj,s,1}}/2,
\ene
which is \eqref{theta}. If instead of subtracting \eqref{n5.41} from \eqref{5.40}, we add them, we obtain 
\bee\label{-xi}
||\bxi_0+\bk^*||> E^{1-(l+\mu+2)\sigma_{\oldj,s,1}}/2.
\ene
Now we look at the triangle with the vertices $0$, $\bk^*$, and $\bxi_0$. Two sides of this triangle have lengths $||\bk^*||$ and $||\bxi_0||$ which are close to $E$, while the third side has length 
$||\bn_0\vec\boldom||\gg E^{1-(l+\mu+2)\sigma_{\oldj,s,1}}$. This implies that corresponding angle $\al$ is bounded below by $cE^{-(l+\mu+2)\sigma_{\oldj,s,1}}$. Similar argument using \eqref{-xi} gives us \eqref{bphi2}. 
\enp
\bec\label{eta1}
Suppose, $\bn\in\check\CC_\oldj(\bn_0)$. Denote $\bth=(\th_1,\dots,\th_d):=\bn\vec\boldom$. Then the first coordinate $\th_1$  satisfies $|\th_1|\gg  E^{1-(l+\mu+2)\sigma_{\oldj,s,1}}$. The sign of $\th_1$ is the same for all points $\bn\vec\boldom$ with $\bn\in\check\CC_\oldj(\bn_0)$.
\enc
\bep
This follows from lemma \ref{eta} and \eqref{diameter''}. 
\enp
Now we want to start moving $\bk$ so that the coordinate $\phi_1$ may become complex (to emphasise this, we use letter $\bka$ instead of $\bk$), but the modulus of $\bk$ is still the same (independent on $\vec\phi$); the rest of $\phi_j$ (when $j>1$) are assumed to be real. We also denote by $\bn$ any vector from $\check\CC_\oldj(\bn_0)$. Then
\bee\label{bk}
\bka=k(\phi_1,\dots,\phi_{d-1},\sqrt{1-\phi_1^2-\dots-\phi_{d-1}^{2}}).
\ene
We have:
\bee\label{square}
\Im||\bka+\bn\vec\boldom||^2_{\R}=2k(\th_1+k\Re\phi_1)\Im\phi_1+2k\Im((1-\phi_1^2-|\hat\phi|^2)^{1/2})(\th_d+k\Re((1-\phi_1^2-|\hat\phi|^2)^{1/2})).
\ene
Corollary \ref{eta1} implies that $|k\th_1|\gg E^{2-(l+\mu+2)\sigma_{\oldj,s,1}}$; the rest of the terms in the RHS of \eqref{square} are much smaller assuming $|\Re\phi_1|\le 10E^{-2}$ and $|\hat\phi|\le 10E^{-2}$. This implies the following result:
\bel\label{new:5.26}
Suppose, $s\ge 1$, $|\Re\phi_1|\le 10E^{-2}$ and $|\hat\phi|\le 10E^{-2}$. Then $\Im(||\bka+\bn\vec\boldom||^2_{\R})$ has the same sign for all $\bn\in\check\CC_\oldj(\bn_0)$ when the sign of $\Im\phi_1$ is fixed. Also, the following inequality holds:
\bee\label{1per}
|\Im(||\bka+\bn\vec\boldom||^2_{\R})|\gg E^{2-(l+\mu+2)\sigma_{\oldj,s,1}}|\Im\phi_1|. 
\ene
\enl
\bec\label{new:5.27}
If all conditions of the previous lemma are satisfied and $\vec\phi$ is real with $|\vec\phi|\leq 10E^{-2}$, then $H(\check\CC_{\oldj}(\bn_0);\bk)
=\CP_p(\bk)H(\bk)\CP_p(\bk)$
is monotone in $\phi_1 $ and all its eigenvalues $\lambda _q(\vec\phi )$ satisfy the estimates:
\begin{equation} 
\label{1per*} 
\left|\frac{\partial \lambda _q(\vec\phi ) }{\partial \phi_1 }\right|\gg  E^{2-(l+\mu+2)\sigma_{\oldj,s,1}}.
\end{equation}
\enc
\bep

Since $||\bk+\bn\vec\boldom||^2_{\R}$ is a holomorphic function of $\phi_1$, Cauchy-Riemann equation and 
inequality \eqref{1per} imply that 
\begin{equation}\label{1per**}
\left|\frac{\partial  ||\bk+\bn\vec\boldom||^2_{\R}}{\partial \phi_1 }\right|\gg  E^{2-(l+\mu+2)\sigma_{\oldj,s,1}}
\end{equation} 
for real $\phi_1 $ and the derivative has the same sign for all $\bn\vec\boldom\in\check\CC_\oldj(\n_0)$.  Taking into account that $V$ does not depend on $\phi_1$ {\newred and, thus, $\phi_1$ enters only the diagonal part of $H(\check\CC_{\oldj}(\bn_0);\bk)$,} we obtain \eqref{1per*}. 
\enp

Let us now fix the values of all angular variables except the first one: we assume that $\hat\phi$ is fixed and consider the determinant of the matrix $(H-\rho^2)(\check\CC_{\oldj}(\bn_0);\bka)=\CP_p(\bka)(H(\bka)-\rho^2)\CP_p(\bka)$ as a function of $\phi_1$. This determinant has a form $P_1(\phi_1)+P_2(\phi_1)\sqrt{1-|\hat\phi|^2-\phi_1^2}$, where $P_1$ and $P_2$ are polynomials of degree at most $2E^{d^2(l+\mu)\sigma_{\oldj,s}}$ (due to \eqref{importantest}). Obviously, this determinant vanishes only if we have
\bee
P_1^2-P_2^2(1-|\hat\phi|^2-\phi_1^2)=0,
\ene
and this is a polynomial in $\phi_1$. Thus, the number of poles $\phi_1$ of the resolvent $(H-\rho^2)^{-1}(\check\CC_{\oldj}(\bn_0);\bka)$
does not exceed $5E^{d^2(l+\mu)\sigma_{\oldj,s}}$. 
Let $\tilde\OO$ be the union of the discs of radius $E^{-2-\sigma_{\oldj,s,1}}$ around the poles in the $2E^{-2}$-neighbourhood of $0$. Obviously, there is a disc 
$D(r):=\{|\phi_1|\le r\}$ with $E^{-2}\leq r\leq\frac32 E^{-2}$ such that the boundary of $D(r)$ does not intersect $\tilde\OO$. Note that $D(r)$  contains the disk $|\phi_1|<E^{-2}$, and any connected component of $\tilde\OO$ which has common points with $D(r)$ is completely inside it. 

After this preparatory work, we move from $\bk$ to $\nbka^{(0)}$. Recall that for $\BPhi\in\CG^{\BPhi(0)}$, the point $\nbka^{(0)}$ is a vector having the same direction as $\BPhi$ such that the eigenvalue at the zeroth step $\la^{(0)}$ satisfies $\la^{(0)}(\nbka^{(0)})=\rho^2$. 
Similarly to \eqref{bk}, we write 
\bee
\nbka^{(0)}=\nka^{(0)}(\phi_1,\dots,\phi_{d-1},\sqrt{1-\phi_1^2-\dots-\phi_{d-1}^{2}}).
\ene
and consider this as a function of $\phi_1$ with all other variables fixed.

\bel\label{rescluster} Suppose, the $\oldj$-rank ($\oldj=0,1$) of $\bn_0$ is $s\ge 1$ and $\BPhi^*\in\CG^{\BPhi(0)}$. Then,  
as a function of $\phi_1$, the resolvent 
\bee
(H(\check\CC_{\oldj}(\bn_0);\nbka^{(0)})-\rho^2)^{-1}=(\CP_p(\nbka^{(0)})(H(\nbka^{(0)})-\rho^2)\CP_p(\nbka^{(0)}))^{-1}, \ \ s\geq1, 
\ene
has no more than $5E^{d^2(l+\mu)\sigma_{\oldj,s}}$ poles in $D(r)$. On the boundary $\partial D(r)$ we have the estimate
\begin{equation}\label{resclusterest}
\|(H(\check\CC_{\oldj}(\bn_0);\nbka^{(0)})-\rho^2)^{-1}\|\ll  E^{2(l+\mu+3)\sigma_{\oldj,s,1}}.
\end{equation}
\enl
\begin{proof} We start with the corresponding statements for $\bk$ and prove that
\begin{equation}\label{resclusterest11}
\|(H(\check\CC_{\oldj}(\bn_0);\bk)-\rho^2)^{-1}\|\ll E^{2(l+\mu+3)\sigma_{\oldj,s,1}}.
\end{equation}
Suppose, $\phi_1\in \partial D(r)$. 
 If $|\Im \phi_1 |>E^{-2-(l+\mu+4)\sigma_{\oldj,s,1}}$, then \eqref{resclusterest11} follows from \eqref{1per}. Suppose $|\Im \phi_1 |\leq E^{-2-(l+\mu+4)\sigma_{\oldj,s,1}}$. 
 Put  $\phi _*:=\Re \phi_1 $. 
 Then $\phi_*$ is at least $\frac12 E^{-2-\sigma_{\oldj,s,1}}$ away from the nearest pole. Thus (see \eqref{1per*}), 
\bee 
 \|(H(\check\CC_{\oldj}(\bn_0);\bk(\phi _*,\hat\phi))-\rho^2)^{-1}\|\ll E^{(l+\mu+3)\sigma_{\oldj,s,1}}. 
\ene 
 Using perturbative arguments again, we obtain
 \bee
\|(H(\check\CC_{\oldj}(\bn_0);\bk(\phi _1,\hat\phi))-\rho^2)^{-1}\|\ll E^{(l+\mu+3)\sigma_{\oldj,s,1}}
\ene
 and, hence, \eqref{resclusterest11} holds. Now, the statements of the lemma follows from \eqref{new17} and the standard perturbation arguments.
 \end{proof}

Now we will extend the results of Lemma \ref{rescluster} to a bigger projections $\tilde\CP$.
The following notation we will use only in this subsection: for $\Lambda\subset \Z^l$ we denote
\bee\label{d(n)}
\dis(\Lambda):=d(\Lambda,\check\CC_p(\bn_0,\bk))=\min\{|\bn-\bm|:\ \ \bn\in\Lambda,\ \bm\in\check\CC_p(\bn_0,\bk)\}.
\ene
We also put 
\begin{equation}\label{trivial3}
\dis(\Lambda,\Lambda'):=\frac{\dis(\Lambda)+\dis(\Lambda')}{20Q}.
\end{equation}

\begin{lem}\label{resclustershrink}
As a function of $\phi_1$, the resolvent 
\bee
(H(\tilde\CC_{\oldj}(\bn_0);\nbka^{(0)})-\rho^2)^{-1}=(\tilde \CP(\nbka^{(0)})(H(\nbka^{(0)})-\rho^2)\tilde \CP(\nbka^{(0)}))^{-1}, \ \ s\geq1, 
\ene
has no more than $5E^{d^2(l+\mu)\sigma_{\oldj,s}}$ poles in $D(r)$. Let $\OO=\OO(\varepsilon)$ be the union of the discs of radius $\varepsilon<E^{-2}$ around each pole. Then on the boundary of $\OO$ we have
\begin{equation}\label{resclusterestshrink}
\|(H-\rho^2)^{-1}(\tilde\CC_{\oldj}(\bn_0);\nbka^{(0)})\|\leq  E^{2(l+\mu+3)\sigma_{\oldj,s,1}}\left(\frac{E^{-2}}{\varepsilon}\right)^{5E^{d^2(l+\mu)\sigma_{\oldj,s}}}.
\end{equation}
Moreover, on the boundary of $\OO$ the following estimate for the truncated resolvent holds. Let 
$\Lambda$ and $\Lambda'$ be two subsets of $\tilde\CC_{\oldj}(\bn_0)$ satisfying 
 $\dis(\Lambda)+\dis(\Lambda')>20Q$.  Then  
\begin{equation}\label{trivial2}
\begin{split} &
\left \|\CP(\Lambda';\nbka^{(0)})
(H(\tilde\CC_{\oldj}(\bn_0);\nbka^{(0)})-\rho^2)^{-1}
\CP(\Lambda;\bka^{(0)})
\right\|\leq \cr & E^{2(l+\mu+3)\sigma_{\oldj,s,1}}\left(\frac{E^{-2}}{\varepsilon}\right)^{5E^{d^2(l+\mu)\sigma_{\oldj,s}}}
E^{-\sigma_{\oldj,s}\dis(\Lambda,\Lambda')}+E^{-\sigma_{\oldj,s}}.
\end{split}
\end{equation}
 \end{lem}

\begin{proof}
For simplicity, we will omit writing $\nbka^{(0)}$ as an argument during the proof of this lemma.
The plan of the proof is as follows. To begin with, we are going to obtain estimates assuming that $\phi_1$ is on the boundary of $D(r)$ and then, we will apply abstract lemma \ref{complexlemma} from Appendix 5. So, suppose that $\phi_1$ is on the boundary of $D(r)$. 

Let us denote  
$$ H':=\CP H\CP+(\tilde \CP-\CP)H_0(\tilde \CP-\CP)$$ 
($\CP$ and $\tilde\CP$ were defined in \eqref{3.27}). 
Then we obviously have: 
$$
\tilde \CP H\tilde \CP=H'+W,
$$
where we have denoted
\bee
W:=\tilde \CP V\tilde \CP-\CP V\CP.
\ene
We also denote 
$$A:=-(H'-\rho^2)^{-1}W(H'-\rho^2)^{-1}.$$
Let us prove first the following statement:
\bel \label{A}
The following estimate holds when $\phi_1\in\partial D(r)$:
\begin{equation}
\label{||A||2hat}
\|A\|<E^{-\sigma_{\oldj,s}}.
\end{equation}
\enl
\bep
To prove \eqref{||A||2hat} it suffices to check
\begin{equation}\label{||A||2-2hat}
\|(\tilde\CP-\CP)A(\tilde \CP-\CP)\|<\|V\|E^{-2\sigma_{\oldj,s}}
\end{equation}
and
\begin{equation}\label{||A||2-3hat}
\|(\tilde \CP-\CP)A\CP\|<6\|V\|E^{-2\sigma_{\oldj,s}}.
\end{equation}
Estimate \eqref{||A||2-2hat} follows from Lemma \ref{Z0new}, so we proceed to 
 \eqref{||A||2-3hat}. Lemmas \ref{Z0new} and \ref{rescluster} imply 
\bee\label{est1}
\|(\tilde \CP-\hat\CP)A\CP\|\ll\|V\|E^{-\sigma_{\oldj,s+1}+2(l+\mu+3)\sigma_{\oldj,s,1}},
\ene
and thus what remains is to estimate $(\hat\CP-\CP)A\CP$. 
We represent $(H'-\rho^2)^{-1}\CP$ using multiple resolvent identities as follows:
\begin{align}\label{May11-13bhat}
&(H'-\rho^2)^{-1}\CP=
\sum_{t=0}^{R_0}\left(-(H_0-\rho^2)^{-1}\CP V\CP\right)^t(H_0-\rho^2)^{-1}\CP+\cr &
\left(-(H_0-\rho^2)^{-1}\CP V\CP\right)^{R_0+1}(H'-\rho^2)^{-1}\CP, 
\end{align}
where we put 
\bee\label{R_0}
R_0:=[E^{\sigma_{\oldj,0}}Q^{-1}]-2.
\ene
This implies (remember that all projections involved commute with $H_0$): 
\begin{equation}\label{May11-13bhat1}
(\hat\CP-\CP)W\CP(H'-\rho^2)^{-1}\CP=
\sum_{t=0}^{R_0}F_t+
(\hat\CP-\CP)W\CP\left(-(H_0-\rho^2)^{-1}\CP V\CP\right)^{R_0+1}(H'-\rho^2)^{-1}\CP,
\end{equation}
where we have defined
\bee
F_t:=(\hat\CP-\CP)W\CP\left(-(H_0-\rho^2)^{-1}\CP V\CP\right)^{t}(H_0-\rho^2)^{-1}\CP.
\ene
Therefore,  
\begin{equation}\label{step25hat}
\|(\hat\CP-\CP)W\CP(H'-\rho^2)^{-1}\|\leq 
\sum_{t=0}^{R_0}\left\|F_t\right\|+
\left\|(\hat\CP-\CP)W\left((H_0-\rho^2)^{-1}\CP W\CP\right)^{R_0+1}\right\|\|(H'-\rho^2)^{-1}\CP\|.
\end{equation}
Note that matrix elements $(F_t)_{\bn\bn'}$ are
equal to zero if $|\bn-\bn'|>Q(t+1)$  
(see \eqref{V_q=0}).
Thus, the only non-trivial elements $(F_t)_{\bn\bn'}$ should satisfy
 $$\bn\in\hat\CC_\oldj(\bn_0)\setminus\check\CC_\oldj(\bn_0),\ \ \
\bn'\in\check\CC_\oldj(\bn_0), \ \ |\bn-\bn'|\leq Q(t+1). $$
In particular, we have $\bn,\bn'\in\Ga^\Z(\bxi_0)+\bn_0$ (recall that $\bxi_0=\bk^*+\bn_0\vec\boldom$) and  $|\bn-\bn'|\leq E^{\sigma_{\oldj,0}}-1$ (this follows from our choice \eqref{R_0}). 
Definition \eqref{ext} implies $\bn\not\in \CR^{\bm}(E^{\sigma_{p,s}},\bxi_0)$. Now, 
 $|\nbka^{(0)}-\bk|=O(E^{-1})$ implies that $\left|||\nbka^{(0)}+\bn'\vec\boldom||_{\R}^{2}-\rho^2\right|>E^{\sigma_{p,s}}/2$. Therefore, all the matrix elements $\left((H_0-\rho^2)^{-1}\CP V\CP\right)_{\bn_1,\bn_2}$ that give a non-zero contribution to 
$F_t$ must satisfy the same assumptions as $\bn'$, namely, $\bn_j\in \check\CC_\oldj(\bn_0)$, but within $\Z$-distance less than $Q(t+1)$ from $\bn_j$ there should be  a point from $\hat\CC_\oldj(\bn_0)\setminus\check\CC_\oldj(\bn_0)$ ($j=1,2$). This means that such matrix elements satisfy 
$$\bigm|\left((H_0-\rho^2)^{-1}\CP V\CP\right)_{\bn_1,\bn_2}\bigm|\ll \|V\|E^{-\sigma_{\oldj,s}}.$$
All this implies 
that for $t\leq R_0$ we have:
\begin{equation}\label{Br}
\begin{split}& \|F_t\|\leq 
(2\|V\|E^{-\sigma_{\oldj,s}})^{t+1},\cr &
\left\|(\hat\CP-\CP)W\left((H_0-\rho)^{-1}\CP W\CP\right)^{R_0+1}\right\|\leq
\|V\|\big(2\|V\|E^{-\sigma_{\oldj,s}}\big)^{R_0+1}.
\end{split}
\end{equation} 
This and \eqref{resclusterest} imply 
\bee
\bes
& 
 \|(\hat\CP-\CP)W\CP(H'-\rho^2)^{-1}\|\leq\\
 &
\sum_{t=0}^{R_0}
(2\|V\|E^{-\sigma_{p,s}})^{t+1}+\|V\|\big(2\|V\|E^{-\sigma_{p,s}}\big)^{R_0+1}E^{2(l+\mu+3)\sigma_{\oldj,s,1}}.
\end{split}
\ene

By lemma \ref{Z0new} and definition \ref{def5.19}, all points in 
$\hat\CC(\bn_0)\setminus  \check\CC(\bn_0)$ are $E^{\sigma_{p,s}}$-good. 
This means 
\bee\label{new5.74}
\|(\hat\CP-\CP)(H'-\rho^2)^{-1}\|\leq 2E^{-\sigma_{\oldj,s}}, 
\ene
which implies  \eqref{||A||2-3hat}. Here we also assumed that $E>E_*$,  so that, in particular,  
$$\sigma_{p,0}E^{\sigma_{\oldj,0}}Q^{-1}>1.$$ 
This finishes the proof of lemma \ref{A}. 

\enp
Now we go back to the proof of Lemma \ref{resclustershrink}. 
Let us consider the perturbation series 
\begin{equation}\label{step2dvahat}
(\tilde\CP(H-\rho^2)\tilde\CP)^{-1}=\sum_{t=0}^\infty(H' -\rho^2)^{-1}\left(-W(H'-\rho^2)^{-1}\right)^t.
\end{equation}
Estimates \eqref{resclusterest} and \eqref{||A||2hat} imply that when $\phi_1\in\partial D(r)$, we have 
\begin{equation}\label{step2raz*hat}
\left\|(\tilde\CP(H-\rho^2)\tilde\CP)^{-1}\right\| \ll E^{2(l+\mu+3)\sigma_{\oldj,s,1}}.
\end{equation}

Estimate \eqref{resclusterestshrink} (for $\phi_1\in\partial \OO$) now follows from Lemma \ref{complexlemma} if we take into account the estimate for the number of poles (see Lemma~\ref{rescluster}). We notice that perturbative arguments also imply that the number of poles inside $D(r)$ is the same for $H(\check\CC_{\oldj}(\bn_0))$ and $H(\tilde\CC_{\oldj}(\bn_0))$.

It remains to estimate the truncated resolvent. Let us assume, as we can without loss of generality, that ${\dis}(\Lambda')\leq{\dis}(\Lambda)$. We have 
\begin{equation}\label{decomp}
\begin{split}
& 
\CP({\Lambda'})(\tilde\CP(H-\rho^2)\tilde\CP)^{-1}\CP({\Lambda})=
\sum\limits_{t=0}^R\CP({\Lambda'})( H'-\rho^2)^{-1}\left(-W(H'-\rho^2)^{-1}\right)^t\CP({\Lambda})+
\cr &
\CP({\Lambda'})(\tilde\CP(H-\rho^2)\tilde\CP)^{-1}\left(-W(H'-\rho^2)^{-1}\right)^{R+1}\CP({\Lambda}),
\end{split}
\end{equation}
where $R:=\left[\frac{\dis({\Lambda})}{Q}\right]-2$. Since $W_{\bn_1\bn_2}=0$ when $|\bn_1-\bn_2|>Q$,  it is easy to see that the first term in the right hand side  is holomorphic when $\phi_1\in D(r)$. Indeed, if $|\bm|\le QR$, then we obviously have that $(\Lambda+\bm)\cap\check\CC_{\oldj}(\bn_0)=\emptyset$. 
Now, estimate \eqref{trivial2} follows from Lemma~\ref{Z0new},  \eqref{resclusterestshrink}, and \eqref{new5.74}. 

\end{proof}

The case $s=0$ is much simpler and completely analogous to Lemma 3.22 from \cite{KS}. Here we formulate the corresponding result (in a shorter and more convenient form) without proof.

\begin{lem}\label{s=0}
Let $s=0$. As a function of $\phi_1$, the resolvent $(H-\rho^2)^{-1}(\tilde\CC_{\oldj}(\bn_0);\nbka^{(0)})=(\tilde \CP(\nbka^{(0)})(H(\nbka^{(0)})-\rho^2)\tilde \CP(\nbka^{(0)}))^{-1}$ has at most $2$ poles in the disc $|\phi_1|<E^{-2}$. If $\varepsilon\leq E^{-2-\sigma_{\oldj,0}}$ is the distance to the nearest pole, then we have: 
\begin{equation}\label{resclusterestshrink0}
\|(H(\tilde\CC_{\oldj}(\bn_0);\nbka^{(0)})-\rho^2)^{-1}\|\leq \|\bn_0\vec\boldom\|^{-1}
\varepsilon^{-2}.
\end{equation}
\end{lem}

\subsection{Operator restricted to several clusters}\label{section5.3}
Now, for the rest of this section we put $p=0$. In this subsection, we consider $\bk$  fixed (with all the clusters defined with $\bk$ as a `base point') and $\bkappa$ is within the distance $E^{-1}$ from $\bk$. The integer vector $\bn_0$ is no longer fixed, however. Instead, we consider all clusters of the form $\tilde\CC(\bn)$ defined in the previous subsection which have a non-trivial intersection with the ball $\Om(E^{r_{1,2}}/2)$ and label them as $\{\tilde\CC_{l}\}_{m=1}^M$. We also put 
\bee
\CP_m=\CP_{m}(\bk):=\CP(\tilde\CC_{m},\bk).
\ene
and
\bee\label{CP}
\CP_{\patch}=\CP_{\patch}(\bk):=\sum_m\CP_m
\ene
(`res' stands for `resonant'). 
\ber
Note that the dependence of $\CP_{\patch}$ on $\bk$ is absent in certain sense. Namely, suppose that $\bk'\in\BB(E^{-1},\bk)$. Then there is a natural isometry between $\GH(\bk)$ and $\GH(\bk')$ (the shift by $\bk'-\bk$). This isometry presents the unitary equivalence between $\CP_{\patch}(\bk)$ and $\CP_{\patch}(\bk')$. 
\enr
\ber
To avoid discussion about the clusters near the boundary of $\Omega(E^{r_{1,2}}/2)$, we will consider the following extended ball:
\bee\label{extball}
\hat K^{(1)}=\Om_{\ext}(0,E^{r_{1,2}}/2):=\Om(E^{r_{1,2}}/2)\bigcup \cup_m\tilde\CC_{m}
\ene
(called  the central cube at level one) 
and later consider the restriction of $H$ onto this central cube. 
We will use similar convention in our further steps as well. We also notice that for any $\bn\in\Omega(E^{r_{1,2}})$ we have $\|\bn\vec\boldom\|\geq E^{-\mu r_{1,2}}$.
\enr
The following result immediately follows from \eqref{intersecttilde}:
\bel\label{PHP*}
We have, for $m\ne m'$:
\bee
\CP_m\CP_{m'}=\CP_m V\CP_{m'}=0.
\ene
Moreover,
\bee\label{PHP}
\CP_{\patch} H\CP_{\patch}=\sum_m \CP_m H\CP_m.
\ene
The number of clusters can be trivially estimated by $E^{lr_{1,2}}$. We also note that since the point $\bk$ is assumed to have been good at step zero (i.e. $\bk\|\bk\|^{-1}\in\CG^{\BPhi (0)}$), the clusters cannot occur sooner than at the distance $\E^{\sigma_{1,d-1,1}}$ from $0$, which means that 
\bee
\CP^{(0)}(\bk)V\CP_{\patch}(\bk)=0.
\ene
\enl
\bec
We have:
\bee
\CP_{\patch}=\CP(\tilde\CC,\bk),
\ene
where of course we have denoted
\bee
\tilde\CC:=\cup_m \tilde\CC_{m}.
\ene
\enc
Recall that the coordinates $\vec\phi$ we have introduced around the point $\BPhi^*$ are dependent on the cluster $\tilde\CC_{j}$ we were considering: the first variable $\phi_1$ is going from $\BPhi^*$ towards this cluster. In order to emphasize this, we will write $\vec\phi_j=((\phi_1)_j,\hat\phi_j)$ to indicate the set of coordinates generated by the $j$-th cluster (of course, the centre of this cluster also depends on $j$: ($\BPhi^*=\BPhi^*_j$). Then, for each $j$ and each (real) value of $\hat\phi_j$ we define the good and  bad sets of (complex) $(\phi_1)_j$: the bad sets are complex $\varepsilon$-neighbourhoods of the poles  with $\varepsilon:=E^{-r_{1,3}}$.  The union of the real parts 
of these bad sets (paired with the corresponding $\hat\phi$) forms the resonant (bad) set of angles corresponding to cluster $j$; we denote it by $ \CNN_{j}^{\vec\phi (1)}
\subset\Pi_j^{(0)}$. 
The image on the sphere of the bad set under the mapping $\BPsi_j$  is called the bad set of spherical angles corresponding to the patch $(\Pi_j^{(0)},\BPsi_{j})$:
\bee
\CNN_{j}^{\sph (1)}:=\BPsi_{j}(\CNN_{j}^{\vec\phi (1)}),
\ene
and the overall bad set of spherical angles is of course 
\bee\label{sphericalpatch}
\CNN_{}^{\sph (1)}:=\cup_j\CNN_{j}^{\sph(1)}.
\ene
The superscript $(1)$ refers to the fact that these are sets introduced at the first step. Finally, the set of good angles is 
\bee\label{goodPhi1}
\CG_{}^{\sph (1)}=\CG_{}^{\sph (1)}(\rho):=\S\setminus \CNN^{\sph (1)}.
\ene

Recall that we have defined  
\bee \label{r1}
r_{1,3}:=E^{\sigma_0},\ \ \ r_{1,2}:= \cs^{2} r_{1,3},\ \ \ r_{1,1}:=\frac{\cs}{10} r_{1,2},
 \ene
 where $\cs$ is a small constant to be defined later. At the moment we use one specific property of this constant: $(\mu+d+l)\cs<1$. 
 We also choose $\varepsilon$ in Lemmas~\ref{resclustershrink} and \ref{s=0} to be equal $E^{-r_{1,3}}$. Then we have

\bel
The measures of bad spherical and angular sets satisfy
\bee\label{angle1}
\meas(\CNN_j^{\vec\phi (1)})\ll E^{-2(d-2)-r_{1,3}+d^2(l+\mu)\sigma_{0,d-1}}<E^{-5r_{1,3}/6}
\ene
and
\bee\label{angle2}
\meas(\CNN_j^{\sph (1)})<E^{-5r_{1,3}/6}.
\ene
\enl
Finally, the good complex set of angles in patch $(\Pi_j^{(0)},\BPsi_j)$ is the $E^{-r_{1,3}E^{2d^2(l+\mu)\sigma_{0,d-1}}}$-neighbourhood of 
$\CG_{j}^{\vec\phi(1)}$ in $\Pi_{j,\C}^{(0)}$:
\bee\label{good1}
\CG_{j,\C}^{\vec\phi(1)}:=\{\vec\phi\in\Pi_{j,\C}^{(0)}:\ 
d(\CG_{j}^{\vec\phi (1)},\vec\phi)<E^{-r_{1,3}E^{2d^2(l+\mu)\sigma_{0,d-1}}}\}.
\ene 
Obviously, perturbations of size $E^{-r_{1,3}E^{2d^2(l+\mu)\sigma_{0,d-1}}}$ preserve the estimates of the resolvent from Lemmas~\ref{resclustershrink} and \ref{s=0} with, possibly, an extra factor $2$. This and basic perturbation theory imply: 
\begin{lem}\label{respatchshrink}
Suppose, $\vec\phi\in\CG_{j,\C}^{\vec\phi (1)}$ and $\k\in\C$ satisfies
\bee
|\nka-\nka^{(0)}(\vec\phi)|<E^{-r_{1,3}E^{2d^2(l+\mu)\sigma_{0,d-1}}}. 
\ene
Then
\begin{equation}\label{respatchesshrink}
\|( \CP_{\patch}(H(\nbka)-\rho^2) \CP_{\patch})^{-1}\|\le 
E^{r_{1,3}5E^{d^2(l+\mu)\sigma_{0,d-1}}}
\end{equation}
and
\begin{equation}\label{respatchesshrink1}
\|( \CP_{\patch}(H(\nbka)-\rho^2) \CP_{\patch})^{-1}\|_1\le
E^{lr_{1,2}+r_{1,3}5E^{d^2(l+\mu)\sigma_{0,d-1}}}.
\end{equation}
Here, of course, $\nbka:=\nka\BPsi_j(\vec\phi)$. 
\end{lem}

\subsection{Good and Bad angles for Step one}
Now we extend the good and bad sets from one patch onto the entire sphere $\S$. Recall that the sphere is covered by the patches  
$\CA^{\BPhi(0)}_j\subset \S$ so that $\S=\cup_j \CA^{\BPhi(0)}_j$ and 
 each patch $\CA^{\BPhi(0)}_j$
 is the image of $\Pi_j^{(0)}$ under the mapping $\BPhi=\BPsi_j(\vec\phi)$; each $\Pi_j^{(0)}$ is a neighbourhood of the origin of diameter smaller than $E^{-2}$ (the size of the `patch' in $\bk$ was $E^{-1}$, which corresponds to the size $E^{-2}$ of patches in $\BPhi$), and the mapping $\BPsi_j$ has the following form in a proper coordinate system: 
\bee\label{BPsil}
\BPsi_j(\phi_1,\dots,\phi_{d-1})=(\phi_1,\dots,\phi_{d-1},\sqrt{1-\phi_1^2-\dots-\phi_{d-1}^2})_j
\ene
(the index $j$ in the RHS means that this expression is considered in the natural coordinates around $\BPhi_j^*:=\BPsi_j(0)$ -- see \eqref{coordinates2} and the discussion afterwards). 


Summing estimates \eqref{angle1} for all patches, we obtain
\bel
The measure of the bad spherical  set satisfies
\bee\label{angle2n}
\meas(\CNN_{}^{\sph (1)})<E^{-r_{1,3}/2}.
\ene
\enl

We will also need the estimates for the resolvent in the
neighbourhood of $0$. Recall that the projection $\CP_{\patch}$ is defined in \eqref{CP}. 
Let $\cont_1$ be a circle in the complex plane:
 \begin{equation}  \label{C-2} 
 \cont_1=\{z\in \C: |z-\rho^{2}|=
\frac 12 E^{-r_{1,3}5E^{d^2(l+\mu)\sigma_{0,d-1}}}\}.
\end{equation}
Using the definition of $\bka^{(0)}(\vec\phi)$, 
we obtain the following lemma; recall that $H^{(0)}$ is the zero step `central box restriction' of $H$ defined by \eqref{newH_0}.
\bel\label{estnonres0} 
Let
$\vec\phi\in\CG_{j,\C}^{\vec\phi (1)}$ for some $j$, $\nka\in\C$: $|\nka-\k^{(0)}(\vec\phi)|<E^{-r_{1,3}E^{2d^2(l+\mu)\sigma_{0,d-1}}}$ and $z\in \cont_1$. Then,
\bee\label{new5.98} 
\|(H^{(0)}(\bka)-z)^{-1}
\|
\leq
4E^{r_{1,3}5E^{d^2(l+\mu)\sigma_{0,d-1}}}. 
\ene 
\enl
The proof is completely analogous to the proof of the corresponding Lemma 3.21 from \cite{KaSh} and we omit it here.

Let
\begin{equation}\label{defP_j} 
\tilde \CP_{\patch}(\nbka):=\CP_{\patch}(\nbka)+\CP^{(0)}(\nbka).  
\end{equation}
\bel 
Let
$\vec\phi\in\CG_{j,\C}^{\vec\phi (1)}$ for some $j$, $\nka\in\C$: $|\nka-\k^{(0)}(\vec\phi)|<E^{-r_{1,3}E^{2d^2(l+\mu)\sigma_{0,d-1}}}$ and $z\in \cont_1$.
Then 
\begin{equation} \label{May29-14a}
\CP_{\patch}V\CP^{(0)}=0, 
\end{equation}
\begin{equation} \label{May29-14b}
\tilde \CP_{\patch}H(\nbka)\tilde
\CP_{\patch}=  H^{(0)}(\nbka)+\CP_{\patch}H(\nbka)\CP_{\patch}
\end{equation} 
and
\begin{equation}\label{estfull}
\left\|\left(\tilde\CP_{\patch}\big(H(\nbka)-z\big)\tilde\CP_{\patch}\right)^{-1}\right\|\leq
4E^{r_{1,3}5E^{d^2(l+\mu)\sigma_{0,d-1}}}.
\end{equation}
\enl
%
\begin{proof} The proof is a straightforward corollary of \eqref{respatchesshrink}, \eqref{C-2} and \eqref{new5.98}. 
\end{proof}

\section{Step one}\label{section5}
{\newred
In this section, we will study our operator $H$ as a perturbation of the operator restricted to clusters constructed in the previous Section. Since the distance between these clusters is relatively large, multiple applications of the resolvent identity will allow us to prove that the error in such approximation is exponentially small. This is done in Theorem \ref{Thm2}. Later, in  
Lemma \ref{ldk-2}, we study the properties of the isoenergetic surface obtained at this step. 
}
\subsection{Operator $H^{(1)}$. Perturbation Formulas}

Throughout this section, we set $p=0$. Denote 
\bee
\CP^{(1)}=\CP^{(1)}(\nbka)=\CP(\hat K^{(1)};\nbka)
\ene
and
\bee
H^{(1)}=H^{(1)}(\nbka)=H(\hat K^{(1)};\nbka)
\ene
(recall that $\hat K^{(1)}$ was defined in \eqref{extball}).

Let us fix the patch $(\Pi_j^{(0)},\BPsi_j)$ centred at $\BPhi^*_j$ and assume  $\vec\phi \in\CG_{j,\C}^{\vec\phi (1)}$.  We
consider $H^{(1)}(\nbka)$ as a perturbation of
\begin{equation}\label{gulf1} 
\tilde H^{(1)}:=\tilde \CP_{\patch}H(\bka)\tilde
\CP_{\patch}+\left(\CP^{(1)}-\tilde \CP_{\patch}\right)H_0(\bka)\left(\CP^{(1)}-\tilde \CP_{\patch}\right),
\end{equation}
where $\tilde \CP_{\patch}$ is defined in \eqref{defP_j}.
By \eqref{May29-14b} and \eqref{PHP},  the first term on the right-hand side of \eqref{gulf1} has a block structure. 
The second term in \eqref{gulf1} is, obviously, diagonal. Thus, $\tilde H^{(1)}$ has a block-diagonal structure.
Let $W$ be the perturbation, i.e. 
\begin{equation}
W:=H^{(1)}-\tilde H^{(1)}=\CP^{(1)}V\CP^{(1)}-\tilde \CP_{\patch}V\tilde \CP_{\patch}. \label{W}\end{equation}
By analogy with \eqref{g}, \eqref{G}, we put:
\begin{equation}\label{g2} g^{(1)}_r({\nbka}):=\frac{(-1)^r}{2\pi
ir}\hbox{Tr}\oint_{\cont_1}\left(W(\tilde
H^{(1)}({\nbka})-zI)^{-1}\right)^rdz,
\end{equation} \begin{equation}\label{G2}
G^{(1)}_r({\nbka}):=\frac{(-1)^{r+1}}{2\pi i}\oint_{\cont_1}(\tilde
H^{(1)}({\nbka})-zI)^{-1}\left(W(\tilde
H^{(1)}({\nbka})-zI)^{-1}\right)^rdz
\end{equation}
(recall that the contour $\cont_1$ is defined in \eqref{C-2}). 

The next theorem is the analogue of Theorem 4.1 from \cite{KaSh}.

\begin{thm} \label{Thm2} Suppose, $\vec\phi\in \CG_{j,\C}^{\vec\phi(1)}\cap\R^{d-1}$, 
$\nka\in\R$,
$|\nka-\k^{(0)}(\vec\phi )|\leq E^{-r_{1,3}E^{2d^2(l+\mu)\sigma_{0,d-1}}}$,
$\nbka=\nka\BPsi_j(\vec\phi)$. Then, for sufficiently
large $E$, there exists a single eigenvalue of
$H^{(1)}({\nbka})$ in the interval\\ 
\bee
I_1:=\left( \rho^2-E^{-r_{1,3}E^{2d^2(l+\mu)\sigma_{0,d-1}}+1},
\rho^{2}+E^{-r_{1,3}E^{2d^2(l+\mu)\sigma_{0,d-1}}+1}\right). 
\ene
It is given by the absolutely
converging series:
\begin{equation}\label{eigenvalue-2}\lambda^{(1)}({\nbka})=\lambda^{(0)}({\nbka})+
\sum\limits_{r=2}^\infty g^{(1)}_r({\nbka}).
\end{equation} 
For
coefficients $g^{(1)}_r({\nbka})$ the following estimates hold:
\begin{equation}\label{estg2} |g^{(1)}_r({\nbka})|<E^{-E^{\sigma_{1,d-1,1}}(2Q)^{-1}}
E^{-\sigma_0r/4}.
\end{equation}
The corresponding spectral projection is given by the series:
\begin{equation}\label{sprojector-2}
\E ^{(1)}({\nbka})=\E^{(0)}({\nbka})+\sum\limits_{r=1}^\infty
G^{(1)}_r({\nbka})
\end{equation} 
(recall that $\E^{(0)}({\nbka})$ is the spectral
projection of $H^{(0)}(\nbka)$). The operators $G^{(1)}_r({\nbka})$ satisfy
the estimates:
\begin{equation}
\label{Feb1a} \left\|G^{(1)}_r({\nbka})\right\|_1<E^{-E^{\sigma_{1,d-1,1}}(4Q)^{-1}}
E^{-\sigma_0r/4}
\end{equation}
and
\begin{equation}\label{Feb6a}
G^{(1)}_r({\nbka})_{\bn\bn'}=0,\ \ \mbox{if}\ \ 4\sqrt{d}\cdot rE^{\sigma_{1,d-1,1}}<|\bn|+|\bn'|. 
\end{equation}
Coefficients $g^{(1)}_r({\nbka})$ and operators
$G^{(1)}_r({\nbka})$ can be analytically extended to the complex neighbourhood $\CG^{\vec\phi (1)}_{j,\C}$ as functions of $\vec\phi $ and to the complex $E^{-r_{1,3}E^{2d^2(l+\mu)\sigma_{0,d-1}}}-$
neighbourhood of $\nka^{(0)}(\vec\phi )$ as functions
of $\k$, estimates \eqref{estg2}, \eqref{Feb1a}
being preserved.
\end{thm}
\begin{cor} \label{corthm2} 
For the perturbed eigenvalue and its spectral
projection the following estimates hold:
 \begin{equation}\label{perturbation-2}
\lambda^{(1)}({\nbka})=\lambda^{(0)}({\nbka})+ O\left(E^{-E^{\sigma_{1,d-1,1}}(2Q)^{-1}}\right),
\end{equation}
\begin{equation}\label{perturbation*-2}
\left\|\E^{(1)}({\nbka})-\E^{(0)}({\nbka})\right\|_1<E^{-E^{\sigma_{1,d-1,1}}(4Q)^{-1}},
\end{equation}
\begin{equation}
\left|\E^{(1)}({\nbka})_{\bn\bn'}\right|<E^{-\dis^{(1)}(\bn,\bn')}\ \
\mbox{when}\ |\bn|>4\sqrt{d}\,E^{\sigma_{1,d-1,1}} \mbox{\ or }
|\bn'|>4\sqrt{d}\,E^{\sigma_{1,d-1,1}},\label{Feb6b}
\end{equation}
where we have defined 
$$\dis^{(1)}(\bn,\bn'):=\frac{\sigma_0}{16\sqrt{d}}(|\bn|+|\bn'|)E^{-\sigma_{1,d-1,1}}+E^{\sigma_{1,d-1,1}
}(4Q)^{-1}.$$
\end{cor}
Estimates \eqref{perturbation-2} and \eqref{perturbation*-2} easily
follow from \eqref{eigenvalue-2}, \eqref{sprojector-2} and
\eqref{estg2} and \eqref{Feb1a}. Formula \eqref{Feb6b} follows
from \eqref{sprojector-2}, \eqref{Feb1a} and \eqref{Feb6a}. Indeed,
using these estimates, we obtain
$\left|\left(\E^{(1)}({\nbka})-\E^{(0)}({\nbka})\right)_{\bn\bn'}\right|
<E^{-\dis^{(1)}(\bn,\bn')}$. Taking into account that $\E^{(0)}({\nbka})_{\bn\bn'}=0$
when $|\bn|>E^{\sigma_{1,d-1,1}}$ or $|\bn'|>E^{\sigma_{1,d-1,1}}$, we
arrive at  \eqref{Feb6b}.
{\newred
\ber While the proof of Theorem \ref{Thm2} is very similar to the proof of Theorem \ref{Thm1} and uses the same resolvent identities, there are several extra technical subtleties in the proof of the former. Therefore, we are providing the full proof of Theorem \ref{Thm2}. Note that we will use part of this proof as the start of induction when proving Lemma \ref{abstractlemma}.
\enr}
\begin{proof}
Put 
\bee
\CP':=\CP^{(1)}-\tilde{\CP}_{\patch}. 
\ene
By \eqref{gulf1}, \eqref{W} we have:
$$\tilde{H}^{(1)}(\bka)=\tilde{\CP}_{\patch}H(\bka)\tilde{\CP}_{\patch}+\CP'H_0(\bka)\CP',\ \ \
W=\CP'V\CP'+\CP'V\tilde{\CP}_{\patch}+\tilde{\CP}_{\patch}V\CP'.$$
We will often omit $\bka $ in the arguments when it
cannot lead to confusion. Let $z\in\cont_1$. By \eqref{estfull}, we have
\begin{equation}\label{step2raz}
\left\|(\tilde{H}^{(1)}-zI)^{-1}\right\| <4E^{r_{1,3}5E^{d^2(l+\mu)\sigma_{0,d-1}}}.
\end{equation} Let us consider the perturbation series
\begin{equation}\label{step2dva}
(H^{(1)}-z)^{-1}=\sum_{r=0}^\infty(\tilde H^{(1)}-z)^{-1}\left(-W(\tilde H^{(1)}-z)^{-1}\right)^r.
\end{equation}
Put
$$\tilde A:=-(\tilde H^{(1)}-z)^{-1}W(\tilde H^{(1)}-z)^{-1}.$$
To check the convergence it is enough to show that
\begin{equation}\label{||A||2}\|\tilde A\|<E^{-\sigma_0}.
\end{equation}
Estimates \eqref{step2raz} and \eqref{||A||2} yield
\begin{equation}\label{step2raz*}
\left\|({H}^{(1)}-zI)^{-1}\right\| <8E^{r_{1,3}5E^{d^2(l+\mu)\sigma_{0,d-1}}}.
\end{equation}
To prove \eqref{||A||2} it suffices to check
\begin{equation}\label{||A||2-2}\|\CP'\tilde A\CP'\|<4\|V\|E^{-2\sigma_0}
\end{equation}
and
\begin{equation}\label{||A||2-3}\|\CP'\tilde A\tilde {\CP}_{\patch}\|<8\|V\|E^{-2\sigma_0}.
\end{equation}

Remark \ref{neweverything} implies that
$$
||{\newred (\CP'H\CP'-zI)^{-1}}||\ll E^{-\sigma_0},
$$
after which estimate \eqref{||A||2-2} becomes trivial. 
Therefore, we proceed to 
 \eqref{||A||2-3}. By Lemma~\ref{PHP*}, it is enough to check

\begin{equation}
\label{||A||2-4}\|\CP'\tilde A\CP^{(0)}\|<8\|V\|E^{-2\sigma_0} 
\end{equation}
and
\begin{equation}
\label{||A||2-7}\|\CP'\tilde A\tilde {\CP}_j\|<8\|V\|E^{-2\sigma_0}.
\end{equation}

To prove  \eqref{||A||2-4} we represent $(\tilde H^{(1)}-z)^{-1}\CP^{(0)}$ as follows:
\begin{align}\label{May11-13b}
&(\tilde H^{(1)}-z)^{-1} \CP^{(0)}=
\sum_{r=0}^{R_0}\left(-(H_0-z)^{-1}\CP^{(0)}V\CP^{(0)}\right)^r(H_0-z)^{-1}\CP^{(0)}+\cr &
\left(-(H_0-z)^{-1}\CP^{(0)}V\CP^{(0)}\right)^{R_0+1}(\tilde
H^{(1)}-z)^{-1}\CP^{(0)}, \end{align}
 where $R_0$ will be fixed later.
Then,
\begin{equation}\label{step25}
\begin{split}& \|\CP'V\CP^{(0)}(\tilde H^{(1)}-z)^{-1}\|\leq
\sum_{r=0}^{R_0}\left\|B_r\right\|+
\left\|\CP'V\left((H_0-z)^{-1}\CP^{(0)}V\CP^{(0)}\right)^{R_0+1}\right\|\|(\tilde
H^{(1)}-z)^{-1}\CP^{(0)}\|, \cr
&
B_r:=\CP'V\left((H_0-z)^{-1}\CP^{(0)}V\CP^{(0)}\right)^r(H_0-z)^{-1}\CP^{(0)}.
\end{split} \end{equation}
Note that matrix elements $(B_r)_{\bn\bn'}$ are
equal to zero if $|\bn-\bn'|>Q(r+1)$ (see \eqref{V_q=0}).
Thus, the only non-trivial elements $(B_r)_{\bn\bn'}$ are such that $$
\bn\in\Omega(\frac12 E^{\sigma_{1,d-1,1}}+Q)\setminus\Omega (\frac12 E^{\sigma_{1,d-1,1}}),\ \ \
\bn'\in \Omega (\frac12 E^{\sigma_{1,d-1,1}}), \ \ |\bn-\bn'|\leq Q(r+1). $$
Let $r$ be chosen in such a way that 
\bee\label{Qr+1}
Q(r+1)\leq E^{\sigma_{1,d-1,1}}/6.
\ene
It follows that  $(B_r)_{\bn{\bf 0}}=0$. If $\bn' \neq 0$, then (see \eqref{G1-2})
$$\left|\|\bka^{(0)}(\vec\phi) +\bn'\vec\boldom\|^{2}-z\right|>E^{1-(l+\mu+2)\sigma_{1,d-1,1}}.$$  
Therefore, for $r$ satisfying \eqref{Qr+1} we have:
$$\|B_r\|\leq
(\|V\|E^{-1/2})^{r+1},
$$
$$
\left\|\CP'V\left((H_0-z)^{-1}\CP^{(0)}V\CP^{(0)}\right)^{r+1}\right\|\leq
\|V\|\big(\|V\|E^{-1/2}\big)^{r+1}. $$ 
Now, we fix $R_0:=[E^{\sigma_{1,d-1,1}} /(6Q)]-1$. Then the condition $Q(r+1)\leq
E^{\sigma_{1,d-1,1}}/6$ is satisfied for all $r\leq R_0$ and using Lemma~\ref{estnonres0} we get 
$$ \|\CP'V\CP^{(0)}(\tilde H^{(1)}-z)^{-1}\|\leq
\sum_{r=0}^{R_0}(\|V\|E^{-1/2})^{r+1}+\|V\|\big(\|V\|E^{-1/2}\big)^{R_0+1}4E^{r_{1,3}5E^{d^2(l+\mu)\sigma_{0,d-1}}}.
$$  
Assuming that $E>E_*$  (so that, in particular, $\frac{
E^{\sigma_{1,d-1,1}}}{6Q}>20r_{1,3}E^{d^2(l+\mu)\sigma_{0,d-1}}$), we obtain \eqref{||A||2-4}.

The proof of \eqref{||A||2-7} is quite similar, so we just outline the difference. Consider the slightly more difficult case $s:=\rank(\CC_j)\geq1$. Then instead of \eqref{step25}, we write 

\begin{equation}\label{step25a}
\begin{split}& \|\CP'V\tilde\CP_j(\tilde H^{(1)}-z)^{-1}\|\leq
\sum_{r=0}^{R_0}\left\|B_r\right\|+
\left\|\CP'V\left((H_0-z)^{-1}\tilde\CP_jV\tilde\CP_j\right)^{R_0+1}\right\|\|(\tilde
H^{(1)}-z)^{-1}\tilde\CP_j\|, \cr
&
B_r:=\CP'V\left((H_0-z)^{-1}\tilde\CP_j V\tilde\CP_j\right)^r(H_0-z)^{-1}\tilde\CP_j.
\end{split} \end{equation}

The only non-zero elements $(B_r)_{\bn\bn'}$ are those when  $\bn$ lies in $Q$-neighbourhood of $\tilde\CC_j$, but not in $\tilde\CC_j$, and $\bn'\in\tilde\CC_j$ (and 
$|\bn-\bn'|\leq Q(r+1)$). 

Let $r$ be such that $Q(r+1)\leq E^{\sigma_{0,s,1}}/6$. It follows that  $(B_r)_{\bn{\bn'}}=0$ if $\bn'\in \check\CC_j$. However, if $\bn'\in \tilde\CC_j\setminus\check\CC_j$, then (see Lemma~\ref{Z0new} and Remark \ref{neweverything})
$$\left|\|\bka^{(0)}(\vec\phi) +\bn'\vec\boldom\|^{2}-z\right|>E^{\sigma_{0,s}}/2.$$  
Now we choose $R_0=[E^{\sigma_{0,s,1}} /(6Q)]-1$ and continue as above. We also use the estimate \eqref{resclusterestshrink} (with $\varepsilon=E^{-r_{1,3}}$). The proof of  \eqref{||A||2-7} easily follows.

To prove \eqref{Feb1a}, we consider the operator
$$A:=W\left(\tilde H^{(1)}-z\right)^{-1}$$ 
and
represent it as 
$$
A=A_0+A_1+A_2,
$$ 
where
$$
A_0=\left(\CP^{(1)}-\E^{(0)}\right)A \left(\CP^{(1)}
-\E^{(0)}\right),
$$ 
$$
A_1=\left(\CP^{(1)}-\E^{(0)}\right)A
\E^{(0)},
$$
and
$$A_2= \E^{(0)}A
\left(\CP^{(1)}-\E^{(0)}\right).
$$
Note that we have
$\E^{(0)}W\E^{(0)}=0$ because of \eqref{W}. It is easy to see that, by construction, $A_0$ is holomorphic inside $\cont_1$ (see, e.g. Theorem~\ref{Thm1}). Hence,
$$\oint _{\cont_1}\left(\tilde H^{(1)}-z\right)^{-1}A_0^r dz=0.$$
Therefore,
\begin{equation} \label{Feb1} 
G^{(1)}_r=\frac{(-1)^{r}}{2\pi i}\sum
_{j_1,...j_r=0,1,2,\ j_1^2+...+j_r^2\neq 0}I_{j_1...j_r},
\ene
where
\bee
I_{j_1...j_r}:=\oint _{\cont_1}\left(\tilde
H^{(1)}-z\right)^{-1}A_{j_1}.....A_{j_r} dz.
\end{equation} 
At least one of indices in each term is equal to 1 or 2.
Let us show that
\begin{equation} \label{A_2}
\|A_2\|_1<\|V\|E^{-E^{\sigma_{1,d-1,1}}(3Q)^{-1}}.
\end{equation} First, we notice that
$\E^{(0)}W(\CP^{(1)}-\E^{(0)})=\E^{(0)}W\CP'$ by \eqref{W} and Lemma~\ref{PHP*}. It suffices
to show that
\begin{equation}\label{Feb6}
\|\E^{(0)}W\CP'\|_1<\|V\|E^{-E^{\sigma_{1,d-1,1}}(3Q)^{-1}},
\end{equation} 
since $\|\CP'\left(\tilde
H^{(1)}-z\right)^{-1}\|=\|\CP'\left(
H_0-z\right)^{-1}\|<2E^{-\sigma_0}$ for $z\in \cont_1 $.
To show \eqref{Feb6}, we write 
$$
\left(\E^{(0)}W\CP'\right)_{\bn\bn'}=\sum _{\bn'':\ |\bn''|\leq \frac12E^{\sigma_{1,d-1,1}},\ |\bn''-\bn'|\leq
Q}\E^{(0)}_{\bn\bn''}W_{\bn''-\bn'}
$$ 
when $|\bn'|>\frac12E^{\sigma_{1,d-1,1}}$
and it is equal to zero otherwise. Hence,
$$
\left|\left(\E^{(0)}W\CP'\right)_{\bn\bn'}\right|\leq \|W\|
\sum _{\bn'':\ \frac12E^{\sigma_{1,d-1,1}}-Q\leq |\bn''|\leq \frac12E^{\sigma_{1,d-1,1}}
}|\E^{(0)}_{\bn\bn''}|
$$ 
if $|\bn'|<\frac12E^{\sigma_{1,d-1,1}}+Q$ and is zero
otherwise. Using \eqref{matrix elements}, we obtain
\begin{equation} \label{E1-1}
\left|\left(\E^{(0)}W\CP'\right)_{\bn\bn'}\right|<
\|V\|E^{d\sigma_{1,d-1,1}}\max _{|\bn''|>\frac12E^{\sigma_{1,d-1,1}}-Q}E^{-\dis^{(0)}(\bn,\bn'')}.
\end{equation}
 It easily follows:
$$
\left|\left(\E^{(0)}W\CP'\right)_{\bn\bn'}\right|<\|V\|E^{d\sigma_{1,d-1,1}}E^{-E^{\sigma_{1,d-1,1}}(2Q)^{-1}+1}
$$
when  $\frac12E^{\sigma_{1,d-1,1}}<|\bn'|<\frac12E^{\sigma_{1,d-1,1}}+Q$, and is zero otherwise. It
follows that 
$$\left\|\E^{(0)}W\CP'\right\|<\|V\|E^{-E^{\sigma_{1,d-1,1}}(3Q)^{-1}}.$$ 
Taking into account that $\E^{(0)}$ is a
one-dimensional projection, we obtain the same estimate for 
the trace-norm, namely, \eqref{Feb6}. Thus, we have proved \eqref{A_2}.
Let us estimate $I_{j_1...j_r}$. Suppose one of the indices is equal to 2.
Substituting \eqref{A_2} into \eqref{Feb1} and taking into account \eqref{||A||2}, \eqref{step2raz},
we obtain:
$$\left\|I_{j_1...j_r}\right\|<\|V\|^2E^{-E^{\sigma_{1,d-1,1}}(3Q)^{-1}}E^{-\sigma_0(\frac{r}{2}-1)}8E^{r_{1,3}5E^{d^2(l+\mu)\sigma_{0,d-1}}}<E^{-E^{\sigma_{1,d-1,1}}(4Q)^{-1}}E^{-\sigma_0 r/2}.$$
More precisely, our $A_2$ splits the integrand in $I_{j_1\dots j_r}$ into two parts; for each part we use \eqref{||A||2} for every product of two $A_{j_k}$ and \eqref{step2raz} for the last single $A_{j_s}$ or $(\tilde H^{(1)}-z)^{-1}$; we also take into account the length of the circle.
Note that the operator
$A_1$ is always followed by $A_2$ unless $A_1$ occupies the very
last position in the product. Thus, it remains to consider the case
$A_{j_1}.....A_{j_r}= A_0^{r-1}A_1$. It is easy to see that
$$A_{0}^{r-1}A_{1}=\left(A_2(\bar z)A_{0}^{r-1}(\bar z)\right)^*.$$
This implies the estimate for this case too.
Therefore (taking into account the sum over all permutations),
$$
\left\|G^{(1)}_r\right\|<E^{-E^{\sigma_{1,d-1,1}}(4Q)^{-1}}E^{-\sigma_0 r/4}.
$$ The same estimate
can be written for the trace norm of this operator, since
$\E^{(0)}$ is one-dimensional.

Let us obtain the estimate for $g_r^{(1)}$.
Obviously,\begin{equation} \label{Feb1'}
g^{(1)}_r=\frac{(-1)^r}{2\pi ir}\sum _{j_1,...j_r=0,1,2,\
j_1^2+...+j_r^2\neq 0}Tr\oint _{\cont_1}A_{j_1}.....A_{j_r} dz.
\end{equation}
Note that each non-trivial term contains both $A_1$ and $A_2$, since we compute
the trace of the integral. Using \eqref{Feb6} and \eqref{step2raz}, we obtain:
$\|A_1\|_1<2b^{-1}_1\|V\|E^{-E^{\sigma_{1,d-1,1}}(3Q)^{-1}}$, where $b_1$ is the radius of $\cont_1$.  Combining
this estimate with \eqref{A_2} and \eqref{||A||2}, \eqref{step2raz}, we obtain
\eqref{estg2} for $r\geq 2$. Finally, applying \eqref{g2} in the
case $r=1$, we see that $g^{(1)}_1=0$, since
$\E^{(0)}W\E^{(0)}=0$.

To prove \eqref{Feb6a} it is enough to notice that (see Lemma~\ref{2'}) the biggest block
of $\tilde H^{(1)}$ has the size not greater than $\frac12E^{\sigma_{1,d-1,1}}$ (the number of elements in $\hat K^{(0)}$).

{\newred By construction, $\E ^{(1)}({\nbka})$ is the spectral projection corresponding to  the interval $I_1$. Since the series  \eqref{sprojector-2} converges in the trace class, see \eqref{Feb1a}, we have  $\mbox{Tr}(\E^{(1)}({\nbka}))=\mbox{Tr}(\E^{(0)}({\nbka}))+o(1)$
for large $\rho $. Since both $\mbox{Tr}(\E^{(1)}({\nbka}))$ and $\mbox{Tr}(\E^{(0)}({\nbka}))$ are integers, for those $\rho$ we have  $\mbox{Tr}(\E^{(1)}({\nbka}))=\mbox{Tr}(\E^{(0)}({\nbka}))$. By Theorem \ref{Thm1},  $\mbox{Tr}(\E^{(0)}({\nbka}))=1$. Therefore, there exists a single eigenvalue  $\lambda^{(1)}({\nbka})$  in $I_1$. Now we carry on as in the proof of Theorem 2.1 in \cite{Ka} and obtain \eqref{eigenvalue-2}.}
\end{proof}


By Theorem \ref{Thm2}, the  coefficients $g^{(1)}_r({\nbka})$ and operators
$G^{(1)}_r({\nbka})$, $r\in \N$, can be analytically extended onto the complex neighbourhood $\CG^{\vec\phi (1)}_{j,\C}$ as functions of $\vec\phi $ and to the complex $E^{-r_{1,3}E^{2d^2(l+\mu)\sigma_{0,d-1}}}-$
neighbourhood of $\nka^{(0)}(\vec\phi )$ as functions
of $\k$, estimates \eqref{estg2}, \eqref{perturbation-2}
being preserved. Now, we use formulae \eqref{g2},
\eqref{eigenvalue-2} to extend
$\lambda^{(1)}({\nbka})=\lambda^{(1)}(\nka,\vec\phi)$ as an
analytic function. Obviously, series \eqref{eigenvalue-2} is
differentiable. Using the Cauchy integral we get the following lemma.

\bel \label{L:derivatives-2}
Under conditions of Theorem \ref{Thm2}, expansions \eqref{eigenvalue-2} and \eqref{sprojector-2}   can be analytically extended onto the complex neighbourhood $\CG^{\vec\phi (1)}_{j,\C}$ as functions of $\vec\phi $ and to the complex $E^{-r_{1,3}E^{2d^2(l+\mu)\sigma_{0,d-1}}}-$
neighbourhood of $\nka^{(0)}(\vec\phi )$ as functions
of $\k$. The following
estimates hold when $\vec\phi \in \CG_{j,\C}^{\vec\phi (1)}$ and $\nka\in \C:$
$|\nka-\k^{(0)}(\vec\phi )|<E^{-r_{1,3}E^{2d^2(l+\mu)\sigma_{0,d-1}}}:$
\begin{equation}\label{perturbation-2c}
\lambda^{(1)}({\nbka})=\lambda^{(0)}({\nbka})+ O\left(E^{-E^{\sigma_{1,d-1,1}}(2Q)^{-1}}\right),
\end{equation}
\begin{equation}\label{estgder1-2k}
\frac{\partial\lambda^{(1)}}{\partial\nka}=\frac{\partial\lambda^{(0)}}{\partial\nka}
+ O\left(E^{-E^{\sigma_{1,d-1,1}}(2Q)^{-1}}E^{r_{1,3}E^{2d^2(l+\mu)\sigma_{0,d-1}}}\right). \end{equation}

Similar estimates can be written for all derivatives of $\lambda^{(1)}$ and $\E^{(1)}$ with respect to $\nka$ and $\vec \phi $.

\enl

\subsection{\label{IS2}Isoenergetic Surface for Operator $H^{(1)}$}
The above estimates produce the following statement: 

\bel\label{ldk-2} \begin{enumerate}
\item For every $j$ and every $E>E_*$,  $\rho\in[E-1,E+1]$, $\lambda :=\rho^{2}$, and $\vec\phi \in\CG_{j,\C}^{\vec\phi (1)}\cap\R^{d-1}$, there is a unique
$\k^{(1)}( \vec\phi,\rho )$ in the interval
$$\tilde I_1:=[\k^{(0)}(\vec\phi,\rho)-E^{-r_{1,3}E^{2d^2(l+\mu)\sigma_{0,d-1}}},\k^{(0)}(\vec\phi,\rho
)+E^{-r_{1,3}E^{2d^2(l+\mu)\sigma_{0,d-1}}}],$$ 
such that
    \begin{equation}\label{2.70-2}
    \lambda^{(1)} \left(\bka
^{(1)}(\vec\phi,\rho )\right)=\rho^2 ,\ \ \bka ^{(1)}(\vec\phi,\rho ):=\k^{(1)}(\vec\phi,\rho ) \BPsi_j(\vec\phi).
    \end{equation}
\item  Furthermore, there exists an analytic in $ \vec\phi $ continuation  of
$\k^{(1)}(\vec\phi,\rho )$ to the complex set $\CG_{j,\C}^{\vec\phi (1)}$
such that $\lambda^{(1)} (\bka ^{(1)}(\vec\phi,\rho ))=\lambda $.
Function $\k^{(1)}(\vec\phi,\rho )$ can be represented as
$\k^{(1)}(\vec\phi,\rho )=\k^{(0)}(\vec\phi,\rho
)+h^{(1)}(\vec\phi,\rho )$, where
\begin{equation}\label{dk0-2} |h^{(1)}(\vec\phi )|=
O\left(E^{-E^{\sigma_{1,d-1,1}}(2Q)^{-1}-1}\right),
\end{equation}
\begin{equation}\label{dk-2}
\frac{\partial{h}^{(1)}}{\partial\vec\phi}= O\left(E^{-E^{\sigma_{1,d-1,1}}(2Q)^{-1}-1}E^{r_{1,3}E^{2d^2(l+\mu)\sigma_{0,d-1}}}\right),\ \ \ \ \
\ene
\bee
\frac{\partial^2{h}^{(1)}}{\partial\vec\phi^2}=O\left(E^{-E^{\sigma_{1,d-1,1}}(2Q)^{-1}-1}E^{2r_{1,3}E^{2d^2(l+\mu)\sigma_{0,d-1}}}\right).
\end{equation} \end{enumerate}\enl
\begin{proof}  The proof is completely analogous to that of Lemma 3.11 from \cite{KS}, 
only now we use estimates from Lemma \ref{L:derivatives-2}.
\end{proof}

\ber
As before, we sometimes will, slightly abusing notation, write $\bka
^{(1)}(\BPhi,\rho ):=\bka
^{(1)}(\BPsi_j(\vec\phi),\rho )$ when $\BPhi\in\CA^{\BPhi(1)}_j$. 
\enr

Let us consider the set of points in $\R^d$ given by the formula:
$\bka=\bka^{(1)} (\vec\phi), \ \ \vec\phi \in \CG_{j,\C}^{\vec\phi (1)}\cap\R^{d-1}$ for some $j$. By Lemma \ref{ldk-2} this set of points is a slight distortion
of ${\mathcal D}^{(0)}$. All the points of this curve satisfy
the equation $\lambda^{(1)} (\bka ^{(1)}(\vec\phi ))=\rho^{2}$. We call
it isoenergetic surface of the operator $H^{(1)}$ and denote by
${\mathcal D}^{(1)}={\mathcal D}^{(1)}(\rho)$.


\section{Step one: preparation for the induction. \label{goodset-3}}

We have discussed how to perform steps zero and one of our approximating procedure. Now we will start describing the induction process -- how to perform step $n+1$ assuming that step $n$ ($n\ge 1$) has been made. What we have done so far, can be considered as the base of our induction. It will be convenient for us to prove one more result that should be considered as a part of the base of induction, since the inductive step will be slightly different from the base. The result proved in this section is similar to the Bourgain Lemma \cite{Bou1}, but requires a more careful analysis than in \cite{Bou1} due to the fact that we need to make sure that all the (large) energies are covered by our result. Once we prove this result, in the next section we will describe the full inductive step. From now on, unless we state otherwise, we will always assume that pre-clusters $\BUps$ {\newred (defined in definitions \ref{7BUps}, \ref{new:ext} and \ref{def5.19})} and clusters $\CC$ {\newred (defined in \eqref{7CC} and \eqref{7CC1})} are taken with $p=1$.

Let us also introduce the following recursive notation. We denote by $\cs$ and $Z_0$ two  constants introduced in Lemma \ref{MGL}. The precise values of these constants will be fixed once and for all at the end of this section; now, we just remark that $\cs$ is so small that, in particular,  $\cs<(100l\mu^2)^{-1}$.
We put (recall that $r_{1,1}, r_{1,2}, r_{1,3}$ have been defined in \eqref{new5.1}): 
\bee\label{r13}
r_{1,3}=r_{1,3}(E):=E^{\sigma_0}.
\ene 
We also put for $n\ge 1$
\bee\label{rn2}
r_{n,2}=r_{n,2}(E):=\cs^2r_{n,3}
\ene
and
\bee\label{rn1}
r_{n,1}=r_{n,1}(E):=\frac{\cs}{10} r_{n,2}.
\ene
Finally, for any $n>1$ we define 
\bee\label{rn3}
r_{n,3}=r_{n,3}(E):=E^{\de_0r_{n-1,1}}, \ \ \de_0:=\cs/(100Z_0).
\ene
This, in particular, means that $r_{n,j}>r_{n',j'}$ assuming $n>n'$, or $n=n'$ and $j>j'$. 
We also introduce the parameter $r_n'$. For $n=1$ we define it as 
\bee\label{r1'}
r_1':=r_{1,3},
\ene
and for $n>1$ we postpone the definition of $r_n'$ until the next section.

\subsection{Bourgain type Lemma}

Let us consider a fixed patch $\CA^{\bxi,\rho,\vec\boldom(0)}=\CA^{\bxi,\rho,\vec\boldom(0)}_{\tim}$ (we have introduced such patches in section \ref{section1}) of  size $\En^{-1}$ in $\bxi$ and $\rho$ and $\En^{-2}$ in $\vec\boldom$. We make sure that $\{\rho\in [E-1,E+1]\}$ and $\{\bxi\in\R^d:\ |\  \|\bxi\|^2-\rho^2| <E^{\sigma_0}\}$ are covered by these patches. 
 Next, we call the patch in $\vec\boldom$ bad, if it contains no frequency vectors $\vec\boldom$ that satisfy SDC (with $\hat\mu$ and $B_0$ fixed in Section \ref{section3}). 
Obviously, $\vec\boldom$-part of such patches is a subset of $\CNN^{\vec\boldom(0)}(B_0)$, so their total volume is less than $CB_0^{1/d}$. Since any good patch contains a frequency vector satisfying \eqref{strong}, we can, as discussed in section \ref{section1}, assume that it is the centre $\vec\boldom^*$ that satisfies \eqref{strong}. 
We construct super-extended pre-clusters $\tilde\BUps^{\Z}_1(\bxi)=\tilde\BUps^{\Z}_{1,\tim}(\bxi)$ from  Section 5 based on the points $(\bxi^*,\rho^*,\vec\boldom^*)=(\bxi^*_{\tim},\rho^*_{\tim},\vec\boldom^*_{\tim})$. 
We also denote by 
\bee
\tilde\CP=\tilde\CP_{\tim}=\CP(\tilde\BUps^{\Z}_1(\bxi),\bxi) 
\ene
the projection onto this pre-cluster. Note that this includes the (simple) case when  $\rank_1(\bxi)=0$, in which case $\BUps^{\Z}_1(\bxi)=\{0\}$, and $\tilde\BUps^{\Z}_1(\bxi)=\Om(E^{\sigma_{1,0,1}})$.

Let $S_{\tim}=S_{\tim}^{(0)} \subset \R^{ld+1+d}$  be defined as:
\bee
S_{{\tim}}:=
\left \{(\bxi,\rho,\vec \boldom  ) \in \CA_{\tim}^{\bxi,\rho,\vec \boldom (0)}: \left \|
(H(\tilde\BUps^{\Z}_1(\bxi),\bxi )-\rho^2)^{-1}\right\|_2> \En^{r_{1}'}\right\} 
\label{Sn}
\end{equation}
(recall that $\|\cdot\|_2$ is the Hilbert-Schmidt norm). 
 Let
\begin{equation}
S=S_{total}=S_{total}^{(0)}:=\cup _{{\tim}=1}^MS_{\tim}. \label{S}
\end{equation}
Here we take the union over all good (in $\vec\boldom$) patches; the number of such patches $M$ can be trivially estimated by 
\bee\label{Nest}
M\leq E^{2d(l+1)}.
\ene 
Recall that we have denoted 
$$\BS\BL(\rho,\En^{\sigma _0}):=\left\{ \bxi :  \left| \|\bxi\|^2-\rho^2\right|<\En^{\sigma _0}\right\}.$$
Recall also that $\Om(R)$ is a ball of radius $R$ in $l^{\infty}$-norm in $\Z^l$ and $r_\iota=r_\iota(\En)$ is a growing function of $\En$.  
Here is the main lemma that we prove in this section; we will call it {\it the Main {\newred Semi-Algebraic} Lemma at level one}: 
\bel\label{MGL} For every sufficiently large $Z_0$ and every $\En>\En_*$, there is a set ${G' }_{}(\En,Z_0)={G'}^{(1)}(\En,Z_0)\subset \CG^{\vec\boldom(0)}(B_0)$ 
 and $\gamma _0=\gamma _0(Z_0)$, $0<\gamma_0<1/2$,
such that for any $(\vec \boldom, \rho, \bxi )\in S_{total}$ with $\vec \boldom \in { G' } (\En,Z_0)$ and $\rho\in[E-1,E+1]$ 
there is a $\gamma $, $\gamma =\gamma (\vec \boldom, \rho , \bxi, Z_0)$ with the following properties:
\bee \label{gamma} \gamma _0<\gamma <1-\gamma _0,
\ene
\bee \label{no}
\left\{ \bn\in\Z^l:   (\vec \boldom, \rho, \bxi +\bn \vec \boldom )\in S_{total}, \ \  \bn \in \Om(\En^{\gamma r_{1,1}}) \setminus \Om\left(\En^{\frac{\gamma r_{1,1}}{Z_0}}\right)\right\}=\emptyset .
\ene
The set ${ G' } (\En, Z_0)$ has asymptotically full measure in $\CG^{\vec\boldom(0)}$:
\bee \meas({G' } )=\meas(\CG^{\vec\boldom(0)})-O(\En ^{- C_1r_{1,1}}), \ {\En\to \infty },\ \ C_1=C_1(Z_0).\label{mesLambda0}
\ene
The value of $\gamma $ can be taken constant in the $\En^{-r_{1}'-2 }$-neighbourhood of every $(\rho, \bxi )$
and in the $\En^{-2r_{1}' }$-neighbourhood of every $\vec \boldom $.
\enl
\bel\label{MGL'}
Similar statement holds for the `enlarged structure' (the nature of this terminology will become clear in the next section). This means that we can find a (possibly different set) $G''=G''^{(1)}(E,Z_0)$ and a (possibly different) $\tilde\gamma$ such that instead of \eqref{no} when $\vec\boldom\in G''^{(1)}(E)$ we have 
\bee \label{no'}
\left\{ \bn\in\Z^l:   (\vec \boldom, \rho, \bxi +\bn \vec \boldom )\in S_{total}, \ \  \bn \in \Om(\En^{\tilde\gamma r_{1,2}}) \setminus \Om\left(\En^{\frac{\tilde\gamma r_{1,2}}{Z_0}}\right)\right\}=\emptyset .
\ene 
The set ${ G'' } (\En, Z_0)$ has an asymptotically full measure in $\CG^{\vec\boldom(0)}$:
\bee \meas({G''(E) } )=\meas(\CG^{\vec\boldom(0)})-O(\En ^{- C_1r_{1,2}}), \ {\En\to \infty },\ \ C_1=C_1(Z_0).\label{mesLambda00}
\ene
The value of $\tilde\gamma $ can be taken constant in the $\En^{-r_{1}'-2 }$-neighbourhood of every $(\rho, \bxi )$
and in the $\En^{-2r_{1}' }$-neighbourhood of every $\vec \boldom $.
\enl
{\newred
\ber Apart from the constants $\gamma$, the only difference between \eqref{no} and \eqref{no'} are the scales: $r_{1,1}$ and $r_{1,2}$. Since both these scales are much smaller than the scale $r'_1$ used in the definition of bad sets  $S$, 
the proof of lemma \ref{MGL'} is practically identical to the proof of lemma \ref{MGL}.
\enr}

{\newred 
\ber\label{Bourgainlargescale} We would like to emphasise once again that these Bourgain-type Lemmas are based on intrinsically exponential estimates with the large parameter $Z_0$ playing the crucial role. Because of that, we can apply these lemmas now (i.e. in step 2) when we have reached the exponential scale (in the spectral parameter) of the remainder. Attempts to apply these lemmas in Step 1 would lead to estimates being much weaker than required. See also Remark \ref{new:rmBourgain}. 
\enr}

{\newred 
\ber\label{new:B}
We would like to comment on the properties of the set  $S_{total}$ that we use in the proof of these Lemmas. Of course, we use the estimates of the measure and degree of it (Corollaries \ref{6.2} and \ref{new:19}). The only thing we use apart from this is a certain algebraic  structure of $S$ that intertwines the variables $\bxi$ and $\vec\boldom$. 
\enr}

Now we define 
\bee
{G } ^{(1)}(E):={G' } ^{(1)}(E)\cap{G'' } ^{(1)}(E).
\ene
Then 
\bee \meas({G^{(1)} }(E) )=\meas(\CG^{\vec\boldom(0)})-O(\En ^{- C_1r_{1,1}}), \ {\En\to \infty }.\label{mesLambda}
\ene

We also denote $E_q:=E_*+q$, $q=0,1,...$, and define the set ${\bf \mathcal G }^{{\vec\boldom }(1)}={\bf \mathcal G }^{{\vec\boldom }(1)}(\tilde E)$ by 
\bee\label{CG1}
{\bf \mathcal G }^{{\vec\boldom }(1)}(E_q)
:=\CG^{\vec\boldom(0)}\cap\left(\cap _{k=q}^{\infty}{G^{(1)} }_{} (\En _k,Z_0)\right);
\ene
for $\tilde E\in[E_q,E_{q+1})$ we put ${\bf \mathcal G }^{{\vec\boldom }(1)}(\tilde E):={\bf \mathcal G }^{{\vec\boldom }(1)}(E_q)$. 

\bec\label{MGC1}
We have:
\bee \meas({\mathcal G }^{\vec\boldom (1)}(\tilde E))=\meas({\mathcal G }^{\vec\boldom (0)})-O(\tilde E ^{- C_1r_{1,1}}),\ {\tilde E\to \infty }.
\label{mesbfLambda1}
\ene
\enc
The corollary easily follows from \eqref{mesLambda} and a power  growth of $r_{1,1}(\En_q)$ with $q$. 
\ber
Notice that $\gamma$ (and $\tilde\gamma$)  depend on $\vec\boldom$ and $\bxi$ highly non-trivially, and this dependence cannot be easily controlled. This lack of control makes further proof of our results much more involved. 
\enr

As we have already stated, we will assume $p=1$ and omit writing index $p$ in the rest of this section.  Since the proof of lemma \ref{MGL'} is similar to the proof of lemma \ref{MGL}, we concentrate on that proof. It is  based on Lemmas 1.18 and 1.20 in \cite{Bou1} for semi-algebraic sets. Most technical complication in our case come from the necessity to treat
vectors $\bn $ located  relatively close to coordinate hyperplanes.  Semi-algebraic sets needed for the proof are introduced in Section \ref{S6.2}. The set $G'$ is described in Section
\ref{S6.3}, where estimate \eqref{mesLambda0} is also proved. Proofs of \eqref{no} and  of the stability of $\gamma $ with respect to $(\vec \boldom, \rho, \bxi )$ are in Section \ref{S6.4}.

Until the end of this section we will use several implicit constants. We will denote them by $\cun_j$ (for small constants) and $\Cun_j$ (for large constants). Those constants are closely related to the constants $\hat C_j$, $j=1,2,3$ from Appendices 2-4. While they are implicit, they can be chosen uniform and depend only on $d$ and $l$.

\subsection{Semi-algebraic sets.} \label{S6.2} Here we investigate properties of the set $S$  given by \eqref{S} and introduce new semi-algebraic sets, see \eqref{0,s}, \eqref{k,s},
needed for proving our lemmas. 
{\newred Recall that a set $S\subset \R^d$ is a semi-algebraic set of degree (not greater than) $N$ means that $S$ can be defined by finitely many polynomial inequalities of degrees $j_1,j_2,...,j_l$ and $j_1+...+j_l\le N$.} 

\bel \label{6.1} The set $S_{\tim}$ is a semi-algebraic subset in $\R^{ld+1+d}$ of degree  $E^{\sigma _{1,d}}$. 
\enl
\bec \label{6.2} The set $S$ is a semi-algebraic subset in $\R^{ld+1+d}$ of degree  $E^{2d(l+1)+1}$. 
\enc

The corollary follows from the estimate \eqref{Nest}  and $\sigma _{1,d}<1$.

\bep Using  Cramer's rule we can rewrite the inequality for the resolvent in the definition of $S_{\tim}$ in terms of determinants, the biggest matrix being that of the operator
$(H(\tilde\BUps^{\Z}(\bxi),\bxi )-\rho^2)$ itself. The size of this matrix is equal to the number of elements $\#\{ \tilde\BUps^{\Z}(\bxi)\}$ in our pre-cluster. Using \eqref{importantest} and \eqref{superext}, we obtain
\bee
\#\{ \tilde\BUps^{\Z}(\bxi)\}<E^{d^2(l+\mu)\sigma _{1,d-1}+d\sigma _{1,d-1,1}}.   \label{Oct29} 
\ene 
Taking into account that each diagonal term is a quadratic polynomial, we obtain that the determinant of the biggest matrix is a polynomial (in $\vec\boldom,\rho,\bxi$) of degree less than 
$$2E^{d^2(l+\mu)\sigma _{1,d-1}+\sigma _{1,d-1,1}d}.$$ 
Since we square the determinants to compute the Hilbert-Schmidt norm of the resolvent,  the inequality for the resolvent is an inequality for  a polynomial of degree less than $4E^{d^2(l+\mu)\sigma _{1,d-1}+\sigma _{1,d-1,1}d}$. We have $2d(l+1)+3$ inequalities in the definition of $S_{\tim}$. Therefore, the degree of $S_{\tim}$ does not exceed $4(2dl+2d+3)E^{d^2(l+\mu)\sigma _{1,d-1}+\sigma _{1,d-1,1}d}<E^{\sigma _{1,d}}$. \enp

Let $(S_{\tim})_{\mathrm {cs}}(\vec \boldom ,\rho )\subset \R^d$  be a cross section of $S_{\tim}$:
\bee
(S_{\tim})_{\mathrm {cs}}(\vec \boldom ,\rho ):=\left\{\bxi : (\vec \boldom ,\rho, \bxi )\in S_{\tim}\right\}. \label{CS-Sn}
\ene
We will use analogous notations for cross sections of other sets.


\bel \label{L5.5} The Lebesgue measure of $(S_{\tim})_{\mathrm {cs}}(\vec \boldom, \rho )$ satisfies the following estimate:
\bee \label{(7.21)}
\meas((S_{\tim})_{\mathrm {cs}}(\vec \boldom, \rho ))
<
 E^{-r_1'-d-1+\sigma _{1,d}}. 
 \ene 
 \enl
 \bec The Lebesgue measure of $S_{\mathrm {cs}}(\vec \boldom, \rho )$ satisfies the following estimate:
 \bee
\meas((S)_{\mathrm {cs}}(\vec \boldom, \rho ))
<
 E^{-r_1'+2d(l+1)-d}.  \label{mesSrho}\ene \enc
 To obtain this corollary we use \eqref{S} , \eqref{Nest} and $\sigma _{1,d}<1$.
 
 \bec\label{new:19} 
 The Lebesgue measures of $S_{\tim}$, $S$ satisfy the estimates:
 \bee
\meas(S_{\tim})
<
 E^{-r_1'-d -1+\sigma _{1,d}}, \ene 
 \bee
\meas( S)
<
 E^{-r_1'+2d(l+1)-d}.  \label{mesS}\ene \enc
 This corollary is obtained by integrating with respect to $\vec \boldom $ and $\rho $.
 
 \bep Obviously, $\nabla _{\bxi }H(\tilde\BUps^{\Z}(\bxi),\bxi )=\nabla _{\bxi }H_0(\tilde\BUps^{\Z}(\bxi),\bxi )=2\bxi I+\CE$, $\CE $ being a diagonal operator.
Using Lemmas \ref{1}, \ref{2'} and formulas \eqref{ext}, \eqref{superext}, we easily obtain that both the norm of $\CE$ and its gradient with respect to $\bxi$ are $o(E)$. Hence, $\nabla _{\bxi }\lambda _i(\bxi )=2\bxi +o(E)$ for every single eigenvalue of $H(\tilde\BUps^{\Z}(\bxi),\bxi )$. The derivative of every eigenvalue in the direction $\bxi$ is piece-wise continuous as a function of one variable $\|\bxi \|$, while each eigenvalue is a continuous function of $\|\bxi \|$. Therefore, the inequality
$\left|\lambda _i(\bxi )-\rho ^2\right|<E^{-r_1'}$ holds on a set of measure less than $E^{-r_1'-2-(d-1)}$ in $\CA_{\tim}^{\bxi(0)}$. By \eqref{Oct29}, the number of eigenvalues does not exceed 
$E^{d^2(l+\mu)\sigma _{1,d-1}+\sigma _{1,d-1,1}d}<E^{\sigma _{1,d}}$. Hence, the inequality $\min _{j}\left|\lambda _j(\bxi )-\rho ^2\right|<E^{-r_1'}$ holds on a set of measure less than $E^{-r_1'-2-(d-1)+\sigma _{1,d}}$. \enp

\ber
This proof used the fact that the size of the pre-cluster $\tilde\BUps^{\Z}(\bxi)$ was not very large. Unfortunately, we will not be able to use this argument at the further stages of our procedure, so the proof of the corresponding inductive statement in section \ref{section8} will be a bit more complicated. 
\enr

Now, we want to introduce artificial variables $\by_1,\dots,\by_s\in\R^d$, with $s$ to be determined later. 
We define the set
$\tilde S ^{(s)}\subset \R^{ld+1+(s+1)d}$ by 
\bee \label{5.15}
\tilde S^{(s)}:=\left \{(\vec \boldom,\rho,\bxi, \by _1,...,\by _{s}): (\vec \boldom, \rho, \bxi )\in S, \ (\vec \boldom, \rho, \bxi+\by _j)\in S,\ j=1,...s\right\}.
\ene
Obviously,
\bee
\tilde S^{(s)}=\left \{(\vec \boldom, \rho,\bxi, \by _1,...,\by _{s}):  (\vec \boldom, \rho, \bxi, \by _j)\in \tilde S^{(1)},\ j=1,...s\right\}.
\ene

\bel \label{L5.6} The set $\tilde S^{(s)}$ is a semi-algebraic set of degree  $(s+4)E^{2d(l+1)}$ and
\bee
\meas(\tilde S^{(s)}_{\mathrm {cs}}(\vec \boldom, \rho ,\bxi ))
<E^{-\left(r_1'-2d(l+1)+d\right)s}.
\ene \enl
\bec
\bee 
\meas(\tilde S^{(s)})
<E^{-\left(r_1'-2d(l+1)+d\right)s+d}.
\label{mestS}
\ene 
\enc
\bep  We integrate with respect to $\by _j$, taking into account \eqref{mesSrho}. Integration over $\vec \boldom $, $\rho $ and $\bxi $ produces a factor $E^d$ (even slightly better). 
\enp
Let $(\tilde S_{\mathrm {cs}}^{(s)}(\vec \boldom ))_{\mathrm pr}^{\by _1, ...,\by _{s}}$ be the projection of $\tilde S^{(s)}(\vec \boldom )$ on the space $\R^{sd}$ of vectors $(\by _1, ...,\by _{s})$:
\bee \label{5.21}
(\tilde S_{\mathrm {cs}}^{(s)}(\vec \boldom ))_{\mathrm pr}^{\by _1, ...,\by _{s}}:=\left\{ (\by _1, ...,\by _{s}): \exists (\rho, \bxi )\mbox{ such that } (\vec \boldom, \rho, \bxi, \by _i) \in \tilde S^{(1)} , i=1,...,s\right\}. 
\ene
It is easy to see that

\bee  \label{5.19}
(\tilde S_{\mathrm {cs}}^{(s)}(\vec \boldom ))_{\mathrm pr}^{\by _1, ...,\by _{s}} =\left\{ ( \by _1, ...,\by _{s}): \tilde S^{(1)}_{\mathrm {cs}}(\vec \boldom, \by _1)\cap...\cap \tilde S^{(1)}_{\mathrm {cs}}(\vec \boldom, \by _{s})\neq \emptyset \right\}, \ene
where $\tilde S^{(1)}_{\mathrm {cs}}(\vec \boldom, \by )$ is the cross section of $\tilde S^{(1)}$ by $\vec \boldom, \by $.
\bel \label{5.10}
Let $s=2^{d+1}$. Then the set $(\tilde S_{\mathrm {cs}}^{(s)}(\vec \boldom ))_{\mathrm pr}^{\by _1, ...,\by _{s}}$ is a semi-algebraic set of degree $E ^{\Cun_1}$ in $\R^{sd}$. Its Lebesgue measure satisfies:
\bee 
\meas \left((\tilde S_{\mathrm {cs}}^{(s)}(\vec \boldom ))_{\mathrm pr}^{\by _1, ...,\by _{s}}\right)<
 E ^{\Cun_1}E^{-r_1' \cun_1}.  \label{mesS-pr-omega}\ene 
 \enl
 \bep 
 We apply Lemma 1.18 in \cite{Bou1}, formulated in Appendix 3.  We take there $A=\tilde S^{(1)}_{\mathrm {cs}}(\vec \boldom )$, $r=1+d$, $t=(\rho , \bxi )$, $x_i=\by _i $, 
 $B=5E^{2d(l+1)}$, $\eta =E^{-r_1'+2d(l+1)-d}$, the last estimate being given by Lemma \ref{L5.6} with $s=1$.
 \enp
 
 Let $(\tilde S^{(s)})_{\mathrm pr}^{\vec \boldom, \by _1, ...,\by _{s}}$ be the projection of $\tilde S^{(s)}$ on the space $\R^{ld+sd}$ of vectors $(\vec \boldom, \by _1, ...,\by _{s})$:
 \bee \label{Sspr} 
 (\tilde S^{(s)})_{\mathrm pr}^{\vec \boldom, \by _1, ...,\by _{s}}:=\left\{ (\vec \boldom ,\by _1, ...,\by _{s}): \exists (\rho, \bxi )\mbox{ such that } (\vec \boldom, \rho, \bxi, \by _1, ...,\by _{s}) \in \tilde S^{(s)}\right\}.
 \ene
 It is easy to see that
\bee \label{5.22}
(\tilde S^{(s)})_{\mathrm pr}^{\vec \boldom, \by _1, ...,\by _{s}}=\left\{ (\vec \boldom ,\by _1, ...,\by _{s}): \exists (\rho, \bxi )\mbox{ such that } (\vec \boldom, \rho, \bxi, \by _j) \in \tilde S^{(1)} ,  j=1,...,s\right\} 
\ene
and
\bee 
(\tilde S^{(s)})_{\mathrm pr}^{\vec \boldom, \by _1, ...,\by _{s}}=\left\{ (\vec \boldom ,\by _1, ...,\by _{s}): \tilde S^{(1)}_{\mathrm {cs}}(\vec \boldom, \by _1)\cap...\cap \tilde S^{(1)}_{\mathrm {cs}}(\vec \boldom, \by _{s})\neq \emptyset \right\}, \ene
where as usual $\tilde S^{(1)}_{\mathrm {cs}}(\vec \boldom, \by )$ is the cross section of $\tilde S^{(1)}$ by $\vec \boldom, \by $.
The set $(\tilde S_{\mathrm {cs}}^{(s)}(\vec \boldom ))_{\mathrm pr}^{\by _1, ...,\by _{s}} $ introduced by \eqref{5.21} is a cross section of $(\tilde S^{(s)})_{\mathrm pr}^{\vec \boldom, \by _1, ...,\by _{s}}$.
\bel Let $s=2^{d+1}$. Then, the set $(\tilde S^{(s)})_{\mathrm pr}^{\vec \boldom, \by _1, ...,\by _{s}}$ is a semi-algebraic set of degree $E ^{\Cun_2}$ in $\R^{ld+sd}$. Its Lebesgue measure satisfies:
\bee 
\meas((\tilde S^{(s)})_{\mathrm pr}^{\vec \boldom, \by _1, ...,\by _{s}})
<
 E ^{\Cun_1}E^{-r_1'\cun_1}.  \label{mesS-pr}\ene 
 \enl
 \bep 
 The set $(\tilde S^{(s)})_{\mathrm pr}^{\vec \boldom, \by _1, ...,\by _{s}}$ is a semi-algebraic set of  degree $(s+4)^{\hat C_3}E^{\hat C_3 2d(l+1)}$ by Lemma \ref{L5.6} and the Tarski-Seiderberg principle, see Appendix 4. Integrating 
 \eqref{mesS-pr-omega} with respect to $\vec \boldom $, we obtain \eqref{mesS-pr}. \enp 
 
  Finally, we can introduce a  set  we are going to use  to treat vectors $\bn$ (see \eqref{no}) that are sufficiently far from all the coordinate hyperplanes. Let ${\mathbf {T}}\in \R^{ld}$, ${\mathbf {T}} =(\bt _1,....\bt _l)$,  $\bt _j\in \R^d$, $j=1,...,l$. 
Let $\tilde {\mathbf {T}}$ be a linear mapping of ${\mathbf {T}}$ to $\R^d$: 
  \bee
   {\mathbf {\tilde T}}= {\mathbf {\tilde T}}({\mathbf {T}}):=\sum _{j=1}^l\bt _j. 
\ene  
  Let $L$ be a large parameter (to be fixed later). We  introduce the set which we denote by
  $\tilde S_{pr}^{(0,s)}\subset \cube^{ld}\times [-L,L]^{sld}$:
  \bee \label{0,s}
  \tilde S_{pr}^{(0,s)}:=
  \left\{ \left(\vec \boldom ,\bold T_1, ...,\bold T_s\right): 
 \left(\vec\boldom, \tilde {\bold T }_1, ...,\tilde{ \bold T}_s\right)\in 
 (\tilde S^{(s)})_{\mathrm pr}^{\vec \boldom, \by _1, ...,\by _{s}}, \bold T_i\in (-L,L )^{ld},\, i=1,...,s\right\}, 
 \ene
 where we have denoted $\tilde {\bold T }_j=\tilde {\bold T }({\bold T }_j)$.

 \bel \label{Lemma6.10} Let $s=2^{d+1}$. Then the set $\tilde S_{pr}^{(0,s)}$ is a semi-algebraic set of degree $E ^{\Cun_2}$ in $\R^{ld(s+1)}$. Its Lebesgue measure satisfies:
\bee 
\meas(\tilde S_{pr}^{(0,s)})
<
 E ^{\Cun_1}E^{-r_1'\cun_1}(2L)^{(l-1)ds}.  \label{mesS-pr-0}\ene 
 \enl
 \bep  
 Obviously,  set $\tilde S_{pr}^{(0,s)}$ is a semi algebraic set of degree just $2sld$ larger than $\tilde S_{pr}^{(s)}$. For convenience, we still denote the estimate by $E^{\Cun_2}$.
 Estimate \eqref{mesS-pr-0} follows from \eqref{mesS-pr} and the condition $\bold T_i\in (-L,L )^{ld}$, $i=1,...,s$.\enp

  We need yet more sets to treat vectors $\bn$ in \eqref{no} that lie relatively close to the coordinate hyperplanes. These sets will be denoted $\tilde S_{pr}^{(k,s)}$; each such set will take care of  vectors $\bn$ relatively close to coordinate hyperplanes of dimension ${\newred l}-k$. These sets will be properly defined in \eqref{k,s}, but first we need to define more objects similar to those discussed above.

 
 

Define $S_{i,n_i}^{(1)}\subset \R^{2ld+1+d}$ ($n_i\in\Z$) as follows:
$$
S_{i,n_i}^{(1)}:=\left\{ \left(\vec \boldom, \rho, \bxi , \bold T \right): \left(\vec \boldom, \rho,  \bxi +n_i\boldom _i+\sum _{j\neq i}\bt _j \right)\in S,\ \bold T\in (-L,L )^{ld}\right\};
$$
recall that $\bold T =(\bt _1,....\bt _l)$ and $\bt_j\in \R^d$, $j=1,..,l$. Obviously, the definition of $S_{i,n_i}^{(1)}$ does not include any conditions on $\bt _i$. Hence if $ \left(\vec \boldom, \rho, \bt _1,...,\bt_i,..,.\bt _l\right)\in S_{i,n_i}^{(1)}$,
then $ \left(\vec \boldom, \rho, \bt _1,...,\bt_i+\tilde \bt_i,...,\bt _l\right)\in S_{i, n_i}^{(1)}$ for any $\tilde \bt_i\in \R^d$, assuming that we still have $\bold T\in (-L,L )^{ld}$. 
 Clearly, $S_{i,n_i}^{(1)}$ is a semi-algebraic set and it has degree by $2ld$ larger than $S$. By Corollary \ref{6.2}, the degree of $S_{i,n_i}^{(1)}$ is at most $E^{2d(l+1)+2}$.

Next, for $\gamma_1>0$ we define  
\bee 
S_i^{(1)}(\gamma _1):=\cup _{|n_i|<E^{\gamma _1r_{1,1}}}S_{i,n_i}^{(1)}
\ene
and
\bee 
S^{(1)}(\gamma _1):=\cup _{i=1,...,l}S_i^{(1)}(\gamma _1). 
\ene

Similarly, let $S^{(k)}_{(i_1,...i_k),(n_{i_1}, ..., n_{i_k})}\subset \R^{2ld+1+d}$, $k\le {\newred l}$, be defined by the following formula: 
\bee \label{Jan23-b}
\bes
&S^{(k)}_{(i_1,...i_k),(n_{i_1}, ..., n_{i_k})}:=\\
&\left\{ \left(\vec \boldom, \rho, \bxi, \bold T\right): \left(\vec \boldom, \rho,  \bxi +\sum _{j=1}^k n_{i_j}\boldom _{i_j}+\sum _{j\neq i_1,...i_k}\bt _j \right)\in S, \ \bold T\in (-L,L )^{ld}\right\}.
\end{split}
\ene
Here and below, we use the convention $i_j\neq i_{j'}$ if $j\neq j'$.
Further,  $S^{(k)}_{(i_1,...i_k),(n_{i_1}, ..., n_{i_k})}$ is a semi-algebraic set of degree by $2ld$
larger than $S$. By Corollary \ref{6.2}, the degree of $S^{(k)}_{(i_1,...i_k),(n_{i_1}, ..., n_{i_k})}$ is at most $E^{2d(l+1)+2}$.

Next, for $\gamma_1,...,\gamma_k>0$ we define  
\bee \label{6.22} 
S^{(k)}_{i_1,...i_k}(\gamma _1,...,\gamma _k):=\cup _{|n_{i_j}|<E^{\gamma _jr_{1,1}}, j=1,...,k}\ S^{(k)}_{(i_1,...i_k),(n_{i_1}, ..., n_{i_k})} 
\ene
and 
\bee \label{6.23} 
S^{(k)}(\gamma _1,...,\gamma _k):=\cup _{i_1,...i_k=1,\ i_j\neq i_{j'}}^l S^{(k)}_{i_1,...i_k}(\gamma _1,...,\gamma _k). 
\ene
\bel \label{L6.10} The set $S^{(k)}(\gamma _1,...,\gamma _k)$ is a semi-algebraic set of degree 
$$
 \Cun_3 E^{(\sum _{j=1}^k\gamma _{j})r_{1,1}+2d(l+1)+2}.
$$

 Its cross-section satisfies the estimate:
\bee 
\meas((S^{(k)}(\gamma _1,...,\gamma _k))_{\mathrm cs}(\vec \boldom, \rho, \bxi ))
< \Cun_3 E^{-r_1'+(\sum _{j=1}^k\gamma _{j})r_{1,1}+2d(l+1)-d}(2L)^{(l-1)d}. \label{measure-S-k}
\ene
\enl
\bec 
 The measure of the set $S^{(k)}(\gamma _1,...,\gamma _k)$ can be estimated like the measure of its cross-section in \eqref{measure-S-k}:
\bee 
\meas(S^{(k)}(\gamma _1,...,\gamma _k))
\leq \Cun_3 E^{-r_1'+(\sum _{j=1}^k\gamma _{j})r_{1,1}+2d(l+1)}(2L)^{(l-1)d}. \label{measure-S-k*}
\ene
\enc 

\bep The degree of $S^{(k)}(\gamma _1,...,\gamma _k)$ is, obviously,  the degree of $S^{(k)}_{(i_1,...i_k),(n_{i_1}, ..., n_{i_k})}$ mutiplied by $M_0$, where $M_0$ is the number of items in the union \eqref{6.22}, \eqref{6.23}.  Taking into account that $M_0=O(E^{(\sum _{j=1}^k\gamma _{j})r_{1,1}})$, we obtain the above estimate for the degree of 
$S^{(k)}(\gamma _1,...,\gamma _k)$. By definition,
\bee \label{Jan23-a}
\bes
& (S^{(k)}_{(i_1,...i_k),(n_{i_1}, ..., n_{i_k})})_{\mathrm cr}(\vec \boldom, \rho, \bxi )\\
 &=
\left\{  \bold T \in [-L,L]^{ld}: 
\bxi +\sum _{j=1}^k n_{i_j}\boldom _{i_j}+\sum _{j\neq i_1,...i_k}\bt _j \in S(\vec \boldom, \rho )\right\}.
\end{split}
\ene
Integrating with respect to $\bt=\sum _{j\neq i_1,...i_k}\bold{t}_{j}$ first and using \eqref{mesSrho}, we obtain
\bee 
\meas((S^{(k)}_{(i_1,...i_k),(n_{i_1}, ..., n_{i_k})})_{\mathrm cr}(\vec \boldom, \rho, \bxi ))
<E^{-r_1'+2d(l+1)-d}(2L)^{(l-1)d}. \label{measure-S-k*a'}
\ene
This implies
\bee 
\meas((S^{(k)}_{(i_1,...i_k)}(\gamma _1,...,\gamma _k))_{\mathrm cr}(\vec \boldom, \rho, \bxi ))
<\Cun_3 E^{-r_1'+(\sum _{j=1}^k\gamma _{j})r_{1,1}+2d(l+1)-d}(2L)^{(l-1)d}, \label{measure-S-k*a}
\ene
and, hence, \eqref{measure-S-k}.

 \enp

Next, let 
$\tilde S ^{(k,s)}=\tilde S ^{(k,s)}(\gamma _1,...,\gamma _k)\subset \R^{(s+1)ld+d+1}$, 
\bee \label{5.38}
\bes
&\tilde S^{(k,s)}(\gamma _1,...,\gamma _k)\\
&:=\left \{(\vec \boldom,\rho,\bxi, \bold{T} _1,...\bold T _{s}): (\vec \boldom, \rho, \bxi )\in S, (\vec \boldom, \rho, \bxi, \bold T _i)\in S^{(k)}(\gamma _1,...,\gamma _k),\   i=1,...s\right\}.
\end{split}
\ene
Obviously,
\bee \label{k=k1} \tilde S^{(k,1)}= S^{(k)} 
\ene
and
\bee
\bes
&\tilde S^{(k,s)}(\gamma _1,...,\gamma _k)\\
&=\left \{(\vec \boldom, \rho,\bxi, \bold T _1,...\bold T_{s}):  (\vec \boldom, \rho, \bxi, \bold T_i)\in \tilde S^{(k,1)}(\gamma _1,...,\gamma _k),\ i=1,...s.\right\}.
\end{split}
\ene

\bel \label{L5.15} The set $\tilde S^{(k,s)}(\gamma _1,...,\gamma _k)$ is a semi-algebraic set of  degree  
\bee
s\Cun_3  E^{(\sum _{j=1}^k\gamma _{j})r_{1,1}+2d(l+1)+2}
\ene 
and
\bee
\meas(\tilde S^{(k,s)}_{\mathrm {cs}}(\vec \boldom, \rho ,\bxi ))
<\left(\Cun_3 E^{-r_1'+(\sum _{j=1}^k\gamma _{j})r_{1,1}+2d(l+1)-d}(2L)^{(l-1)d}\right)^s \label{6.32}
\ene \enl
\bec
\bee 
\meas(\tilde S^{(k,s)})
<\left(\Cun_3 E^{-r_1'+(\sum _{j=1}^k\gamma _{j})r_{1,1}+2d(l+1)-d}(2L)^{(l-1)d}\right)^s E^{d}. \label{6.33}
\ene 
\enc
\bep  Obviously, the degree of $\tilde S^{(k,s)}$ is equal to that of $\tilde S^{(k)}$ multiplied by $s$, see Lemma \ref{L6.10}. The relation $(\vec \boldom, \rho, \bxi, \bold T_1,\dots, \bold T_s)\in \tilde S^{(k,s)}(\gamma _1,...,\gamma _k)$ can be rewritten as  $(\vec \boldom, \rho, \bxi, \bold T _i)\in S^{(k)}(\gamma _1,...,\gamma _k),$ for $i=1,...,s$; i.e.,
$\bold T_i\in S^{(k)}(\gamma _1,...,\gamma _k)(\vec \boldom, \rho, \bxi ).$ Taking into account \eqref{measure-S-k}, we arrive to \eqref{6.32}.
\enp
Let $(\tilde S^{(k,s)}_{\mathrm {cs}}(\vec \boldom ))_{\mathrm pr}^{\bold T _1, ...,\bold T _{s}}$ be the projection of the cross-section $\tilde S^{(k,s)}_{\mathrm {cs}}(\vec \boldom )$ on the space $\R^{dls}$ of vectors $(\bold T _1, ...,\bold T _{s})$:
\bee 
\label{5.41} (\tilde S^{(k,s)}_{\mathrm {cs}}(\vec \boldom ))_{\mathrm pr}^{\bold T _1, ...,\bold T _{s}}=\left\{ (\bold T _1, ...,\bold T _{s}): \exists (\rho, \bxi )\mbox{ such that } (\vec \boldom, \rho, \bxi, \bold T _i) \in \tilde S^{(k,1)} , i=1,...,s.\right\}. \ene
It is easy to see that
\bee 
(\tilde S^{(k,s)}_{\mathrm {cs}}(\vec \boldom ))_{\mathrm pr}^{\bold T _1, ...,\bold T _{s}} =\left\{ ( \bold T_1, ...,\bold T _{s}): \tilde S^{(k,1)}_{\mathrm {cs}}(\vec \boldom, \bold T _1)\cap...\cap \tilde S^{(k,1)}_{\mathrm {cs}}(\vec \boldom, \bold T _{s})\neq \emptyset \right\},
 \ene
where $\tilde S^{(k,1)}_{\mathrm {cs}}(\vec \boldom, \bold T )$ is the cross section of $\tilde S^{(k,1)}$ by $\vec \boldom, \bold T $.

\bel \label{L5.17} 
Let $s=2^{d+1}$ and assume that $\sum _{j=1}^k\gamma _{j}<\frac12$. Then, the set $(\tilde S^{(k,s)}_{\mathrm {cs}}(\vec \boldom ))_{\mathrm pr}^{\bold T _1, ...,\bold T _{s}}$ is a semi-algebraic set of  degree $E^{\Cun_4}E^{\Cun_4(\sum _{j=1}^k\gamma _{j})r_{1,1}}$ in $\R^{sdl}$. Its Lebesgue measure satisfies:
\bee 
\meas( (\tilde S^{(k,s)}_{\mathrm {cs}}(\vec \boldom ))_{\mathrm pr}^{\bold T _1, ...,\bold T _{s}})<
 E^{-r_1'\cun_2}L^{\Cun_4}E^{\Cun_4}E^{\Cun_4(\sum _{j=1}^k\gamma _{j})r_{1,1}}.  \label{mesS-pr-omega*}\ene 
 \enl
 \bep We apply Lemma 1.18 in \cite{Bou1} (see Appendix 3), where we take there 
 $A=\tilde S^{(k,1)}(\vec \boldom )$, 
 $r=1+d$, $t=(\rho , \bxi )$, $x_i=\bold T_i $, 
 $B=\Cun_3 E^{(\sum _{j=1}^k\gamma _{j})r_{1,1}+2d(l+1)+2}$ (see Lemma \ref{L5.15}), 
 $\eta =\Cun_3 E^{-r_1'+(\sum _{j=1}^k\gamma _{j})r_{1,1}+2d(l+1)-d}(2L)^{(l-1)d}$.
 \enp
 
 Let $\tilde S_{\mathrm pr}^{(k,s)}=\tilde S_{\mathrm pr}^{(k,s)}(\gamma_1,...,\gamma_k)$ be the projection of $\tilde S^{(k,s)}$ on the space $\R^{ld(1+s)}$ of vectors $(\vec \boldom, \bold T _1, ...,\bold T _{s})$:
\bee 
\label{5.44} 
\bes
\tilde S_{\mathrm pr}^{(k,s)}:&=(\tilde S^{(k,s)})_{pr}^{\vec \boldom, \bold T _1, ...,\bold T _{s}}\\
&=\left\{ (\vec \boldom ,\bold T _1, ...,\bold T _{s}): \exists (\rho, \bxi )\mbox{ such that } (\vec \boldom, \rho, \bxi, \bold T _i) \in \tilde S^{(k,s)} , i=1,...,s\right\}. 
\end{split}
\ene
It is easy to see that

\bee \label{k,s}
\tilde S_{\mathrm pr}^{(k,s)}=\left\{ (\vec \boldom ,\bold T _1, ...,\bold T _{s}): \tilde S^{(k,1)}_{\mathrm {cs}}(\vec \boldom, \bold T _1)\cap...\cap \tilde S^{(k,1)}_{\mathrm {cs}}(\vec \boldom, \bold T _{s})\neq \emptyset \right\}, \ene
where $\tilde S^{(k,1)}_{\mathrm {cs}}(\vec \boldom, \bold T )$ is the cross section of $\tilde S^{(k,1)}$ by $\vec \boldom, \bold T$.
Obviously, the set $(\tilde S^{(k,s)}_{\mathrm {cs}}(\vec \boldom ))_{\mathrm pr}^{\bold T _1, ...,\bold T _{s}}$ introduced above is a cross section of $\tilde S_{\mathrm pr}^{(k,s)}$. We also note that the set defined in \eqref{0,s} is a special case of \eqref{k,s} when $k=0$. 

\bel \label{L6.15} Let $s=2^{d+1}$. Then the set $\tilde S_{pr}^{(k,s)}(\gamma _1,...,\gamma _k)$ is a semi-algebraic set of degree $E^{\Cun_5}E^{\Cun_5(\sum _{j=1}^k\gamma _{j})r_{1,1}}$ in $\R^{(s+1)dl}$. Its Lebesgue measure satisfies:
\bee 
\meas(\tilde S_{pr}^{(k,s)})<
 E^{-r_1'\cun_2}L^{\Cun_4}E^{\Cun_4}E^{\Cun_4(\sum _{j=1}^k\gamma _{j})r_{1,1}}. \label{mesS-pr*}\ene 
 \enl
 \bep The estimate of the degree of the set $\tilde S_{pr}^{(k,s)}$ follows from the Tarski-Seiderberg principle, see (1.9) in \cite{Bou1}. Integrating 
 \eqref{mesS-pr-omega*} with respect to $\vec \boldom $, we obtain \eqref{mesS-pr*}. \enp

Next, for a given positive $\gamma$ we consider 
\bee\label{K-gamma}
K(\gamma ):=\Om(\En^{\gamma r_{1,1}}).
\ene
 We split it into several regions. Let  $\gamma _1<\gamma $ and put 
 \bee \label{Pi0} \Piy ^0 (\gamma _1, \gamma ):=\left \{\bn \in K(\gamma ): \min _{j{\newred =1,...,l}}|n_j|\geq E^{\gamma_1r_{1,1}}\right\}.
 \ene
 Obviously,  $\Piy ^0$ is the union of $2^l$ disjoint cubes with the side length $E^{\gamma r_{1,1}}-E^{\gamma _1r_{1,1}}$.
 Put 
 \bee \label{Pi1i} \Piy ^1 _i(\gamma _1, \gamma ):=\left \{\bn \in K(\gamma ): |n_i| <E^{\gamma _1r_{1,1}} \right\},\ i=1,...,{\newred l}
 \ene
 and
 \bee \label{Pi1} 
\Piy ^1= \Piy ^1(\gamma _1, \gamma ):=\cup _{i=1}^{{\newred l}}\Piy ^1_i.
 \ene
 Clearly,
 \bee \label{Jan22} K(\gamma )=\Piy ^0(\gamma_1 , \gamma )\cup \Piy ^1(\gamma_1 , \gamma).\ene 
 Note that $\Piy ^1 _i\cap \Piy ^1 _j \neq \emptyset $ if ${\newred l}>1$. Our goal is to split $K(\gamma )$ into a disjoint union of sets of similar type. Assume that $\gamma _1<\gamma _2<\gamma $ and define 
 $$\Piy ^2 _{i,j}(\gamma _1, \gamma _2, \gamma ):=\left \{\bn \in K(\gamma ): |n_i| <E^{\gamma _1r_{1,1}}, \  |n_j| <E^{\gamma _2r_{1,1}}\right\}, \ i\neq j,\ i,j=1,...,{\newred l}$$

 $$\Piy ^2 (\gamma _1, \gamma _2, \gamma ):=\cup _{i,j=1, i\neq j}^{{\newred l} }\Piy ^2 _{i,j}(\gamma _1, \gamma _2, \gamma ),$$
 and
  \bee \label{tildePi1}\tilde \Piy ^1 (\gamma _1, \gamma _2, \gamma ):=\Piy ^1 (\gamma _1, \gamma )\setminus \Piy ^2 (\gamma _1, \gamma _2, \gamma ).\ene
Clearly,
 the sets $\Piy ^0 $, $\tilde \Piy ^1$ and $\Piy ^2$ are all mutually disjoint and
 $$K(\gamma )= \Piy ^0(\gamma_1 , \gamma)\sqcup \tilde \Piy ^1(\gamma _1, \gamma _2, \gamma )\sqcup \Piy ^2 (\gamma _1, \gamma _2, \gamma ).$$
 Further, for $ 2<k\leq {\newred l}$ and 
  $\gamma _1<\gamma _2<...<\gamma _k<\gamma $ we define
 $$\Piy ^k_{i_1...i_k}(\gamma _1,...,\gamma _k,\gamma ):=\left \{\bn \in K(\gamma ): |n_{i_j}| <E^{\gamma _jr_{1,1}},\ j=1,...,k\right\},$$
 (as usual, all indices $i_1,...,i_k$ are different here) and 
  $$\Piy ^k (\gamma _1, ...,\gamma _k, \gamma ):=\cup _{i_1,...i_k=1}^{{\newred l}}\Piy ^k _{i_1,...i_k} (\gamma _1, ...,\gamma _k, \gamma ).$$
  Next,
  \bee \label{Jan28} \tilde \Piy ^k (\gamma _1, ...,\gamma _k, \gamma _{k+1}, \gamma ):=\Piy ^k (\gamma _1, ...,\gamma _k, \gamma )\setminus \Piy ^{k+1} (\gamma _1, ...,\gamma _k, \gamma _{k+1}, 
  \gamma ), \  \  k=1,..., {\newred l}-1.\ene
  {\newred
  In principle,  $\tilde \Piy ^k$ consists of points $\bn$ such that the absolute value of some $l-k$ of their coordinates is much much bigger than the absolute value of the rest $k$  coordinates. This is done with a view of applying Lemma \ref{LB}.
  }
  Obviously, 
  $$K(\gamma )=\Piy ^0(\gamma_1 , \gamma)\sqcup_{k=1}^{{\newred l}-1} \tilde \Piy ^k(\gamma _1, ...,\gamma _{k+1} ,\gamma )\sqcup  \Piy ^{{\newred l}}(\gamma _1,....,\gamma _d, \gamma ).$$
  Clearly,
  $\Piy ^{{\newred l}}(\gamma _1,....,\gamma _{\newred l}, \gamma )\subset K(\gamma _{\newred l})$ and $\gamma $ can be omitted from the formula for $\Piy ^{{\newred l}}$. Further we will write just $\Piy ^{\newred l}(\gamma _1,....,\gamma _{\newred l})$.

\bel \label{L6.18} Let $\vec \boldom \in \cube^{ld}$, $(\vec \boldom, \rho, \bxi )\in S_{total}$. Assume there is $$\bn \in \tilde \Piy ^k (\gamma _1, ...,\gamma _k, \gamma _{k+1}, \gamma )$$ such that $(\vec \boldom, \rho, \bxi+ \bn \vec \boldom) \in S_{total}$. Then there is 
$$\bn ^*\in  \Piy ^0 (\gamma _{k+1}, \gamma )$$ such that $$(\vec \boldom, \rho , \bxi, \bn ^*\vec \boldom )\in S^{(k)}(\gamma _1, ...,\gamma _k),$$ the set 
$S^{(k)}(\gamma _1, ...,\gamma _k)$ being defined by \eqref{Jan23-b}, \eqref{6.22}, \eqref{6.23} with $L=E^{\gamma r_{1,1}}$. \enl
\bep Indeed, let  $(\vec \boldom, \rho, \bxi+ \bn \vec \boldom) \in S_{total}$, $\bn \in \tilde \Piy ^k (\gamma _1, ...,\gamma _k, \gamma _{k+1}, \gamma )$. By the definition of $ \tilde \Piy ^k $, there is a finite sequence
$i_1,...i_k$,  $1\leq i_j \leq {\newred l}$  ($j=1,...,k$), such that $|n_{i_j}|<E^{\gamma _j r_{1,1}}$ for all  $j=1,...,k$, and $|n_i|>E^{\gamma _{k+1}r_{1,1}}$ when $i\neq i_1,...i_k$.  Hence, by \eqref{Jan23-a}, \eqref{6.22},
$$\bn \vec \boldom \in S^{(k)}_{(i_1,...i_k),(n_{i_1}, ...,n_{i _k})}(\vec \boldom, \rho ,\bxi)\subset S^{(k)}(\gamma _1, ...,\gamma _k)(\vec \boldom, \rho ,\bxi).$$ 

Note that the definition \eqref{Jan23-a} of 
$ S^{(k)}_{(i_1,...i_k),(n_{i_1}, ...,n_{i _k})}(\vec \boldom, \rho ,\bxi)$ does not include components $\bt_{i_1},...,\bt_{i_k}$ of ${\bf T}$.
Therefore,  for any $\bn^*=\bn +\sum _{j=1}^k m_{i_j}\be_{i_j}$, such than $\bn^*\in K(\gamma )$, 
the vector $(\vec \boldom, \rho ,\bxi +\bn ^*\vec \boldom)$ belongs to 
$S^{(k)}_{i_1,...i_k}(\gamma _1, ...,\gamma _k)$ too. 
Hence it belongs to $S^{(k)}(\gamma _1, ...,\gamma _k)$. 
Obviously, $m_{i_j}$ can be chosen in such a way that $\bn^*\in \Piy ^0 (\gamma _{k+1}, \gamma )$.
\enp

\subsection{Bourgain's Lemma and a good set of $\vec \boldom $.} \label{S6.3} Here we formulate a lemma which is a direct consequence of Lemma 1.20 in \cite{Bou1}.
 Technical explanations how this lemma follows from Lemma 1.20 in \cite{Bou1} are in Appendix 2.
 
 \bel \label{LB} Assume $A\subset \cube^{ld}\times \R^{lds}$ 
 is a semi-algebraic set of degree $B$ and
$$\left|A\right|_{ld(s+1)}<\eta. $$

Let $ {\mathcal N_1}, ...,  {\mathcal N_{s}}\subset \Z^l$ 
 be finite sets with the following properties:
 \bee \min _{1\leq j\leq l}|n_j|>\left(B\max _{1\leq j\leq l}|m_j|\right)^{\Cun_6}, \label{1.21}
\end{equation}
if $\bn :=(n_1,...,n_l) \in  {\mathcal N_i}$ and $\bm :=(m_1,...,m_l)\in {\mathcal N_{i-1}}$, $i=2,...s$.

Assume also
\bee 
\frac{1}{\eta }>\max _{\bn \in {\mathcal N_{s}}}|\sqrt d\bn |^{\Cun_6}.
\label{1.22}
\ene
Let $\Lambda \subset \cube^{ld}$,
\bee
\Lambda:=
\left\{ \vec \boldom: \left( \vec \boldom , \{n ^{(i)}_j\boldom  _j\}_{j,i=1}^{l,s}\right)\in A\ \mbox{ for some } \bn ^{(i)} \in {\mathcal N_i, i=1,...,s.}\right\}. \label{Lambda}
\ene
Then $\Lambda $ satisfies the estimate
\bee
\meas(\Lambda)< B^{\Cun_6}\delta ,\  \  \  \delta ^{-1}:=\min _{\bn \in {\mathcal N_1}}\min _{1\leq j\leq l}|n_j|. \label{mesLambda*}
\ene
\enl
{\newred
\ber\label{new:rmBourgain}
Loosely speaking, this Lemma says the following. Consider the quasi-periodic lattice $\{\bn\vec\boldom, \bn\in\Z^l\}$. Then we can guarantee that some strategically placed point from this lattice does not lie inside a bad semi-algebraic set $A$ assuming that we control the measure and the degree of this set. The price for this is throwing away a small set of frequencies $\vec\boldom$. We also remark that condition \eqref{1.22} is one of the reasons we had to reach an exponential scale in $\lambda$ before applying the Bourgain machinery. 
\enr
}

From now on we put $s:=2^{d+1}$. Let ${\tilde C}_0$ be a sufficiently large constant to be specified at the end of this section. We put 
\bee \label{Ck}{\tilde C}_k:={\tilde C}_0^{(4s)^k},\ \ \ k=1,...,{\newred l}.\ene
Next, we define the following constants:
$$\gamma _{i_0}^{(0)}:={\tilde C}_{\newred l}^{-i_0},\ \  i_0=1,...,2s;$$
$$\gamma _{i_0,i_1}^{(1)}:=\gamma _{i_0-1}^{(0)}{\tilde C}_{{\newred l}-1}^{-i_1}, \ \ \ i_0=2,...,2s, \  \ i_1=1,...,2s;$$
\bee \label{Nov8} \gamma _{i_0,...,i_{k-1},i_k}^{(k)}:=\gamma _{i_0,..,i_{k-2},i_{k-1}-1}^{(k-1)}{\tilde C}_{{\newred l}-k}^{-i_k}, \ \ \ i_0,..,i_{k-1}=2,...,2s, \  \ i_k=1,...,2s.\ene
Obviously,
\bee \label{Jan23} \gamma _{i_0,...,i_{k-1},i_k}^{(k)}< \gamma _{i_0,...,i_{k-1},i_k-1}^{(k)},\ene
\bee \label{Nov10} \gamma _{i_0,...,i_{k-1},i_k}^{(k)}<\gamma _{i_0,...,i_{k-2},i_{k-1}-1}^{(k-1)}<1.\ene
Using \eqref{Nov8}, it is not difficult to show that
\bee
{\tilde C}_0^{2s}\gamma _{i_0,...,i_{k-1}}^{(k-1)}\leq \gamma _{i_0,...,i_{k-1},i_k}^{(k)},  \ \mbox{for any} \ i_k=1,...,2s. \label{Nov8-a}
\ene
Indeed, by \eqref{Nov8},
$$ \gamma _{i_0,...,i_{k-1}}^{(k-1)}=\gamma _{i_0,...,i_{k-2}-1}^{(k-2)}{\tilde C}_{{\newred l}-k+1}^{-i_{k-1}}, $$
$$ \gamma _{i_0,...,i_{k-1},i_k}^{(k)}=\gamma _{i_0,...,i_{k-2}-1}^{(k-2)}{\tilde C}_{{\newred l}-k+1}^{-i_{k-1}+1}{\tilde C}_{{\newred l}-k}^{-i_k}, \ i_k=1,..., 2s.$$
It follows from \eqref{Ck} that
${\tilde C}_{{\newred l}-k+1}^{-1}{\tilde C}_0^{2s}<{\tilde C}_{{\newred l}-k}^{-i_k}$ for any $i_k=1,...,2s$. Now \eqref{Nov8-a} follows immediately.

Hence, we have a nested sequence of intervals:
$$\cup _{i_k=2}^{2s}\left(\gamma _{i_0,...,i_{k-1},i_k}^{(k)}, \gamma _{i_0,...,i_{k-1},i_k-1}^{(k)}\right)\subset \left (\gamma _{i_0,...,i_{k-1}}^{(k-1)}, \gamma _{i_0,...,i_{k-1}-1}^{(k-1)}
\right)$$
for any $i_0,..., i_{k-1}=2,...,2s$.

Now we define a set $ G' \subset  \cube^{ld}$ of good frequencies:

\bee \label{good-omegas-set}
{ G'}(E,{\tilde C}_0):=\cube^{ld}\setminus \Lambda,
\ene
where a bad set 
 $\La=\Lambda (E, {\tilde C}_0)$
is defined as  
 \bee
 \La:=\cup _{k=0}^{{\newred l}-1}\Lambda ^{(k)}, \ \  \Lambda ^{(k)}\subset  \cube^{ld}, 
 \ene
and  $\Lambda ^{(k)}$ are defined as follows. 
First, $\Lambda ^{(0)}=\Lambda ^{(0)}(E,{\tilde C}_0)$ is defined by  Lemma \ref{LB} (see \eqref{Lambda}) for the set $A:=\tilde S _{pr}^{(0,s)}$ with  $\CN_i(E, {\tilde C}_0)$ being defined by the formula:
\bee 
\CN_{s+1-i_0}:=\Piy ^0\left(\gamma _{2i_0}^{(0)}, \gamma _{2i_0-1}^{(0)}\right), \ i_0=1,...,s \label{Ni}
\ene
so that 
$L:=\max _{\bn\in \CN_{s}}|\bn | = E^{\gamma ^{(0)}_{1}r_{1,1}}$. 
Next, 
\bee \label{Lambda1}\Lambda ^{(1)}(E,{\tilde C}_0):=\cup _{i_0=1}^{s}\Lambda ^{(1)}_{i_0},\ene
where $\Lambda ^{(1)}_{i_0}$ is defined by  Lemma \ref{LB} for the set $A:=\tilde S_{pr}^{(1,s)}(\gamma _{2i_0}^{(0)})$ and $\CN_i$ being defined by the formula: 
\bee 
\label{Ni1}\CN_{s+1-i_1}:=
\Piy ^0\left(\gamma _{2i_0,2i_1}^{(0)}, \gamma _{2i_0,2i_1-1}^{(0)}\right), \ \ i_1=1,...,s. 
\ene 
Finally,
\bee \label{Lambdak} 
\Lambda ^{(k)}(E,{\tilde C}_0):=\cup _{i_0,...,i_{k-1}=1}^{s}\Lambda ^{(k)}_{i_0,...,i_{k-1}},
\ene
where $\Lambda ^{(k)}_{i_0,...,i_{k-1}}$ is defined by Lemma \ref{LB} for $A=\tilde S_{pr}^{(k,s)}(\gamma _{2i_0}^{(0)},...,\gamma _{2i_0,...,2i_{k-1}}^{(k-1)})$, with
\bee 
\label{N}\CN_{s+1-i_k}:=\Piy ^0\left(\gamma _{2i_0,...,2i_{k-1}, 2i_k}^{(k)}, \gamma _{2i_0,...,2i_{k-1}, 2i_{k}-1}^{(k)}\right),\  \  i_k=1,...,s,
\ene
so that  
 $L:=\max _{\bn\in \CN_{s}}|\bn | = E^{\gamma _{2i_0,...,2i_{k-1}, 1}^{(k)}r_{1,1}}$.  
\bel Assume ${\tilde C}_0$ is sufficiently large.
  Then 
\bee
\meas(\Lambda ^{(0)})<E^{-\frac12\gamma _{ 2s}^{(0)}r_{1,1}} \label{good-omegas-1}\ene
and for any $k=1,...,{\newred l}-1$, the set $\Lambda ^{(k)}_{i_0,...,i_{k-1}}$ satisfies the estimate: 
\bee
\meas(\Lambda ^{(k)}_{i_0,...,i_{k-1}})<E^{-\frac 12 \gamma _{2i_0,...,2i_{k-1}, 2s}^{(k)}r_{1,1}}.\label{good-omegas-k}\ene
\enl
\bec
\bee
\meas(\Lambda )<\Cun_7 E^{-\frac12\gamma _{2s}^{(0)}r_{1,1}}. \label{good-omegas}\ene
\enc
We obtain the corollary summing the estimates for $\Lambda ^{(k)}\subset \cube^{ld}$ and using \eqref{Jan23}, \eqref{Nov8-a}.

\bep Assume $\tilde C_0$ is large enough. Let us check that \eqref{1.21} holds for sets $\CN_i$ defined by \eqref{Ni} for $k=0$ and \eqref{N} for $k=1,..., {\newred l}-1$.  Indeed, let $k=0$. By Lemma \ref{Lemma6.10}, the degree of the semi-algebraic set $A=\tilde S _{pr}^{(0,s)}$ obeys the estimate: $B_0<E^{\Cun_2}$. Taking into account that $\gamma _{2i_0}^{(0)}=
{\tilde C}_d\gamma _{2i_0+1}^{(0)}$ and using ${\tilde C}_0>\Cun_6\Cun_2$, we easily check \eqref{1.21}. Let $k>1$
 and $\bn\in \CN_{s+1-i_k}$, $\bm\in \CN_{s-i_k}$ for $A=\tilde S_{pr}^{(k,s)}(\gamma _{2i_0}^{(0)},...,\gamma _{2i_0,...,2i_{k-1}}^{(k-1)})$. By Lemma \ref{L6.15}, the degree $B_k$
 of the semi-algebraic set $A$ admits the estimate:
 \bee \label{Bk} B_k<E^{\Cun_5}E^{{\tilde \gamma _k}r_{1,1}},\ene
 where
 \bee\label{cun1}
 \tilde \gamma _k:=\Cun_5(\gamma _{2i_0}^{(0)}+...+\gamma _{2i_0,...,2i_{k-1}}^{(k-1)}).
 \ene
 By \eqref{Nov8-a},
 \bee \label{Jan29} \tilde \gamma _k<2\Cun_5\gamma _{2i_0,...,2i_{k-1}}^{(k-1)}.\ene
 Let us now assume that 
 \bee\label{ass1}
 (\gamma _{2i_0}^{(0)}+...+\gamma _{2i_0,...,2i_{k-1}}^{(k-1)})r_{1,1}>1.
 \ene
 After we fix the choice of the constants $\gamma$ at the end of this section, \eqref{ass1}  will be one of the new conditions on $E_*$ which we will summarise later.
 Next,  by \eqref{N},
$\min |\bn|= E^{\gamma ^{(k)}_{2i_0,...,2i_{k-1}, 2i_k} r_{1,1}}$, $\max |\bm|=E^{\gamma ^{(k)}_{2i_0,...,2i_{k-1}, 2i_k+1}r_{1,1}}$. Considering \eqref{Bk}, \eqref{Jan29},  and \eqref{Nov8}, \eqref{Nov8-a} and using ${\tilde C}_0>4\Cun_6\Cun_5$, we  easily check that  \eqref{1.21} holds.

 Now, we check \eqref{1.22}. Let $k=0$. By \eqref{Ni}, $\max _{\bn\in \CN_{s}}|\bn | = E^{\gamma ^{(0)}_{1}r_{1,1}}$ and we can apply  Lemma \ref{Lemma6.10} with
 $L=\max _{\bn\in \CN_{s}}|\bn | = E^{\gamma ^{(0)}_{1}r_{1,1}}$. By this lemma $\eta =E^{\Cun_1}E^{-r_1'\cun_1}(2L)^{(l-1)sd}$. Now, we assume
 \bee\label{ass2}
\gamma ^{(0)}_{1}r_{1,1}>\Cun_1.
 \ene
  Taking into account \eqref{new5.1}, \eqref{r1'} and $\gamma_0<1$, we easily see that   \eqref{1.22} holds when ${\tilde C}_0>(\Cun_6+lsd)/\cun_1$. 
 
 Assume now that $k\geq 1$. By \eqref{N}, $\max _{\bn\in \CN_{s}}|\bn | = E^{\gamma _{2i_0,...,2i_{k-1}, 1}^{(k)}r_{1,1}}$, and we can apply  Lemma \ref{L6.15} with
 $L=\max _{\bn\in \CN_{s}}|\bn | = E^{\gamma _{2i_0,...,2i_{k-1}, 1}^{(k)}r_{1,1}}$. By this lemma, $\eta = E^{-r_1'\cun_2}L^{\Cun_4}E^{2\Cun_4(\gamma _{2i_0}^{(0)}+...+\gamma _{2i_0,...,2i_{k-1}}^{(k-1)})r_{1,1}}$. Using ${\tilde C}_0>(\Cun_6+3\Cun_4)/\cun_2$, we obtain that 
 $\gamma _{2i_0,...,2i_{k-1}, 1}^{(k)}$ is small enough and, hence, \eqref{1.22}   holds. 
 
 Next, by \eqref{Ni}, $\delta ^{-1}=E^{\gamma _{2s}^{(0)}r_{1,1}}$ (for $k=0$) and, by \eqref{N}, $\delta ^{-1}=E^{\gamma _{2i_0,...,2i_{k-1}, 2s}^{(k)}r_{1,1}}$ otherwise. Hence, for $k=0$:
 $$B_0^{\Cun_6}\delta\leq E^{\tilde C_0}E^{-\gamma _{2s}^{(0)}r_{1,1}}\leq E^{-\frac12\gamma _{2s}^{(0)}r_{1,1}},$$ 
 where we assumed 
 \bee\label{ass3}
\gamma ^{(0)}_{2s}r_{1,1}>2\tilde C_0.
 \ene
 For $k\geq 1$,
 $$B_k^{\Cun_6}\delta\leq 
 E^{\left({\tilde C_0}(\gamma _{2i_0}^{(0)}+...+\gamma _{2i_0,...,2i_{k-1}}^{(k-1)})-\gamma _{2i_0,...,2i_{k-1}, 2s}^{(k)}\right)r_{1,1}}.$$
 Using \eqref{Nov8-a}, we obtain:
 $$B_k^{\Cun_6}\delta<E^{-\frac 12 \gamma _{2i_0,...,2i_{k-1}, 2s}^{(k)}r_{1,1}}. $$ 
 Lemma \ref{LB} yields \eqref{good-omegas-1} and \eqref{good-omegas-k}.
 \enp
 
 \subsection{Proof of Lemma \ref{MGL}. } \label{S6.4}
 Assume that $\vec\boldom \in G'(E, {\tilde C}_0)$ and $(\vec \boldom, \rho, \bxi )\in S_{total}$. Now we  find  $\gamma =\gamma (\vec \boldom,  \rho , \bxi,\tilde C_0, l,d)$,  satisfying  \eqref{gamma},  such that\eqref{no} holds. 
 We start by considering $\bn \in \cup _{i=1}^s \CN_{i}$, where $\CN_{i}(E,{\tilde C}_0)$ is given by \eqref{Ni}.  Assume 
 \bee (\vec\boldom, \rho , \bxi +\bn \vec \boldom)\in S_{total}\label{main} \ene  
  holds for each $\bn$ from a collection $\{\bn ^{(i)}\}$, $\bn ^{(i)}\in \CN_i$, $i=1,...,s$. We will show that
  \bee \label{Jan25}
  (\vec \boldom, \{n ^{(i)}_j\boldom  _j\}_{j,i=1}^{l,s}) \in \tilde S _{pr}^{(0,s)}.
  \ene
  Indeed, 
  by \eqref{main}, \eqref{5.15},  
  $(\vec\boldom, \rho , \bxi ,\bn ^{(1) 
}\vec \boldom , ...,\bn ^{(s)}\vec \boldom)\in \tilde S^{(s)}$, and by \eqref{Sspr}, \eqref{0,s}, we have \eqref{Jan25}.
This means that $\vec\boldom \in \Lambda^{(0)}$ by the definition of $\Lambda ^{(0)}$ and, hence,  contradicts our assumption $\vec \boldom \in G'$. It follows that
there is
 $i_0=i_0(\vec \boldom,  \rho , \bxi,\tilde C_0, l,d)$, such that $N_{s+1-i_0}$ contains no points $\bn$: $(\vec \boldom, \rho , \bxi +\bn \vec \boldom )\in S_{total}$. 
 From now on, we consider only $\bn \in K(\gamma ^0_{2i_0-1})$, see \eqref{K-gamma}. This means that we are looking for $\gamma \leq \gamma ^0_{2i_0-1}$. Further, by \eqref{Jan22}
$$K(\gamma ^0_{2i_0-1})=\Piy ^0(\gamma ^0_{2i_0}, \gamma ^0_{2i_0-1})\cup \Piy ^1(\gamma ^0_{2i_0}, \gamma ^0_{2i_0-1}).$$  Hence, the only option remaining for \eqref{main}
to hold is  $\bn \in \Piy ^1(\gamma ^0_{2i_0}, \gamma ^0_{2i_0-1})$.
 Let us consider $\cup _{i=1}^s \tilde \Piy ^1(\gamma ^{(0)}_{2i_0}, \gamma ^{(1)}_{2i_0,2i},\gamma ^{(1)}_{2i_0,2i-1})\subset \Piy ^1(\gamma ^{(0)}_{2i_0}, \gamma ^{(0)}_{2i_0-1})$, see \eqref{tildePi1}, \eqref{Nov10}.
 Assume there is a family $\{\bn ^{(i)}\}$, $i=1,...,s$, $\bn ^{(i)}\in \tilde \Piy ^1(\gamma ^{(0)}_{2i_0}, \gamma ^{(1)}_{2i_0,2i},\gamma ^{(1)}_{2i_0,2i-1})$, such that \eqref{main} holds for every  $\bn=\bn ^{(i)}$, $i=1,...,s$. By Lemma \ref{L6.18}  there is $\bn^{(i)*}\in
 \Piy ^0( \gamma ^{(1)}_{2i_0,2i},\gamma ^{(1)}_{2i_0,2i-1})$ such that 
 $(\vec \boldom, \rho, \bxi, \bn^{(i)*}\vec \boldom) \in S^{(1)}(\gamma ^{(0)}_{2i_0})$, $i=1,...,s$.  By \eqref{5.38},
   \bee \label{Jan26}(\vec \boldom, \rho, \bxi, \{n ^{(i)*}_j\boldom  _j\}_{j,i=1}^{l,s}) \in \tilde S ^{(1,s)}(\gamma ^{(0)}_{2i_0}).\ene
  By the definition of the projection,
   \bee \label{Jan26a}(\vec \boldom, \{n ^{(i)*}_j\boldom  _j\}_{j,i=1}^{l,s}) \in \tilde S _{pr}^{(1,s)}(\gamma ^{(0)}_{2i_0}).\ene
   This means that $\vec \boldom  \in \Lambda  ^{\newred (1)}$ by the definition of $\Lambda ^{\newred (1)}$, see  \eqref{Lambda1}, \eqref{Ni1},  and, hence,  contradicts  the assumption $\vec \boldom \in G'$. 
Therefore, there is $i_1$ such that 
$\tilde \Piy ^1(\gamma ^{(0)}_{2i_0}, \gamma ^{(1)}_{2i_0,2i_1},\gamma ^{(1)}_{2i_0,2i_1-1})$ does not contain points satisfying \eqref{main}. From now on, we consider only $\bn \in K(\gamma ^{(1)}_{2i_0,2i_1-1})$.  This means that we are looking for $\gamma \leq \gamma ^{(1)}_{2i_0,2i_1-1}$.
Since 
 $$K(\gamma ^{(1)}_{2i_0,2i_1-1})=\Piy ^0(\gamma ^{(0)}_{2i_0}, \gamma ^{(1)}_{2i_0,2i_1-1})\cup \tilde \Piy ^{1}(\gamma ^{(0)}_{2i_0}, \gamma ^{(1)}_{2i_0,2i_1}, \gamma ^{(1)}_{2i_0,2i_1-1})\cup 
 \Piy ^{2}(\gamma ^{(0)}_{2i_0},\gamma ^{(1)}_{2i_0,2i_1}, \gamma ^{(1)}_{2i_0,2i_1-1}),$$  the only option remaining for \eqref{main} is $\bn \in \Piy ^{2}(\gamma ^{(0)}_{2i_0},\gamma ^{(1)}_{2i_02i_1}, \gamma ^{(1)}_{2i_02i_1-1})$. Further, we describe an induction procedure. Assume we consider $\bn \in \Piy ^{k}\left(\gamma ^{(0)}_{2i_0},...,\gamma ^{(k-1)}_{2i_0,...,2i_{k-1}}, \gamma ^{(k-1)}_{2i_0,...,2i_{k-1}-1}\right)$, $k\geq 2$. Note that
 $$\cup _{i=1}^s\tilde \Piy ^k\left(\gamma ^{(0)}_{2i_0},...,\gamma ^{(k-1)}_{2i_0,...,2i_{k-1}}, \gamma ^{(k)}_{2i_0,...,2i}, \gamma ^{(k)}_{2i_0,...,2i-1}\right)\subset \Piy ^k\left(\gamma ^{(0)}_{2i_0},...,\gamma ^{(k-1)}_{2i_0,...,2i_{k-1}}, \gamma ^{(k-1)}_{2i_0,...,2i_{k-1}-1}\right).$$
 Assume that there is a family $\{\bn ^{(i)}\} $, $\bn ^{(i)} \in\tilde  \Piy ^k\left(\gamma ^{(0)}_{2i_0},...,\gamma ^{(k-1)}_{2i_0,...,2i_{k-1}}, \gamma ^{(k)}_{2i_0,...2i}, \gamma ^{(k)}_{2i_0,...,2i-1}\right)$, $i=1,...,s$, such that \eqref{main} holds for each $\bn=\bn ^{(i)}$. By Lemma \ref{L6.18}, there is $\bn^{(i)*}\in
 \Piy ^0(\gamma ^{(k)}_{2i_0,...,2i}, \gamma ^{(k)}_{2i_0,...,2i-1})$ such that 
 $$(\vec \boldom, \rho, \bxi, \bn^{(i)*}\vec \boldom) \in S^{(k)}\left(\gamma ^{(0)}_{2i_0},...,\gamma ^{(k-1)}_{2i_0,...,2i_{k-1}}\right),\ \ i=1,...,s.$$
  By \eqref{5.38},
   \bee \label{Jan26b}(\vec \boldom, \rho, \bxi, \{n ^{(i)*}_j\boldom  _j\}_{j,i=1}^{l,s}) \in \tilde S ^{(k,s)}\left(\gamma ^{(0)}_{2i_0},...,\gamma ^{(k-1)}_{2i_0,...,2i_{k-1}}\right).\ene
  By the definition of the projection,
   \bee \label{Jan26c}(\vec \boldom, \{n ^{(i)*}_j\boldom  _j\}_{j,i=1}^{l,s}) \in \tilde S _{pr}^{(k,s)}\left(\gamma ^{(0)}_{2i_0},...,\gamma ^{(k-1)}_{2i_0,...,2i_{k-1}}\right).\ene
   This means that $\vec \boldom  \in \Lambda  ^{\newred (k)}$ by the definition of $\Lambda ^{\newred (k)}$, see  \eqref{Lambdak}, \eqref{N},  and, hence,  contradicts the assumption $\vec \boldom \in G'$. 
Hence,  there is $i_k$ such that 
$$\tilde \Piy ^k\left(\gamma ^{(0)}_{2i_0},...,\gamma ^{(k-1)}_{2i_0,...,2i_{k-1}}, \gamma ^{(k)}_{2i_0,...,2i_{k}}, \gamma ^{(k)}_{2i_0,...,2i_{k}-1}\right)$$ 
does not contain points satisfying \eqref{main}. From now on, we consider only $\bn \in K(\gamma ^{(k)}_{2i_0,...,2i_{k}-1})$. 
  This means that we are looking for $\gamma \leq \gamma ^{(k)}_{2i_0,...,2i_{k}-1}$. By \eqref{Jan22}, \eqref{Jan28},
\bee \label{Jan28a} 
\bes
K(\gamma ^{(k)}_{2i_0,...,2i_{k}-1})&=\Piy ^0(\gamma ^{(0)}_{2i_0}, \gamma ^{(k)}_{2i_0,...,2i_{k}-1})\cup \\
&\cup _{j=1}^k\tilde \Piy ^j\left(\gamma ^{(0)}_{2i_0},...,\gamma ^{(j-1)}_{2i_0,...,2i_{j-1}}, \gamma ^{(j)}_{2i_0,...,2i_{j}}, \gamma ^{(k)}_{2i_0,...,2i_{k}-1}\right)\\
&\cup \Piy ^{k+1}\left(\gamma ^{(0)}_{2i_0},...,\gamma ^{(k-1)}_{2i_0,...,2i_{k-1}}, \gamma ^{(k)}_{2i_0,...,2i_{k}}, \gamma ^{(k)}_{2i_0,...,2i_{k}-1}\right).
\end{split}
 \ene 
 This implies that   
 $$\bn \in  \Piy ^{k+1}\left(\gamma ^{(0)}_{2i_0},...,\gamma ^{(k-1)}_{2i_0,...,2i_{k-1}}, \gamma ^{(k)}_{2i_0,...,2i_{k}}, \gamma ^{(k)}_{2i_0,...,2i_{k}-1}\right).$$  Thus, we have arrived at  the next step of the induction procedure. 
Taking $k={\newred l}-1$, we consider   $\bn \in K(\gamma ^{({\newred l-1)}}_{2i_0,...,2i_{{\newred l}-1}-1})$  as in \eqref{Jan28a}, and conclude that \eqref{main} can hold only when $$\bn \in \Piy ^{\newred l}\left(\gamma ^{(0)}_{2i_0},...,\gamma ^{({\newred l}-1)}_{2i_0,...,2i_{{\newred l}-1}}, \gamma ^{({\newred l}-1)}_{2i_0,...,2i_{{\newred l}-1}-1}\right),$$ which is just a subset of $K(\gamma ^{({\newred l}-1)}_{2i_0,...,2i_{{\newred l}-1}})$.
This mean that every $\bn \in K(\gamma ^{({\newred l}-1)}_{2i_0,...,2i_{{\newred l}-1}-1})$ satisfying  \eqref{main} is in fact in  $K(\gamma ^{({\newred l}-1)}_{2i_0,...,2i_{{\newred l}-1}})$.
Thus, we arrive at the conclusion that every $\gamma $ in the interval $({\tilde C}_0\gamma ^{({\newred l}-1)}_{2i_0,...,2i_{{\newred l}-1}}, \gamma ^{({\newred l}-1)}_{2i_0,...,2i_{{\newred l}-1}-1})$ satisfies
\eqref{no}. We now can put 
\bee\label{Z_0}
Z_0:=\tilde C_0
\ene
 and
\bee\label{gamma_0}
\gamma_0:=\tilde C_0 \gamma ^{(0)}_{2s},\ \ \ s:=2^{d+1}.
\ene 
Property \eqref{gamma} easily follows from \eqref{Ck}-\eqref{Nov8}.
 Note that the choice of $i_0,..., i_{{\newred l}-1}$ is stable with respect to perturbations of $\vec \boldom $ of order $E^{-2r'_1}$ and $\rho , \bxi $ of order $E^{-r'_1-2}$. This finishes the proof of Lemma \ref{MGL}. We now  add one more lower bound on $E_*$ so that \eqref{ass1}, \eqref{ass2} and \eqref{ass3} are satisfied. It is enough to assume
 \bee\label{E_*1}
 r_{1,1}(E_*)>2Z_0^2/\gamma_0.
 \ene 
 \ber\label{Zgamma0}
 Since $\tilde C_0$ was any constant satisfying finitely many bounds formulated in the construction above, we have indeed proved lemma \ref{MGL} for all sufficiently large $Z_0$. However, we will be using this lemma (and similar lemmas in the induction step formulated in the next Section) only for one value of $Z_0$. Therefore, from now on we fix $\tilde C_0$ satisfying the conditions needed for the proof above and define $Z_0$ and $\gamma_0$ by \eqref{Z_0} and \eqref{gamma_0}. These values will not change in the rest of our paper. 
 \enr
 
\section{Induction} \label{section7}
 Now we discuss how to extend the results of the previous section to higher scales. This will be done by induction. By superscript $(n)$ (or $(j)$) we denote the information that a  corresponding object is considered at step $n$ (or $j$).  The zeroth and first step have been considered in previous sections, so from now on we assume that $n\ge 2$. For a given number $\gamma\in (0,1)$ and $n\ge 2$, we also define 
\bee\label{rn'}
r'_{n}(\gamma)=r'_{n}(E;\gamma):=\En^{2l\gamma r_{n-1,1}/Z_0}, \ \ \ r'_{n}:=r'_{n}(1).
\ene
Recall that $r'_1=r_{1,3}$; it is easy to see that we have $r_n'>r_{n,3}$ when $n>1$.  Recall as well that various constants $\gamma^{(j)}$ have been defined in \eqref{Nov8} and $r_{n,m}$ were defined in the beginning of section \ref{goodset-3}. 

Now we will construct patches with which we will work during the induction procedure. 
We will need patches in three variables: $\bxi,\vec\boldom,\rho$, where the following `region of interest' will be covered by patches: $\{\bxi\in\R^d, \|\bxi\|\in[E-2,E+2]\}$, $\rho\in [E-1,E+1]$, $\vec\boldom\in[-1/2,1/2]^{ld}$; recall that we eventually will put $\bxi=\bk+\bn\vec\boldom$, but for some time we want to treat $\bxi$ as an independent variable, {\newred see Remark \ref{newrm2}}. In section \ref{section1} we have described conventions we follow when covering the region of interest by patches; here we just state that when performing step $n$, $n\ge 2$, we mostly deal with the patches $\CA^{(n-1)}$ of level $n-1$ and the size of patches in $\rho$ and $\vec\boldom$ is $\En^{-2r_{n-1}'}$; the size of patches in $\bxi$ is $\En^{-2r_{n-1}'+2r_{n-1,1}}$.  

We will also need patches in $\BPhi$. Recall that for $n=1$ we used the patches $\CA^{(0)}$ of size $E^{-2}$ in $\BPhi$ and $E^{-1}$ in other variables.  
The patches   
$\CA^{\BPhi(j)}$  of level $j$  have size 
$\En^{-r_{j,3}\En^{3lr_{j-1,1}/Z_0}}$, $j\geq2$, and $\CA^{\BPhi(1)}$ have size 
$\En^{-r_{1,3}\En^{2d^2(l+\mu)\sigma_{0,d-1}}}$. For each patch $\CA^{\BPhi(j)}_m$ we also consider the corresponding patch $\CA^{\vec\phi(j)}_m$ (so that $\CA^{\BPhi(j)}_m=\BPsi_m(\CA^{\vec\phi(j)}_m)$) and its complexification $\CA^{\vec\phi(j)}_{m,\,\C}$ of the same size. We also consider quasi-patches in $\bk$ associated with patches $\CA^{\BPhi(n)}$: 
\bee\label{CAbkn}
\CA^{\bk(n)}:=\{\bk\in\R^d:\ \bk\|\bk\|^{-1}\in \CA^{\BPhi(n)}\ \&\ |\|\bk\|-\k^{(n-1)}(\BPhi,\rho )|\leq \En^{-r_{n,3}\En^{3lr_{n-1,1}/Z_0}}\}.
\ene

\ber
The reader may wonder why we have labelled the patches where we work to perform step $n$ as $\CA^{(n-1)}$, and not $\CA^{(n)}$. The point is, when we label different objects (patches, good/bad sets, central cubes $\hat K$, cubes $K$ of the \Bourgain structure to be introduced in the next definition, etc.), we will try to synchronise the labelling as much as we can, but whatever we do, there is bound to be some discrepancy somewhere. We have chosen the labelling of all these objects to minimise these discrepancies as much as possible. 
\enr

 We always assume that any patch of level $j\ge 1$ is completely covered by some patch at the previous level. At the beginning of each step $n$ we have declared some patches in $\vec\boldom$,  and $\BPhi$ of level $n-1$ bad. These bad patches depend on $d$, $l$, and the Fourier coefficients $V_\bn$ only; the proportion of bad patches is small (and going to $0$ quickly as $n$ increases). The patches in $\rho$ and $\bxi$ are never going to be bad (and patches in $\bxi$ are playing an axillary role). After we fix a patch of order $n-1$ in $\rho$, $\BPhi$ and $\vec\boldom$, we will declare some patches of order $n$ in $\BPhi$ and $\vec\boldom$ bad and remove them from any further consideration. 
 
Recall that an extended ball was defined in definition \ref{extended} (we are using word cube as a synonym for a ball in $\Z^l$). Finally, recall that the parameters $Z_0$ and $\gamma_0$ were fixed in remark \ref{Zgamma0}.

\begin{defin}\label{8.1}
Let $\rho\in\R$ and $\vec\boldom\in [-1/2,1/2]^{ld}$ be fixed. We say that our operator $H(\bk)$ ($\bk\in\R^d$) has a \Bourgain structure of order $n\in\N$ at the points $\rho, \vec\boldom$ corresponding to our system of patches, if the following holds:

1. Let $j=0,...,n$ and let let  $\vec\boldom^{*(j)}_{\tim^{\vec\boldom}}$ be a centre of a patch of order $j$ that contains $\vec\boldom$, let $\CM^{\rho(j)}_{\tim^{\rho}}$  be a matryoshka of  patches of order $j$ in  $\rho$, and let $\CM^{\bxi(j)}_{\tim^{\bxi}}$ be a matryoshka of  patches of order $j$ in $\bxi$.
For each such triple
 $(\CM^{\bxi(j)}_{\tim^{\bxi}},\CM^{\rho(j)}_{\tim^{\rho}},\vec\boldom^{*(j)}_{\tim^{\vec\boldom}})$, $j=0,...,n$, there is a corresponding  `base  cube' $K^{b(j)}_{\ham}\subset\Zl$, $\ham=\ham(j, \tim^{\rho}, \tim^{\vec\boldom}, \tim^{\bxi})$. 
Each set  $K^{b(j)}_{\ham}$ for $j\ge 1$ contains a ball of radius  $\frac14\En^{\gamma^{(j)}_{\ham}r_{j,1}}$ and is contained inside 
 a ball of radius $\frac12\En^{\gamma^{(j)}_{\ham}r_{j,1}}$ (both balls are centred at the origin).  The factors $\gamma^{(j)}_{\ham}$ are numbers inside $(\gamma_0,1-\gamma_0)$. This means that  $K^{b(j)}_{\ham}$ is an extended cube centred at the origin. 
 For $j=0$ each set $ K^{b(0)}_{\ham}$ is a  super-extended pre-cluster $\tilde \BUps_1^{\Z}(\bxi^*)$ defined in \eqref{superext}.

2. Similarly, for each triple $(\CM^{\bxi(j)}_{\tim},\CM^{\rho(j)}_{\tim},\vec\boldom^{*(j)}_{\tim})$, $j=0,...,n$, there is a corresponding  `small base cube' $K^{b(j),{\mathrm {small}}}_{\ham}\subset\Zl$. For $j\ge 1$, each set  $K^{b(j),{\mathrm {small}}}_{\ham}$ is an extended ball centred at the origin and of radius $\En^{\gamma^{(j)}_{\ham}r_{j,1}/Z_0}$. 
 For $j=0$ each set $ K^{b(0),{\mathrm {small}}}_{\ham}$ is an extended pre-cluster $\check \BUps_1^{\Z}(\bxi^*)$ defined in \eqref{ext}; note that this is not an extended ball in $\Z^l$ according to definition \ref{extended}.

3. For each $j=0,...,n$ we have a collection of sets: `\Bourgain cubes' (or `normal \Bourgain cubes') $\{K^{(j)}_{m}\}$  and `small \Bourgain cubes' $\{K^{(j),{\mathrm {small}}}_{m}\}$ ($m\in\N$); all of them are subsets of $\Z^l$. Each set  $K^{(j)}_{m}$ 
is a shifted base cube. More precisely, given any  $K^{(j)}_{m}$ 
there exists a vector $\bn=\bn(j,m)\in K^{(j)}_{m}$ such that $\bxi=\bxi(j,\bn):=\bk+\bn\vec\boldom$ is $E^{\sigma_0}$-bad, and
\bee\label{(8.3)}
K^{(j)}_{m}=K^{b(j)}_{\ham}+\bn.
\ene
Here, $K^{b(j)}_{\ham}$ is constructed in the following way. First, we notice that $\bxi$ lies in several matryoshkas $\CM^{\bxi(j)}$. Also, the point $\rho$ can belong to several matryoshkas $\CM^{\rho(j)}$. We can choose one of these matryoshkas $\CM^{\bxi(j)}_{\tim^{\bxi}}$ and one of the matryoshkas in $\rho$ $\CM^{\rho(j)}_{\tim^{\rho}}$ so that  $K^{b(j)}_{\ham}$ is the base cube corresponding to these matryoshkas: $\ham=\ham(j, \tim^{\rho}, \tim^{\vec\boldom}, \tim^{\bxi})$. 
Similar property holds for each small \Bourgain cube $\{K^{(j),{\mathrm {small}}}_{m}\}$, with the same vector $\bn(j,m)$ and the same matryoshkas $\CM^{\bxi(j)}_{\tim^{\bxi}}$ and $\CM^{\rho(j)}_{\tim^{\rho}}$ as for the normal  \Bourgain  cube $K^{(j)}_{m}$. In particular, the centres of extended cubes $\{K^{(j)}_{m}\}$ and $\{K^{(j),{\mathrm {small}}}_{m}\}$ coincide.

We refer to $j$ as the order of  $K^{(j)}_m$; note that for $j=0$ our \Bourgain cubes are extended and super-extended clusters $\CC_1$ defined in \eqref{7CC} and \eqref{7CC1} with $p=1$.

4. For each $j=1,...,n$ and $m,m'\in\N$, the $\Z$-distance between $K^{(j)}_m$ and $K^{(j)}_{m'}$ is at least $\frac14\max(\En^{\gamma^{(j)}_mr_{j,1}},\En^{\gamma^{(j)}_{m'}r_{j,1}})$. For $1\leq j<j'$, the $\Z$-distance between $K^{(j)}_m$ and $K^{(j')}_{m'}$ is at least $\frac18\En^{\gamma^{(j)}_mr_{j,1}}$ (unless $K^{(j)}_m\subset K^{(j')}_{m'}$). For $j=0$ we assume that the properties described in Section 5 hold. If $j<j'$ and $K^{(j)}_m\subset K^{(j')}_{m'}$ we assume that $\Z$-distance between $K^{(j)}_m$ and $\Z^l\setminus K^{(j')}_{m'}$ is at least $10Q$.

5. Consider a  \Bourgain  cube  $K^{(j)}_m$, $j=0,...,n$. According to condition 3 of this definition, there is a vector $\bn(j,m)$ that determines a point $\bxi(j,m)$. This point can belong to several matryoshkas,  $\CM^{\bxi(j)}_{\tim}$, $\CM^{\bxi(j)}_{\tim'}$,...; similarly,  the point $\rho$ can belong to several patches $\CM^{\rho(j)}$. Let us list all the possible base cubes corresponding to different choices of matryoshkas in $\bxi$ and $\rho$:  $K^{b(j)}_{\ham}$, $K^{b(j)}_{\ham'}$,... We know that 
$K^{(j)}_{m}=K^{b(j)}_{\ham}+\bn$; let us denote other shifted base cubes by  
$K'^{(j)}_{m}=K^{b(j)}_{\ham'}+\bn$,...

We say that a given \Bourgain  cube  $K^{(j)}_m$, $j=0,...,n-1$, is bad, if for at least one of the shifted cubes $K^{(j)}_{m}$, $K'^{(j)}_{m}$,... described above (say, $K'^{(j)}_{m}$), we have
\bee\label{bad7}
||(H(K'^{(j)}_m;\bk)-\rho^2)^{-1}||_2=||(\CP( K'^{(j)}_m;\bk)(H(\bk)-\rho^2)\CP( K'^{(j)}_m;\bk))^{-1}||_2>\En^{r'_{j+1}(\gamma'^{(j)}_m)}, 
\ene
 and good otherwise.

6. We assume that every point $\bn\in(\Z^l\setminus \{0\})$ that lies outside all small  \Bourgain cubes of order $0$ satisfies $|\|\bk+\bn\vec\boldom\|^2-\rho^2|>\En^{\sigma_{0}}$.

7. For each bad  \Bourgain cube $K^{(j)}_m$, $j=0,...,n-1$ there is a small  \Bourgain cube  $K^{(j+1),{\mathrm {small}}}_{m'}$ such that 
\bee\label{small}
K^{(j)}_m\subset K^{(j+1),{\mathrm {small}}}_{m'}.
\ene
Also, for each small  \Bourgain cube  $K^{(j+1),{\mathrm {small}}}_{m'}$ there is a bad  \Bourgain cube $K^{(j)}_m$ such that inclusion \eqref{small} holds.

We say that our operator $H$ has a \Bourgain structure of order $n$ for given $\rho$ and $\vec\boldom$ on a set $S\subset\R^d$ if for any $\bk\in S$ the operator $H(\bk)$ has a \Bourgain structure. We will often say in this case that $K^{(j)}_m$ are \Bourgain cubes generated by $\bk$. 

Let $\CA^{\bk}\subset\R^d$ and $\CA^{\rho}\subset\R$. We say that a \Bourgain structure is {\it stable} in $\bk$ on $\CA^{\bk}$ and in $\rho$ on $\CA^{\rho}$, if all the cubes $K^{(j)}_m$ (for all $0\leq j\le n$) are the same  subsets of $\Zl$ for all $\bk\in\CA^{\bk}$ and $\rho\in\CA^{\rho}$. We say that a \Bourgain structure is  stable in $\bk$ on matryoshka $\CM^{\bk(n)}=\{\CA^{\bk(j)}\}_{j=0}^n$ and in $\rho$ on matryoshka $\CM^{\rho(n)}=\{\CA^{\rho(j)}\}_{j=0}^n$, if for each $j=0,...,n$ the structure of order $j$ is stable in $\bk$ on $\CA^{\bk(j)}$ and in $\rho$ on $\CA^{\rho(j)}$.
\end{defin}


\ber
This definition makes sure that any bad  \Bourgain cube $K^{j}_m$ is covered by a small  \Bourgain cube of the next level $K^{(j+1), {\mathrm {small}}}_{m'}$. Also, inside any `cubical layer' 
$K^{(j+1)}_{m}\setminus K^{(j+1), {\mathrm {small}}}_{m}$ there are no bad  \Bourgain cubes of order $j$. These two statements imply that any  \Bourgain cube $K^{j'}_{m'}$ ($j'\le j$) inside a spherical layer $K^{(j+1)}_{m}\setminus K^{(j+1), {\mathrm {small}}}_{m}$ is either good, or is covered by a good  \Bourgain cube of order at most $j$. 
\enr



\ber
The reason for the conditions 1 and 2 is as follows: we need to control the number of different `shapes' that \Bourgain cubes can take (and the `shape' is, of course, the base cube). We need this control to make   measure estimates \eqref{measSmn} and \eqref{meastildeSmn} feasible: we make the estimate for each possible `shape' and then multiply by the number of `shapes'. And in our definition the number of `shapes' is estimated by the product of the number of  different patches or matryoshkas in all  variables. 
\enr

\ber
Condition 3 of this definition may also look slightly strange: the reader may wonder why we cannot do the following: For each $\bn\in\Z^l$, we look at the centre $\bxi^{*(j)}$ of the patch containing $\bk+\bn\vec\boldom$, take corresponding base cube, shift it by $\bn$, and then declare all such shifts a \Bourgain structure. The point is, these shifts may (and will) have a non-trivial intersection, while we  want the  \Bourgain cubes of the same level to stay away from each other (property 4). Therefore, when we construct the \Bourgain structure in section \ref{section8}, after this step, we have to throw away some of thus constructed  \Bourgain cubes.
\enr

\ber
This definition does not prescribe anything about the relationship between $K^{(j)}_m$ and   $K^{(j'),{\mathrm {small}}}_{m'}$ assuming that $j<j'$ and $K^{(j)}_m\subset K^{(j'),}_{m'}$. The reason is, when we have the situation that $K^{(j)}_m$ is partially covered by $K^{(j'),{\mathrm {small}}}_{m'}$, we can easily resolve this (for example, by extending $K^{(j'),{\mathrm {small}}}_{m'}$ so that it contains $K^{(j)}_m$). This means that we can always assume without loss of generality that when $j<j'$, for any pair of cubes $K^{(j)}_m$ and   $K^{(j'),{\mathrm {small}}}_{m'}$ we either have $K^{(j)}_m\subset K^{(j'),{\mathrm {small}}}_{m'}$, or the two cubes do not intersect. 
\enr

\ber 
We are going to use this definition in the situation when Theorem~\ref{Thm2} holds. Then,  in condition 3 of this definition we can consider only vectors $\bn$ satisfying $|\bn|>E^{r_{1,2}}/2$. The point is that, strictly speaking, when $\bn\in \hat K^{(1)}$ (which more or less means $|\bn|<E^{r_{1,2}}/2$), we have been using for base  \Bourgain cubes $K^{b(0)}$   super-extended pre-clusters $\tilde \BUps_0(\bxi^*)$
with $p=0$, not $1$. 
\enr

{\newred \ber\label{newrm3}
We emphasise that set of base cubes $K^{b(j)}_{\ham}$ depends on $\bk$ only through $\rho$. In other words, the base cubes do not depend on $\BPhi$. 
\enr
}

Parts $6$ and $7$ of this definition in particular mean the following. Suppose, a point $\bn\in\Z^l$ is not $\En^{\sigma_{0}}$-good 
(i.e. $|\|\bk+\bn\vec\boldom\|^2-\rho^2|<\En^{\sigma_{0}}$). Then $\bn$ is covered by a  small \Bourgain cube of level $0$. If this  \Bourgain cube is bad, then $\bn$ is covered by a small  \Bourgain cube of level $1$, etc. This naturally leads to the following  definition:

\begin{defn}\label{corresponding}
Suppose, $\bn\in\Z^l$ is not $\En^{\sigma_{0}}$-good. The \Bourgain cube $K^{(j)}_m$ with the largest $j$ containing $\bn$  is called  the \Bourgain cube corresponding to $\bn$. The comment above means that if this cube is bad, then we necessarily have $j=n$ (the order of the \Bourgain structure). 
\end{defn}

As we mentioned earlier, the definition of the \Bourgain structure makes sure that any bad  \Bourgain cube $K^{j}_m$ is covered by a  \Bourgain cube of the next level $K^{(j+1), {\mathrm {small}}}_{m'}$ and that in the `cubical layer' 
$K^{(j+1)}_{m'}\setminus K^{(j+1), {\mathrm {small}}}_{m'}$ there are no bad  \Bourgain cubes of order $j$. 
Now, for technical purposes, we want to be able to put any bad  \Bourgain cube $K^{(j)}_m$ inside a cube of much bigger size. More precisely, we need the following additional structure:

\begin{defin}\label{8.2}
We say that our operator has an enlarged \Bourgain structure at the points $\rho, \vec\boldom$   of order $n$, corresponding to our system of patches, if it has a \Bourgain structure of order $n$  together with  the (of larger size than before) sets $\tilde K^{(j)}_m$ and $\tilde K^{(j),{\mathrm {small}}}_{m}$, $j=1,..,n$, called the enlarged  \Bourgain cubes. These cubes satisfy 
similar (but not all) properties to those in Definition \ref{8.1}, but with $r_{j,1}$ replaced with $r_{j,2}$. More specifically, the list of properties in each item is as follows.

$1'$. Each enlarged base cube  $\tilde K^{b(j)}_m$  contains a ball of radius $\frac14\En^{\tilde\gamma^{(j)}_mr_{j,2}}$ and is contained in a ball of radius $\frac12\En^{\tilde\gamma^{(j)}_mr_{j,2}}$ centred at the origin.  
The factors $\tilde\gamma^{(j)}_m$ are numbers inside $(\gamma_0,1-\gamma_0)$.

$2'$. Each small enlarged base cube $\tilde K^{b(j),{\mathrm {small}}}_m$ contains a ball of radius $\frac14\En^{\tilde\gamma^{(j)}_mr_{j,2}/Z_0}$ and is contained in a ball of radius $\frac12\En^{\tilde\gamma^{(j)}_mr_{j,2}/Z_0}$ centred at the origin. 

$3'$ is the same as item $3$ from definition \ref{8.1}, but with usual base cubes replaces with the enlarged base cubes (notice that there are no enlarged base cubes of order $0$). 

$4'$. The distance between $\tilde K^{(j)}_m$ and $\tilde K^{(j)}_{m'}$ is at least $\frac14\max(\En^{\tilde\gamma^{(j)}_mr_{j,2}},\En^{\tilde\gamma^{(j)}_{m'}r_{j,2}})$. The distance between $\tilde K^{(j)}_m$ and $ K^{(s)}_{m'}$, $s\leq j$, is at least $\frac18\En^{\gamma^{(s)}_{m'}r_{s,1}}$ (unless 
$K^{(s)}_{m'}\subset \tilde K^{(j)}_m$). If $s\leq j$ and $K^{(s)}_{m'}\subset \tilde K^{(j)}_{m}$, then the distance between $K^{(s)}_{m'}$ and $\Z^l\setminus \tilde K^{(j)}_{m}$ is at least $10Q$.

$5'$. For each bad  \Bourgain cube $K^{(j)}_m$, $j=0,\dots,n-1$, there is a small enlarged  \Bourgain cube  $\tilde K^{(j+1),{\mathrm {small}}}_{m'}$ such that 
\bee
K^{(j)}_m\subset \tilde K^{(j+1),{\mathrm {small}}}_{m'}.
\ene
Also, for any small enlarged  \Bourgain cube $\tilde K^{(j+1),small}_{m'}$ there exists a bad  \Bourgain cube $K^{(j)}_{m}$ inside it.

\end{defin} 

\ber\label{newrem5}
The enlarged \Bourgain structure is not `parallel' to the usual structure, but rather an additional auxiliary structure needed for technical purposes.  
The difference between enlarged and usual structures is, in particular, that we do not distinguish bad or good enlarged cubes, and there are no enlarged cubes of order zero. {\newred The enlarged \Bourgain structure is serving mostly technical purposes, in particular we use is it to establish estimates of the type  \eqref{Cartan3}.}
\enr

\ber
Note that, as in Definition \ref{8.1}, each enlarged cube $\tilde K^{(j)}_m$, $j=1,\dots,n$, is a shift of some `base' enlarged cube $\tilde K^{b(j)}_{\ham}$ by a corresponding vector $\bn$.
\enr

Our first  objective during the inductive step is to prove that for most $\vec\boldom$ if we assume that we have the enlarged \Bourgain structure  at  $ \bk,\rho,\vec\boldom$ of order $n-1$, then we have the enlarged \Bourgain structures of order $n$. Thus, we want to define recursively a set ${\mathcal G }^{\vec\boldom (n)}(\tilde E)$ that consists of frequencies that allow an enlarged \Bourgain structure of order $n$ for all $\rho\ge \tilde E$ and prove that its measure is sufficiently large.  Recall that the set ${\mathcal G }^{\vec\boldom (0) }={\mathcal G }^{\vec\boldom (0) }_{B_0}$ was introduced in the beginning of Section \ref{section3} (this set does not depend on the Fourier coefficients of the potential), and
 the set ${\mathcal G }^{\vec\boldom (1) }(\tilde E)$ is defined in Corollary~\ref{MGC1}; this set (as well as all the consecutive good sets of frequencies) depend on the choice of the Fourier coefficients $V_\bn$.  The proper inductive definition of the sets ${\mathcal G }^{\vec\boldom (n) }(\tilde E)$, $n\geq 2$, is given by Corollary~\ref{MGC}. Here, we just mention that the set  ${\mathcal G }^{\vec\boldom (n) }(\tilde E)$ consists of  frequencies for which there exist the usual and enlarged \Bourgain structures of order $n$ for all $\rho\ge\tilde E$.

Suppose, $\vec \boldom \in {\mathcal G }^{\vec\boldom (n) }(\tilde E)$ and $\rho\ge\tilde E$, $\rho\in[E-1,E+1]$, and 
define
\bee\label{Smn}
\bes
S_{\ham}^{(n)}=&
\Big\{(\vec \boldom, \rho ,\bxi )\in\CA^{\vec \boldom, \rho ,\bxi, (n) }(\vec\boldom^{*(n)}_{\tim^{\vec\boldom}},\rho^{*(n)}_{\tim^{\rho}},\bxi^{*(n)}_{\tim^{\bxi}}):   \\
&||(H( K^{b(n)}_{\ham};\bxi)-\rho^2)^{-1}||_2>\En^{r'_{n+1}(\gamma^{(n)}_{\ham} )}\Big\} 
\end{split}     
 \ene
 and
 \bee\label{tildeSmn}
 \bes
\tilde S_{\ham}^{(n)}=&
\Big\{(\vec \boldom, \rho ,\bxi )\in\CA^{\vec \boldom, \rho ,\bxi, (n) }(\vec\boldom^{*(n)}_{\tim^{\vec\boldom}},\rho^{*(n)}_{\tim^{\rho}},\bxi^{*(n)}_{\tim^{\bxi}}):   \\
&||(H( \tilde K^{b(n)}_{\ham};\bxi)-\rho^2)^{-1}||_2>\En^{\tilde r'_{n+1}(\tilde\gamma^{(n)}_{\ham} )}
\Big\}.
\end{split}    
 \ene
Here,  
\bee\label{tildern'}
\tilde r'_{n}(\gamma):=\En^{2l\gamma r_{n-1,2}/Z_0}, \ \ \ \tilde r'_{n}:=\tilde r'_{n}(1)
\ene
(compare with \eqref{rn'}) and $\ham=\ham(\tim^{\rho}, \tim^{\vec\boldom}, \tim^{\bxi})$ as described in part 1 of definition \ref{8.1}. Note that as a result there are not more than $E^{2d(l+2)r_n'}$ possible values for $\ham$.  We also define
\bee
S^{(n)}_{total}:=\cup_{\ham} S_{\ham}^{(n)}.
\ene


Another inductive assumption is given by the following definition:
\begin{defin}\label{MGD}
We say that our \Bourgain structure at level $n$ is reasonable on ${\mathcal G }^{\vec\boldom (n)}(\tilde E)$ above energy $\tilde E$, if for each $\vec\boldom\in{\mathcal G }^{\vec\boldom (n) }(\tilde E)$, each $ \rho\ge \tilde E$, $\rho\in [E-1,E+1]$, and each $\ham$, 
the cross-section $(S_{\ham}^{(n)})_{\mathrm {cs}}(\vec\boldom,\rho)$ satisfies 
\bee\label{measSmn}\meas((S_{\ham}^{(n)})_{\mathrm {cs}}(\vec\boldom,\rho))<\En^{-(r'_{n+1}(\gamma^{(n)}_{\ham}))^{1/4}}.
\ene 

We say that our enlarged \Bourgain structure at level $n$ is reasonable above energy $\tilde E$, if, in addition, for each $\vec\boldom\in{\mathcal G }^{\vec\boldom (n) }(\tilde E)$, each $ \rho\geq \tilde E$, $\rho\in [E-1,E+1]$, and each $\ham$, the cross-section  
$(\tilde S_{\ham}^{(n)})_{\mathrm {cs}}(\vec\boldom,\rho)$ satisfies 
\bee\label{meastildeSmn}\meas((\tilde S_{\ham}^{(n)})_{\mathrm {cs}}(\vec\boldom,\rho))<\En^{-(\tilde r'_{n+1}(\tilde\gamma^{(n)}_{\ham}))^{1/4}}.\ene
\end{defin}

{\newred 
Analogously to Corollary \ref{6.2}, we can prove the following estimate:
\bel\label{new:121}
 The set $S^{(n)}_{total}$ is a semi-algebraic subset in $\R^{ld+1+d}$ of degree  $E^{3d(l+2)r_n'}$. 
\enl}

Now we can formulate the first main inductive statement. It concerns the existence of \Bourgain structures:
\bet\label{MGT}

a)  There are reasonable (usual and enlarged) \Bourgain structures  at level $1$ on ${\mathcal G }^{\vec\boldom (1)}(\tilde E)$ above energy $\tilde E$.

b) Suppose, $n\ge 1$ and there are reasonable \Bourgain structures (usual and enlarged) at level $n$ on ${\mathcal G }^{\vec\boldom (n)}(\tilde E)$ above energy $\tilde E$.  Then we can construct a set ${\mathcal G }^{\vec\boldom (n+1)}(\tilde E)\subset{\mathcal G }^{\vec\boldom (n)} (\tilde E)$ such that there are reasonable \Bourgain structures (again, both usual and enlarged) at level $n+1$ on ${\mathcal G }^{\vec\boldom (n+1)}(\tilde E)$ above energy $\tilde E$. Moreover, 
\bee 
\frac{\meas({\mathcal G }^{\vec\boldom (n+1) }(\tilde E))}{\meas({\mathcal G }^{\vec\boldom (n) }(\tilde E))}=_{\tilde E\to \infty }1-O\left(\tilde E^{- \tilde C_1r_{n+1,1}(\tilde E)}\right), \ \ \tilde C_1=\tilde C_1(Z_0).\label{MGTmeasure}
\ene

For any $E\ge \tilde E$ and $n\ge 1$, the \Bourgain structures at level $n$ can be made stable in $\bk$ on any matryoshka of (quasi-)patches  $\CM^{\bk(n)}=\{\CA^{\bk(j)}\}_{j=0}^n$  and in $\rho$ on any matryoshka of patches  $\CM^{\rho(n)}=\{\CA^{\rho(j)}\}_{j=0}^n$ that has a non-empty intersection with $[E-1,E+1]$. 
\ent 

\ber\label{rem2}
The necessity to establish the stability of the \Bourgain structures in $\bk$ and $\rho$ on matryoshkas of patches is the reason why we had to introduce these matryoshkas in the first place, cf. remark \ref{rem1}. The reader may be surprised however that we seemingly do not use this stability anywhere else in our paper. The point is that this stability is needed to properly apply the machinery of \cite{KS} in section \ref{8.1'}.  
\enr

The proof of this statement will be given in Section~\ref{section8}. We just remark here that neither of the estimates \eqref{measSmn} or \eqref{meastildeSmn} on its own is enough to perform an inductive step, even for the price of throwing away a small set of frequencies: we really need both \eqref{measSmn} and \eqref{meastildeSmn} to have a proper inductive statement. More precisely, we use \eqref{measSmn} at level $n$ to prove the existence of the \Bourgain structure at level $n+1$ and then we use \eqref{meastildeSmn} at  level $n$ to establish \eqref{measSmn} at level $n+1$. Luckily, we can also use \eqref{meastildeSmn} at  level $n$ to prove \eqref{meastildeSmn} at  level $n+1$.

Now we are going to define a notion of a good patch in $\BPhi$, $\CA^{\BPhi(n)}_{\tim_n}$. We define it inductively. We have already defined a good set of $\vec\phi$ and $\BPhi$ of order $1$, and this allows us to define good patches of order $1$. We define them as those  patches $\CA^{\BPhi(1)}$ 
that entirely consist of points of the form $\BPsi_j(\CG_{j,\C}^{\vec\phi(1)})\cap\R^{d-1}$ for some $j$; recall definition \eqref{good1}  
(however, check Theorem \ref{Thm7.4} to be sure). Suppose, we have defined a good patch $\CA^{\BPhi(j)}_{\tim_j}$ for all $j\le n-1$. We are going to define a good patch of order $n$. First, we recall that any patch that we declare bad is just thrown away: we do not consider sub-patches at the next level of a bad patch.  Also, we recall the following definition: 

\begin{defin}\label{def7.1bis}
A matryoshka $\CM_{\textrm {patches}}^{\BPhi(q)}$ of patches of level $q$ is a collection of patches $\{\CA^{\BPhi(j)}_{\tim_j}\}_{j=1}^q$, where each patch of level $j+1$ lies inside a patch of level $j$. We say that matryoshka $\CM_{\textrm {patches}}^{\BPhi(q)}$ is good, if each of the patches of it is good.
\end{defin} 
\ber
This definition is proper for $q\le n-1$, since we have assumed that we know the definition of good patches of all levels up to $n-1$. 
\enr

\begin{defin}\label{def7.2}
A matryoshka $\CM_{\textrm {cubes}}^{(n)}$ of central cubes of level $n$ is a collection $\{\hat K^{(j)}\}$, $j=1,...,n$ of extended cubes centred at the origin; $\hat K^{(j)}\subset\Zl$ has size  $\frac12\En^{r_{j,2}}$ (note that it is much bigger than the  size of any \Bourgain cube of level $j$). 
\end{defin} 

\ber\label{sentral}
In general, there is no relationship between a central cube $\hat K^{(j)}$ and \Bourgain cubes $ K^{(j)}_m$ of the same order. However, in the situations we study any central cube $\hat K^{(j)}$ contains not more than one \Bourgain cubes $ K^{(j)}_m$ of the same order (and this is the \Bourgain cube that contains the origin $\bn=0$). This fact will be established in section \ref{section9}. 
\enr

\begin{defin}\label{def7.3} Let $\rho\in[E-1,E+1]$ and $\vec\boldom\in {\mathcal G }^{\vec\boldom (n-1) }(\tilde E)$, $E\geq\tilde E\geq E_*$, be fixed. Consider 
a matryoshka $\CM_{\textrm {patches}}^{\BPhi(n)}$ of patches  of level $n$ and a matryoshka $\CM_{\textrm {cubes}}^{(n)}$ of central cubes. We assume that a sub-matryoshka $\CM_{\textrm {patches}}^{\BPhi(n-1)}$ of patches (i.e. our original  matryoshka with the last patch of order $n$ removed) is good. We say that $\CM_{\textrm {patches}}^{\BPhi(n)}$ and $\CM_{\textrm {cubes}}^{(n)}$ 
 are synchronised (with respect to $\rho$) if for $j=1$ the statement of Theorem \ref{Thm2} holds  for $H(\hat K^{(1)},\bk)$ and for $1<j\le n$ the following recursive Statement (the recursion means that the formulation of the Statement at level $n$ uses objects obtained in the same Statement at the previous step $(n-1)$) holds:

{\bf Statement at level $ n$.} \label{leveln}

Suppose, $\BPhi\in \CA^{\BPhi (n)}_{\tim}, \BPhi=\BPsi_{\tim}(\vec\phi)$,
$\k\in\R$,
$|\k-\k^{(n-1)}(\vec\phi,\rho )|\leq \En^{-r_{n,3}\En^{3l r_{n-1,1}/Z_0}}$,
$\bka=\k\BPsi_{\tim}(\vec\phi)$. Then there exists a single eigenvalue $\lambda^{(n)}({\bka})$ of
$H(\hat K^{(n)},\bka)$ in the interval\\ 
\bee
I_{n}:=\left( \rho^2-\En^{-r_{n,3}\En^{3lr_{n-1,1}/Z_0}},
\rho^{2}+\En^{-r_{n,3}\En^{3lr_{n-1,1}/Z_0}}\right). 
\ene
This eigenvalue is given by the absolutely
converging series:
\begin{equation}\label{eigenvalue-n}\lambda^{(n)}({\bka})=\lambda^{(n-1)}({\bka})+
\sum\limits_{q=2}^\infty g^{(n)}_q({\bka}).
\end{equation} 
The coefficients $g^{(n)}_q({\bka})$ satisfy the following estimates:
\begin{equation}\label{estgn} |g^{(n)}_q({\bka})|<\En^{-\En^{r_{n-1,2}}\En^{-r_{n-2,2}}}
\En^{-\sigma_0q/4}.
\end{equation}
The corresponding spectral projection is given by the series:
\begin{equation}\label{sprojector-n}
\E ^{(n)}({\bka})=\E^{(n-1)}({\bka})+\sum\limits_{q=1}^\infty
G^{(n)}_q({\bka}).
\end{equation} 
 The operators $G^{(n)}_q({\bka})$ satisfy
the estimates:
\begin{equation}
\label{Feb1a-n} \left\|G^{(n)}_q({\bka})\right\|_1<\En^{-\En^{r_{n-1,2}}\En^{-r_{n-2,2}}}
\En^{-\sigma_0q/4}
\end{equation}
and
\begin{equation}\label{Feb6a-n}
G^{(n)}_q({\bka})_{\bn\bn'}=0,\ \ \mbox{if}\ \ 4\sqrt{d}\cdot q\En^{r_{n-1,2}}<|\bn|+|\bn'|. 
\end{equation}
Here, $\E^{(n-1)}({\bka})$ and $\lambda^{(n-1)}({\bka})$ are respectively the spectral
projection and the eigenvalue of $H(\hat K^{(n-1)},\bka)$ obtained at the previous step and $\kappa^{(n-1)}(\vec\phi)=\kappa^{(n-1)}_{\tim}(\vec\phi)$
is the unique $\kappa$-solution of the equation $\la^{(n-1)}(\kappa\BPsi_{\tim}(\vec\phi))=\rho^2$ (the isoenergetic surface of the previous level). 

Coefficients $g^{(n)}_r({\nbka})$ and operators
$G^{(n)}_r({\nbka})$ can be analytically extended to the complex neighbourhood $\CA^{\vec\phi (n)}_{\tim,\C}$ as functions of $\vec\phi $ and to the complex $\En^{-r_{n,3}\En^{3l r_{n-1,1}/Z_0}}-$
neighbourhood of $\nka^{(n-1)}(\vec\phi )$ as functions
of $\k$, estimates \eqref{estgn}, \eqref{Feb1a-n}
being preserved.

{\bf End of the Statement at level $ n$.}

If a patch $\CA^{\BPhi (n)}_{\tim}$ is the last patch in a matryoshka of patches 
 synchronised with a matryoshka of central cubes of order $n$ (i.e., if the Statement at level $n$ holds), we call this patch perfect. 

\end{defin}
\ber
Since this Statement is recursive, we assume that if this statement holds at level $n$, then it holds at all levels $j$, $j\le n$. 
\enr

Assuming this Statement holds, we can make the following conclusions:

\bel \label{corthmn} Suppose, the Statement holds. Then
for the perturbed eigenvalue and its spectral
projection the following estimates hold:
 \begin{equation}\label{perturbation-n}
\lambda^{(n)}({\bka})=\lambda^{(n-1)}({\bka})+ O\left(\En^{-\En^{r_{n-1,2}}\En^{-r_{n-2,2}}}\right),
\end{equation}
\begin{equation}\label{perturbation*-n}
\left\|\E^{(n)}({\bka})-\E^{(n-1)}({\bka})\right\|_1<\En^{-\En^{r_{n-1,2}}\En^{-r_{n-2,2}}},
\end{equation}
\begin{equation}
\left|\E^{(n)}({\bka})_{\bn\bn'}\right|<\En^{-\dis^{(n)}(\bn,\bn')}\ \
{\mathrm {when}}\ |\bn|>4\sqrt{d}\En^{r_{n-1,2}} \mbox{\ or }
|\bn'|>4\sqrt{d}\En^{r_{n-1,2}},\label{Feb6b-n}
\end{equation}
where
$$\dis^{(n)}(\bn,\bn')=\frac{\sigma_0}{16\sqrt{d}}(|\bn|+|\bn'|)\En^{-r_{n-1,2}}+\En^{r_{n-1,2}}\En^{- r_{n-2,2}}.$$
\enl
\bel \label{L:derivatives-n}
Assume the Statement. Then the following
estimates hold when $\vec\phi \in \CA_{{\tim},\,\C}^{\vec\phi(n)}$ and $\k\in \C:$
$|\k-\k^{(n-1)}(\vec\phi )|<\En^{-r_{n,3}\En^{3lr_{n-1,1}/Z_0}}:$
\begin{equation}\label{perturbation-nc}
\lambda^{(n)}({\bka})=\lambda^{(n-1)}({\bka})+ O\left(\En^{-\En^{r_{n-1,2}}\En^{-r_{n-2,2}}}\right),
\end{equation}
\begin{equation}\label{estgder1-nk}
\frac{\partial\lambda^{(n)}}{\partial\k}=\frac{\partial\lambda^{(n-1)}}{\partial\k}
+ O\left(\En^{-\En^{r_{n-1,2}}\En^{-r_{n-2,2}}}\En^{r_{n,3}\En^{3lr_{n-1,1}/Z_0}}\right). \end{equation}

Similar estimates can be written for all derivatives of $\lambda^{(n)}$ and $\E^{(n)}$ with respect to $\nka$ and $\vec \phi $.
\enl
Finally,
\bel\label{ldk-n} \begin{enumerate}
Assume the Statement. Then: 
\item For every $\lambda :=\rho^{2}$ with $\rho>E_*$ and $\vec\phi \in\CA_{{\tim},\,\C}^{\vec\phi(n)}\cap\R^{d-1}$, there is a unique
$\k^{(n)}(\vec\phi,\rho )$ in the interval
$$\tilde I_{n}:=[\k^{(n-1)}(\vec\phi,\rho )-\En^{-r_{n,3}\En^{3lr_{n-1,1}/Z_0}},\k^{(n-1)}(\vec\phi,\rho
)+\En^{-r_{n,3}\En^{3lr_{n-1,1}/Z_0}}],$$ such that
    \begin{equation}\label{2.70-n}
    \lambda^{(n)} \left(\bka
^{(n)}(\vec\phi,\rho )\right)=\rho^2 ,\ \ \bka ^{(n)}(\vec\phi,\rho ):=\k^{(n)}(\vec\phi,\rho )\BPsi_{\tim}(\vec\phi).
    \end{equation}
\item  Furthermore, there exists an analytic in $ \vec\phi $ continuation  of
$\k^{(n)}(\vec\phi,\rho)$ to the complex set $\CA_{{\tim},\,\C}^{\vec\phi(n)}$
such that $\lambda^{(n)} (\bka ^{(n)}(\vec\phi,\rho))=\rho^2 $.
Function $\k^{(n)}(\vec\phi,\rho )$ can be represented as
$\k^{(n)}(\vec\phi,\rho )=\k^{(n-1)}(\vec\phi,\rho
)+h^{(n)}(\vec\phi,\rho)$, where
\begin{equation}\label{dk0-n} |h^{(n)}(\vec\phi )|=
 O\left(\En^{-\En^{r_{n-1,2}}\En^{-r_{n-2,2}}}\right),
\end{equation}
\begin{equation}\label{dk-n}
\bes
\frac{\partial{h}^{(n)}}{\partial\vec\phi}&=  O\left(\En^{-\En^{r_{n-1,2}}\En^{- r_{n-2,2}}}\En^{r_{n,3}\En^{3lr_{n-1,1}/Z_0}}\right),\\
\frac{\partial^2{h}^{(n)}}{\partial\vec\phi^2}&= O\left(\En^{-\En^{r_{n-1,2}}\En^{- r_{n-2,2}}}\En^{2r_{n,3}\En^{3lr_{n-1,1}/Z_0}}\right).
\end{split}
\end{equation} \end{enumerate}
\enl

\begin{defin}\label{def7.4}
We say that an enlarged \Bourgain structure  and a matryoshka of central cubes are consistent, if the following conditions hold:

1. Each $ \tilde K^{(j)}_m$, $j=1,...,n$, is either inside $\hat K^{(j)}$, or is at least $\frac14\En^{\tilde\gamma^{(j)}_mr_{j,2}}$-away from it (the distance is at least the size of the enlarged multiscale cube).

2. Each $ K^{(s)}_m$, $s<j$, is either inside $\hat K^{(j)}$, or is at least $\frac18\En^{\gamma^{(s)}_mr_{s,1}}$-away from it (the distance is at least the size of the smaller cube). 

3. If $ K^{(s)}_m\subset \hat K^{(j)}$, $s\leq j$, then $ K^{(s)}_m$ is at least $10Q$-away from $\Z^l\setminus \hat K^{(j)}$.
\end{defin}

\begin{defin}\label{def7.5}
We call a \Bourgain cube $ K^{(s)}_m$ ($s\ge 1$) generated by $\bk$ a {\bf housewife}, if it contains the origin (point $\bn=0$). 
We call a \Bourgain cube
a {\bf prodigal son}, if it does not contain the origin, but it contains a point $\bq$ such that $||\bq\vec\boldom||<\En^{-r_{s,3}\En^{4lr_{s-1,1}/Z_0}}$. A cube that is neither a housefive, nor a prodigal son, is called a {\bf globetrotter}. Enlarged \Bourgain cubes with these properties are called `enlarged housewife', `enlarged prodigal son' and `enlarged globetrotter' respectively. 
For $s=1$ the estimate defining a prodigal son is $\|\bq\vec\boldom\|<\En^{-r_{1,3}\En^{3d^2(l+\mu)\sigma_{0,d-1}}}$.
\end{defin}
\ber The reason for this terminology is this: a `housewife' stays at home, a `prodigal son' goes away, but then returns (sort of), and a globetrotter goes away and never comes back. 
The reason why we actually need this definition is as follows. First, we notice that the distance that defines the prodigal sons, $\En^{-r_{s,3}\En^{4lr_{s-1,1}/Z_0}}$, is smaller than the size of the patch in $\bk$. This means that prodigal sons can, after small manipulation, be considered as the shifted central cubes. Therefore, we can establish their properties just by taking the properties of the central cube and shifting them. Since, on the other hand, this shift is not extremely small (due to the Diophantine properties), we can easily achieve that each prodigal son has a unique eigenvalue which is close, but not too close, to $\rho^2$. The globetrotters, on the other hand, can be treated using the standard methods (like the Cartan's Lemma), which are unavailable for the prodigal sons. We also remark that, obviously, there is at most one housewife at each level $s\ge 1$.   
\enr

Together with the Statement formulated in Definition \ref{def7.3}, another important inductive assumption will be the following:

\begin{defin}\label{excellent}
We say that a perfect patch $\CA^{\BPhi(n)}_{\tim}=\BPsi^{(n)}_{\tim}(\Pi^{(n)}_{\tim})$ is excellent, if the following property (called the important inductive estimate, or IIE) holds on $\CA^{\BPhi(n)}_{\tim}$.

{\bf Important inductive estimate at level $n$}:

1) Let us fix a globetrotter $K^{(n)}_m$. 
 After throwing away a set $\CNN^{\vec\phi(n+1)}_{m}\subset \Pi_{\tim}^{(n)}$ of measure not greater than $\En^{-r_{n+1,3}}$, for the rest of $\vec\phi\in \Pi_{\tim}^{(n)}\setminus \CNN^{(n+1)}_{m}$ we have 
\begin{equation}\label{Cartan5-n}
\|((H(K^{(n)}_m,\bka^{(n)}(\vec\phi))-\rho^2))^{-1}\|\leq \En^{r_{n+1,3}\En^{2l\gamma^{(n)}_m r_{n,1}/Z_0}}.
\end{equation}

2) Let $\tilde K^{(n)}_m$ be an enlarged globetrotter. 
After throwing away a set $\tilde\CNN^{\vec\phi(n+1)}_{m}\subset \Pi_{\tim}^{(n)}$ of measure not greater than $\En^{-r_{n+1,3}\En^{4l\tilde \gamma^{(n)}_m r_{n,2}/Z_0}}$, for the rest of $\vec\phi\in \Pi_{\tim}^{(n)}\setminus \tilde\CNN^{(n+1)}_{m}$ we have 
\begin{equation}\label{Cartan5-ntilde}
\|((H(\tilde K^{(n)}_m,\bka^{(n)}(\vec\phi))-\rho^2))^{-1}\|\leq \En^{\left(r_{n+1,3}\En^{4l\tilde \gamma^{(n)}_m r_{n,2}/Z_0}\right)\En^{2l\tilde \gamma^{(n)}_m r_{n,2}/Z_0}}.
\end{equation}

{\bf End of Important inductive estimate at level $n$}

\end{defin}

\ber
Despite cases 1 and 2 in the definition above of an excellent patch looking similar, there is an important distinction: inequality \eqref{Cartan5-n} is stable on a patch of the next level, which means that if it holds at one point of that patch, it holds everywhere (with possibly an extra factor $2$ in the RHS). However, inequality \eqref{Cartan5-ntilde} is not stable on our next level patches, which means that a set  $\tilde\CNN^{\vec\phi(n+1)}_{m}$ (where this inequality is not satisfied) is quite difficult to control and, besides its measure being small, we do not know any further properties of it. 
\enr

Now, finally, we can define a good patch of order $n$. 

\begin{defin}\label{goodpatchesPhi}
We say that a patch $\CA^{\BPhi(n)}_{\tim}$ is good, if it is perfect and excellent, i.e. if the Statement (definition  \ref{def7.3}) and IIE (definition \ref{excellent}) hold there. 
\end{defin}
\ber
Since bad patches are always thrown away, and only good patches are covered by patches of the next levels,  the definition of a good patch of order $n$ implies that the Statement and IIE hold not just at level $n$, but also at all levels $j\le n$.   
\enr

Now we can formulate our second  inductive Theorem. Before doing this, note that, strictly speaking, we have cheated a bit when discussing the notion of good patches, since the first level patches $\CA^{\BPhi(1)}$ that we have called good have not been proved to be good. We have proved that they are perfect (i.e., the Statement holds) in section \ref{section5}, but we have not proved that they are excellent (i.e. that IIE holds there). Let us do both things at once:

\bet\label{Thm7.4}

a) The good patches $\CA^{\BPhi(1)}$ are indeed good (the IIE holds there).

b) Suppose $n\geq1$. Let $\rho\in[E-1,E+1]$ and $\vec\boldom\in {\mathcal G }^{\vec\boldom (n) }(\tilde E)$ be fixed, where $\tilde E\le E$. Suppose, we have a good matryoshka of patches (in $\BPhi$) of level $n$, $\CM_{\textrm {patches}}^{\BPhi(n)}$ such that corresponding matryoshka of central cubes $\CM_{\textrm {cubes}}^{(n)}$ is consistent and synchronised with the \Bourgain structure. Suppose that IIE at level $n$ is satisfied at the last patch $\CA^{\BPhi(n)}_{\tim}$. Consider a simple covering of this patch by the patches of the next level
$\{\CA^{\BPhi(n+1)}_{\tim'}\}_{\tim'=1}^M$, $M=M_{n+1}\sim \En^{d r_{n+1,3}\En^{3lr_{n,1}/Z_0}}$. Then   
we can choose at least $M_{n+1}(1-\En^{-r_{n+1,3}/2})$
of these next level patches (which we will call the good patches), such that for each of these patches $\CA^{\BPhi(n+1)}_{\tim'}$ there exists a corresponding central cube $\hat K^{(n+1)}_{\tim'} $ (possibly different for each $\tim'$) such that the new matryoshka of central cubes of level $n+1$ is consistent and synchronised with the \Bourgain structure (in particular, the Statement at level $n+1$ holds). Moreover, at these patches 
the  IIE also holds at level $n+1$. In other words, if a patch is good (i.e. perfect and excellent), then most of the patches at the next level are also perfect and excellent. 
\ent
\ber
The statement of the theorem can be loosely reformulated as follows: The Statement and IIE at level $n$ are persistent at the next level, modulo throwing away a set of spherical angles of small measure. We remark that neither the Statement, nor the IIE on their own are  persistent  at the next level. Even if we consider the Statement with a half of IIE (e.g. \eqref{Cartan5-n}, but not \eqref{Cartan5-ntilde}), this combination cannot be shifted to the next level, modulo a small set of spherical angles: we really need the entire package  of the  Statement, \eqref{Cartan5-n}, and \eqref{Cartan5-ntilde} to make a proper inductive step. 
\enr

The proof of this theorem is given in Section \ref{section9}. It allows us to define the good sets of spherical angles at each level. Namely, the good sets $\CG^{\BPhi(0)}$ and $\CG^{\BPhi(1)}$ are defined by formulas \eqref{W1} and  \eqref{goodPhi1} respectively. Theorem~\ref{Thm7.4} provides the inductive construction of the sets $\CG^{\BPhi(n)}=\CG^{\BPhi(n)}(\rho)$, $n\geq 2$: $\CG^{\BPhi(n)}$ is a union of all good patches at level $n$. Estimates obtained previously (see \eqref{W1} and \eqref{W1'}) allow us to estimate:
\bee\label{CGPhi}
\meas(\CG^{\BPhi(n)})\ge \meas(\S)(1-E^{-\sigma_{0}}),\ \ E\geq E_*.
\ene

 \section{Induction. Proof of Theorem \ref{MGT}}\label{section8}
 
 The strategy of the proof is going to be as follows. First, we will prove that if we assume the existence of a reasonable \Bourgain structure at level $n$, then there is a \Bourgain structure at level $n+1$ (for the price of throwing away a small proportion of frequencies). This is done in subsection \ref{9.1}; we also show there the existence of a \Bourgain structure of order $1$. The first stage of the proof is quite similar to the proof of Main {\newred Semi-Algebraic} Lemma at level one in section \ref{goodset-3}, after which we have to modify the cubes obtained in that lemma slightly to construct a proper \Bourgain structure. 

Then, in subsection \ref{9.2}, we prove that this \Bourgain structure is reasonable. 
We first prove that the structure is reasonable at level $1$: this, together with the existence of \Bourgain structure at level one can be seen as the base of induction. Next, we will finish proving the inductive step and prove that the structure is reasonable at level $n+1$. The proofs of both cases are very similar to each other though. 

Finally, at the end of subsection \ref{9.2}, we will define the good set of frequencies at level $n+1$, ${\bf \mathcal G }^{{\vec\boldom }(n+1)}$. 
 
Thus, to begin with, we assume that the \Bourgain structure (stable in $\rho$ and $\bk$ on the lower part -- up to the level $n$ -- of our matryoshkas of patches) has been established.

\subsection{The main {\newred Semi-Algebraic} Lemma at high levels}\label{9.1}

The first step in the proof of the existence of both (usual and enlarged) reasonable \Bourgain structures of order $n+1$ is based on estimates \eqref{measSmn} and \eqref{meastildeSmn} and is completely similar to the construction in section \ref{goodset-3} following Lemma~\ref{L5.5}. Recall that we have fixed the values of $Z_0$ and $\gamma_0$ in \eqref{Z_0} and \eqref{gamma_0}. In Lemma \ref{MGL} we have constructed a set $G^{(1)}(E,Z_0)$. Now we will state the inductive construction.

\bel\label{MGL-n} Suppose, $n\ge 1$ and there are reasonable \Bourgain structures (usual and enlarged) at level $n$ on ${\mathcal G }^{\vec\boldom (n)}(\tilde E)$ above energy $\tilde E$. For every 
$\En>\tilde E\ (>\En_{*})$, there is a set ${G '}^{(n+1)}(\En)\subset {G }^{(n)}(\En)$,
such that for any $(\vec \boldom, \rho, \bxi  )\in S_{total}^{(n)}$,  with  $\vec \boldom \in { G'^{(n+1)} } (\En)$ and $\rho\in [E-1,E+1]$  
there is a $\gamma $, $\gamma =\gamma ({n}, \vec \boldom, \rho , \bxi)$ with the following properties:
\bee \label{gamma-2-n} \gamma _0<\gamma <1-\gamma _0,
\ene
\bee \label{no-2-n}\left\{ \q:   (\vec \boldom, \rho, \bxi +\q \vec \boldom )\in S_{total}^{(n)}, \ \  \q \in \Omega(E^{\gamma r_{n+1,1}})\setminus \Omega(E^{\gamma r_{n+1,1}/Z_0})\right\}=\emptyset .
\ene
The set ${ G'} ^{(n+1)}(\En)$ has an asymptotically full measure in ${G }^{(n)}(\En)$:
\bee \frac{\meas({G' } ^{(n+1)} )}{\meas({G }^{(n)})}=_{\En\to \infty }1-O(\En ^{- C_1r_{n+1,1}}).\label{mesLambda-2-n}
\ene
The value of $\gamma $ can be taken constant in the $\En^{-r_{n+1}'-2 }$-neighbourhood of every $(\rho, \bxi )$
and in the $\En^{-2r_{n+1}'}$-neighbourhood of every $\vec \boldom $. 
\enl
\bel\label{MGL'-n}
Similar statement holds for the enlarged structure. This means that we can find a (possibly different) set $G''^{(n+1)}(E)$ and a (possibly different) $\tilde\gamma$ such that instead of \eqref{no-2-n} when $\vec\boldom\in G''^{(n+1)}(E)$ we have
\bee \label{no'-2-n}\left\{ \q:   (\vec \boldom, \rho, \bxi +\q \vec \boldom )\in  \tilde S_{total}^{(n)}, \ \  \q \in \Omega(E^{\tilde\gamma r_{n+1,2}})\setminus \Omega(E^{\tilde\gamma r_{n+1,2}/Z_0})\right\}=\emptyset .
\ene
The set ${ G''} ^{(n+1)}(\En)$ has an asymptotically full measure in ${G }^{(n)}(\En)$:
\bee \frac{\meas({G'' } ^{(n+1)} )}{\meas({G }^{(n)})}=_{\En\to \infty }1-O(\En ^{- C_1r_{n+1,2}}).\label{mesLambda-2-n''}
\ene
The value of $\tilde\gamma $ can be taken constant in the $\En^{-r_{n+1}'-2 }$-neighbourhood of every $(\rho, \bxi )$
and in the $\En^{-2r_{n+1}'}$-neighbourhood of every $\vec \boldom $. 
\enl
\bep
As we mentioned above, the proof of both lemmas repeats the arguments from the construction in section \ref{goodset-3} following Lemma \ref{L5.5}, with the set $S^{(0)}$ being replaced by  $S^{(n)}$, estimate \eqref{(7.21)} being replaced by \eqref{measSmn} {\newred and Corollary \ref{6.2} being replaced by Lemma \ref{new:121}, cf Remark \ref{new:B}. The algebraic structure of the sets $S^{(0)}$ and $S^{(n)}$ is the same. 
}
No new restrictions on $E_*$ are imposed during the proof. 
\enp

Let us also define 
\bee
{G } ^{(n+1)}(E):={G' } ^{(n+1)}(E)\cap{G'' } ^{(n+1)}(E).
\ene
Then 
\bee \frac{\meas({G } ^{(n+1)}(E) )}{\meas({G }^{(n)}(E))}=_{\En\to \infty }1-O(\En ^{- C_1r_{n+1,1}(E)}).\label{mesLambda-2-n'}
\ene

Now we will use these lemmas to construct \Bourgain structures at level $n+1$ (stable in $\bk$ and $\rho$). The same construction can be used when we construct \Bourgain structures at level $1$ using lemmas \ref{MGL} and \ref{MGL'}, so for definiteness we concentrate on the inductive step.
We are using the induction assumption that tells us that there is a \Bourgain structure of order $n$ that is stable in $\rho$ and $\bk$ on matryoshkas of order $n$. Now we have to add to this construction cubes of order $n+1$ to create a \Bourgain structure of order $n+1$. We fix $\vec\boldom\in G^{(n+1)}$, $\vec\boldom\in\CA^{\vec\boldom(n+1)}(\vec\boldom^*_{\tim^{\vec\boldom}})$ and $\rho\in\CA^{\rho(n+1)}(\rho^*_{\tim^{\rho}})$. Suppose we have a centre $\bxi^*_{\tim^{\bxi}}$ of the $\bxi$-patch  such that $(\vec\boldom^*_{\tim^{\vec\boldom}},\vec\rho^*_{\tim^{\rho}},\bxi^*_{\tim^{\bxi}})\in S^{(n)}_{total}$. 
Then for $\ham=\ham(n,\vec\boldom^*_{\tim^{\vec\boldom}},\rho^*_{\tim^{\rho}},\bxi^*_{\tim^{\bxi}})$ we initially define the base cubes (the existence of which is postulated in the first and second items of definition \ref{8.1}) as $K^{\mathrm{ib}
(n+1)}_{\ham}:=\Omega(E^{\gamma r_{n+1,1}})$ and  $K^{\mathrm{ib}(n+1){\mathrm {small}}}_{\ham}:=\Omega(E^{\gamma r_{n+1,1}/Z_0})$ (the index $\mathrm{ib}$ stands for `initial base' -- as we will see in a moment, we will have to modify these base cubes).   
Here, $\gamma=\gamma(n,\vec\boldom^*_{\tim^{\vec\boldom}},\rho^*_{\tim^{\rho}},\bxi^*_{\tim^{\bxi}})$ is given by lemma \ref{MGL-n}. That lemma also ensures that the same value of $\gamma$ works for all $\rho\in\CA^{\rho(n+1)}(\rho^*_{\tim^{\rho}})$. Obviously, the corresponding cube depends not only on the patch $\CA^{\bxi(n+1)}(\bxi^*_{\tim^{\bxi}})$, but also on the choice of matryoshka $\CM^{\bxi(n+1)}$ covering this patch. 
  
Suppose now that $\CM^{\bk(n+1)}=\{\CA^{\bk(j)}\}_{j=0}^{n+1}$ is a matryoshka of patches in $\bk$, and we want to construct a \Bourgain structure,  stable in $\bk$ on $\CM$. Stability with respect to patches of order $j<n+1$ in $\CM$ follows from the induction assumption, so now we will discuss how to achieve the stability in $\bk$ on the patch $\CA^{\bk(n+1)}$. 
First, for any $\bn\in\Zl$ we look at the shifted patch $\CA^{\bk(n+1)}+\bn\vec\boldom$. By the construction of the patches discussed in section \ref{section1} (patches in $\bxi$ being ten times bigger than the patches in $\bk$), there is a centre of a $\bxi$-patch, $\bxi^{*}=\bxi^*(\bn)$, such that  
$\CA^{\bk(n+1)}+\bn\vec\boldom\subset\CA^{\bxi(n+1)}(\bxi^{*})$. Now let us first try to choose the base cubes $K^{\mathrm{ib}(n+1)}=K^{\mathrm{ib}(n+1)}(\bn)$ and $K^{\mathrm{ib}(n+1){\mathrm {small}}}=K^{\mathrm{ib}(n+1){\mathrm {small}}}(\bn)$ as just described (i.e. we take $\ham=\ham(n,\vec\boldom^*_{\tim^{\vec\boldom}},\rho^*_{\tim^{\rho}},\bxi^*_{\tim^{\bxi}}(\bn))$). We denote the parameter $\gamma$ of the corresponding cube by $\gamma^{(n+1)}_{\bn}$. We will soon see what is wrong with such a choice of base cubes, after which we will modify them accordingly.

So, let us try to work with `initial base' cubes as defined above. The next step would be to shift them by vectors $\bn$ as prescribed in condition $3$ of definition \ref{8.1}. Let us denote these shifts by 
$K^{\mathrm{i}(n+1)}(\bn):=K^{\mathrm{ib}(n+1)}(\bn)+\bn$ and $K^{\mathrm{i}(n+1){\mathrm {small}}}(\bn):=K^{\mathrm{ib}(n+1){\mathrm {small}}}(\bn)+\bn$, where the index  $\mathrm{i}$ stands for `initial'. 
If we add these cubes to the already existing \Bourgain cubes $K^{(j)}_{m}$ and $K^{(j){\mathrm {small}}}_{m}$ ($j<n+1$), we will get the structure that satisfies all conditions of definition \ref{8.1} except, possibly, condition $4$; we obviously also have stability in $\bk$ and $\rho$. Condition $4$ is satisfied for $j,j'<n+1$ by the inductive assumption. We now explain how to modify $K^{\mathrm{ib}(n+1)}(\bn)$ to satisfy condition $4$ for $j=n+1$ and/or $j'=n+1$.

1) First, consider any two initial cubes of the same level $n+1$ ($K^{\mathrm{i}(n+1)}(\bn_m)=K^{\mathrm{ib}(n+1)}(\bn_m)+\bn_m$ and  $K^{\mathrm{i}(n+1)}(\bn_{m'})=K^{\mathrm{ib}(n+1)}(\bn_{m'})+\bn_{m'}$). Let us assume for definiteness $\gamma^{(n+1)}_{\bn_m}\geq\gamma^{(n+1)}_{\bn_{m'}}$. By construction (no bad points subcubes of order $n$ in a spherical layer $K^{\mathrm{i}(n+1)}(\bn_m)\setminus K^{\mathrm{i}(n+1){\mathrm {small}}}(\bn_m)$) we  either have
$$K^{\mathrm{i}(n+1)\mathrm{small}}(\bn_{m'})\cap K^{\mathrm{i}(n+1)\mathrm{small}}(\bn_{m})\not=\emptyset,$$ 
or 
$$K^{\mathrm{i}(n+1)\mathrm{small}}(\bn_{m'})\cap (\Z^l\setminus K^{\mathrm{i}(n+1)}(\bn_{m}))\not=\emptyset.$$ 

If the former case, we simply discard the cube $K^{\mathrm{ib}(n+1)}(\bn_m')$ (or rather two cubes: $K^{\mathrm{ib}(n+1)}(\bn_m')$ and $K^{\mathrm{ib}(n+1)\mathrm{small}}(\bn_m')$) from our list of base cubes. We continue this procedure for all other cubes of level $n+1$.

2) Now, we have kept only the cubes $K^{\mathrm{ib}(n+1)}(\bn_m)$ such that when we shift them by $\bn_m$, they do not cover any shifted small initial  cube.  Now we rescale all the remaining cubes by $\frac14$, i.e. we put
\bee\label{1/4ball}
K^{\mathrm{ibr}(n+1)}(\bn_m):=\frac14 K^{\mathrm{ib}(n+1)}(\bn_m),
\ene 
where the RHS of \eqref{1/4ball} is a ball in $\Z^l$ (centred at the origin) of the radius $\frac14$ times the radius of the ball in the LHS (here we do abuse the standard notation slightly). The extra index $\mathrm{r}$ stands for `rescaled'.   
This ensures the proper $\Z$-distance between the shifted cubes $K^{\mathrm{ibr}(n+1)}(\bn_m)+\bn_m$. 

3) Now property $4$ of definition \ref{8.1} is satisfied for $j=j'=n+1$. Let us ensure that it holds also for $0\le j<j'=n+1$. To do this, we consider the situation when the  \Bourgain cube of the smaller order $K^{(j)}_{m'}$ is too close to the  cube of order $n+1$ $K^{\mathrm{ibr}(n+1)}(\bn_{m})+\bn_{m}$, and if this happens we just attach the smaller cube to the re-defined bigger cube. More precisely, we proceed as follows. Suppose, there is a cube 
$K^{(n)}_{\bn_{m'}}=K^{b(n)}_{\bn_{m'}}+\bn_{m'}$ that is located within distance 
$\frac14E^{\gamma^{(n)}_{m'}r_{n,1}}$ of $K^{\mathrm{ibr}(n+1)}(\bn_{m})+\bn_{m}$. Then we attach $K^{(n)}_{\bn_{m'}}-\bn_{m}$ together with its $\frac{1}{16}E^{\gamma^{(n)}_{m'}r_{n,1}}$-neighbourhood to $K^{\mathrm{ibr}(n+1)}(\bn_{m})$. Then we do the same for all cubes $K^{(n-1)}_{\bn_{m'}}-\bn_{m}$ that are within distance 
$\frac14E^{\gamma^{(n-1)}_{m'}r_{n-1,1}}$ of $K^{\mathrm{ibr}(n+1)}(\bn_{m})$ (where we have attached all cubes of level $n$ to $K^{\mathrm{ibr}(n+1)}(\bn_{m})$); we also attach $\frac{1}{16}E^{\gamma^{(n-1)}_{m'}r_{n-1,1}}$-neighbourhood of such cubes to $K^{\mathrm{ibr}(n+1)}(\bn_{m})$. We carry on this process until we attach the cubes of level $0$. The resulting cube $K^{\mathrm{ibr}(n+1)}(\bn_{m})$ (with attached cubes of smaller levels) is what we will finally call the base cube $K^{\mathrm{b}(n+1)}(\bn_{m})$. {\newred Now, the shifted cubes are defined by \eqref{(8.3)}, thus establishing the translation invariance. It is also not hard to see that all the properties listed in item 4 of the definition \ref{8.1} are now satisfied. Indeed, the distance between cubes of order $n+1$ is controlled by construction and the estimate of the size of clusters of the cubes of a smaller order. Then, we included close smaller cubes together with their proper neighbourhoods. Since, by induction, the distance between cubes of smaller order is controlled, it ensures the proper distance between a cube of order $n+1$ and all cubes of smaller order.}

The enlarged base cubes $\tilde K^{b(n)}_{m}$ and $\tilde K^{b(n),{\mathrm {small}}}_{m}$ corresponding to this patch are initially defined in a similar way, but using $\tilde\gamma$ from lemma \ref{MGL'-n}; then we modify them in the same way as the usual cubes to satisfy $4'$ from definition \ref{8.2}.  

\subsection{Proof of Theorem \ref{MGT}.}\label{9.2}

It remains to prove that the \Bourgain structures are reasonable, i.e. estimates \eqref{measSmn} and \eqref{meastildeSmn} hold, both for $n=1$ (base) and for $n+1$, assuming they hold for $n$ (step). The two proofs are very similar, but not identical; that is why we had to split these two cases. 
We also would like to emphasise that the proof of the structure being reasonable is the main technical difference compared to  section \ref{goodset-3} (see Lemma~\ref{L5.5}) and this is one of the places where we actually need the enlarged \Bourgain structure. 

We will concentrate on proving \eqref{measSmn}. The proof of \eqref{meastildeSmn} is similar.  

First, we consider the case $n=1$. We are going to use Cartan's Lemma, see Lemma \ref{Cartan}, so we need to define all the objects in that Lemma as well as to check that all the assumptions of that Lemma are satisfied. Suppose, $\bxi\in\CA^{\bxi(1)}(\bxi^*_{\tim})$. Let $K_{\ham}^{b(1)}$ be a base cube corresponding to $\CM_{\tim^{\bxi}}^{\bxi(1)}$ according to Definition \ref{8.1}, $\ham=\ham(\tim)$. Put
  $\Lambda :=K_{\ham}^{b(1),\mathrm{small}}$ and  $\tilde\Lambda :=K_{\ham}^{b(1)}$.
  Obviously, 
  \begin{equation}
  \label{M}|\Lambda |<M:=2^l\En^{l\gamma _{\ham}^{(1)} r_{1,1}/Z_0}.
  \end{equation}
  
Denote 
 $$A({\bf z}):=(H(\tilde\Lambda,\bxi )-\rho^2)=
 \CP(\tilde\Lambda, \bxi)(H(\bxi )-\rho^2)\CP(\tilde\Lambda, \bxi ), 
 \ \ {\bf z}=(\bxi -\bxi ^*_{\tim})\En^{2r_1'}.$$
 Obviously, $A$ is an analytic function of ${\bf z}$ in $D^d$, $D:=\{z\in \C, |z|<1\}$ and $N$ -- the size of the matrix $A$ -- is bounded above by $2^l\En^{l\gamma _{\ham}^{(1)} r_{1,1}}$. It is easy to see that
 $\| A({\bf z})\|<2^{2l}\En^{2l\gamma _{\ham}^{(1)} r_{1,1}}$. Therefore, \eqref{Cartan1} holds with $B_1:=2^{2l}\En^{2l\gamma^{(1)} _{\ham} r_{1,1}}$. 
 Let us check \eqref{Cartan2}.

Let $\bq$ be any point from $\tilde\Lambda$ and consider the point $\bxi ^*_{\tim}+\bq \vec\boldom$. By Definition \ref{8.1}, we either have $|\|\bxi ^*_{\tim}+\bq \vec\boldom\|^2-\rho^2|>E^{\sigma_0}$, or there is a cube $K_{\ham'}^{b(0)}$, $\ham'=\ham'(\tim,\bq)$ corresponding to this point. 

Suppose first that $\bq\in   \tilde\Lambda\setminus\Lambda$ . Then 
 Lemma \ref{MGL} tell us that  
  $$ (\vec \boldom , \rho , \bxi ^*_{\tim}+\bq \vec\boldom )\not \in S^{(0)}_{total}, $$
  which means that  in the latter case we have:
  \bee \label{11!}
  \left\|\left(H(K_{\ham'}^{b(0)},\bxi ^*_{\tim}+\bq \vec\boldom)-\rho^2)\right)^{-1}\right\|_2<\En^{r_1'}.
  \end{equation}
Then we use the perturbation theory arguments to move from \eqref{11!} to
 \bee \label{11}
\left\|\left(H(K_{\ham'}^{b(0)},\bxi +\bq \vec\boldom)-\rho^2)\right)^{-1}\right\|_2<2\En^{r_1'}.
 \end{equation}
The last formula can obviously be re-written as
\bee \label{11n}
\left\|\left(H(K_{\ham'}^{b(0)}+\bq,\bxi )-\rho^2)\right)^{-1}\right\|_2<2\En^{r_1'}.
 \end{equation}

  Now, applying Lemma \ref{abstractlemma} with $K^{(n+1)}:=\tilde \Lambda \setminus \Lambda$ (see also Theorem \ref{Thm2}), we obtain that \eqref{Cartan2} holds with $B_2:=4\En^{r_1'}$. 
    
  Next we check \eqref{Cartan3}. Let us introduce somewhat longer notation for $S_{\ham}^{(0)}$ in \eqref{Sn}: $S^{(0)}_{\ham}=:S_{\bxi^* _{\tim}, r_1'}$ and consider, for $\bp\in\Lambda$, 
a modification of this set: 
\bee
S_{\bxi ^*_{\tim}+\bp \vec \boldom ,3lr_1'}:=
\left \{(\vec \boldom, \rho ,\bxi ) \in \CA_{\tim}^{(0)}: \left \|
(H( K_{\ham'}^{b(0)}+\bp,\bxi)-\rho^2)^{-1}\right\|_2> \En^{3lr_{1}'}\right\}, 
\ene
i.e. we use $3lr_1'$ in \eqref{Sn} instead of $r_1':=r_{1,3}$, and we put $\ham'=\ham'(\tim,\bp)$, so that $K_{\ham'}^{b(0)}$ is a base cube corresponding to matryoshka $\CM_{\tim^{\bxi}}^{\bxi(1)}$ that covers $\bxi+\bp\vec\boldom$.  By Lemma \ref{L5.5} (with the obvious adjustment of the power),
\bee
\meas( (S_{\bxi^* _{\tim}+\bp \vec \boldom ,3lr_1'})_{\mathrm {cs}}(\vec \boldom , \rho ))<
 \En^{-3lr_1'-d-1+\sigma _{1,d}}. 
 \ene 
 Next, put
 $$S' (\vec \boldom , \rho ):=\cup _{\bp \in \Lambda } (S_{\bxi^* _{\tim}+\bp \vec \boldom ,3lr_1'})_{\mathrm {cs}}(\vec \boldom , \rho ).$$
 Obviously,
 \bee \meas( S'(\vec \boldom , \rho ))<
 \En^{-3lr_1'-d-1+\sigma _{1,d}}\En^{l r_{1,1}/Z_0}<\En^{-2lr_1'}.
 \ene 
It easily follows that
\bee \label{set} 
\BB(\bxi ^*_{\tim};\frac{1}{2}\En^{-2r_1'})\setminus S'(\vec \boldom , \rho )\neq \emptyset.
 \ene

Let us choose $\tilde \bxi _{\tim}$ to be any point from the set in the LHS of \eqref{set}.
Then for any $\bp\in\Lambda$ we have:
\bee \label{12}
\left\|\left(H(K_{\ham'}^{b(0)}+\bp,\tilde\bxi_{\tim} )-\rho^2\right)^{-1}\right\|_2<\En^{3lr_1'}.
\ene
Notice that, unlike when we were deriving \eqref{11} from \eqref{11!}, here we cannot replace $\tilde \bxi_{\tim}$ by $\bxi$ using perturbation arguments. However, we do not need to do this, since we need to establish \eqref{12} only at one point to be able to apply the Cartan's Lemma. 
Next, we consider \eqref{11} for $\bn \not \in  \Lambda $ and \eqref{12} for $\bn \in  \Lambda $. Now, Lemma \ref{abstractlemma} implies \eqref{Cartan3} with $B_3:=2\En^{3lr_1'}$ and 
 $${\bf a }:=(\bxi ^*_{\tim}-\tilde \bxi _{\tim})\En^{2r_1'}
 \in (-\frac{1}{2},\frac{1}{2})^d.$$ 
 
Finally, we put $e^t=\En^{r_2'(\gamma_{\ham}^{(1)})}$. Now \eqref{Cartan4} implies
\bee
\meas(( S_{\ham}^{(1)})_{\mathrm {cs}}(\vec \boldom, \rho ))<C\En^{\frac{-cr_2'(\gamma_{\ham}^{(1)})}{M\ln (B_1B_2B_3)}}.
\ene
 Estimate \eqref{measSmn} for $n=1$ easily follows. The proof of \eqref{meastildeSmn} for $n=1$ is similar with obvious changes due to the fact that we consider the enlarged cubes. 

 Now let us prove the inductive step, i.e. we assume that \eqref{measSmn},  \eqref{meastildeSmn} hold for some $n\geq1$ and prove them for $n+1$.  Notice that $r_{n}'(\gamma_{\ham}^{(n-1)})>4r_{n,3} \En^{\frac32 l\gamma_{\ham}^{(n-1)} r_{n-1,1}/Z_0}$. 
 As above, we use Cartan's Lemma. The definition of matrix $A({\bf z})$ and the proof of \eqref{Cartan1}, \eqref{Cartan2} is analogous to the proof in the case $n=1$.
 Indeed, suppose that $\bxi\in\CA^{\bxi(n+1)}(\bxi^*_{\tim})$. Let $K_{\ham}^{b(n+1)}$ be a base cube corresponding to $\CM_{\tim^{\bxi}}^{\bxi(n+1)}$. Put
  $\Lambda :=K_{\ham}^{b(n+1),\mathrm{small}}$ and  $\tilde\Lambda :=K_{\ham}^{b(n+1)}$.
  Obviously, 
  \begin{equation}
  \label{M-n}|\Lambda |\le M:=2^l\En^{l\gamma _{\ham}^{(n+1)} r_{n+1,1}/Z_0}.
  \end{equation}
  
Let
 $$A({\bf z}):=(H(\tilde\Lambda,\bxi )-\rho^2)=\CP(\tilde\Lambda, \bxi)(H(\bxi )-\rho^2)\CP(\tilde\Lambda, \bxi ), \ \ {\bf z}=(\bxi -\bxi ^*_{\tim})\En^{2r_{n+1}'}.$$
 Obviously, it is an analytic function of ${\bf z}$ in $D^d$, $D:=\{z\in \C, |z|<1\}$, and the size of the matrix $N$ is $2^l\En^{l\gamma _{\ham}^{(n+1)} r_{n+1,1}}$. It is easy to see that
 $\| A({\bf z})\|<2^{2l}\En^{2l\gamma _{\ham}^{(n+1)} r_{n+1,1}}$. Therefore, \eqref{Cartan1} holds with $B_1:=2^{2l}\En^{2l\gamma _{\ham}^{(n+1)} r_{n+1,1}}$. 
 
Let $\bq$ be any point from $\tilde\Lambda$ and consider the point $\bxi ^*_{\tim}+\bq \vec\boldom$. By Definition \ref{8.1}, we either have $|\|\bxi ^*_{\tim}+\bq \vec\boldom\|^2-\rho^2|>E^{\sigma_0}$, or there is a \Bourgain cube $K_{\ham'}^{b(j)}$, $0\le j\le n$, $\ham'=\ham'(\tim,\bq,j)$ corresponding to this point. Recall (see Definition \ref{corresponding}) that if $j<n$, then $K_{\ham'}^{b(j)}$ is good. On the other hand, if  $\bq\in\tilde\Lambda\setminus\Lambda$, then Lemma \ref{MGL-n} implies that $K_{\ham'}^{b(n)}$ is good.
 Overall,  for all $j\le n$ and
  for $\bq\in\tilde\Lambda\setminus\Lambda$ we have that
  $$ (\vec \boldom , \rho , \bxi ^*_{\tim}+\bq \vec\boldom )\not \in S^{(j)}_{total}.$$
  Using again the perturbation arguments as in \eqref{11!}, \eqref{11}, we obtain:
  \bee \label{11-n}
  \left\|\left(H(K_{\ham'}^{b(j)}+\bq,\bxi )-\rho^2\right)^{-1}\right\|_2<\En^{r_{j+1}'(\gamma_{\ham'}^{(j)})}.  
  \end{equation}
  
 Now, applying Lemma \ref{abstractlemma} for $H(\tilde\Lambda\setminus\Lambda,\bxi)$ (notice that $r_{j+1}'(\gamma_{\ham'}^{(j)})\ll\En^{\gamma_{\ham'}^{(j)}r_{j,1}}$) we see that \eqref{Cartan2} holds with $B_2:=2\En^{r_{n+1}'}$. 
 
Next, we need to check \eqref{Cartan3}. Consider the enlarged \Bourgain cubes $\tilde K_{\ham'}^{b(n)}+\bq$ located inside (or at least having a non-empty intersection with) $K_{\ham}^{b(n+1),{\mathrm {small}}}$. At the same time, we consider the collection of sets $\tilde S_{\ham'}^{(n)}$ (see \eqref{tildeSmn}).
Using \eqref{meastildeSmn} for $n$ (our induction assumption), we have
\bee
\meas( (\tilde S_{\ham'}^{(n)})_{\mathrm{cs}}(\vec \boldom , \rho ))< \En^{-\En^{\frac12 l\tilde\gamma_{\ham'}^{(n)} r_{n,2}/Z_0}}<\En^{-\En^{4l  r_{n,1}/Z_0}}.
 \ene 
 Next, put
 $$S' (\vec \boldom , \rho ):=\cup(\tilde S_{\ham'}^{(n)})_{\mathrm{cs}}(\vec \boldom , \rho ),$$
where the union is over all extended cubes $\tilde K_{\ham'}^{b(n)}+\bq$ inside (or at least intersecting) $K_{\ham}^{b(n+1),{\mathrm {small}}}$.  Then we have:
 \bee 
 \meas( S'(\vec \boldom , \rho ))<
{\En^{l\gamma _{\ham}^{(n+1)} r_{n+1,1}}}\En^{-\En^{4l  r_{n,1}/Z_0}}<\En^{-\frac12\En^{4l  r_{n,1}/Z_0}}.\ene 
It easily follows that
\bee \label{set-n'} \BB(\bxi ^*_{\tim};\frac{1}{2}\En^{-2r_{n+1}'})\setminus S'(\vec \boldom , \rho )\neq \emptyset. \ene
We now choose any point $\tilde \bxi _{\tim}$ from the set in the LHS. 
Then we have:
\bee \label{12-n'}
\left\|\left(H(\tilde K_{\ham'}^{b(n)}+\bp,\tilde\bxi_{\tim} )-\rho^2\right)^{-1}\right\|_2<\En^{\tilde r_{n+1}'(\tilde\gamma_{\ham'}^{(n)})}
\ene
for all $\bp\in\Lambda$. For $\bp\in\tilde\Lambda\setminus\Lambda$ we still use \eqref{11-n}.

Recall that by definition \ref{8.2} all cubes $\tilde K_{\ham'}^{b(n)}+\bq$ and $K_{\ham''}^{b(j)}+\bq'$ are well-separated from each other. Now Lemma \ref{abstractlemma} implies \eqref{Cartan3} with $B_3:=\En^{\tilde r_{n+1}'}$. Finally, we choose $t$ so that $e^t=\En^{r_{n+2}'(\gamma_{\ham}^{(n+1)})}$ and 
$${\bf a }:=(\bxi ^*_{\tim}-\tilde \bxi _{\tim})\En^{2r_{n+1}'}
 \in (-\frac{1}{2},\frac{1}{2})^d.$$ 

By \eqref{Cartan4},
 $$\meas( S_{\ham}^{(n+1)}(\vec \boldom, \rho ))<C\En^{\frac{-cr_{n+2}'(\gamma_{\ham}^{(n+1)})}{M\ln( B_1B_2B_3)}}.$$ 
 Estimate \eqref{measSmn} (for $n+1$) easily follows. The proof of \eqref{meastildeSmn} is similar. 
 
Finally, we recall that $E_q=E_*+q$, $q=0,1,...$ and define the set ${\bf \mathcal G }^{{\vec\boldom }(n+1)}={\bf \mathcal G }^{{\vec\boldom }(n+1)}(\tilde E)$ by 
\bee\label{CGn}
{\bf \mathcal G }^{{\vec\boldom }(n+1)}(\En _q)
:=\CG^{\vec\boldom(n)}(\En _q)\cap\left(\cap _{k=q}^{\infty}{G^{(n+1)} }_{} (\En _k)\right);
\ene
for $\tilde E\in[E_q,E_{q+1})$ we put 
$${\bf \mathcal G }^{{\vec\boldom }(n+1)}(\tilde E):={\bf \mathcal G }^{{\vec\boldom }(n+1)}(E_q)=\CG^{\vec\boldom(n)}(\En _q)\cap\left(\cap _{k=q}^{\infty}{G^{(n+1)} }_{} (\En _k)\right).$$ 
The set ${\bf \mathcal G }^{{\vec\boldom }(n+1)}(\tilde E)$ consists of the frequencies for which there is a \Bourgain structure (usual and enlarged) of order $n+1$ for all $\rho\ge\tilde E$.

\bec\label{MGC}
We have:
\bee \meas({\mathcal G }^{\vec\boldom (n+1)}(\tilde E))=\meas({\mathcal G }^{\vec\boldom (0)})-O(\tilde E ^{- C_1r_{n+1,1}(\tilde E)}),\ {\tilde E\to \infty }.
\label{mesbfLambda}
\ene
\enc
Theorem~\ref{MGT} is proven.

\section{Induction. Proof of Theorem \ref{Thm7.4}}\label{section9}

The scheme of the proof is somewhat reminiscent of the scheme of the proof from the previous section.  

Part I. Here, we will use the properties of good matryoshkas (the Statement) and IIE (or rather the half of it) \eqref{Cartan5-n}  at level $n$ to obtain that the Statement at the level $n+1$ holds for most patches of level $n+1$.  The proof is relatively simple and is a direct application of the abstract resolvent  lemma \ref{abstractlemma} from Appendix 1.

Part II. Here, we first prove (a) the IIE on perfect patches of level one. Next,  (b) we assume that Statement and IIE hold at level $n$ and prove the IIE  at level $n+1$ for the patches described in Part I. The proofs of parts (a) (which again could be seen as the base of induction) and (b) (which is an inductive statement) are similar, but not identical, so we describe both of them in detail. 
The proofs of \eqref{Cartan5-n} and \eqref{Cartan5-ntilde} are similar to each other, so we will only prove one of them.

{\it Part I.} Here we show how to obtain the next ($n+1$)-st level of matryoshka of central cubes consistent and synchronised with the \Bourgain structure using the estimates \eqref{Cartan5-n} and the properties of the central cube of level $n$ listed in the Statement. To begin with, we fix a patch at level $n$, $\CA^{\BPhi(n)}_{\tim}$, $\BPhi=\BPsi_{\tim}(\vec\phi)$, and assume that $\bk$ is associated with it. We assume that this patch is good, so in particular 
the Statement and the IIE at the level $n$ hold. We also assume that the \Bourgain structure is stable in $\bk$ associated with $\CA^{\BPhi(n)}_{\tim}$ (Theorem \ref{MGT} states that this is possible).
We want to prove that for most of the patches at the next level inside $\CA^{\BPhi(n)}_{\tim}$, the Statement at the level $n+1$ holds.  
 
Consider first the prodigal sons $K^{(n)}_m$ (see Definition~\ref{def7.5}) and notice that 
$$\En^{-r_{n,3}\En^{4lr_{n-1,1}/Z_0}}=\En^{-r_{n,3}(r_n')^2}.$$ 
In particular, all good \Bourgain cubes of level $n-1$ are stable under such perturbation (meaning that after this perturbation the inequalities defining the good cubes will still be valid, possibly with an extra factor of $2$). Next, we modify all prodigal sons. Suppose, $\bq\in\Z^l$ is such that $||\bq\vec\boldom||<\En^{-r_{n,3}\En^{4lr_{n-1,1}/Z_0}}$. We then construct the cube $\hat K_\bq^{(n)}$ around it which is just the shifted central cube: $\hat K_\bq^{(n)}:=\hat K^{(n)}+\bq$. We also naturally denote $\hat K_{\bf 0}^{(n)}:=\hat K^{(n)}$. The  Diophantine condition implies that the distance between different prodigal sons $\hat K_\bq^{(n)}$ and $\hat K_{\bq'}^{(n)}$ is at least $\En^{r_{n,2}}$ (the size of the central cube). We also note that all the points $\bk+\bq\vec\boldom$ are associated with the same quasi-patch $\CA^{\BPhi(n)}_{\tim}$ (strictly speaking, we have to increase the size of this patch by a factor $2$). Therefore, as explained above, we can choose base \Bourgain cubes corresponding to the shifts of all such points $\bk+\bq\vec\boldom$ by any vector $\bn\in\Z^l$ to be the same. This means that 
the properties from definition \ref{def7.4} hold with respect to each $\hat K_\bq^{(n)}$, not only for $\hat K^{(n)}$. This means, in particular, that each cube $\hat K_\bq^{(n)}$ is well-separated from any \Bourgain cube that does not lie inside it. 
Finally, using the Statement for level $n$, we conclude that our operator restricted to each prodigal son, $H(\hat K_\bq^{(n)},\bka)$, has a single eigenvalue in the interval $I_{n}$. This eigenvalue is given by $\lambda^{(n)}(\bka+\bq\vec\boldom)$ and the properties listed in the Statement hold. In other words, the operator $H(\hat K_\bq^{(n)},\bka)$ is very similar to $H(\hat K^{(n)},\bka)$.

Given $\bq$, we denote $\bq\vec\boldom=:\|\bq\vec\boldom\|(x_1,\dots,x_d)$, so that $\|(x_1,\dots,x_d)\|=1$. Recall that we choose coordinates so that 
$$\BPhi=\left(\phi_1,\dots,\phi_{d-1},\sqrt{1-\sum_{j=1}^{d-1}\phi_j^2}\right).$$ 
Assume, as we can  without loss of generality, that $|x_1|\geq|x_j|$ for all $j$, $2\leq j\leq d-1$. Consider the operator $H(\hat K_\bq^{(n)},\bka^{(n)})$ as an analytic function of $\phi_1$, assuming that $\hat\phi:=(\phi_2,...,\phi_{d-1})$ is fixed and real. The next Lemma is a simple consequence of the properties of the function $\lambda^{(n)}$:

\bel\label{lemmasimple-n}
Let $\bq\in\Z^l$, $\bq\ne 0$ be such that $||\bq\vec\boldom||<\En^{-r_{n,3}\En^{4lr_{n-1,1}/Z_0}}$. Let us choose any $\vec\phi \in\Pi^{(n)}_{\tim}$. We fix $\hat\phi:=(\phi_2,...,\phi_{d-1})$ and start varying $\phi_1\in\C$ so that $\vec\phi \in\Pi^{(n)}_{\tim,\C}$. Then the resolvent $(H(\hat K_\bq^{(n)},\bka^{(n)})-\rho^2)^{-1}$  has at most one pole in $\phi_1$ for every choice of other variables and we have
\bee\label{estsimple-n}
\|(H(\hat K_\bq^{(n)},\bka^{(n)})-\rho^2)^{-1}\|\leq 2\sqrt{d}\En^{-1}\|\bq\vec\boldom\|^{-1}\varepsilon^{-1},
\ene
whenever $\phi_1$ lies  outside the $\varepsilon$-neighbourhood of the pole.
\enl
\bep
Using the Statement at all levels $1,...,n$, we obtain: 
\bee\label{10.2}
\begin{split}&
\lambda^{(n)}(\bka^{(n)}(\vec\phi)+\bq\vec\boldom)-\rho^2=\lambda^{(n)}(\bka^{(n)}(\vec\phi)+\bq\vec\boldom)-\lambda^{(n)}(\bka^{(n)}(\vec\phi))\cr &
=(\bka^{(n)}(\vec\phi)+\bq\vec\boldom)^2-(\bka^{(n)}(\vec\phi))^2+f_1(\bka^{(n)}(\vec\phi)+\bq\vec\boldom)-f_1(\bka^{(n)}(\vec\phi))\cr &
=2\kappa^{(n)}(\vec\phi)\|\bq\vec\boldom\|\left(x_1 \phi_1+x_d(1-\phi_1^2-\sum_{j=2}^{d-1}\phi_j^2)^{1/2}+\sum_{j=2}^{d-1}x_j\phi_j\right)+\|\bq\vec\boldom\|f_2.
\end{split}
\ene
Here, $f_1$ is the sum of the corrections for all $j\le n$ in the RHS of \eqref{eigenvalue-n}:
$$
f_1(\bka):=\sum\limits_{j=0}^{n}\sum\limits_{q=2}^\infty g^{(j)}_q({\bka}).
$$
Properties formulated in the Statement imply that  $f_1$ is a holomorphic function of $\phi_1$, 
and $f_2$ and its first and second derivatives are $O(\En^{-3/2})$. Notice that $x_d=o(1)$ as $E\to\infty$; otherwise, there are no zeros in the LHS of \eqref{10.2} at all. Now, it follows from the estimates for the first derivatives of $\ka^{(n)}$  that the first derivative of $\lambda^{(n)}(\bka^{(n)}(\vec\phi)+\bq\vec\boldom)$ with respect to $\phi_1$ has modulus bounded below by $E\|\bq\vec\boldom\|/\sqrt{d}$ when $\vec\phi \in\Pi^{(n)}_{\tim,\C}$. This completes the proof of the lemma by standard analytic arguments. 
\enp

Now, we consider the ball $\Omega(\En^{r_{n+1,2}})$ and all the prodigal sons $\hat K_\bq^{(n)}$ that are inside this ball. Then Diophantine estimate tells us that $\|\bq\vec\boldom\|\geq \En^{-\mu r_{n+1,2}}$. Therefore, if we want estimate \eqref{estsimple-n} with $\varepsilon:=\En^{-r_{n+1,3}}$ to hold on all prodigal sons, lemma \ref{lemmasimple-n} tells that we can achieve this by throwing away a subset of $\CA^{\BPhi(n)}_{\tim}$ of total measure $\En^{-r_{n+1,3}}\En^{lr_{n+1,2}}$ (the second term estimating the number of \Bourgain cubes of order $n$ inside $\Omega(\En^{r_{n+1,2}})$). Next, using IIE at level $n$ we can throw away another subset of $\CA^{\BPhi(n)}_{\tim}$ of total measure $\En^{-r_{n+1,3}}\En^{lr_{n+1,2}}$ to ensure that the estimate \eqref{Cartan5-n} holds for {\it any} globetrotter of level $n$ that is inside $\Omega(\En^{r_{n+1,2}})$. Moreover, since both estimates \eqref{Cartan5-n} and \eqref{estsimple-n} (with $\varepsilon:=\En^{-r_{n+1,3}}$ and  $\|\bq\vec\boldom\|\geq \En^{-\mu r_{n+1,2}}$) are stable on a patch $\CA^{\BPhi(n+1)}_{\tim'}$ (as usual, possibly with an extra factor 2), one can easily see that the set we have just thrown away has a non-empty intersection with at most $M_{n+1}\En^{-r_{n+1,3}/2}$ patches at the next level ($n+1$) (those are the patches that we will  declare non-perfect); recall that $M_{n+1}$ is the total number of patches at the next level. 

If we have a \Bourgain cube $K^{(j)}_p$ of level $j<n$, then such cube is either good (and so estimate opposite to \eqref{bad7} holds), or bad, and then $K^{(j)}_p$ is covered by either a prodigal son, or a globetrotter of level $j+1$. 

Finally, we construct $\hat K^{(n+1)}$ as follows. We consider $\Omega(\frac14\En^{r_{n+1,2}})$ and modify it as prescribed in item $3$ of the modification process described in Section 9. The small difference is that first we  add all {\em enlarged} \Bourgain cubes of level $n+1$ that are near the boundary of $\Omega(\frac14\En^{r_{n+1,2}})$, and then continue to include all {\em usual} \Bourgain cubes of all levels $n,n-1,\dots,0$ near the boundary of the already updated cube. This way we ensure the consistency (properties from Definition~\ref{def7.4}) on level $n+1$. Now, we repeat the construction from the proof of Theorem~\ref{Thm2} (see also Lemma~\ref{abstractlemma}) to obtain the Statement for $n+1$. Thus, the patches that have not been thrown away are indeed perfect.

\medskip

{\it Part II.}

(a) First, we need to establish the base of induction - IIE at level 1.

We are fixing a globetrotter $K^{(1)}_m$. Let us look at all clusters $K^{(0)}_{m'}:=\tilde\CC_1(\bn_{m'})$ of level $n=0$, with $p=1$, that are inside $K^{(1)}_m$. Suppose first that the rank $s$ of $\tilde\CC_1(\bn_{m'})$ is positive. Then we can use Lemma \ref{resclustershrink} with $\varepsilon=\En^{-r_{1,3}\En^{3d^2(l+\mu)\sigma_{0,d-1}}}$ to show that we can throw away a bad set $S'_{m'}\subset \Pi_{\tim}^{(1)}$ with
\bee\label{p=0'}
\meas(S'_{m'})\le E^{1-r_{1,3}E^{3d^2(l+\mu)\sigma_{0,d-1}}}
\ene
so that for $\vec\phi\in \Pi_{\tim}^{(1)}\setminus S'_{m'}$ we have: 
\begin{equation}\label{respatchesshrinknew1}
\|( (H(K^{(0)}_{m'},\bka^{(0)}(\vec\phi))-\rho^2))^{-1}\|\le 
\En^{r_{1,3}\En^{2d^2(l+\mu)\sigma_{1,s}}}
\end{equation}
(here we also use the estimate $10\sigma_{0,d-1}<\sigma_{1,s}$).

\ber\label{p=0}
This is the place where we use the fact that our constructions in Sections \ref{section5.3} --\ref{section5} (including definition \eqref{good1}) were defined using $p=0$: this allows us to guarantee that the difference $\Pi_{\tim}^{(1)}\setminus S'_{m'}$ is non-empty. 
\enr

Suppose that the rank  of $\tilde\CC_1(\bn_{m'})$ is $s=0$. Then we use Lemma \ref{s=0}  with $\varepsilon=\En^{-r_{1,3}\En^{3d^2(l+\mu)\sigma_{0,d-1}}}$ together with the (globetrotter) estimate $\|\bn_{m'}\vec\boldom\|\geq\En^{-r_{1,3}\En^{3d^2(l+\mu)\sigma_{0,d-1}}}$ to show that we can throw away a bad set $S'_{m'}\subset \Pi_{\tim}^{(1)}$ with 
$$
\meas(S'_{m'})\le 2E^{-r_{1,3}E^{3d^2(l+\mu)\sigma_{0,d-1}}}
$$
so that for $\vec\phi\in \Pi_{\tim}^{(1)}\setminus S'_{m'}$ we have 
\begin{equation}\label{respatchesshrinknew}
\|( (H(K^{(0)}_{m'},\bka^{(0)}(\vec\phi))-\rho^2))^{-1}\|\le 
\En^{3r_{1,3}\En^{3d^2(l+\mu)\sigma_{0,d-1}}}\leq\En^{r_{1,3}\En^{2d^2(l+\mu)\sigma_{1,0}}}.
\end{equation}

Using \eqref{dk0-2} and the usual perturbation arguments, we deduce that \eqref{respatchesshrinknew1} and \eqref{respatchesshrinknew} hold with $\bka^{(1)}$ instead of $\bka^{(0)}$ in the LHS and an extra factor $2$ in the RHS. These estimates hold simultaneously assuming $\vec\phi\not\in S'$, where $S':=\cup_{m'}S'_{m'}$. Taking into account the trivial bound on the number of clusters $K^{(0)}_{m'}$ inside $K^{(1)}_m$, we obtain
$$
\meas(S')\le \En^{-r_{1,3}\En^{3d^2(l+\mu)\sigma_{0,d-1}}}E^{lr_{1,1}+1},
$$
which implies that the set $\Pi^{(1)}\setminus S'$ is not empty. Applying  lemma~\ref{abstractlemma}, we see that for any  
$\vec\phi\in\Pi^{(1)}\setminus S'$ we have (here it is important that we deal with $p=1$!)
\bee\label{nonsimple}
\|( (H(K^{(1)}_{m},\bka^{(1)}(\vec\phi))-\rho^2))^{-1}\|\leq2\En^{r_{1,3}\En^{2d^2(l+\mu)\sigma_{1,d-1}}}.
\ene

We plan to apply Cartan's Lemma \ref{Cartan} once again (with \eqref{Cartan3} provided by the estimate we have just obtained). We denote 
  $\Lambda :=K_{m}^{(1),\mathrm{small}}$ and  $\tilde\Lambda :=K_{m}^{(1)}$. 
  Obviously, 
  \begin{equation}
  \label{M-n'}|\Lambda |\le M:=\En^{l\gamma_m^{(1)} r_{1,1}/Z_0}
  \end{equation}
and
 \begin{equation}
  |\tilde \Lambda |\le N:=\En^{l\gamma_m^{(1)} r_{1,1}}
  \end{equation}

Let us denote, as before, 
 $$A({\bf z}):=(H(\tilde\Lambda,\bka^{(1)}(\vec\phi))-\rho^2)=\CP(\tilde\Lambda, \bka^{(1)}(\vec\phi))(H(\bka^{(1)}(\vec\phi))-\rho^2)\CP(\tilde\Lambda, \bka^{(1)}(\vec\phi) ),$$
 where ${\bf z}:=\vec\phi \En^{r_{1,3}\En^{2d^2(l+\mu)\sigma_{0,d-1}}}.$
 Obviously, it is an analytic function of ${\bf z}$ in $D^{d-1}$, $D:=\{z\in \C, |z|<1\}$. It is easy to see that
 $\| A({\bf z})\|<2^{2l}\En^{2l r_{1,1}}$. Therefore, \eqref{Cartan1} holds with $B_1:=2^{2l}\En^{2l r_{1,1}}$. 
 
Applying Lemma \ref{abstractlemma} for $H(\tilde\Lambda\setminus\Lambda,\bka^{(1)}(\vec\phi))$ we prove the estimate \eqref{Cartan2} with $B_2=
E^{r'_1}=E^{r_{1,3}}$. Finally, \eqref{nonsimple} gives the estimate \eqref{Cartan3} with $B_3=2\En^{r_{1,3}\En^{2d^2(l+\mu)\sigma_{1,d-1}}}$. Now, we can apply Lemma~\ref{Cartan}. We define $t$ by requiring that the right hand side of \eqref{Cartan4} is equal to $\En^{-r_{2,3}}$. Then $e^t\leq E^{r_{2,3}E^{2l\gamma_m^{(1)}r_{1,1}/Z_0}}$ This proves \eqref{Cartan5-n} for $n=1$. The proof of \eqref{Cartan5-ntilde} is, as we have stated, similar.

(b) Here, we assume that the Statement and IIE at level $n$ hold on a fixed patch $\CA^{\BPhi(n)}_{\tim}$, $\BPhi=\BPsi_{\tim}(\vec\phi)$, $\vec\phi\in\Pi_{\tim}^{(n)}$, and prove that IIE at level $n+1$ holds on any patch $\CA^{\BPhi(n+1)}_{\tim}$ that were declared perfect during Part I, so that patch is also excellent.
 By Part I we can also assume that the Statement holds at the level $n+1$. {\newred Therefore,  formulas \eqref{eigenvalue-n}, \eqref{estgn} and Lemma \ref{ldk-n} hold for the step $n+1$. In particular, this means that there is  $\bka^{(n+1)}(\vec \phi)$ such that 
 $\k^{(n+1)}(\vec\phi)-\k^{(n)}(\vec\phi
)=
 O\left(\En^{-\En^{r_{n,2}}\En^{-r_{n-1,2}}}\right).$}


As above, we use Cartan's Lemma. The definition of matrix $A({\bf z})$ and the proof of \eqref{Cartan1}, \eqref{Cartan2} is analogous to the proof above. So, let $K_{m}^{(n+1)}$ be a globetrotter. 
 
 Put
  $\Lambda :=K_{m}^{(n+1),\mathrm{small}}$ and  $\tilde\Lambda :=K_{m}^{(n+1)}$.
  Obviously, 
  \begin{equation}
  |\Lambda |\le M:=\En^{l\gamma_m^{(n+1)} r_{n+1,1}/Z_0}
  \end{equation}
and
 \begin{equation}
  |\tilde\Lambda |\le N:=\En^{l\gamma_m^{(n+1)} r_{n+1,1}}.
  \end{equation}

Let
 $$A({\bf z}):=(H(\tilde\Lambda,\bka^{(n+1)}(\vec\phi))-\rho^2), \ \ {\bf z}:=\vec\phi E^{r_{n+1,3}E^{3lr_{n,1}/Z_0}}.$$
 Obviously, it is an analytic function of ${\bf z}$ in $D^{d-1}$, $D:=\{z\in \C, |z|<1\}$. It is easy to see that
 $\| A({\bf z})\|<2^{2l}\En^{2l\gamma _{m}^{(n+1)} r_{n+1,1}}$. Repeating the same arguments as above we see that  \eqref{Cartan1} holds with $B_1:=2^{2l}\En^{2l r_{n+1,1}}$ and \eqref{Cartan2} holds with $B_2:=
\En^{r_{n+1}'}=\En^{\En^{2l r_{n,1}/Z_0}}$. 

It remains to prove \eqref{Cartan3}. To do this, we modify the \Bourgain cubes in $\Lambda$. Similar to Part I, we consider new prodigal sons $\hat K_\bq^{(n)}$ instead of old  prodigal sons $K_{m'}^{(n)}$ inside $\Lambda$. The size of such cubes is $\En^{r_{n,2}}$. For those cubes we use Lemma~\ref{lemmasimple-n}, where we put $\varepsilon=\En^{-r_{n+1,3}\En^{4l r_{n,1}/Z_0}}$ and use the  estimate $\|\bq\vec\boldom\|\geq\En^{-r_{n+1,3}\En^{4l r_{n,1}/Z_0}}$, coming from the fact that $K_{m}^{(n+1)}$ is a globetrotter. As a result, we have that the inequality 
\bee\label{estsimple-n''}
\|(H(\hat K_\bq^{(n)},\bka^{(n)}(\vec\phi))-\rho^2)^{-1}\|\leq \En^{2r_{n+1,3}\En^{4l r_{n,1}/Z_0}}
\ene
holds for every $\vec\phi\in\Pi^{(n+1)}_{\tim}\setminus S_{\bq}^{(n)}$ with $$\meas(S_{\bq}^{(n)})\leq\En^{-r_{n+1,3}\En^{4l r_{n,1}/Z_0}}.$$ 
Then, using simple perturbation we can replace $\bka^{(n)}$ by $\bka^{(n+1)}$ (see \eqref{dk0-n}) to obtain
\bee\label{estsimple-n'}
\|(H(\hat K_\bq^{(n)},\bka^{(n+1)}(\vec\phi))-\rho^2)^{-1}\|\leq 2\En^{2r_{n+1,3}\En^{4l r_{n,1}/Z_0}}.
\ene

Next, instead of the globetrotters $K_{m'}^{(n)}$ inside $\Lambda$ we consider enlarged cubes $\tilde K_{m''}^{(n)}$. Now we use \eqref{Cartan5-ntilde} to show that the inequality 
\begin{equation}\label{Cartan5-n'}
\|((H(\tilde K^{(n)}_{m''},\bka^{(n+1)}(\vec\phi))-\rho^2))^{-1}\|\leq 2\En^{\left(r_{n+1,3}\En^{4l\tilde \gamma^{(n)}_{m''} r_{n,2}/Z_0}\right)\En^{2l\tilde \gamma^{(n)}_{m''} r_{n,2}/Z_0}}
\end{equation}
holds whenever $\vec\phi\in\Pi^{(n+1)}_{\tim}\setminus \tilde S_{m''}^{(n)}$ and 
\bee 
\meas(\tilde S^{(n)}_{m''})\leq\En^{-r_{n+1,3}\En^{4l\tilde \gamma^{(n)}_{m'} r_{n,2}/Z_0}}.
\ene
 Here, as before we also used a simple perturbation theory and \eqref{dk0-n} to replace $\bka^{(n)}(\vec\phi)$ with $\bka^{(n+1)}(\vec\phi)$.
 
 Next, we put together all bad sets thrown away in the last two paragraphs and define $\tilde S':=\cup\tilde S^{(n)}_{m''}\cup  S_{\bq}^{(n)}$, where the first union is taken over all globetrotters cubes $\tilde K_{m''}^{(n)}$ and the second union is over all prodigal sons $\hat K_\bq^{(n)}$ inside $\Lambda$. Obviously,
 \bee 
 \meas(\tilde S')<\En^{-\frac12r_{n+1,3}\En^{4l r_{n,1}/Z_0}}.
 \ene 
Thus, the set $\Pi^{(n+1)}_{\tim}\setminus \tilde S'$ is not empty.
Now using again Lemma~\ref{abstractlemma}, we obtain \eqref{Cartan3} with $B_3:=4\En^{r_{n+1,3}\En^{6l r_{n,2}/Z_0}}$.

Finally, we can apply Lemma~\ref{Cartan}. We define $t$ to be such that the right hand side of \eqref{Cartan4} is equal to $\En^{-r_{n+2,3}}$. Then $e^t\leq\En^{r_{n+2,3}\En^{2l\gamma_m^{(n+1)} r_{n+1,1}/Z_0}}$. This completes the proof of \eqref{Cartan5-n} for $n+1$; the proof of \eqref{Cartan5-ntilde} is similar.

\section{Final touches to the proof of our main results.}\label{8.1'}


Theorem~\ref{Thm7.4} implies that the Statement from Definition~\ref{def7.3} as well as Lemmas~\ref{corthmn}, \ref{L:derivatives-n}, \ref{ldk-n} are valid at all scales $n$. The rest of the proof is a straightforward (some readers may even call it standard), though rather technical argument very similar to the construction in Sections 8 and 9 from \cite{KS}. We briefly explain the argument here, while referring the reader interested in the details to \cite{KS}; after each statement here, we will refer to an analogous statement that has been proved in \cite{KS}. Some of the statements formulated here are not exactly required for the proof of our Main Theorem, but they may be used in our further work, so we state them here for convenience.

\subsection{Limit Set of Good Frequencies.}

Recall that the sets of good frequencies at step $j$ were chosen depending on the parameter $B_0$ and they satisfy the following properties (we now emphasise that all our constructions depend on $B_0$): 

\bee\label{CG0'}
\CG^{\vec\boldom(0)}=\CG^{\vec\boldom(0)}_{B_0}:=\BOm_0(B_0),
\ene
\bee
\meas(\CG^{\vec\boldom(0)})>1-CB_0^{1/d}.
\ene

The set $\CG^{\vec\boldom(0)}_{B_0}$ consists of the frequencies for which we can perform the zeroth step of our procedure for all energies $\tilde E$, assuming $\tilde E\geq E_*$ with $E_*$ defined in \eqref{E_*}.

We also have defined sets ${\mathcal G }^{\vec\boldom (n)}_{B_0}(\tilde E)$, $n\ge 1$; they consist of  frequencies for which the enlarged \Bourgain structure exists for all $\rho\geq\tilde E\,(\geq E_*)$, where $E_*$ also satisfies \eqref{E_*1}. We have proved that 
\bee \meas({\mathcal G }^{\vec\boldom (n)}_{B_0}(\tilde E))=\meas({\mathcal G }^{\vec\boldom (n-1)}_{B_0}(\tilde E))-O(\tilde E ^{- C_1r_{n,1}(\tilde E)}),\ {\tilde E\to \infty };
\label{mesbfLambda'}
\ene

Let us now define  
$${\mathcal G }^{\vec\boldom (\infty)}_{B_0}(\tilde E):=\cap_{n=0}^{\infty}{\mathcal G }^{\vec\boldom (n)}_{B_0}(\tilde E).$$
This set consists of  frequencies for which there exists \Bourgain structures of all levels for all $\rho\ge \tilde E$. Then we obviously have:
\bee 
\meas({\mathcal G }^{\vec\boldom (\infty)}_{B_0}(\tilde E))=\meas({\mathcal G }^{\vec\boldom (0)}_{B_0})-O(\tilde E ^{- C_1r_{1,1}(\tilde E)}),\ {\tilde E\to \infty }.
\label{mesbfLambda''}
\ene
Next, we define
$${\mathcal G }^{\vec\boldom }_{B_0}:=\cup_{\tilde E=E_*}^{\infty}{\mathcal G }^{\vec\boldom (\infty)}_{B_0}(\tilde E).$$
This set consists of  frequencies for which there exists \Bourgain structure of all levels for all sufficiently large $\rho$.
Then we have
\bee 
\meas({\mathcal G }^{\vec\boldom }_{B_0})=\meas({\mathcal G }^{\vec\boldom (0)}_{B_0})>1-CB_0^{1/d}. 
\label{mesbfLambda'''}
\ene
Finally, we put
\bee
\BOm_*={\mathcal G }^{\vec\boldom }_{0}:=\cup_{B_0>0}{\mathcal G }^{\vec\boldom }_{B_0}.
\ene
Then this set is of full measure. Suppose, $\vec\boldom\in\BOm_*$. Then $\vec\boldom\in {\mathcal G }^{\vec\boldom (\infty)}_{B_0}(\tilde E)$ for some $\tilde E\geq E_*$ and some $B_0$. We define $\rho_*$ to be this value of $\tilde E$ and put $\lambda_*:=\rho_*^2$. Let us prove that then the absolutely continuous spectrum of $H$ contains $[\lambda_*,+\infty)$.

\subsection{Limit set of Good Angles. } 
At every step $n$, we have constructed  a set
$\CG^{\BPhi(n)}(\rho)\subset \S$ (defined at the end of Section \ref{section7}; strictly speaking, these sets depend not just on $\rho>\rho_*$, but also on the choice of $\vec\boldom\in \BOm_*$; we omit mentioning the latter dependence in our notation). Next, we introduce the limiting set 
\begin{equation}
\CG^{\BPhi(\infty)}(\rho)=\cap_n\CG^{\BPhi(n)}(\rho) \subset \S. \label{good-angles} 
\end{equation}
 Estimates \eqref{CGPhi} imply that $\CG^{\BPhi(\infty)}(\rho)$ is non-empty and, moreover, 
\begin{equation}
\meas(\CG^{\BPhi(\infty)}(\rho))> \meas(\S)(1-\rho^{-\sigma_0}).
\end{equation}

\subsection{Construction of the Limit Isoenergetic Set} All steps of the inductive procedure hold on $\CG^{\BPhi(\infty)}(\rho )$.
At step $n$ we have constructed a function $\ka^{(n)}(\BPhi,\rho)$, $\BPhi \in \CG^{\BPhi(n)}(\rho)$
with the following properties.  For any $\nbka^{(n)} (\BPhi, \rho)=\ka^{(n)}(\BPhi ,\rho )\BPhi $ there is a
single eigenvalue $\lambda^{(n)}(\nbka^{(n)})$ of
 $H^{(n)}(\nbka^{(n)})$
  given by the perturbation series in Theorems \ref{Thm1},   \ref{Thm2} and the inductive statement \eqref{eigenvalue-n}. This eigenvalue is equal to $\rho ^2$.
  \ber
  Strictly speaking, we have defined the function $\ka^{(n)}(\vec \phi,\rho)$ for $\vec \phi\in \CG^{\vec \phi(n)}_m(\rho)$. We will, slightly abusing notation, write 
  $
  \ka^{(n)}(\BPhi,\rho):=\ka^{(n)}(\vec \phi,\rho)
  $
  if $\BPhi=\BPsi_m(\vec \phi)$ in this case. If a point $\BPhi$ belongs to several patches, we chose one of them (for example, the one that minimises the distance from $\BPhi$ to the centre of the patch) to define the function $\ka^{(n)}(\BPhi,\rho)$ at that point. This definition also allows us to differentiate $\ka^{(n)}(\BPhi,\rho)$ with respect to $\vec\phi$. 
  \enr

By Lemma \ref{ldk-n}, the sequences $\ka^{(n)}(\BPhi, \rho)$ and $\nabla _{\vec \phi }\ka^{(n)}(\BPhi, \rho)$ are Cauchy sequences in $L_{\infty}\left( \CG^{\BPhi(\infty)}(\rho)\right)$.
Let us define 
    $\ka^{(\infty)}( \BPhi,\rho):=\lim_{n \to \infty}\ka^{(n)}(\BPhi, \rho)$, 
    $\nabla _{\vec \phi} \ka^{(\infty)}:=\lim \nabla _{\vec \phi} \ka^{(n)}$, and
    $    \bka^{(\infty)}( \BPhi,\rho):=\ka^{(\infty)}( \BPhi,\rho)\BPhi$
     for $\BPhi \in \CG^{\BPhi(\infty)}(\rho)$.
    Note that  $\nabla _{\vec \phi} \ka^{(\infty)}$ is not quite the derivative of $\ka^{(\infty)}( \BPhi,\rho)$ in a usual sense, since the set $\CG^{\BPhi(\infty)}(\rho)$ is most likely to be a Cantor set and have no interior points. However, $\nabla _{\vec \phi} \ka^{(\infty)}$ can be thought of as a `derivative' of $\ka^{(\infty)}$ if we define the derivative as a limit over sequences inside $\CG^{\BPhi(\infty)}(\rho)$. 
        The following lemma is a straightforward consequence of this definition.
    
    \bel The function $\ka^{(\infty)}( \BPhi,\rho)$ satisfies the following estimates for $ \BPhi \in \CG^{\BPhi(\infty)}(\rho)$: 
   
    \begin{align}\label{n6.1}& \left|\ka^{(\infty)}(\BPhi ,\rho)
-\rho\right|\ll \rho ^{-2}.
    \end{align}
    Moreover,
    \begin{align}\label{6.1a}& \left|(\nabla _{\vec \phi})^q \ka^{(\infty)}(\BPhi ,\rho)
\right|\ll \rho ^{-2},\ \ \hbox{if}\ \ q<\frac{1}{3(l+\mu+1)\sigma_{1,d-1,1}},
   \cr &\left|\ka^{(\infty)}(\BPhi ,\rho)-\ka^{(0)}(\BPhi ,\rho)\right|\ll \rho^{-\rho^{\sigma_{1,d-1,1}}(2Q)^{-1}-1},
                             \cr &\left|\ka^{(\infty)}(\BPhi ,\rho)-\ka^{(n)}(\BPhi ,\rho)\right| \ll
   \rho^{-\rho^{r_{n,2}}\rho^{-r_{n-1,2}}},\ \ n\geq1.
\end{align}
\enl

We now define the following set: 
$$
{\CD}^{(n)}(\rho)=\{\nbka^{(n)} (\BPhi, \rho):\  \BPhi \in \CG^{\BPhi(n)}(\rho)\}\subset\R^d. 
$$ 
Since  all the points of this set satisfy
the equation $\lambda^{(n)} (\nbka ^{(n)}(\BPhi;\rho))=\rho^2$,
we call this set the isoenergetic surface of the operator $H^{(n)}$. The ``radius"
$\nka^{(n)}(\BPhi;\rho )$ 
increases with $\rho $ (for fixed $\BPhi$).
The set ${\CD}^{(n)}(\rho)$
    is a slightly distorted $(d-1)$-dimensional sphere with holes, see (\ref{new17})
    and Lemmas  \ref{ldk-2},
\ref{ldk-n}.

Further, we define 
\begin{equation}
{\CD}^{(\infty)}(\rho)=\left\{\ka^{(\infty) }(\BPhi, \rho)\BPhi , \BPhi \in \CG^{\BPhi(\infty)}(\rho)\right\}.\label{Dinfty} 
\end{equation} 

Let us also define
$$\CG ^{\bka(n) }:=\cup _{\rho
>\rho_*}{ \CD}^{n}(\rho)\  $$
and
$$\CG ^{\bka(\infty) }:=\cup _{\rho
>\rho_*}{ \CD}^{\infty}(\rho)\  $$
(the good sets of $\bka$ at step $n$ and in the limit respectively). 
We have proved that
\begin{equation}
\frac{\meas(\CG ^{\bka(\infty) }\cap B(R))}{\meas (B(R))}=1-O(R^{-\sigma_0}) \label{full}
\end{equation}
as $R\to\infty$. 

Next, we will show that ${\CD}^{(\infty)}(\rho)$ is in fact the isoenergetic
surface for $H$. 
Namely, for every $\bka \in {
\CD}^{(\infty)}(\rho)$ there is a generalized eigenfunction $\fun
^{(\infty) }(\bka ,\bx)$ such that $H\fun
^{(\infty )}=\rho^2
\fun^{(\infty )}$.

\subsection{Generalized Eigenfunctions of H}

Let $n\ge 1$. By (\ref{6.1a}), any $\bka \in {\CD}^{(\infty)}(\rho
)$ belongs to the $c\rho^{-\rho^{r_{n,2}}\rho^{-r_{n-1,2}}}-$ neighbourhood of ${\CD}^{(n)}(\rho)$. 
Let us consider the spectral projections $\E^{(n)}(\bka)$. They are one-dimensional spectral projections of $H^{(n)}(\bka)$ with the corresponding eigenvalue $\la^{(n)}(\bka)$ given by the perturbation series \eqref{eigenvalue-n}.  Each of these projections has initially been defined as an operator acting on $\GH( \hat K^{(n)},\bka)$. 
We will now extend them to operators acting in the entire space $\GH(\bka)$ (that we have identified with  $\ell^{2}(\Z^{l})$); they will remain orthogonal projections after such an  extension. 
In the previous sections (see \eqref{perturbation*},\eqref{perturbation*-2}, \eqref{perturbation*-n}) we have obtained the following
inequalities (recall that $\ka:=||\bka||$):
 \begin{equation}   
    \begin{split}& \left\|\E^{(0)}({\nbka})-\E_{\mathrm {unp}}({\nbka})\right\|_1\ll \ka^{-1+(2l+\mu+1)\sigma_{1,d-1,1}}.
    \cr &
\left\|\E^{(1)}({\nbka})-\E^{(0)}({\nbka})\right\|_1 \ll \ka^{-\ka^{\sigma_{1,d-1,1}}(4Q)^{-1}},
    \cr &\left\|\E^{(n)}({\bka})-\E^{(n-1)}({\bka})\right\|_1\ll \ka^{-\ka^{r_{n-1,2}}\ka^{-r_{n-2,2}}},\quad n \geq 2,\label{6.2.2}
    \end{split}
    \end{equation}
    \begin{equation}
    \begin{split}&
    \bigl|\lambda ^{(0)}(\bka)-\ka^{2}
     \bigr|
     \ll\ka^{-2+(3l+2\mu+2)\sigma_{1,d-1,1}},
    \cr &
     \bigl|\lambda ^{(1)}(\bka)-\lambda ^{(0)}(\bka)\bigr| \ll
   \ka^{-\ka^{\sigma_{1,d-1,1}}(2Q)^{-1}},
     \cr &
     \bigl|\lambda ^{(n)}(\bka )-\lambda ^{(n-1)}(\bka )\bigr|\ll \ka^{-\ka^{r_{n-1,2}}\ka^{-r_{n-2,2}}},
     \quad n \geq 2,
    \label{6.2.3}
    \end{split}
    \end{equation}
where  $\
{\E}_{unp}(\bka)$ is the unperturbed projection of $H_0(\bka)$. In all these formulas we assume that $r_{n,j}=r_{n,j}(\kappa)$. 


Using analyticity arguments one can easily obtain estimates for derivatives of the above objects (cf. Corollary~\ref{Cor4.8},   Lemma~\ref{L:derivatives-2},  and Lemma~\ref{L:derivatives-n}) valid in the corresponding neighbourhoods of the non-resonant sets. Although we do not need those estimates to prove our main result, we would like to state them here for future reference.

\bel\label{analyticity}
The following estimates hold:

\begin{equation}   
    \begin{split}& \left\|D^{\bm}_{\bka}(\E^{(0)}({\nbka})-\E_{\mathrm {unp}}({\nbka}))\right\|_1\ll \ka^{-1+(2l+\mu+1)\sigma_{1,d-1,1}}\ka^{|\bm|(l+\mu+3)\sigma_{1,d-1,1}}.
    \cr &
\left\|D^{\bm}_{\bka}(\E^{(1)}({\nbka})-\E^{(0)}({\nbka}))\right\|_1 \ll \ka^{-\ka^{\sigma_{1,d-1,1}}(4Q)^{-1}}\ka^{|\bm|r_{1,3}\ka^{2d^2(l+\mu)\sigma_{0,d-1}}},
    \cr &\left\|D^{\bm}_{\bka}(\E^{(n)}({\bka})-\E^{(n-1)}({\bka}))\right\|_1\ll \ka^{-\ka^{r_{n-1,2}}\ka^{-r_{n-2,2}}}\ka^{|\bm|r_{n,3}\ka^{3lr_{n-1,1}/Z_0}},\quad n \geq 2,\label{6.2.2analytic}
    \end{split}
    \end{equation}
    \begin{equation}
    \begin{split}&
    \bigl|D^{\bm}_{\bka}(\lambda ^{(0)}(\bka)-\ka^{2})
     \bigr|
     \ll\ka^{-2+(3l+2\mu+2)\sigma_{1,d-1,1}}\ka^{|\bm|(l+\mu+3)\sigma_{1,d-1,1}},
    \cr &
     \bigl|D^{\bm}_{\bka}(\lambda ^{(1)}(\bka)-\lambda ^{(0)}(\bka))\bigr| \ll
   \ka^{-\ka^{\sigma_{1,d-1,1}}(2Q)^{-1}}\ka^{|\bm|r_{1,3}\ka^{2d^2(l+\mu)\sigma_{0,d-1}}},
     \cr &
     \bigl|D^{\bm}_{\bka}(\lambda ^{(n)}(\bka )-\lambda ^{(n-1)}(\bka ))\bigr|\ll \ka^{-\ka^{r_{n-1,2}}\ka^{-r_{n-2,2}}}\ka^{|\bm|r_{n,3}\ka^{3lr_{n-1,1}/Z_0}},
     \quad n \geq 2,
    \label{6.2.3analytic}
    \end{split}
    \end{equation}
where $\bm=(m_1,\dots,m_d)$ and $D^{\bm}_{\bka}=D^{m_1}_{\ka_1} \dots D^{m_d}_{\ka_d}$.
\enl

 \ber \label{R:Dec9} We see from (\ref{6.1a}) that any $\bka
\in {\CD}^{\infty}(\rho)$ lies within distance $c\rho^{-\rho^{r_{n,2}}\rho^{-r_{n-1,2}}}$ 
from ${\CD}^{(n)}(\rho),\
n\geq1$. Applying perturbation formulae for the $n$-th step, we easily
obtain that  our sequence of eigenvalues $\lambda^{(n)}(\bka )$ of $H^{(n)}(\bka )$ satisfies the following property: 
\begin{equation} \lim _{n\to \infty }\lambda^{(n)}(\bka
)=\rho^2.
\label{n6.2}
\end{equation}
\enr

Let $\bv^{(n)}$ be a unit vector generating the range of  the projection
${\E}^{(n)}(\bka )$, ${\E}^{(n)}(\bka)=(\cdot ,\bv^{(n)})\bv^{(n)}$. As an element of  $\ell^2(\Z^{l})$,
 we can express $\bv^{(n)}$ as  
\bee\label{lincom}
\bv^{(n)}=\{ v^{(n)}_{\bs}\}_{\bs \in \Z^{l}}\in\ell^2(\Z^{l}), 
\ene
where our construction implies that $v^{(n)}_{\bs}=0$ when $\bs \not \in \hat K^{(n)}$. If we prefer to consider  $\bv^{(n)}$ as an element of $B^2(\R^d)$ (more precisely, of $\GH(\bka)$), the expression \eqref{lincom} corresponds to the linear combination of exponentials:  
\bee\label{lincom1}
\fun^{(n)}(\bka,\bx):=\sum _{\bs \in \hat K^{(n)}} v^{(n)}_{\bs}\be_{\bka+\bs \vec \boldom},\ \ n\geq0.
\ene
\bel\label{prelimit} The function $\fun^{(n)}(\bka,\bx),\ n\geq 0,$ satisfies the equation:
$$-\Delta \fun^{(n)}(\bka ,\bx)+V(\bx)\fun^{(n)}(\bka,\bx)=
\lambda ^{(n)}(\bka)\fun^{(n)}(\bka,\bx)+g^{(n)}(\bka,\bx),$$
 with $g^{(n)}$ satisfying the following estimates (as a function of $x$): 
 \begin{equation}
 \|g^{(0)}\|_{B^1(\R^d)}<\ka^{-1+3(l+\mu+1)\sigma_{1,d-1,1}},\ \ \|g^{(1)}\|_{B^1(\R^d)}<\ka^{-\ka^{\frac 12 \sigma_{1,d-1,1}}},\ \ \|g^{(n)}\|_{B^1(\R^d)}<\ka^{-\ka^{\frac 12 r_{n-1,2}}},\ \ n\geq2.
\label{Aug13-2} 
\end{equation} 
Consequently, we have
\begin{equation}\label{g_n}
\|g^{(0)}\|_{L_{\infty }(\R^d)}<\ka^{-1+3(l+\mu+1)\sigma_{1,d-1,1}},\ \ \|g^{(1)}\|_{L_{\infty }(\R^d)}<\ka^{-\ka^{\frac 12 \sigma_{1,d-1,1}}},\ \ \|g^{(n)}\|_{L_{\infty }(\R^d)}<\ka^{-\ka^{\frac 12 r_{n-1,2}}},\ \ n\geq2.
\end{equation}

Fourier coefficients $\lu g^{(n)},\be_{\bka+\bs\vec\boldom}\ru_{B^2(\R^d)}$, $\bs\in\Z^l$, 
can differ from zero only when $\bs\not\in\hat K^{(n)}$ but is inside $Q$-neighbourhood of $\hat K^{(n)}$. 
 
\enl
\bel The functions $\fun^{(n)}(\bka,\bx)$ satisfy the following inequalities:
\begin{equation}
\bes
\label{psi1}
    \|\fun^{(0)}-\be_{\bka}\|_{L_{\infty }(\R^d)}& \ll \ka^{l\sigma_{1,d-1,1}}\ka^{-1+(2l+\mu+1)\sigma_{1,d-1,1}},\\
   \|\Delta\fun^{(0)}+\ka^2\be_{\bka}\|_{L_{\infty }(\R^d)}&\ll \ka^{l\sigma_{1,d-1,1}}\ka^{1+(2l+\mu+1)\sigma_{1,d-1,1}}.
 \end{split}
    \end{equation}
    Moreover,
   \begin{equation}\label{psi2}
   \begin{split}
\|\fun^{(1)}-\fun^{(0)}\|_{L_{\infty}(\R^d)}&\ll \ka^{lr_{1,2}}\ka^{-\ka^{\sigma_{1,d-1,1}}(4Q)^{-1}},\\
\|\Delta \fun^{(1)}-\Delta \fun^{(0)}\|_{L_{\infty}(\R^d)}&\ll \ka^{2+lr_{1,2}}\ka^{-\ka^{\sigma_{1,d-1,1}}(4Q)^{-1}}.
    \end{split}
    \end{equation}
   Finally, 
   \begin{equation}
    \begin{split}
    &\|\fun^{(n)}-\fun^{(n-1)}\|_{L_{\infty}(\R^d)}
    \ll  \ka^{lr_{n,2}}\ka^{-\ka^{r_{n-1,2}}\ka^{-r_{n-2,2}}}, \ \cr &\|\Delta\fun^{(n)}-\Delta\fun^{(n-1)}\|_{L_{\infty}(\R^d)}
    \ll \ka^{2+lr_{n,2}}\ka^{-\ka^{r_{n-1,2}}\ka^{-r_{n-2,2}}} ,\ \ n \geq 2. \label{Dec10}
\end{split}
    \end{equation}
    \enl
        \begin{corollary} \label{C:Psi}
        All functions $\fun^{(n)}$, $n=0,1,....$ enjoy the estimate 
        $$ \|\fun^{(n)}\|_{L_{\infty
}(\R^d)}<1+C\ka^{-1+(3l+\mu+1)\sigma_{1,d-1,1}}$$ 
uniformly in $n$. 
\end{corollary}

\bet \label{T:Dec10}
For every  $\lambda >\rho_*^{2}$ and $\bka \in {
\CD}^{\infty}(\rho)$ the sequence of functions $\fun^{(n)}(\bka,\bx)$
converges  in $L_{\infty}(\R^d)$ and $W_{2,loc}^{2}(\R^d)$. The
limit function $\fun^{(\infty )}(\bka,\bx):=\lim_{n\to \infty }
    \fun^{(n)}(\bka,\bx)$ is a quasi-periodic function:
  \begin{equation} \label{quasi-periodic}
  \fun^{(\infty )}(\bka,\bx)=\sum _{\bs \in \Z^{l}}  v^{(\infty)}_{\bs}\be_{\bka +\bs\vec \boldom}, 
  \end{equation}
  where $\bv^{(\infty)}:=\{v^{(\infty)}_\bs\}_{\bs\in\Z^{l}}\in \ell^1(\Z^{l})$ and $\|\bv^{(\infty)}\|_{\ell^2 (\Z^{l})}=1$.
  The function $\fun^{(\infty )}(\bka,\bx)$   satisfies the equation
    \begin{equation}\label{6.7}
     -\Delta \fun^{(\infty )}(\bka, \bx)+V(\bx)\fun^{(\infty )}(\bka,
    \bx)= \lambda \fun^{(\infty )}(\bka, \bx).
    \end{equation}
 It can be represented in the form
   \begin{equation}\label{6.4}
    \fun^{(\infty )}(\bka,\bx)=\be_{\bka}(\bx)\bigl(1+u^{(\infty)}(\bka, \bx)\bigr),
    \end{equation}
where $u^{(\infty)}(\bka, \bx)$ is an almost-periodic
function:
   \begin{equation}\label{6.5}
    u^{(\infty)}(\bka, \bx)=\sum_{n=0}^{\infty} u^{(n)}(\bka,
    \bx),
    \end{equation}
    each $u_n$ being a finite sum of exponentials,
 \begin{equation}
 u^{(n)}(\bka, \bx)=\sum _{\bs \in \hat K^{(n)}}
c^{(n)}_{\bs}\be_{\bs\vec \boldom}. 
\label{u_n} 
\end{equation}
 The  functions $u^{(n)}$ satisfy the estimates:
\begin{equation}
    \|u^{(0)}\|_{L_{\infty}(\R^2)} \ll \ka^{l\sigma_{1,d-1,1}}\ka^{-1+(2l+\mu+1)\sigma_{1,d-1,1}}, \ \
    \|u^{(1)}\|_{L_{\infty}(\R^2)} \ll \ka^{lr_{1,2}}\ka^{-\ka^{\sigma_{1,d-1,1}}(4Q)^{-1}},\label{6.6a}
    \end{equation}
    \begin{equation}\label{6.6}
    \|u^{(n)}\|_{L_{\infty}(\R^d)} \ll \ka^{lr_{n,2}}\ka^{-\ka^{r_{n-1,2}}\ka^{-r_{n-2,2}}}, \ \ \ n\geq
    2 .\end{equation}
\ent

\begin{proof} 
By
\eqref{Dec10}, we obtain that $\fun^{(n)}(\bka,\bx)$ is a Cauchy sequence
in $L_{\infty }(\R^d)$ and $W_{2,loc}^{2}(\R^d)$. Let $\fun^{(\infty)}(\bka
,\bx):=\lim_{n \to \infty}\fun^{(n)}(\bka, \bx).$ This limit is
defined pointwise uniformly in $\bx$ and in $W_{2,loc}^{2}(\R^d)$.
Noting also that $\lim \lambda ^{(n)}(\bka )=\rho^2=\la $, and
taking into account Lemma~\ref{prelimit} we obtain that \eqref{6.7}
holds.

 Next, we have: 
    $$ \fun^{(n)}=\be_{\bka}+(\fun^{(0)}-\be_{\bka})+\sum_{n=1}^{\infty}(\fun^{(n)}-\fun^{(n-1)}),
$$ 
the series converging in $L_{\infty}(\R^d)$ by (\ref{Dec10}).
Introducing $u^{(n)}:=e^{-i \langle \bka ,\bx\rangle}(\fun^{(n)}-\fun^{(n-1)}),$ we arrive at (\ref{6.4}), (\ref{6.5}). Note that
 $u^{(n)}$ has a form \eqref{u_n}. Estimates (\ref{6.6a}), (\ref{6.6})  follow from (\ref{psi1})--(\ref{Dec10}).
 \end{proof}

\bet[Bethe-Sommerfeld Conjecture]
The spectrum of  operator $H$ contains a semi-axis.
\ent
\begin{proof}
    The theorem immediately follows from the fact that the equation \eqref{6.7} has a bounded solution $\fun^{(\infty) }(\bka ,\bx )$ for every  sufficiently large $\lambda $.
\end{proof}

\subsection{Proof of Absolute Continuity of the Spectrum}\label{chapt8}
Recall that we have defined $\CG^{\bka(n)}:=\cup _{\rho >\rho_*}\CD ^{(n)}(\rho )$. This is a good set of momenta on step $n$.
 There is a family of  eigenfunctions $\fun^{(n)}(\bka)=\fun^{(n)}(\bka,\bx)$, $\bka \in \CG^{\bka(n)}$, of the operator
    $H^{(n)}$ as described above. Suppose, $\tilde\CG^{\bka(n)}$ is a measurable and bounded subset of $\CG^{\bka(n)}$.
    Let us define the approximate spectral projection $E^{(n)}\left(\tilde\CG^{\bka(n)}\right)$ in the following way. 
 First, for   $F\in C_0^{\infty}(\R^d)$ we put 
    \begin{equation} E^{(n)}\left( \tilde\CG^{\bka(n)}\right)F:=\frac{1}{(2\pi)^{d}}\int
    _{ \tilde\CG^{\bk(n)}}\bigl( F,\fun ^{(n)}(\bka )\bigr) \fun ^{(n)}(\bka) d\bka \label{s},
    \end{equation}
     where $\bigl( \cdot ,\cdot \bigr)$
    is the canonical scalar product in $L_2(\R^d)$, i.e.
    $$
    \bigl( F,\fun ^{(n)}(\bka )\bigr)=\int _{\R^d}F(\bx)\overline{\fun ^{(n)}(\bka ,\bx)}d\bx.
    $$
Note that the operator $E^{(n)}\left( \tilde\CG^{\bka(n)}\right)$ can be expressed as a composition:
    \begin{equation} E^{(n)}\left( \tilde\CG^{\bka(n)}\right)=S^{(n)}\left(\tilde\CG^{\bka(n)}\right)T^{(n)} \left(
    \tilde\CG^{\bka(n)}\right), \label{ST}
    \end{equation}
    where 
    $$T^{(n)}=T^{(n)}\left(\tilde\CG^{\bka(n)}\right):L_2(\R^d) \to L_2\left(  \tilde\CG^{\bka(n)}\right), \ \
    \ \ S^{(n)}=S^{(n)}\left(\tilde\CG^{\bka(n)}\right):L_2\left( \tilde\CG^{\bka(n)}\right)\to L_2(\R^d)$$
    are given by 
    \begin{equation}
    T^{(n)}F=\frac{1}{(2\pi)^{d/2} }\bigl( F,\fun ^{(n)}(\bka )\bigr) \mbox{\ \ for any $F\in C_0^{\infty}(\R^d)$}
    \label{eq2}
    \end{equation}
    (note that then $T^{(n)}F\in L_{\infty }\left( \tilde\CG^{\bka(n)}\right)$) and
    \begin{equation}S^{(n)}f = \frac{1}{(2\pi)^{d/2} }\int _{ \tilde\CG^{\bka(n)}}f (\bka)\fun ^{(n)}(\bka ,\bx)d\bka  \mbox{\ \ for any $f \in L_{\infty }\left(
    \tilde\CG^{\bka(n)}\right)$.} \label{ev}
    \end{equation}
    Note that $S^{(n)}f \in L_2(\R^d)$, since $\fun ^{(n)}$ is a finite combination of exponentials $\be _{\bka +\bn \vec \boldom}.$

   \bel\label{n9.1} Let $\tilde\CG^{\bka(n)}$ be  bounded and $f, g\in L_{\infty }\left(
    \tilde\CG^{\bka(n)}\right)$. Denote 
    $$\xi _*:=\inf _{\bxi \in \tilde\CG^{\bka(n)}}||\bxi ||.$$ 
    Then
    \begin{equation}| \left( S^{(n)} f, S^{(n)} g \right)_{L_2(\R^d)}-
    \left(f, g \right)_{L_2(\tilde\CG^{\bka(n)})}|<\xi_*^{-{r_{n-2,2}(\xi_*)}}\|f\|_{L_2(\tilde\CG^{\bka(n)})}\|g\|_{L^2(\tilde\CG^{\bka(n)})}. 
    \label{May14a-12}
    \end{equation}
In particular, we have, uniformly in   $f$, $g$  and $\tilde\CG^{\bka(n)}$:   
\begin{equation} \left( S^{(n)} f, S^{(n)} g \right)_{L_2(\R^d)}=
    \left(f, g \right)_{L_2(\tilde\CG^{\bka(n)})}+o(1)\|f\|_{L_2(\tilde\CG^{\bk(n)})}\|g\|_{L^2(\tilde\CG^{\bka(n)})}. 
    \label{May14a-12aa}
    \end{equation}    
    as $n\to\infty$.  
\enl
\begin{corollary} The following relation holds:
     \begin{equation} \left|\left( S^{(n)} f, S^{(n)} g \right)_{L_2(\R^d)}\right|<(1+o(1))
     \left\|f\right\|_{L_{\infty}(\tilde\CG^{\bka(n)})}\left\|g\right\|_{L_{\infty}(\tilde\CG^{\bka(n)})}\meas(\tilde\CG^{\bka(n)}). \label{Vienna-1}
     \end{equation}
     \end{corollary}
\begin{corollary} The operator $S^{(n)}$ is bounded and $\lim_{n\to \infty}\|S^{(n)}\|= 1$. \end{corollary}

The proof of the above lemma is analogous to the proof of Lemmas 9.1 in \cite{KS}.

It is easy to see that $T^{(n)}\subset (S^{(n)})^*$. Therefore, $\|T^{(n)}\|\leq 1+o(1)$ and it can be extended to the whole space $L_2(\tilde\CG^{\bka(n)})$. We still denote the extended operator by $T^{(n)}$, $T^{(n)}=(S^{(n)})^*$. Therefore, $E^{(n)}$ is a self-adjoint operator. The proofs of the next three lemmas are analogous to those of Lemmas 9.4, 9.7 and 9.8  in \cite{KS}.

\bel \label{L:10.4}  Let $\tilde\CG^{\bka(n)}\subset \tilde{\tilde{\CG}}^{\bka(n)}\subset \CG^{\bka(n)}$. The following relations hold as $n\to \infty $:
\begin{equation} E^{(n)}(\tilde\CG^{\bka(n)})E^{(n)}(\tilde{\tilde{\CG}}^{\bka(n)})=E^{(n)}(\tilde\CG^{\bka(n)})+o(1), \label{Vienna-6} \end{equation}
\begin{equation}\label{C:10.4-1}  E^{(n)}(\tilde{\tilde{\CG}}^{\bka(n)})E^{(n)}(\tilde\CG^{\bka(n)})=E^{(n)}(\tilde\CG^{\bka(n)})+o(1),\end{equation}
\begin{equation} \label{C:10.4-2} (E^{(n)})^2(\tilde\CG^{\bka(n)})=E^{(n)}(\tilde\CG^{\bka(n)})+o(1) .\end{equation}
\enl

Let \begin{equation}
\CG^{(n)}_{\lambda}:=\{ \bka \in {\CG}^{\bka(n)}:
\lambda ^{(n)}(\bka) < \lambda\}; \label{d} 
\end{equation}
we have skipped writing the superscript $\bka$ for simplicity. 
 This set is bounded and Lebesgue measurable, since ${\CG}^{\bka(n)} $ is
open and $\lambda ^{(n)}(\bka)$ is continuous on $
{\CG}^{\bka(n)}$.

\bel\label{L:abs.6}
We have $\meas({\CG}^{(n)}_{\lambda+\varepsilon} \setminus
{\CG}^{(n)}_{\lambda}) \leq C(d) \lambda ^{-1+\frac{d}{2}}\varepsilon $ when $0\leq
\varepsilon \leq 1$. 
\enl

 By (\ref{s}),
$E^{(n)}\left({\CG}^{(n)}_{\lambda+\varepsilon}\right)-E^{(n)}\left({\CG}^{(n)}_{\lambda}\right)
=E^{(n)}\left({\CG}^{(n)}_{\lambda+\varepsilon}\setminus
    {\CG}^{(n)}_{\lambda}\right)$. 
\bel\label{L:abs.7} For any $F \in C_0^{\infty}(\R^d)$ and
$0\leq \varepsilon \leq 1$, \begin{equation}
\left\|\bigl(E^{(n)}({\CG}^{(n)}_{\lambda+\varepsilon})-E^{(n)}({\CG}^{(n)}_{\lambda})\bigr)F\right\|^2_{L_2(\R^d)}
 \leq C( F)\lambda ^{-1+\frac d2} \varepsilon , \label{tootoo1}
 \end{equation}
 where $C(F)$ is uniform with respect to $n$ and $\lambda$.
\enl

By construction,
$
    \CG^{\bka(n+1)}\subset \CG^{\bka(n)}$ and 
    $\CG^{\bka(\infty)}=\bigcap_{n=0}^{\infty}{\CG}^{\bka(n)}. $
Therefore, the perturbation formulas for $\lambda ^{(n)}(\bka )$ and $\fun^{(n)}(\bka )$ hold in
$\CG^{\bka(\infty)}$ for all $n$. We denote $\lambda^{(\infty)}(\bka )=\lim_{n\to\infty}\lambda^{(n)}(\bka )$ for $\bka\in\CG^{\bka(\infty)}$.
 Let
     \begin{equation}
     \CG^{(\infty)}_{ \lambda }:=\left\{\bka \in
    \CG^{\bka(\infty)}: \lambda^{(\infty)}(\bka )<\lambda
    \right\}. \label{dd}
    \end{equation}
The function $\lambda^{(\infty)}(\bka )$ is a Lebesgue
measurable function, since it is a limit of a sequence of
measurable functions. Hence, the set  $\CG^{(\infty)}_{\lambda
}$ is measurable.

\bel\label{add6*} The measure of the symmetric difference of
two sets $\CG^{(\infty)}_{\lambda }$ and $\CG^{(n)}_{
\lambda}$ converges
 to zero as $n \to
\infty$ uniformly in $\lambda$ in every bounded interval:
    $$\lim _{n\to \infty }\left|\CG^{(\infty)}_{\lambda }\triangle \CG^{(n)}_{ \lambda
    }\right|=0.$$
 \enl
 The proof is completely analogous to the proof of Lemma 4 in \cite{KL2}.

Next, we will show that the operators
$E^{(n)}(\CG^{(\infty)}_{\lambda })$ have a strong limit
$E^{(\infty)}(\CG^{(\infty)}_{ \lambda })$ in $L_2(\R^d)$ as $n$
tends to infinity. The operator $E^{(\infty)}(\CG^{(\infty)}_{
\lambda })$ is a spectral projection of $H$. It can be represented
in the form $E^{(\infty)}(\CG^{(\infty)}_{\lambda })=S^{(\infty)}
T^{(\infty)}$, where $S^{(\infty)}$ and $T^{(\infty)}$ are norm 
limits of $S^{(n)}(\CG^{(\infty)}_{\lambda })$ and
$T^{(n)}(\CG^{(\infty)}_{ \lambda })$ respectively.   
 For any $F\in
C_0^{\infty }(\R^d)$, we can show: 
\begin{equation} 
E^{(\infty)}\left(
\CG^{(\infty)}_{ \lambda }\right)F=\frac{1}{(2\pi)^d}\int
    _{ \CG^{(\infty)}_{\lambda }}\bigl( F,\fun ^{(\infty)}(\bka )\bigr) \fun ^{(\infty)}(\bka,\bx)\,d\bka \label{s1u}
    \end{equation}
    and 
    \begin{equation} HE^{(\infty)}\left(
\CG^{(\infty)}_{\lambda }\right)F=\frac{1}{(2\pi)^d}\int
    _{ \CG^{(\infty)}_{\lambda }}\lambda ^{(\infty)}(\bka )\bigl( F,\fun^{(\infty)}(\bka )\bigr) \fun ^{(\infty)}(\bka,\bx)\,d\bka .\label{s1uu}
    \end{equation}
  
The proof of the next lemma is analogous to the one of Lemma 9.10 in \cite{KS}.
\bel We have
\begin{equation} \label{Vienna-2}
\left\|(S^{(n)}(\CG^{(\infty)}_{ \lambda})-S^{(n-1)}(\CG^{(\infty)}_{\lambda}))
f\right\|_{L_2(\R^d)}< C\|f\|_{L_2(\CG^{(\infty)}_{\lambda
})}\rho_*^{-\rho _*^{\frac12 r_{n-1,2}(\rho_*)}}.
\end{equation}
In particular, the convergence of $S^{(n)}(\CG^{(\infty)}_{\lambda
    })$ to $S^{(\infty)}(\CG^{(\infty)}_{\lambda
    }) $  is uniform in $\lambda $ when $\lambda >\lambda _*$.
\enl


\bel \label{Lem1} The operator $T^{(\infty)}(\CG^{(\infty)}_{ \lambda
    })$  can be described by the
    formula $T^{(\infty)}(\CG^{(\infty)}_{ \lambda
    })F=\frac{1}{(2\pi)^{d/2} }\bigl( F,\fun^{(\infty)}(\bka )\bigr) $
for any $F\in C_0^{\infty }(\R^d)$.
\enl
\begin{proof} The lemma easily follows from  Theorem \ref{T:Dec10} and formula \eqref{eq2}. \end{proof}

The details for the next lemma can be found in Lemmas 9.11, 9.17 from \cite{KS} and Lemmas 10, 11 from \cite{KL2}.

\bel\label{May8} \begin{enumerate}  \item The sequence 
$E^{(n)}(\CG^{(\infty)}_{ \lambda })$ has a norm limit
$E^{(\infty)}(\CG^{(\infty)}_{\lambda })$, uniformly over $\lambda >\lambda _*$. The
operator $E^{(\infty)}(\CG^{(\infty)}_{\lambda })$ is an orthogonal 
projection. Its action on any $F\in C_0^{\infty }(\R^d)$ it is given by
(\ref{s1u}).
\item There is a strong limit
$E^{(\infty)}(\CG^{(\infty)})$ of the projections $E^{(\infty)}
(\CG^{(\infty)}_{\lambda })$ as $\lambda $ goes to infinity.
 \item The operator $E^{(\infty)}(\CG^{(\infty)})$ is an orthogonal projection.
\item Projections
$E^{(\infty)}(\CG^{(\infty)}_{\lambda })$ and
$E^{(\infty)}(\CG^{(\infty)})$ reduce the operator $H$.
\item The family of projections
$E^{(\infty)}(\CG^{(\infty)}_{\lambda} )$ is the resolution of the
identity of the operator $H$ 
restricted to the range of 
$E^{(\infty)}(\CG^{(\infty)})$. 
\item Formula \eqref{s1uu} holds when $F\in C_0^{\infty}(\R^d)$. 
\end{enumerate}
\enl

Finally, we can show that the  branch of spectrum (semi-axis) corresponding
to $\mathcal G^{(\infty)}$  is absolutely continuous.

\begin{thm}\label{T:abs} For any $F\in C_0^{\infty }(\R^d)$ and
$0\leq \varepsilon \leq 1$,
    \begin{equation}
    \left| \left(E^{(\infty)}(\CG^{(\infty)}_{\lambda+\varepsilon})F,F\right)-
    \left(E^{(\infty)}(\CG^{(\infty)}_{\lambda})F,F
    \right) \right| \leq C_F\lambda ^{-1+\frac d2}\varepsilon .\label{May10*}
    \end{equation}

\begin{proof} By  (\ref{s1u}),
    $$ | \left(E^{(\infty)}(\CG^{(\infty)}_{\lambda+\varepsilon})F,F\right)-\left(E^{(\infty)}(\CG^{(\infty)}_{\lambda})F,F
    \right) | \leq C_F\left| \CG^{(\infty)}_{ \lambda +\varepsilon
    }\setminus \CG^{(\infty)}_{ \lambda } \right| .$$
Applying Lemmas \ref{L:abs.6} and \ref{add6*}, we immediately get
(\ref{May10*}).

\end{proof}

\end{thm} \begin{corollary} The spectrum of the operator
$HE^{(\infty)}(\CG^{(\infty)})$ is absolutely continuous.
\end{corollary}

{\newred
\ber\label{newrm1}
Finally, we would like to discuss the possibility of extending our Theorem \ref{mainth} to prove that the spectrum of $H$ is purely absolutely continuous for large energies (analogously to the one-dimensional case). 
Doing this would require constructing more general Bloch-Floquet solutions than those constructed in Theorem \ref{T:Dec10}. More precisely, instead of restricting ourselves to solutions $\fun^{(\infty )}(\bka,\bx)$ corresponding to $\bka \in {
\CD}^{\infty}(\rho)$, we would have to construct such solutions for all (or, possibly, almost all) $\bka\in\R^d$ and prove that $\{\fun^{(\infty )}(\bka,\bx),\bka\in\R^d\}$ forms a distorted Fourier basis. Of course, as we explained in the introduction, we cannot possibly hope that all such solutions will be small perturbations of one exponential function satisfying   \eqref{6.4}. Indeed, since some $\bka$ will be resonant for our frequencies, the best we can hope for is that such solutions will be small perturbations of a finite linear combination of exponentials, with the number of terms in this linear combination unbounded above. In other words, instead of throwing away some resonant $\bk$ at each step of our procedure, we would be forced to keep them and, instead of dealing only with solutions with one bump in $l^2(\Z^l)$, look at solutions with any number of bumps. This, of course, will make all the estimates like the Cartan Lemma significantly worse. Still, it does not seem totally  infeasible to prove the complete absolute continuity of the spectrum for large energies in this way. We plan to make an attempt of doing this in the future.   
\enr 
}

\section{Appendices}\label{section10}
\subsection{Appendix 1}
Here, we formulate a useful abstract perturbation type lemma. It has appeared in various shapes and forms many times before and can be considered as "folk knowledge". 
{\newred
Loosely speaking, this Lemma states that if we have an operator on a lattice, we know how to estimate the resolvent of this operator restricted to smaller cubes, and the distances between these cubes are sufficiently large (compared to the estimates of corresponding resolvents), then we can estimate the resolvent of the operator (in entire space or restricted to a much bigger cube).
}
We state it in the  most convenient for us setting and using notation close to the one used in this paper. The proof is based on the arguments from the proof of Theorem~\ref{Thm2} with the use of the multiple resolvent identity. We need one definition before we formulate this theorem. Suppose, $K\subset\Z^l$ is an extended cube with centre $a$ and radius $r$ (recall that this means that $\Omega(a,r)\subset K\subset \Omega(a,2r)$). By $int(K)$ we denote the ball $\Omega(a,r/2)$. Of course, since $a$ and $r$  are not uniquely defined by $K$, $int(K)$ is also not uniquely defined. Note that
 the $\Z$-distance from $int(K)$ to $\Z^l\setminus K$ is at least $r/2$.

\bel\label{abstractlemma}
Let $H=H_0+V$ be a self-adjoint operator acting on $l^2(\Z^l)$ with diagonal $H_0$ (i.e. $(H_0)_{\bn\bn'}=0$ if $\bn\ne\bn'$) and $V$ has range $Q$ (i.e. $V_{\bn\bn'}\ne 0$ only if  $0<|\bn-\bn'|<Q$) and is bounded. Let $z\in\C$ be any complex number. 

Let $n\in\N$ and $K^{(n+1)}\subset\Z^l$ (note that we do not assume that $K^{(n+1)}$ is an extended cube). Let $E$ be a big constant (see \eqref{Ebig}) and  $K^{(j)}_m\subset K^{(n+1)}$, $1\leq j\leq n$, $1\leq m\leq m_j<\infty$, be  extended cubes with the following properties:

1) the size of $K^{(j)}_m$ is $E^{s_{j,m}}$ with $\min_m s_{j,m}>10\max_m s_{j-1,m}$ ($1<j\leq n$), 
and $\min_m E^{s_{1,m}}>10Q$.

2) $$dist\{K^{(j)}_m,\,K^{(j)}_{m'}\}>2Q\ \ \hbox{for}\ \ 1\leq j\leq n,$$
and
$$dist\{K^{(j)}_m,\,K^{(j')}_{m'}\}>2Q\ \ \hbox{for}\ \ 1\leq j'<j\leq n,$$ unless $K^{(j')}_{m'}\subset K^{(j)}_{m}$. Also, if $K^{(j')}_{m'}\subset K^{(j)}_{m}$, then 
$$dist\{\Z^l\setminus K^{(j)}_m,\,K^{(j')}_{m'}\}>2Q.$$

3) There are positive numbers $p_0$, $p_{j,m}$ ($1\leq j\leq n$, $1\leq m\leq m_j$) satisfying  $p_{j,m}<\frac{p_0}{1000}E^{s_{j,m}/2}$ such that 
the following estimates hold: 
$$\|(H(K^{(n)}_m)-z)^{-1}\|\leq E^{p_{n,m}}.$$
Also, for any cube $K^{(j)}_m$ with $1\leq j<n$, either we have the estimate 
$$\|(H(K^{(j)}_m)-z)^{-1}\|\leq E^{p_{j,m}},$$ 
or $K^{(j)}_m\subset int(K^{(j+1)}_{m'})$ for some larger cube $K^{(j+1)}_{m'}$. Finally, for any point $\bq\in K^{(n+1)}$ either we have 
$$|(H_0)_{\bq\bq}-z|>E^{p_0},$$
 or $\bq\in int(K^{(1)}_{m'})$ for some cube $K^{(1)}_{m'}$.

Assume that  
\bee\label{Ebig}
E^{-p_0}<\frac{1}{10}\min\{1,\|V\|^{-1}\}.
\ene 
Then we have 
\bee\label{abslest}
\|(H(K^{(n+1)})-z)^{-1}\|\leq 2E^{p},
\ene
where $p:=\max_{j,m}p_{j,m}$.

\enl

\bep

First, without loss of generality, we may assume that for any $K^{(j')}_{m'}\subset K^{(j)}_m$, $j'<j$, we either have $K^{(j')}_{m'}\subset int(K^{(j)}_m)$ or $K^{(j')}_{m'}\cap int(K^{(j)}_m)=\emptyset$ (otherwise, we can just add to $int(K^{(j)}_m)$ all $K^{(j')}_{m'}$ which nontrivially intersect with it).  
We also introduce the following notation: $P^{int(j)}_m:=P(int(K^{(j)}_m))$, $P^{(j)}_m:=P(K^{(j)}_m)$, $P^{(n+1)}:=P(K^{(n+1)})$, $H^{int(j)}_m:=H(int(K^{(j)}_m))$, $H^{(j)}_m:=H(K^{(j)}_m)$.  Let us establish that for any $H^{(j)}_m$ satisfying the estimate from condition 3) of the lemma, we have
\bee\label{absl1}
||(P^{(n+1)}-P^{(j)}_m)V(H^{(j)}_m-z)^{-1}||<2E^{-p_0}||V||. 
\ene
We proceed by induction. For $j=1$ this inequality has, essentially, been proved in theorem \ref{Thm2}. Suppose, $j>1$. 
We denote 
$$\tilde P^{(j)}_m:=P^{int(j)}_m+\sum P(K^{(j')}_{m'}), \ \ \ \check P^{(j)}_m:=P^{(j)}_m-\tilde P^{(j)}_m,$$
where the sum is over all $j',m'$ with $j'<j$ and $K^{(j')}_{m'}\subset  K^{(j)}_m\setminus int(K^{(j)}_m)$. We put  
\bee\label{tildeH}
\tilde H^{(j)}_m:=\tilde P^{(j)}_m H \tilde P^{(j)}_m=P^{int(j)}_m H P^{int(j)}_m+\sum P(K^{(j')}_{m'})H P(K^{(j')}_{m'})
\ene 
and 
\bee\label{hatH}
\hat H^{(j)}_m:=\tilde H^{(j)}_m+\check P^{(j)}_m H_0\check P^{(j)}_m.
\ene
Then $H^{(j)}_m=\hat H^{(j)}_m+W^{(j)}_m$, where
$$W^{(j)}_m:= P^{(j)}_m V P^{(j)}_m-\tilde P^{(j)}_m V\tilde P^{(j)}_m=\check P^{(j)}_mV\check P^{(j)}_m+\tilde P^{(j)}_mV\check P^{(j)}_m+\check P^{(j)}_mV\tilde P^{(j)}_m.$$
We have
\bee\label{abspert}
(P^{(n+1)}-P^{(j)}_m)V(H^{(j)}_{m}-z)^{-1}=\sum_{r=0}^{R_0}B_r+(P^{(n+1)}-P^{(j)}_m)V((\hat H^{(j)}_{m}-z)^{-1}W^{(j)}_{m})^{R_0+1}( H^{(j)}_{m}-z)^{-1},
\ene
where
\bee
B_r:=(P^{(n+1)}-P^{(j)}_m)V((\hat H^{(j)}_{m}-z)^{-1}W^{(j)}_{m})^{r}(\hat H^{(j)}_{m}-z)^{-1}
\ene
and $R_0:=[E^{s_{j,m}}E^{-2\max_{m'}s_{j-1,m'}}/4]-1$. We notice that any block in $\hat H^{(j)}_m$ (except $int(K^{(j)}_m)$) has size smaller than than $E^{\max_{m'}s_{j-1,m'}}$. Thus, we never reach $int(K^{(j)}_m)$ in \eqref{abspert} with this number of steps, i.e. 
$$
(P^{(n+1)}-P^{(j)}_m)V((\hat H^{(j)}_{m}-z)^{-1}W^{(j)}_{m})^{r}P^{int(j)}_m=0,\ \ \ r\leq R_0+1.
$$
Using \eqref{absl1} for $j',m'$ (which holds due to the induction assumption), we obtain 
\bee\label{absbr}
\|B_r\|\leq (2\|V\|E^{-p_0})^{r+1},\ \ \ r\leq R_0.
\ene
For the last term in \eqref{abspert} we have
\bee\label{absbr'}
\|(P^{(n+1)}-P^{(j)}_m)V((\hat H^{(j)}_{m}-z)^{-1}W^{(j)}_{m})^{R_0+1}( H^{(j)}_{m}-z)^{-1}\|\leq (2\|V\|E^{-p_0})^{R_0+1}\|V\|E^{p_{j,m}}.
\ene
Now, we only need to adjust the estimate \eqref{abspert} for $r=0$. Indeed, in this case we don't even reach any $K^{(j')}_{m'}$ inside $K^{(j)}_{m}$, i.e. if $K^{(j')}_{m'}\subset K^{(j)}_{m}$, $1\leq j'<j$, then (see condition 2) $(P^{(n+1)}-P^{(j)}_m)VP^{(j')}_{m'}=0$. So, we can use better estimate from condition 3 of the lemma. We have 
\bee\label{absb0}
\|B_0\|\leq \|V\|E^{-p_0}.
\ene
Combining \eqref{absbr}, \eqref{absbr'} and \eqref{absb0} we obtain \eqref{absl1}.

Now, the statement of the lemma follows by perturbation arguments similar to those from the proof of Theorem~\ref{Thm2}. Indeed, we have (cf. \eqref{abspert})
\bee\label{abspertnew}
(H^{(n+1)}-z)^{-1}=\sum_{r=0}^{\infty}((\hat H^{(n+1)}-z)^{-1}W^{(n+1)})^{r}(\hat H^{(n+1)}-z)^{-1}.
\ene
Here, $\hat H^{(n+1)}$ is defined by \eqref{tildeH} and \eqref{hatH} assuming $int(K^{(n+1)})=\emptyset$ and, correspondingly, $P^{int(n+1)}=0$. The series is convergent due to \eqref{absl1}. The estimate \eqref{abslest} follows.

\enp

\subsection{Appendix 2. Lemma \ref{LB}}
First, for completeness, we formulate Lemma 1.20 in \cite{Bou1} with insignificant change of notations.

 \bel \label{L6.3} Assume ${\mathcal A}\subset \cube^{rq}$ is a semi-algebraic subset of the degree $B$ and
$\left|{\mathcal A}\right|_{rq}<\eta. $

Consider frequency vectors $\boldom \in \cube^r$ with components $(\omega _1,...\omega _r)$. 
For $\bn=(n_1,...,n_r)\in \Z^r$, denote 
$\bn\boldom _f=(\{n_1\omega _1\},...,\{n_r\omega _r\})$. Here, $\{x\}$ is a fractional part of a real number $x$. 

Let $\tilde {\mathcal N}_1, ...,  \tilde {\mathcal N}_{q-1}\subset \Z^r$ be finite sets with the following properties:
\bee \min _{1\leq p\leq r}|n_p|>\left(B\max _{1\leq p\leq r}|m_p|\right)^{\hat C_1}, \label{1.21B}
\end{equation}
if $\bn \in  \tilde {\mathcal N}_i$ and $\bm \in\tilde {\mathcal N}_{i-1}$, i=2,...q-1, and where $\hat C_1=\hat C_1(q,r)$.

Assume also
\bee 
\frac{1}{\eta }>\max _{\bn \in \tilde {\mathcal N}_{q-1}}|\bn |^{\hat C_{1}}
\label{1.22B}
\ene
Then
\bee
\left|\left\{ \vec \boldom\in \cube^r \big| \left( \boldom , \bn ^{(1)}\boldom _f, ...,  \bn^{(q-1)}\boldom  _f\right)\in {\mathcal A}\ \mbox{ for some } \bn ^{(i)} \in \tilde{\mathcal N}_i\right\}\right|_r< B^{\hat C_1}\delta ,\label{LambdaB}
\ene
where
\bee
 \delta ^{-1}=\min _{\bn \in \tilde{\mathcal N}_1}\min _{1\leq p\leq r}|n_p|. \label{mesLambda*B}
\ene
\enl

Next, we use the above lemma to prove Lemma \ref{LB}.
\bep In Lemma \ref{L6.3} we take $r:=ld$, $q:=1+s$, ${\mathcal A}=\{A\}\subset \cube^{ld(1+s)}$, where $\{A\}$ is the fractional part of $A$. Let $ \tilde {\mathcal N}_1, ...,  \tilde {\mathcal N}_{q-1}\subset \Z^{ld}$ be defined as:
$$\tilde {\mathcal N_i}=\left\{ \tilde \bn =\{n_{jk}\}_{j=1,k=1}^{l,d}:\ n_{j,k}=n_{j,k'}:=n_j \mbox{ for all\ }k,k';\ \bn=(n_1,...,n_l)\in  {\mathcal N}_i\right\},$$
where ${\mathcal N}_i$, $i=1,...,q-1$, are defined in the statement of Lemma \ref{LB}.
Considering \eqref{1.21}, we see that $ \tilde {\mathcal N}_1, ...,  \tilde {\mathcal N}_{q-1}$ have the property  \eqref{1.21B}:
\bee \min _{1\leq j\leq l, 1\leq k\leq d}|n_{jk}|>\left(B\max _{1\leq j\leq l, 1\leq k\leq d}|m_{jk}|\right)^{\hat C_{1}}, \label{1.21a}
\end{equation}
here a double index $jk$ is taken instead of $p$ in \eqref{1.21B}. Furthermore, by \eqref{1.22}, we have \eqref{1.22B}:
\bee 
\frac{1}{\eta }>\max _{\bn \in \tilde {\mathcal  N}_{q-1}}|\tilde \bn |^{\hat C_{1}}.
\label{1.22a}
\ene
Thus, the conditions of Lemma \ref{L6.3} hold. 
Let $\tilde \Lambda \in \cube^{r}$,
\bee
\tilde \Lambda:=\left\{ \vec \boldom: \left( \boldom , \tilde \bn ^{(1)}\boldom _f,...\tilde \bn ^{(q-1)}\boldom _f\right)\in {\mathcal A} \mbox{ for some } \tilde \bn ^{(i)} \in \tilde {\mathcal N_i},\\,i=1,...,q-1\right\},\label{tildeLambda}
\ene
where $\tilde \bn ^{(i)}\boldom _f$ is the fractional part of vector 
$\left\{n_{jk}^{(i)}\omega _{jk}\right\}_{j=1,k=1}^{l,d}$.
Applying  \eqref{LambdaB}, we obtain:
\bee
\meas(\tilde \Lambda)< B^{\hat C_1}\delta ,\  \  \  \delta ^{-1}=\min _{\bn \in \tilde{\mathcal  N}_1}\min _{1\leq j\leq l, 1\leq k\leq d}|n_{jk}|. \label{mes-tildeLambda}
\ene
Using \eqref{Lambda}, \eqref{tildeLambda}, it is easy to see that $\Lambda \subseteq \tilde \Lambda $ and, hence \eqref{mesLambda*} follows from \eqref{LambdaB}. \enp

\subsection{Appendix 3. Lemma 1.18 in \cite{Bou1}}
\bel Let $A\subset \cube^{d+r}{\newred \subset\R^d_{x}\times\R^r_t}$ be a semi-algebraic set of degree B. 
{\newred  For each $ t\in \cube^r$  
we put
  $$ A_{\mathrm {cs}}(t):=\{x\in \cube ^{d}: (t,x)\in A\} .$$
Similarly, for  $ x \in \cube^d $ we put
$$ A_{\mathrm {cs}}(x):=\{t\in \cube ^{r}: (t,x)\in A\} .$$ 
}
Assume that  for each $t$ {\newred
\bee\  \  \meas(A_{\mathrm {cs}}(t))<\eta . \ene
}
Then the set
{\newred
\bee
\left\{ \left(x_1,...,x_{2^r}\right) \big| A_{\mathrm {cs}}(x_1)\cap ....\cap A_{\mathrm {cs}}(x_{2^r})\neq \emptyset \right\}\subset \cube^{d2^r} \ene
 }
is semi-algebraic of the degree at most $B^{\hat C_{2}}$ and measure at most
\bee
\label{8.17} \eta _r=B^{\hat C_{2}}\eta ^{d^{-r}2^{-r(r-1)/2}}\ \ \mbox{with} \ \ {\hat C_{2}}={\hat C_{2}}(r). \ene
\enl

\subsection{Appendix 4. The Tarski-Seidenberg Principle} If $S\subset \R^{d_1+d_2}$ is a semi-algebraic set of degree $B$, then its projections $\pi _1(S)\subset \R^{d_1}$ and
$\pi _2(S)\subset \R^{d_2}$ are semialgebraic of degree at most $B^{\hat C_{3}(d_1,d_2)}$, see e.g. \cite{Bou1}.

\subsection{Appendix 5. Rouch\'e's type Theorem}

\bel\label{complexlemma}
Let $f$ be a meromorphic function in the disc $\{|z|\leq r\}$ such that the number of poles (counting multiplicity) in this disc is $N$ and on the boundary $\{|z|=r\}$ we have the estimate $|f(z)|\leq C$.
 Then for any $z$ inside the disc being $\epsilon$-away of any pole of $f$ we have
$$
|f(z)|\leq C\left(\frac{2r}{\epsilon}\right)^N.
$$
\enl
\begin{proof} Let $z_j,\ j=1,\dots,N$ be the poles of $f$. Consider $g(z):=f(z)\prod_{j=1}^N (z-z_j)$. Obviously, $g$ is analytic in $|z|\leq r$ and for $|z|=r$
$$
|g(z)|\leq C(2r)^N.
$$
By the maximum principle the same estimate holds for all $|z|\leq r$. If, in addition, $|z-z_j|\geq\epsilon$, then
$$
|f(z)|\leq |g(z)|\epsilon^{-N}\leq C\left(\frac{2r}{\epsilon}\right)^N.
$$
\end{proof}
We will apply this lemma to the inverse of some analytic matrix-valued function $H(z)$. Then $H(z)^{-1}$ has the form $S(z)/(\det(H(z)))$, where $S(z)$ is analytic. By definition, poles of $H(z)^{-1}$ are zeros of $\det(H)$; let us denote them by $z_j,\ j=1,\dots,N$. Then $H(z)^{-1}\prod_{j=1}^N (z-z_j)$ is analytic and we can apply the proof of our Lemma to it. As a result, we have
\bec
Let $H(z)$ be a meromorphic matrix-function in the disc $\{|z|\leq r\}$ such that the number of poles (counting multiplicity) is $N$ and on the boundary $\{|z|=r\}$ we have the estimate $||H(z)||\leq C$. Then for any $z$ inside the disc being $\epsilon$-away of any pole of $H$ we have
$$
||H(z)||\leq C\left(\frac{2r}{\epsilon}\right)^N.
$$ 
\enc

\subsection{Apppendix 6. Cartan's Lemma}

 Here, we formulate the analogue of Cartan's Lemma for matrices (see \cite{B1}, Lemma 2).

\begin{lem}\label{Cartan}
Let $A(\bx)$ be an real-analytic self-adjoint $N\times N$ matrix-function of $\bx\in[-1/2,1/2]^d$, satisfying the following conditions (with $M\ll N$ and $B_1,\ B_2,\ B_3>1$).

1) $A(\bx)$ has an analytic extension $A(\bz)$ to $\bz\in D^d$ (recall that $D$ is a unit disk in $\C$) with
\begin{equation}\label{Cartan1}
\|A(\bz)\|<B_1\ \ \ \hbox{for}\ \bz\in D^d.
\end{equation}

2) There is a subset $\Lambda$ of $[1,N]$ such that $|\Lambda|\leq M$ and for all $\bz\in D^d$
\begin{equation}\label{Cartan2}
\|(R_{[1,N]\setminus\Lambda}A(\bz)R_{[1,N]\setminus\Lambda})^{-1}\|<B_2
\end{equation}
(here $R_S$ denotes coordinate restriction to $S$).

3) For some $\ba\in[-1/2,1/2]^d$ we have
\begin{equation}\label{Cartan3}
\|A(\ba)^{-1}\|<B_3.
\end{equation}
Then
\begin{equation}\label{Cartan4}
\meas\{\bx\in[-1/2,1/2]^d:\ \|A(\bx)^{-1}\|>e^t\}\leq Cde^{-\frac{ct}{M\log (B_1B_2B_3)}}.
\end{equation}
\end{lem}


\subsection*{Notation Index}

\begin{center}
 \begin{tabular}{| c | c | c|} 
 \hline
 Notation & Meaning & Place where this object is defined/remarks \\ [0.5ex] 
 
  \hline
 $\CA$ & Various patches &  \\
 \hline
 $\CA^{\BPhi(n)}$ & A patch of variables $\BPhi$ at step $n$ & Similar convention is used  \\ 
 & &  for other variables  \\ 
  \hline
$B_0$  & Constant from the SDC  &  \eqref{strong} \\ 
  \hline
 $B(a,r)$ & Ball in $l^2$-norm in $\R^d$ & {\newred  \eqref{newballs} and above}\\
  \hline
  $\CB$ & Various bad sets &  \\
  \hline
 $\CB^{(n)}$ & Various bad sets obtained at step $n$ &  \\
 \hline
 $\CB^{\BPhi}$ & Bad set of variables $\BPhi$ & Similar convention is used  \\ 
 & &  for other variables  \\ 
 \hline
  $\CC$,$\check\CC$,$\tilde\CC$& Cluster, (multiscale cube of level $0$)&\eqref{7CC}, \eqref{7CC1}\\
  \hline
 $D$ & $\{z\in\C:\ |z|<1\}$ &   \\ 
  \hline

 $E_*$ & A large number starting from which & \eqref{E_*}, \eqref{E_*1}  \\ 
 
  & our constructions work &   \\ 
   \hline
 $\CG$ & Various good set &  \\
  \hline
 $\CG^{(n)}$ & Various good sets obtained at step $n$ &  \\
 \hline
 $\CG^{\BPhi}$ & Good set of variables $\BPhi$ & Similar convention is used  \\ 
 & &  for other variables  \\ 
 \hline
$\gamma$  &  &  Lemma \ref{MGL} \\  
 \hline
 $\gamma_0$ &  & \eqref{gamma_0}  \\  
 \hline
{\newred $\gamma_m^{(j)}$} &  & {\newred\eqref{Nov8} } \\ 
\hline

 Notation & Meaning & Place where this object is defined/remarks \\ [0.5ex]

 \hline
 $H(\Lambda,\bk)$  & $\CP(\Lambda,\bk)H\CP(\Lambda,\bk)$  &  Definition \ref{proj} \\ 
 \hline
 $\GH(\bk)$ & The fibre generated by $\bk$ & \eqref{7fibre}  \\
 \hline
 $\GH(\Lambda,\bk)$ & The subspace of $\GH(\bk)$ & Definition \ref{proj}  \\
  & spanned by elements of $\Lambda$ &   \\ 
  \hline
    \end{tabular}
\end{center}

\begin{center}
 \begin{tabular}{| c | c | c|} 

 \hline
 $K^{b(j)}$ & Base cube & Definition \ref{8.1}, Condition 1 \\
  \hline
 $K^{b(j),{\mathrm{small}}}$ & Small base cube & Definition \ref{8.1}, Condition 2 \\
 \hline
 $K^{(j)}$ & Multiscale cube & Definition \ref{8.1}, Condition 3 \\
  \hline
 $K^{(j),{\mathrm{small}}}$ & Small multiscale  cube & Definition \ref{8.1}, Condition 3 \\
  \hline
 $\tilde K^{b(j)}$ & Enlarged base cube & Definition \ref{8.2}, Condition 1 \\
   \hline
 $\tilde K^{(j)}$ & Enlarged multiscale cube & Definition \ref{8.2}, Condition 3 \\
 \hline
 $\hat K^{(j)}$ & Central cube of order $j$ & Definition \ref{def7.2}   \\
 \hline 
 {\newred $L$-good (for $\bxi$)} &{\newred $\bxi\in\BS\BL(\sqrt{\rho^2-L},\sqrt{\rho^2+L})$ } &{\newred  Definition \ref{newdef1}} \\ 
 \hline
{\newred $L$-good (for $\bn$)} &{\newred $\bk+\bn\vec\boldom\in\BS\BL(\sqrt{\rho^2-L},\sqrt{\rho^2+L})$}  &{\newred  Definition \ref{newdef1}} \\ 
 \hline
$\CM$  & Matryoshka of patches or central cubes & Definitions \ref{def7.1xi}, \ref{def7.2}   \\ 
 \hline
 $\mu$ &   & Lemma \ref{angles}  \\ 
  \hline
$\hat\mu$  & Exponent from the SDC &  \eqref{strong} \\ 
 \hline
 $\Omega(a,r)$ & Ball in $l^{\infty}$-norm in $\Z^l$ &  {\newred \eqref{newballs} and above }\\
 \hline
 
{\newred $\Omega'(r)$ }& {\newred $\Om(0;r)\setminus\{0\}$}&  
{\newred  \eqref{newballs}} \\
 \hline
$\CP(\Lambda,\bk)$  & Projection onto $\GH(\Lambda,\bk)$  &  Definition \ref{proj} \\ 
 \hline
$\Pi$  & $(-\tilde\phi,\tilde\phi)^{d-1}$ & \eqref{coordinates3}  \\ 
 \hline
{\newred $Q$}  &{\newred $\max\{|\bn|, V_{\bn}\ne 0\}$}  & {\newred\eqref{V_q=0} } \\ 
 \hline
$r_{n,j}$  &  & \eqref{new5.1}, \eqref{rn3}  \\ 
 \hline
$r_{n}'(\gamma)$, $r_n'$  &  & \eqref{r1'},  \eqref{rn'}  \\ 
 \hline
$\tilde r_{n}'(\gamma)$, $\tilde r_{n}'$ &  &  \eqref{tildern'}  \\ 
 \hline
$\sigma_{p,s,q}$ &  &  \eqref{7sigma} and the text above it.  \\ 
\hline
$\BUps^{\Z}_p(\bxi)$& Primitive pre-cluster &Definition \ref{7BUps}\\
\hline
$\check\BUps^{\Z}_p(\bxi)$& Extended pre-cluster &\eqref{ext}\\
\hline
$\hat\BUps^{\Z}_p(\bxi)$& Intermediate pre-cluster &\eqref{haterext}\\
\hline
$\tilde\BUps^{\Z}_p(\bxi)$& Super-extended pre-cluster &\eqref{superext}\\
  \hline
 $Z_0$ &  & \eqref{Z_0}  \\ 

 \hline
\end{tabular}
\end{center}

\end{document}